\newif\ifanonym

\documentclass[11pt]{article}

\usepackage[english]{babel}

\usepackage[utf8]{inputenc}
\usepackage[margin=1in]{geometry}
\usepackage{algorithm}
\usepackage{algpseudocodex}
\algrenewcommand\algorithmicrequire{\textbf{Input:}}
\algrenewcommand\algorithmicensure{\textbf{Output:}}

\usepackage{amsmath,amsthm,amssymb,bbm}
\usepackage{thmtools}
\allowdisplaybreaks
\usepackage{xspace}
\usepackage{graphicx}
\usepackage{todonotes}
\usepackage[colorlinks=true, allcolors=blue]{hyperref}
\usepackage[capitalize,noabbrev]{cleveref}
\usepackage{enumitem}

\newcommand{\E}{\mathbb{E}}
\newcommand{\Z}{\mathbb{Z}}
\newcommand{\R}{\mathbb{R}}

\newcommand{\OPT}{\mathrm{OPT}}
\newcommand{\dist}{\mathrm{dist}}
\newcommand{\cost}{\textsc{cost}}
\newcommand{\poly}{\mathrm{poly}}
\newcommand{\diam}{\mathrm{diam}}
\newcommand{\polylog}{\mathrm{polylog}}

\newcommand{\ALG}{\mathrm{\textsc{ALG}}}
\newcommand{\Sparsifier}{\mathrm{\textsc{Sparsifier}}}

\newcommand{\calA}{\mathcal{A}}
\newcommand{\calB}{\mathcal{B}}

\newcommand{\calD}{\mathcal{D}}
\newcommand{\calE}{\mathcal{E}}
\newcommand{\calF}{\mathcal{F}}
\newcommand{\calG}{\mathcal{G}}

\newcommand{\calW}{\mathcal{W}}

\newcommand{\calY}{\mathcal{Y}}

\theoremstyle{plain}
\newtheorem{theorem}{Theorem}[section]
\newtheorem{lemma}[theorem]{Lemma}
\newtheorem{claim}[theorem]{Claim}

\newtheorem{corollary}[theorem]{Corollary}

\newtheorem{invariant}{Invariant}

\theoremstyle{definition}
\newtheorem{definition}[theorem]{Definition}

\theoremstyle{remark}
\newtheorem{remark}[theorem]{Remark}

\crefname{claim}{Claim}{Claims}

\newtheoremstyle{restate}{}{}{\itshape}{}{\bfseries}{~(Restated).}{.5em}{\thmnote{#3}}
\theoremstyle{restate}
\newtheorem*{restate}{}

\newcommand{\init}{\textnormal{\texttt{init}}}
\newcommand{\final}{\textnormal{\texttt{final}}}
\newcommand{\nn}[0]{\textsc{NN}}

\newcommand{\ball}[0]{\textsc{Ball}}

\newcommand{\MakeRbst}[0]{\texttt{MakeRobust}}
\newcommand{\Robustify}[0]{\texttt{Robustify}}

\newcommand{\old}[0]{\text{old}}
\newcommand{\new}[0]{\text{new}}
\DeclareMathOperator{\set}{bk}

\newcommand{\Deld}[0]{[\Delta]^{d}}

\newcommand{\polyup}[0]{p_\epsilon}
\newcommand{\neps}[0]{\tilde{O}(n^\epsilon)}

\newcommand{\Unif}{\operatorname{Unif}}
\newcommand{\eqdef}{:=}

\newcommand{\argmin}{\mathrm{argmin}}
\renewcommand{\Tilde}{\widetilde}
\renewcommand{\hat}{\widehat}
\renewcommand{\Hat}{\widehat}
\newcommand{\tO}{\tilde{O}}

\newcommand{\ignore}[1]{}

\title{Fully Dynamic Euclidean $k$-Means}
\ifanonym
\author{Anonymous Authors}
\else
\author{Sayan Bhattacharya\thanks{
    Email: \texttt{S.Bhattacharya@warwick.ac.uk}
} \\
University of Warwick
\and Mart\'{i}n Costa\thanks{
    Supported by a Google PhD Fellowship.
    Email: \texttt{Martin.Costa@warwick.ac.uk}
} \\
University of Warwick
\and Ermiya Farokhnejad\thanks{
    Email : \texttt{Ermiya.Farokhnejad@warwick.ac.uk}
} \\
University of Warwick
\and Shaofeng H.-C. Jiang\thanks{
    Email: \texttt{shaofeng.jiang@pku.edu.cn}
} \\
Peking University
\and Yaonan Jin\thanks{
    Email: \texttt{jinyaonan@huawei.com}
} \\
Huawei TCS Lab
\and Jianing Lou\thanks{
    Email: \texttt{loujn@pku.edu.cn}
} \\
Peking University
}
\fi

\date{}

\begin{document}
\maketitle

\begin{abstract}

We consider the Euclidean $k$-means clustering problem in a {\em dynamic} setting, where
we have to explicitly maintain a solution (a set of $k$ centers) $S \subseteq \mathbb{R}^d$ subject to
point insertions/deletions in $\mathbb{R}^d$.
We present a dynamic algorithm for Euclidean $k$-means with $\mathrm{poly}(1/\epsilon)$-approximation ratio, $\tilde{O}(k^{\epsilon})$ update time, and $\tilde{O}(1)$ recourse, for any $\epsilon \in (0,1)$, even when $d$ and $k$ are both part of the input.
This is the first algorithm to achieve a constant ratio with $o(k)$ update time for this problem,
whereas the previous $O(1)$-approximation runs in $\tO(k)$ update time~[Bhattacharya, Costa, Farokhnejad; STOC'25].

In fact, previous algorithms cannot go beyond $O(k)$ update time precisely because they are designed for general metrics where an $\Omega(k)$ lower bound is known.
We break this $O(k)$ barrier by devising new fundamental data structures to utilize Euclidean properties:
a structure that (implicitly) maintains a clustering subject to both center and data point updates,
and a range query structure that can evaluate a mergeable function over any metric ball range given as a query.
To obtain these structures, we devise the first {\em consistent hashing scheme} [Czumaj, Jiang, Krauthgamer, Vesel{\'{y}}, Yang; FOCS'22] that achieves $\tilde O(n^{\epsilon})$ running time per point evaluation with competitive parameters.

Our final algorithm exploits the framework of [Bhattacharya, Costa, Farokhnejad; STOC'25] for general metrics.
The key change is to redesign several critical subroutines so that they reduce to our new Euclidean data structures, replacing the general-metric implementations that are unlikely to run efficiently even when Euclidean properties are provided.

\end{abstract}

\newpage

\pagenumbering{gobble}

\tableofcontents

\newpage

\pagenumbering{arabic}

\part{Extended Abstract}

\section{Introduction}

Euclidean $k$-means clustering, popularized by the celebrated Lloyd's algorithm~\cite{Lloyd82} in the 1950s, has long been one of the most fundamental problems in data analysis and machine learning, and its algorithmic study has been a central topic in theoretical computer science.
In $k$-means,
given as input a set of $n$ points $X \subseteq \R^d$ and an integer $k\geq 1$,
the goal is to find a $k$-point center set $S \subseteq \R^d$ that minimizes the cost
\begin{equation*}
    \cost(X, S) := \sum_{x \in X} \dist^2(x, S),
\end{equation*}
where $\dist(x, S) := \min_{s \in S} \dist(x, s)$ and $\dist(x, s) := \|x - s\|_2$.
A related problem, $k$-median, is the same as $k$-means except that the cost function sums over $\dist(x, S)$ without the square.
We consider the general high-dimensional setting where $d$ may be arbitrary, and in the remainder of this section, we assume $d = O(\log n)$ without loss of generality via a standard Johnson-Lindenstrauss (JL) transform~\cite{JL84}.

We study $k$-means (and the related $k$-median) in a dynamic setting,
where the algorithm is asked to maintain a solution of $k$ centers,
subject to point insertions and deletions in $X$.
While there has recently been a flourishing study of dynamic clustering~\cite{LattanziV17,ChanGS18,GoranciHL18,Cohen-AddadHPSS19,HenzingerK20,FichtenbergerLN21,GoranciHLSS21,BhattacharyaLP22,nips/BhattacharyaCLP23,BateniEFHJMW23,BhattacharyaCGLP24,BhattacharyaGJQ24,CrucianiFGNS24,esa/TourHS24,LackiHGJR24,BCF24,icml25,ForsterS25,GJK+26},
dynamic algorithms for \emph{Euclidean} $k$-means are less understood.
On the one hand, in \emph{general} metrics, a nearly tight $\tO(k)$\footnote{
    We use $\tilde{O}(\cdot)$ notation to hide polylogarithmic factors in $n$ and the aspect ratio, as well as polynomial factors in the dimension $d$ in Euclidean space. \label{footnote:tilde}
}
update time for $O(1)$-approximation has recently been obtained~\cite{BCF24},
and this $O(k)$ cannot be improved without utilizing Euclidean structure,
since in general metrics any $O(1)$-approximate static algorithm already requires $\Omega(nk)$ time~\cite{BadoiuCIS05}.
On the other hand, in Euclidean space, a static $O(1)$-approximation can be computed in $\tilde{O}(n + k^{1+\epsilon})$ time ($0 < \epsilon < 1$)~\cite{laTourS24,abs-2504-03513}, suggesting that obtaining $o(k)$ update time may be possible.
However, an $o(k)$ update time has only been achieved for an $O(\log n)$-approximation for Euclidean $k$-median~\cite{GJK+26}.
Therefore, breaking the $O(k)$ update time barrier while still achieving $O(1)$-approximation for Euclidean $k$-means remains a research gap.

We make significant progress by obtaining the first fully dynamic algorithm for Euclidean $k$-means with sublinear $\tilde{O}(k^{\epsilon})$ amortized update time for any $\epsilon \in (0,1)$, breaking the $O(k)$ barrier.
Our algorithm also addresses the \emph{recourse} (the number of changes to the solution per update), achieving a $\tilde{O}(1)$ bound that nearly matches the state-of-the-art for general metrics~\cite{BCF24}. 
We summarize our algorithm in the theorem below.

\begin{theorem}
\label{th:main}
    For any $\epsilon \in (0, 1)$, there exists a randomized fully dynamic algorithm for Euclidean $k$-means that achieves $\poly(1/\epsilon)$-approximation, $\tilde{O}(k^\epsilon)$ amortized update time and $\tilde{O}(1)$ amortized recourse, with high probability.\footnote{In this Theorem statement, the $\tilde{O}(\cdot)$ notation hides $\poly(1/\epsilon)$ factors as well as those mentioned in Footnote \ref{footnote:tilde}.}
\end{theorem}

Overall, it is unlikely that one can improve our result by a significant margin.
Specifically, our result immediately implies a static $\poly(1/\epsilon)$-approximation algorithm running in $\tilde{O}(nk^\epsilon)$ time,
and this further leads to $\tilde{O}(n + k^{1+\epsilon})$ running time via coresets~\cite{DraganovSS24}.
This matches the state-of-the-art static algorithms~\cite{laTourS24,abs-2504-03513} up to $\poly(1/\epsilon)$ factors in the approximation ratio.
More importantly, as noted in~\cite{abs-2504-03513}, such a ratio-time tradeoff is believed to be essentially the best possible for static $k$-means algorithms in the worst case, i.e., when $k = \Omega(n)$,
which in turn suggests that substantially improving our result to, for example, polylogarithmic update time while maintaining an $O(1)$-approximation, would likely encounter a fundamental barrier.

Our algorithm also extends to the related $k$-median problem, while retaining a $\poly(1/\epsilon)$-approximation and $\tilde{O}(k^{\epsilon})$ update time.
For $k$-median,
a recent result offers an alternative tradeoff: 
an $O(\poly(1/\epsilon)\cdot \log n)$-approximation with the same $\tilde{O}(k^{\epsilon})$ update time~\cite{GJK+26}.
Compared with this, when $\epsilon$ is a constant, our result achieves a constant approximation, which improves upon their $O(\log n)$ ratio, while for $\epsilon = 1 / \log n$, both results achieve a $\polylog(n)$ approximation, although theirs has a better polynomial degree.

Finally, while our main result works only for an \textit{oblivious adversary},
it is nonetheless obtained via an intermediate claim that is robust to an \emph{adaptive adversary}  
albeit with an update time of
$2^{O(\epsilon d)} \cdot \tilde O(n^\epsilon)$ (see \Cref{remark:robust-adaptive}).
Our final $\tilde{O}(k^{\epsilon})$ update time is obtained 
by applying a JL transform, which reduces the dimension $d$ to $O(\log n)$,
as well as using a dynamic coreset algorithm \cite{esa/TourHS24} to sparsify the input to size $\tilde O(k)$.
Unfortunately, even this very basic JL transform (or any linear sketch) cannot be robust, as shown in~\cite{HardtW13}, 
and we do not know whether \cite{esa/TourHS24} can be replaced with a robust version,
which forces our robust algorithm to remain stuck with this $2^{O(\epsilon d)} \cdot \tilde O(n^\epsilon)$ form.
In general, this highlights a major challenge in devising robust dynamic algorithms in high-dimensional Euclidean spaces, where very basic tools lack a robust counterpart.

\subsection{Technical Contributions}

As we mentioned, the key technical challenge in breaking the $O(k)$ barrier is to utilize the Euclidean structure. 
However, dynamic data structures specifically designed for clustering in high dimensions are lacking.
Along the line of dynamic clustering algorithms that exploit Euclidean structure,
it is recently known that a new dynamic tree embedding method can be employed to obtain a dynamic $k$-median algorithm~\cite{GJK+26}, but the key drawback is an $\Omega(\log n)$ distortion lower bound, making it unsuitable for $O(1)$-approximation.
Another notable line of research leverages \emph{coreset} techniques~\cite{DBLP:conf/stoc/Har-PeledM04,HenzingerK20,esa/TourHS24}, and the state-of-the-art can dynamically maintain a coreset of size $\tilde{O}(k)$ with a very fast polylogarithmic update time~\cite{esa/TourHS24}. 
Unfortunately, coreset technique is unlikely to be useful for breaking the $O(k)$ update time barrier for dynamic $k$-means: even if a coreset is maintained efficiently, one still needs to run a static algorithm on the coreset to obtain a solution after each update, and the $\Omega(k)$ lower bound on coreset size~\cite{HuangV20,DBLP:conf/stoc/Cohen-AddadLSS22,Huang0024} therefore limits the overall update time to $\Omega(k)$.

\paragraph{New Fundamental Structures for Dynamic Euclidean Clustering.}
Our main technical contribution is the development of fundamental dynamic data structures for clustering in high dimensions.
The most notable one addresses a very basic task that appears in almost all clustering algorithms: for a dynamic point set $X$ and a dynamic center set $S$, maintain a clustering of $X$ with respect to $S$, which is an assignment of each point in $X$ to an approximately nearest center in $S$.
This clustering is generally impossible to maintain explicitly and efficiently, since updates to $S$ may repeatedly insert and delete a single center, causing $\Omega(n)$ reassignments per update.
Nonetheless, we show that it is possible to implicitly maintain the clustering such that we may obtain relevant statistics, for instance, the cluster sizes and the sum of squared distances.
In fact, an interesting corollary of this is a dynamic algorithm for implementing the $D^2$-sampling process, which is the technical core of the seminal $k$-means++ algorithm~\cite{ArthurV07}.
This $D^2$-sampling process samples each data point with probability proportional to the squared Euclidean distance to the current center set.

Another important feature of our data structures is that they work even against an adaptive adversary.
This is required since the final algorithm, which utilizes these structures, may invoke them in a manner that depends on their past outputs.
Still taking the $D^2$-sampling as an example, the standard use of it is to draw a sample $x$ with respect to the current center set $S$, then update $S \gets S \cup \{x\}$ and draw another sample with respect to this new $S$.
This necessarily makes the dataset for the second sampling/query depend on the result of the first query.
Such a robustness requirement makes the design of data structures more challenging, as it renders many widely used \emph{space partitioning} techniques inapplicable, including locality-sensitive hashing (LSH) and tree embeddings, whose performance relies crucially on randomness.
Moreover, even without the robustness requirement, it is already nontrivial to devise space partitioning techniques to implement the desired clustering maintenance structure.

\paragraph{New Data Structures via New Consistent Hashing.}

We consider another space partitioning technique, \emph{consistent hashing}~\cite{CJKVY22}, which has recently been shown to be useful in various settings for clustering (and beyond)~\cite{CJKVY22,DBLP:conf/icalp/CzumajGJK024,CzumajG0J25,random-shift,JiangL25,GJK+26}.
Roughly speaking, consistent hashing induces a partition of the point set such that
(a) each part has a diameter bounded by $r$, and
(b) every ball of radius $\poly(\epsilon)\cdot r$ intersects at most $n^{\epsilon}$ parts.
Consistent hashing has a key advantage over other space decomposition techniques, such as tree embeddings and LSH:
its guarantee is deterministic rather than probabilistic, which makes it a viable approach for designing data structures that are robust against adaptive adversaries.

We show that consistent hashing can indeed be used to design the clustering maintenance structure as described,
but this requires an additional efficiency guarantee for \emph{ball evaluation}; specifically,
given any ball of radius $\poly(\epsilon)\cdot r$, we need to evaluate all hash values on that ball efficiently (note that there are at most $n^{\epsilon}$ distinct values).
However, existing consistent hashing schemes require $2^{d}$ time even to evaluate at a single point~\cite{CJKVY22}, and this running time becomes superlinear $\poly(n)$ in the typical case $d = O(\log n)$.
On the other hand, a recent work~\cite{random-shift} improves the evaluation time to $\poly(d)$, but only provides probabilistic guarantees, which are not robust to adaptive inputs.

In this paper, we give a new consistent hashing scheme that, without introducing any probabilistic guarantees, significantly improves the $2^{d} = \poly(n)$ evaluation time of~\cite{CJKVY22} to sublinear $\tO(n^{\epsilon})$, even for evaluating over balls, which may be of independent interest.
As mentioned, using our new consistent hashing, we are able to design a clustering maintenance structure that has efficient $\tO(n^{\epsilon})$ update time and is robust against adaptive adversaries.
In addition, our new hashing also yields a \emph{range query} structure that,
given query $x \in \R^d$ and $r \ge 0$,
approximates certain statistics of the metric ball $\ball_X(x,r):=\{y\in X: \dist(x,y)\le r\}$.
This is a fundamental task with applications such as computing the size of the metric ball and approximate nearest neighbor search, and it also plays a key role in our dynamic $k$-means algorithm.
See \Cref{ex-ab:sec:data-structures-main} for more details on our consistent hashing scheme and data structures.

\paragraph{Final Clustering Algorithm and Two Key Subproblems.}
Our final clustering algorithm is an augmentation of a recent framework~\cite{BCF24} which was designed for general metrics.
This particularly requires us to identify two key dynamic subproblems in~\cite{BCF24}, called \emph{restricted} $k$-means and \emph{augmented} $k$-means, and give completely new implementations via our new data structures to achieve sublinear update time; this is necessary because their original implementation was designed for general metrics without utilizing Euclidean structures.
Since both subproblems are natural generalizations of the standard $k$-means problem, 
particularly restricted $k$-means, which has also appeared as a subproblem in several notable works (e.g.,~\cite{Thorup04,ChrobakKY06,LinNRW10,LiS16,Cohen-Addad0LSS25,CharikarCGGLW25}),
we believe our new Euclidean implementation may find applications in future studies.

In the following, we introduce both subproblems,
discuss the conceptual barriers and how our new algorithms work.

\paragraph{Restricted $k$-Means.}
In the restricted $k$-means problem, the goal is to remove a given number $r\ge 1$ of centers from the solution set $S$ to minimize the cost among all possible deletions of $r$ centers out of $S$.
The algorithm for this subproblem in~\cite{BCF24} uses randomized local search,
which appears to require explicit maintenance of the clustering, and hence an $\Omega(kr)$ time bound seems difficult to circumvent.
Here, we take a completely different \emph{coreset-based} approach.
In this approach, we define a new notion of coreset, whose definition is different from standard coresets for clustering in that our coreset preserves the cost with respect to every \emph{deleted} set of size $r$, instead of directly working with the center set of size $k - r$.
We show that there is such a coreset of size $O(r)$, and that solving (approximately) the restricted $k$-means on this coreset suffices for restricted $k$-means on the solution $S$.
This coreset algorithm also excels in leveraging the Euclidean structure.
Specifically, we give a coreset construction via standard approximate nearest neighbor search, which is supported by our data structures and hence is directly made dynamic.
The final restricted $k$-means is obtained by applying an efficient algorithm for Euclidean $k$-means~\cite{laTourS24,abs-2504-03513} on this coreset, and the total query time (including the construction) is $\tO(r^{1+\epsilon})$, which significantly improves upon the $O(kr)$ bound in the general metrics case, and we show that this suffices for our eventual $\tO(k^\epsilon)$ update time.
See \Cref{sec:ex-ab:restricted} for a more complete explanation of this part.

\paragraph{Augmented $k$-Means.}
In the augmented $k$-means subproblem, the goal is to add $r$ more centers to the solution set $S$ to minimize the cost among all possible additions of $r$ centers to $S$, for a given $r > 0$.
The algorithm proposed in~\cite{BCF24} contracts the centers in $S$ and introduces a new metric space such that any $(r+1)$-means solution in the new space corresponds to a solution for the augmented $k$-means problem.
It is obvious that the new metric space is far from being Euclidean, and we cannot use our Euclidean data structures to solve the $(r+1)$-means problem in the new metric space.
Hence, we need a new algorithm to solve augmented $k$-means.
Our key observation here is that the $k$-means++ algorithm~\cite{ArthurV07}, which iteratively performs $D^2$-sampling to add centers, can in fact be used to yield a good approximation for augmented $k$-means.
Building on the data structure for implicitly maintaining the clustering discussed above, we design a new structure to implement a single (approximate) $D^2$-sampling step in sublinear $\tO(n^{\epsilon})$ time, which is thus sufficient to solve the augmented $k$-means problem (approximately).
See \Cref{sec:ex-ab:augmented-k-means-explain} for more details.

\paragraph{Structural Difficulty Caused by Approximate Distances.}

Finally, we discuss a subtle yet significant issue caused by the nontrivial notion of approximation introduced by our new data structures, which is harder to handle than the ordinary multiplicative approximation:
for instance, when we use our approximate range query structure to estimate the size of a metric ball $\ball_X(x,r)$, its return value may include points within a radius $O(r)$ that are outside the ball.
In fact, there is essentially no guarantee of how the approximate metric ball looks, and if one considers the distorted distance, then it may not satisfy the triangle inequality either.
This largely breaks the nice geometry of balls, and particularly forces us to modify the definition of a central notion called \emph{robust center} to deal with this approximation, whose definition precisely depends on the metric ball.
In turn, this requires us to consider the approximate metric ball throughout the analysis, which unfortunately leads to a rebuild of the framework of~\cite{BCF24} almost from scratch.

\subsection{Related Work}
Beyond $k$-median and $k$-means,
dynamic algorithms have also been well studied for two closely related clustering problems, $k$-center~\cite{ChanGS18,GoranciHLSS21,BateniEFHJMW23,CrucianiFGNS24,LackiHGJR24,icml25,ForsterS25} and facility location~\cite{GoranciHL18,Cohen-AddadHPSS19,BhattacharyaLP22,BhattacharyaGJQ24}.
In general metrics, for $k$-center, there is an algorithm that simultaneously achieves an $O(1)$-approximation ratio, $\tO(1)$ recourse, and $\tO(k)$ update time~\cite{icml25};
for uniform facility location, it is known how to achieve an $O(1)$-approximation ratio, $O(1)$ recourse, and $O(n\log n)$ update time~\cite{Cohen-AddadHPSS19}, and similar results for the non-uniform version have also been obtained~\cite{BhattacharyaLP22}.
In Euclidean space, both problems admit a $\poly(1/\epsilon)$-approximation with $\tO(n^{\epsilon})$ update time; see~\cite{BateniEFHJMW23} for $k$-center and~\cite{BhattacharyaGJQ24} for facility location.

Dynamic clustering algorithms have also been studied in other metric spaces, such as shortest-path metrics on dynamic graphs~\cite{CrucianiFGNS24,abs-2602-08542} and metrics with bounded doubling dimension~\cite{GoranciHL18,GoranciHLSS21}.

\section{Notations and Preliminaries}

\paragraph{Parameters and Assumptions.}
We fix the following parameters in the paper: $\epsilon$ is a sufficiently small constant, $n$ is the number of points in the data-set, $d$ is the dimension of the space, and $\Delta$ is the aspect ratio, i.e., the maximum distance divided by the smallest distance between distinct points.
As a typical setup, we assume $\Delta \leq \poly(n)$ is an integer.
Without loss of generality, the input points that we consider
are from the discrete grid $[\Delta]^d := \{ (x_1,\ldots,x_d) \mid 1 \leq x_i \leq \Delta \ \forall 1 \leq i \leq d \}$, by a standard rounding and rescaling.
Also without loss of generality, we assume $d = O(\log n)$ by a standard Johnson-Lindenstrauss transform~\cite{JL84} (see the first paragraph in \Cref{ex-ab:preprocessing-transformation} for a detailed discussion).

\medskip
\noindent
\textbf{Notations.}
We now introduce some key notations that will be used in subsequent sections. 
For integer $n \geq 1$, denote $[n] := \{1, \ldots, n\}$.
Consider a generic function $f : U \to V$; for $T \subseteq U$ we write $f(T) := \{ f(u) : u \in T \}$, and for $v \in V$ we write $f^{-1}(v) := \{u \in U \mid f(u) = v\}$.
For every two sets $A$ and $B$, we use $A-B$ instead of $A\setminus B$, and whenever $B = \{b\}$ is a singleton, we write $A - b$ instead of $A - \{b\}$ for simplicity. 
For $x \in [\Delta]^d, r \geq 0$, let $\ball(x, r) := \{ y \in [\Delta]^d : \dist(x, y) \leq r \}$ be the metric ball.
For $T \subseteq [\Delta]^d$, let $\diam(T) := \max_{x, y \in T} \dist(x, y) $ be the diameter of $T$.
For every weighted dataset $X \subseteq [\Delta]^d$ and center set $S \subseteq [\Delta]^d$, the cost of $S$ on $X$ is denoted by $\cost(X,S)$ and is defined as $\cost(X,S) := \sum_{x \in X} w(x) \cdot \dist(x,S)^2$.
Accordingly, we define $\textsc{AverageCost}(X, S) := \cost(X,S)/w(X)$.
When the size of $S$ is one, i.e.~$S = \{c\}$, we abbreviate $\cost(X,\{c\})$ by $\cost(X,c)$ and $\textsc{AverageCost}(X, \{c\})$ by $\textsc{AverageCost}(X, c)$.
We denote by $\OPT_k(X)$ the cost of the optimum solution of size $k$.
For every $W \subseteq [\Delta]^d$, we also denote by $\OPT_k^W(X)$ the cost of the optimum solution of size $k$ which is restricted to open centers only in $W$, i.e.~$\OPT_k^W(X) := \min_{S \subseteq W, |S|=k} \cost(X,S)$.

\medskip
\noindent
\textbf{Organization.}
In \Cref{ex-ab:sec:data-structures-main}, we introduce our new geometric data structures.
In \Cref{sec:ex-ab:restriced-augmented}, we describe our new algorithms for two important subroutines: restricted and augmented $k$-means.
In \Cref{ex-ab:sec:implementation}, we provide an overview of the algorithm of~\cite{BCF24} and describe our modifications to their algorithm.
Finally, in \Cref{ex-ab:sec:alg-analysis}, we provide the sketch of the analysis of our algorithm.
We refer the reader to \Cref{part:data-structures} for a thorough description and analysis of our geometric data structures and new algorithms for the subroutines, and to \Cref{part:full-alg} for a complete description and analysis of all parts of our algorithm.

\section{Data Structures in High-Dimensional Euclidean Spaces}
\label{ex-ab:sec:data-structures-main}

In this section, we present two new data structures for high-dimensional Euclidean spaces: the clustering maintenance structure (\Cref{sec:ex-ab:approx_assign}) and the range query structure (\Cref{sec:ex-ab:range_query}).
Both data structures are built on our new consistent hashing scheme, which we explain in \Cref{sec:ex-ab:consistent-hashing}.
These data structures are designed to be robust against adaptive adversaries, which may be of independent interest.

There are some parameters we use throughout this section, whose final values in our algorithm are set as follows.
\begin{equation}\label{eq:consistent-hashing-parameter}
    \Gamma = \epsilon^{-3/2}, \quad \Lambda = 2^{O(\epsilon d)}\cdot \poly(d \log \Delta) = \tO(n^{O(\epsilon)}).
\end{equation}
In this section, when we say ``with high probability'' (w.h.p.), we mean the algorithm satisfies all stated guarantees simultaneously with probability $1 - 1 / \poly(\Delta^d)$.

\subsection{Efficient Consistent Hashing}
\label{sec:ex-ab:consistent-hashing}

We give in \Cref{def:intro-consistent} the notion of consistent hashing,
which was introduced in~\cite{CJKVY22}
and was also studied under a closely related notion called sparse partitions~\cite{JiaLNRS05,Filtser24}.
Roughly speaking, a consistent hashing is a data-oblivious map on $[\Delta]^d$ that partitions the space $[\Delta]^d$ into hash buckets,
such that each bucket has some bounded diameter $\rho$,
and every small ball of diameter $\frac{\rho}{\Gamma}$ can intersect only $\Lambda$ parts.

\begin{definition}[Consistent Hashing]
\label{def:intro-consistent}
A map $\varphi$ on $\Deld$ is a $(\Gamma, \Lambda, \rho)$-\emph{hash}
if 
\begin{enumerate}[font = \bfseries]
    \item\label{def:intro-consistent:closeness}
    (diameter) $\forall z\in \varphi(\Deld), \ \diam(\varphi^{-1}(z)) \le \rho$.

    \item\label{def:intro-consistent:consistency}
    (consistency) $\forall x\in \Deld, \ |\varphi(\ball(x, \frac{\rho}{\Gamma}))| \le \Lambda$.
\end{enumerate}
\end{definition}
We also need the hash function to evaluate efficiently.
In particular, given a point $x \in [\Delta]^d$,
we need to evaluate $\varphi(x)$ 
and $\varphi(\ball(x, \rho / \Gamma))$ (allowing approximation) in $\tO(\Lambda)$ time 
(noticing that $|\varphi(\ball(x, \rho / \Gamma))| \leq \Lambda$ by the consistency guarantee), as follows.

\begin{definition}[Efficient consistent hashing]
\label{ex-ab:def:efficient-consistent}
    We call a $(\Gamma, \Lambda, \rho)$-hash $\varphi$ efficient, if for every $x \in [\Delta]^d$, (a)
    $\varphi(x)$ can be evaluated in $\tO(\Lambda)$ time
    and
    (b) a set $\Phi$ can be computed in $\tO(\Lambda)$ time such that
    $\varphi(\ball(x, \rho / \Gamma)) \subseteq \Phi \cap \varphi([\Delta]^d) \subseteq \varphi(\ball(x, 2\rho))$.
\end{definition}

We devise efficient consistent hashing with parameter $\Lambda = 2^{O(d/\Gamma^{2/3})} \cdot \poly(d)$.
This tradeoff is close to the optimal existential bound of $\Lambda = 2^{\Omega(d / \Gamma)}$ from~\cite{Filtser24} (which runs in $\exp(d)$ time).
Previously, efficient construction of consistent hashing was only known from~\cite{random-shift},
but their consistency guarantee is in expectation, which is too weak to satisfy \Cref{def:intro-consistent}.
While such a weaker guarantee may be useful in many applications,
it is nonetheless insufficient for our application, especially for devising structures robust to an adaptive adversary.

\begin{lemma}
\label{lem:intro-consistent}
For any $\Gamma \ge 2\sqrt{2\pi}$, $\rho > 0$ and $\Lambda = 2^{{\Theta(d/\Gamma^{2/3})}}\cdot \poly(d \log \Delta)$,
there exists a random hash $\varphi$ on $[\Delta]^d$, such that
it can be sampled in $\poly(d \log \Delta)$ time,
and with probability at least $1 - 1 / \poly(\Delta^d)$,
$\varphi$ is an efficient $(\Gamma, \Lambda, \rho)$-hash.

\end{lemma}

\noindent
\textbf{Construction.}
We sketch the proof of \Cref{lem:intro-consistent}.
Our starting point is the above-mentioned weaker construction from~\cite{random-shift} which only achieves an expected consistency guarantee.
We aim to boost this expected guarantee to a high-probability one with only $\tilde O(1)$ loss on the consistency bound.
Specifically,
let $\tilde \varphi$ be a random hashing on $\mathbb{R}^d$,
the expected consistency guarantee of $\tilde \varphi$ is 
\begin{align*}
    \forall x\in [\Delta]^d, \quad \E\big[|\tilde \varphi(\ball(x, {\rho}/{\Gamma}))| \big]  ~\le~ \Lambda.
\end{align*}
Our construction is as follows.
We first draw $m = O(d\log\Delta)$ independent weaker hashing $\{\tilde \varphi_i\}_{i \in [m]}$ \cite{random-shift} with expected consistency guarantee. Then, with high probability, for all points $x \in [\Delta]^d$,
there exists $i \in [m]$ such that the strict consistency condition $|\tilde \varphi_i(\ball(x, \rho/\Gamma))| \le 2\Lambda$ holds.
We pick for each point $x$ the smallest such $i \in [m]$
and call this the \emph{color} of $x$, denoted as $c_x$.
Our hash $\varphi$ is defined as $\varphi(x) \eqdef (c_x, \tilde \varphi_{c_x}(x))$, $\forall x\in\Deld$.

\medskip
\noindent
\textbf{Consistency.} 
The diameter bound $\diam(\varphi^{-1}(z)) \le \rho$ carries over from $\tilde \varphi_i$'s.
We next show the consistency $|\varphi(\ball(x, \rho/(2\Gamma)))| \leq O(m \Lambda)$ for $x \in [\Delta]^d$.
Let $Y := \varphi(\ball(x, \rho/(2\Gamma)))$, and we need $|Y| \leq O(m\Lambda)$.
For $i \in [m]$, let $Y_i := \{ (c_x, \tilde \varphi_{c_x}(x)) \in Y : c_x = i \}$
be the subset of hash values coming from color $i$.
Then $|Y| \leq \sum_i |Y_i|$.
We show that $|Y_i| \leq O(\Lambda)$,
and this leads to $|Y| \leq O(m \Lambda) = \tO(\Lambda)$ as desired.

Fix some $i\in [m]$.
By our construction, an immediate bound is $|Y_i| \leq |\tilde \varphi_i(\ball(x, \rho / (2\Gamma)))|$,
so we are done if $|\tilde \varphi_i(\ball(x, \rho / (2\Gamma)))| \leq 2\Lambda$.
However, it could be that $ |\tilde \varphi_i(\ball(x, \rho / (2\Gamma)))| > 2\Lambda$.
We show that if this happens, then no point in $\ball(x, \rho / (2\Gamma))$ has color $i$,
which implies $|Y_i| = 0$.
This follows from the contrapositive: if some $y \in \ball(x, \rho/(2\Gamma))$ has color $c_y = i$,
then $\ball(x, \rho/(2\Gamma)) \subseteq \ball(y, \rho/\Gamma)$, and thus the consistency guarantee at $y$ implies $|\tilde\varphi_i(\ball(x, \rho/(2\Gamma)))| \le 2\Lambda$.

\medskip
\noindent
\textbf{Efficiency.}
Our new construction makes the hashing more complicated to evaluate,
particularly to determine the color $c_x$ for $x \in [\Delta]^d$.
By the definition of $c_x$, this reduces to checking whether $N_i := |\tilde \varphi_i(\ball(x, \rho / \Gamma))| \le 2\Lambda$ for each color $i$.
To do this efficiently, we leverage the special structure of $\tilde \varphi_i$ in~\cite{random-shift}:
each $\tilde \varphi_i$ is defined using a randomly shifted grid over $[\Delta]^d$, where each point is hashed based on the grid cell that it belongs to.
Under this grid structure, $\tilde \varphi_i(\ball(x, \rho / \Gamma))$ corresponds to a connected set of grid cells,
where each cell is adjacent to $O(d)$ others (i.e., differs in only one coordinate).
We can explore these cells using a depth-first search, 
and once the depth-first search takes more than $\tilde{O}(\Lambda)$ time, we can terminate the search and conclude that $N_i > 2\Lambda$, and no point around $x$ has color $i$.
This entire process takes total time $\tO(\Lambda) \cdot \poly(d)$.
The computation of $\varphi(\ball(x, \rho / \Gamma))$ follows from similar ideas.

\subsection{Clustering Maintenance Structure}
\label{sec:ex-ab:approx_assign}
One novel use of our consistent hashing is to devise a data structure for maintaining clustering.
This task asks to take two dynamic sets $X, S \subseteq [\Delta]^d$ as input,
and at any time maintain an explicit assignment $\sigma : X \to S$, such that each $x \in X$ is assigned to an approximately nearest neighbor in $S$.
However, there is no finite ratio with a small update time for this task, since one can keep on inserting/deleting a point of $S$ to change the assignment drastically.
To bypass this barrier, instead of maintaining the full assignment $\sigma : X \to S$, 
our idea is to hash the dataset into buckets $\varphi(X)$
and explicitly maintain only an intermediate assignment $\sigma'$ between bucket (ID) and $S$.
The final assignment $X \to S$ is \emph{implicitly} defined by assigning all points in the same bucket to the same bucket partner in the intermediate assignment.

We informally describe our data structure in \Cref{ex-ab:lem:approx_assignment},
and the actual structure we maintain needs to be more sophisticated than simply hash buckets of $\varphi$; see \Cref{sec:approximate-assignment} for more details.
In this lemma, we take a two-level hashing approach, namely, low-level and high-level hashes.
Denote by $L$ the ID set of low-level buckets and $H$ the ID set of high-level buckets.
The set of data points in each low-level bucket is explicitly maintained,
and we maintain a subset $H' \subseteq H$ of high-level bucket IDs,
each of which consists of a (dynamic) list of low-level bucket IDs (via map $f$).
The intermediate assignment $\sigma'$ is defined between $H'$ and $S$.
Importantly, one can only afford to maintain the list of low-level bucket \emph{ID} in a high-level bucket $h \in H$,
since it is possible that a small change to this ID list results in big changes in the \emph{point} set in $h$.

\begin{lemma}
    \label{ex-ab:lem:approx_assignment}
    There exist data-oblivious sets $L, H$, and an algorithm that takes dynamic sets $X, S \subseteq [\Delta]^d$ as input, and maintains\footnote{The size of $L, H$ may be $\Delta^d$, so full description of $\set(\ell)$ and $f : H \to 2^L$ is not efficient.
        Here, we assume both are described explicitly only on their support, i.e., $\set(\ell) \neq \emptyset$ and $f(h) \neq \emptyset$.
    }
    in $\tO(\Lambda^2)$ time per update against an adaptive adversary w.h.p.:
    \begin{itemize}
        \item (low-level bucket) $\set(\ell) \subseteq X$ for every $\ell \in L$;
        \item (high-level bucket) map $f : H \to 2^L$,
such that $\forall h \in H, \ell_1 \neq \ell_2 \in f(h)$, $\set(\ell_1) \cap \set(\ell_2) = \emptyset$ and $\forall h_1 \neq h_2 \in H$, 
        $f(h_1) \cap f(h_2) =\emptyset$;
\item (partition) $H' \subseteq H$ such that for every $x \in X$, there is a unique $H'(x) \in H'$ with $x \in \set(H'(x))$, where $\set(h) := \bigcup_{\ell \in f(h)} \set(\ell)$ for every $h\in H$;
        \item (assignment) $\sigma' : H' \to S$.
\end{itemize}
    Let $\sigma : X \to S$ be defined by $\sigma(x) = \sigma'(H'(x))$ for every $x \in X$.
    The following properties hold:
    \begin{enumerate}
        \item \label{prop:ann} (ANN) $\forall x \in X$, $\dist(x, \sigma(x)) \leq O(\Gamma^2)\cdot \dist(x, S)$;
        \item \label{prop:equidist} (equidistant) $\forall h \in H'$, $\max_{x \in \set(h)} \dist(x, S) \leq O(\Gamma^2) \min_{x \in \set(h)} \dist(x, S)$.
    \end{enumerate}
\end{lemma}

\medskip
\noindent
\textbf{Application: Cluster Size Maintenance.}
An immediate application is to maintain the injectivity $w_S(s)$ of points $s \in S$ in the assignment $\sigma$ via maintaining $|\set(h)|$ (for every $h \in H'$) and $w_S(s) := \sum_{h \in H' : \sigma'(h) = s} |\set(h)|$.
Since $\sigma$ is $\poly(\Gamma)$-approximate by Property~\ref{prop:ann} of \Cref{ex-ab:lem:approx_assignment},
$\{w_S(s)\}_{s \in S}$ represents the cluster sizes of an approximate clustering of $X$ with respect to $S$.
It is well known that the weighted set $(S, w_S)$ is a good proxy for $X$, which we will use in our dynamic $k$-means algorithm (discussed in later sections).

\medskip
\noindent
\textbf{Application: $D^2$-Sampling.}
Another more interesting application is to implement $D^2$-sampling, where the task is to sample an $x \in X$ with probability proportional to $\dist^2(x, S)$~\cite{ArthurV07}.
To implement this, we compute for each high-level bucket $\set(h)$ $(h \in H')$ and each low-level buckets therein
the sum of squared distance to the assigned $\sigma'(h)$.
Next, to do the sampling, we first sample a high-level bucket ID $h \in H'$ proportional to the sum-of-squared distance to $\sigma'(h)$,
then further sample a low-level bucket $\ell \in f(h)$ within $f(h)$ again proportional to the sum-of-squared distance,
and finally do a uniform sample in the sampled (low-level) bucket $\set(\ell)$.
Here, uniform sampling works since Property~\ref{prop:equidist} in \Cref{ex-ab:lem:approx_assignment} ensures that points in the same bucket have similar distance to $S$,
and by the same property 
we can estimate for high-level bucket $h \in H$ the value $\sum_{x \in \set(h)} \dist^2(x, \sigma'(h))$ by $|\set(h)| \cdot \dist^2(x, \sigma'(h))$, where $x$ is an arbitrary point in $\set(h)$ (note that both $|\set(h)|$ and this point $x$ can be maintained); similarly for a low-level bucket $\ell \in L$.
This eventually enables us to sample $x$ proportional to $\dist^2(x, S)$, up to $O(\poly(\Gamma))$ factor.

\medskip
\noindent
\textbf{Proof Sketch for \Cref{ex-ab:lem:approx_assignment}.}
Our dynamic algorithm is a dynamization of the following static algorithm.
For every distance scale $i$, we impose a ($\Gamma, \Lambda, \rho_i := \Gamma 2^i)$-hash $\varphi_i$.
We define the high-level bucket IDs as $H := \bigcup_i \varphi_i([\Delta]^d)$.
The low-level buckets are not needed at this point and will be discussed later.
For $i = 1, 2\ldots$, for each $s \in S$, take all non-empty (i.e., with non-empty intersection with $X$) buckets (ID) in $\varphi_i(\ball(s, 2^i))$ into $H'$,
assign them in $\sigma'$ to $s$, remove the data points in all these buckets and continue to the next scale.
This ensures an invariant that a $\varphi_i$-bucket, whose diameter is $\rho_i = \Gamma 2^i$, is picked only when its distance to the closest $s \in S$ is roughly $2^i$,
and this can imply both Property~\ref{prop:ann} and Property~\ref{prop:equidist}.

\medskip
\noindent
\textbf{Representing Partial High-level Buckets by Low-level Buckets.}
The caveat is that the level-$i$ buckets $\varphi_i(\ball(s, 2^i))$ which we add to $H'$
should actually exclude data points already handled/deleted from levels $\leq i - 1$.
In other words, only a part of a high-level bucket should be added.
To describe these partial bucket, we introduce low-level buckets, which are ``refinement'' of high-level buckets,
defined as follows: for each $i$, for each bucket ID $z \in \varphi_i([\Delta]^d)$,
partition the level-$i$ high-level bucket $\varphi^{-1}_{i}(z)$ with respect to $\varphi_{i - 1}$ (i.e., the level-$(i - 1)$ high-level buckets).
The set of all these parts (i.e., the intersections with $(i-1)$-level buckets) forms the low-level buckets $L$.
Now, each high-level bucket with ID $h \in H$ should maintain a list of existing/non-empty intersecting low-level buckets,
in order to determine the set of points belonging to $h$.
This is precisely our mapping $f : H \to 2^L$.

\medskip
\noindent
\textbf{Locality of Low-level Buckets.}
One might wonder if we also need to consider partitioning a level-$i$ high-level bucket by buckets from $i - 2, i - 3, \ldots$ levels and include all these into $L$ and potentially also to $f(h)$.
We show that this is not necessary and only level-$(i - 1)$ is needed;
specifically, when we remove a high-level bucket $h$ from some level $i$, then all points in this $h$ must all be removed (possibly from levels less than $i$), without leaving partially some low-level buckets behind.
This is a crucial property to control the time complexity.

\medskip
\noindent
\textbf{Dynamic Implementation.}
These properties lead to efficient dynamic implementation of the static algorithm.
Consider the insertion operation as an example. Suppose $s \in S$ is inserted;
we iterate over levels $i = 1,2,\ldots$, examine the intersecting high-level buckets around $s$, and update the subset $H'$ as well as the assignment.
Now, before we proceed to insert $s$ to level-$(i + 1)$, we populate the change of $H'$ (caused by level $i$) to level $i + 1$,
by updating the list of associated low-level buckets in affected high-level buckets.
We then proceed to level-$(i+1)$ and repeat the process.
The deletion works similarly and we omit the discussion.

\medskip
\noindent
\textbf{Time Complexity.}
The running time of the dynamic algorithm is governed by the consistency guarantee of consistent hashing.
In particular, when we examine $\varphi_i(\ball(s, 2^i))$, we only need to consider at most $O(\Lambda)$ intersecting high-level buckets.
Moreover, the consistency guarantees that each (high-level) bucket at level $(i-1)$ results in at most $O(\Lambda)$ low-level buckets at level $i$.

\subsection{Approximate Range Query Structure}
\label{sec:ex-ab:range_query}

We give a brief introduction to an approximate \emph{range query} structure,
which is another application of our consistent hashing; see \Cref{sec:range-query} for details.

This data structure works for any query problem that is \emph{mergeable}, which intuitively means that, for two disjoint sets $Y_1, Y_2$, the answer to the query on $Y_1 \cup Y_2$ can be computed from the answers on $Y_1$ and that on $Y_2$.
We note that many fundamental query problems are mergeable, such as min/max/sum, as well as more sophisticated ones, such as coresets.

Our construction is, in fact, a general reduction: suppose there exists a data structure that maintains a dynamic set $X \subseteq [\Delta]^d$ and can answer a mergeable query problem over the \emph{entire} set $X$, then we can design a data structure that also maintains $X$ and, for any \emph{query} $x \in [\Delta]^d$ and $r \ge 0$, answers the same problem on a \emph{subset} that approximates the ball $\ball(x,r) \cap X$;
the update and query time increase by a factor of $\tO(\Lambda)$.

The construction follows almost directly via consistent hashing, so we omit it here and discuss some notable examples below.

\medskip
\noindent
\textbf{Application: Range Query Coreset.}
A notable example is the \emph{coreset}~\cite{DBLP:conf/stoc/Har-PeledM04}, which is a popular data reduction technique.
A coreset of a point set $X$ for $k$-means is a subset of $X$ that approximately preserves the $k$-means cost of \emph{every} possible solution $S \subseteq \R^d$ of $k$ centers.
It is well known that coresets are mergeable~\cite{DBLP:conf/stoc/Har-PeledM04}.
Moreover, there is a robust data structure for maintaining coresets against adaptive adversaries~\cite{HenzingerK20}; therefore, our reduction yields a robust data structure that can return a coreset for any given (approximate) ball subset.

In this paper, we use this dynamic coreset range query only in the $k = 1$ case (i.e., for the $1$-means problem) to compute approximate solutions and cost estimates over any given approximate ball subset.
We state the guarantee for this application below.

\begin{lemma}[$1$-Means Estimation]
\label{ex-ab:lem:one-median-on-a-ball}
There is an algorithm on a dynamic dataset $X \subseteq [\Delta]^d$,
such that given a point $x \in \Deld$ and a radius $r \ge 0$, it can compute in time $\tO(\Lambda)$ a real number $b_{x} \ge 0$, a center $c^{\star} \in \Deld$, and two estimations $\Hat{\cost}_{c^{\star}}$ and $\Hat{\cost}_{x}$ such that there exists some point subset $B_{x} \subseteq X$ with $\ball(x, r) \cap X \subseteq B_{x} \subseteq \ball(x, O(\Gamma) \cdot r) \cap X$ satisfying the following.
\begin{enumerate}[font = \bfseries]
    \item $b_{x} = w(B_{x})$.
    
    \item $\cost(B_{x}, c) \leq \Hat{\cost}_{c} \leq O(1) \cdot \cost(B_{x}, c)$, for either $c \in \{c^{\star},\ x\}$.
    
    \item $\cost(B_{x}, c^{\star}) \leq O(1) \cdot \OPT_{1}(B_{x})$.
\end{enumerate}
The algorithm has amortized update time $\tO(\Lambda)$ (w.r.t.~updates on $X$), works against an adaptive adversary, and succeeds w.h.p.
\end{lemma}

\medskip
\noindent
\textbf{Application: ANN.}
Another application is an $O(\Gamma)$-approximate nearest neighbor (ANN) structure, which maintains a point set and, for any query point, returns an $O(\Gamma)$-ANN in the maintained set.
This structure supports update and query time $\tO(\Lambda)$ and works against adaptive adversaries.
Similar guarantee of ANN has also been obtained in e.g.,~\cite{DBLP:conf/nips/0001DFGHJL24}, via more specific techniques, although they may potentially yield better tradeoff beyond consistent hashing.

The ANN data structure is a fundamental data structure with wide applications. 
In this paper, one particular useful application is to plug it into the framework of~\cite{BhattacharyaGJQ24} to obtain a new data structure that maintains, for each point, its distance to an approximate nearest neighbor.
We summarize this result in the following lemma, which is needed in our main algorithm.

\begin{lemma}[ANN Distance]
\label{ex-ab:lem:nearest-neighbor-distance}
There exists a data structure that explicitly maintains, for every $s\in S$, a value $\Hat{\dist}(s, S - s)$ such that 
\begin{align*}
    \dist(s, S - s) \le \Hat{\dist}(s, S - s) \le O(\Gamma)\cdot \dist(s, S - s).
\end{align*}
The data structure has amortized update time $\tO(\Lambda)$ (w.r.t. updates on $S$), works against an adaptive adversary, and succeeds w.h.p.
\end{lemma}

\section{\texorpdfstring{Main Subroutines: Restricted and Augmented $k$-Means}{}}
\label{sec:ex-ab:restriced-augmented}

In this section, we present our new algorithms for the two main subroutines in our final dynamic algorithm: \emph{restricted} $k$-means and \emph{augmented} $k$-means.
Both problems take a parameter $r > 0$, as well as a dataset $X$ and a center set $S$ as dynamic input.
The restricted $k$-means aims to remove $r$ points from $S$, and augmented $k$-means aims to add $r$ points to $S$,
both to minimize the resultant clustering cost.
While in general metrics both subroutines are known to be solvable in $\tilde{O}(kr)$ time (using certain data structures), which is fast enough for the general-metric case,
they are too slow for our Euclidean setting, where our goal is to achieve sublinear $o(k)$ time.

\subsection{\texorpdfstring{Restricted $k$-Means}{}}
\label{sec:ex-ab:restricted}

\begin{lemma}[Restricted $k$-Means] 
\label{ex-ab:lem:restricted-k-means}
    There exists a data structure that handles two dynamic sets $X, S\subseteq \Deld$, where updates are generated by an adaptive adversary. For every query with an integer $1\le r\le |S|$, it runs in time $\tO(2^{\epsilon d}\cdot r + r^{1 + \epsilon})$ to return a set $R\subseteq S$ of $r$ points such that the following holds:
    \begin{align*}
        \cost(X, S - R) \le \poly(\epsilon^{-1})\cdot \OPT_{|S| - r}^{S}(X),
    \end{align*}
    where $\OPT_{|S| - r}^{S}(X)\eqdef \min_{S'\subseteq S:|S'| = |S| - r}\cost(X,S')$.
Each update is handled in amortized time $\tO(2^{\epsilon d })$, and the algorithm succeeds w.h.p.
\end{lemma}

For ease of illustration, we may assume that $X = S$; this is natural, since each point $x \in X$ can be moved to its (approximate) nearest center in $S$ with only a constant-factor increase in the clustering objective. 
Under this assumption, the goal of the restricted $k$-means problem becomes finding a set $R$ of $r$ points in $S$ that minimizes

\begin{equation*}
    \cost(S, S - R) ~=~ \cost(R, S - R) ~=~ \sum_{s\in R}\dist^2(s, S - R).
\end{equation*}

The first observation here is that the restricted $k$-means problem with $r$ centers to delete is equivalent to the $k'$-means problem with $k' = k - r$, where the solution is restricted to be selected from $S$.
Hence, a na\"ive approach would be to directly apply an efficient algorithm for Euclidean $k$-means (e.g.~\cite{laTourS24,abs-2504-03513}); however, this would only achieve a superlinear $|S|^{1+\epsilon}$ running time.

\paragraph{Our Coreset-based Algorithm.}
Our algorithm, at a high level, is to speed up this na\"ive approach by designing a \emph{coreset} for the restricted $k$-means problem.
Roughly speaking, this coreset guarantees that running an $O(1)$-approximation algorithm for the restricted $k$-means on it yields a solution whose ratio is larger but still constant for the full set $S$.
We give a coreset construction with size $O(r)$, and this immediately helps to speed up the na\"ive approach to $\tilde{O}(r^{1+\epsilon})$ time while still obtaining an $O(1)$-approximation for the restricted $k$-means on $S$.

This notion of coreset is new and is specifically devised for our restricted $k$-means problem.
This notion is particularly different than the standard coreset for clustering in the literature (see e.g.~\cite{
DBLP:conf/stoc/Har-PeledM04}),
and a main conceptual difference is that ours is a coreset with respect to the points to be deleted (i.e., $R$),
whereas a standard coreset preserves costs for the center set (i.e., $S - R$).
Moreover, the techniques for constructing our new coreset is also very different than commonly used sampling techniques~\cite{DBLP:journals/siamcomp/Chen09, DBLP:conf/stoc/FeldmanL11, DBLP:conf/stoc/Cohen-AddadSS21, DBLP:conf/stoc/Cohen-AddadLSS22, DBLP:conf/focs/0001CPSS24,Huang0024, DBLP:conf/soda/Cohen-AddadD0SS25} as in standard coresets,
and we instead employ a greedy algorithm to build the coreset.

Next, we give an overview of our coreset construction,
which consists of two major steps.
To start, in the first step, we greedily select $6r$ points $s \in S$ with the smallest values of $\dist(s, S - {s})$ as the candidate centers.
Denoting the resulting set by $T_1$,
we show (in \Cref{claim:restricted-intuition-1}) that the candidate set $T_1$ contains, as a subset, a good solution to the restricted $k$-means problem.

\begin{claim}\label{claim:restricted-intuition-1}
    There exists a subset of $r$ points from $T_1$ that is an $O(1)$-approximation to the restricted $k$-means on the full set $S$.
\end{claim}

However, the existence of such a good solution does not imply that $T_1$ is a coreset, since the clustering objective may be significantly distorted on $T_1$, and thus an approximation algorithm run on $T_1$ may not yield the good solution guaranteed by \Cref{claim:restricted-intuition-1}.
Therefore, the second step is to address this issue, and we proceed in the following way.
For each selected point $s \in T_1$, we find its (approximate) nearest neighbors among the unselected points $S - T_1$.
Denoting these (approximate) nearest neighbors by $T_2$ and letting $T := T_1 \cup T_2$, we prove the following claim.

\begin{claim}\label{claim:restricted-intuition-2}
    For any subset $R\subseteq T_1$, it holds that 
    \begin{equation*}
        \cost(S, S - R) ~\le~ \cost(T, T - R) ~\le~ O(1)\cdot \cost(S, S - R). 
    \end{equation*}
\end{claim}

\Cref{claim:restricted-intuition-2} implies that the objectives of all candidate solutions are preserved on $T$ within a constant factor.
Therefore, by combining \Cref{claim:restricted-intuition-1,claim:restricted-intuition-2} (as well as some argument), it can be shown that $T = T_1 \cup T_2$ is the desired coreset.

In what follows, we explain the intuition behind \Cref{claim:restricted-intuition-1,claim:restricted-intuition-2} separately, and then discuss the dynamic implementation of this construction.

\paragraph{Candidate Centers via Greedy Selection (\Cref{claim:restricted-intuition-1}).}
Recall that $T_1$ consists of the $6r$ points $s \in S$ with the smallest values of $\dist(s, S - {s})$.
Here, we consider the case where the optimal solution $R^{*}$ to the restricted $k$-means problem on $S$ is disjoint from $T_1$, for ease of illustration.
(The case where $R^{*}$ and $T_1$ overlap is handled similarly.)

The key to the proof is that, among the points in $T_1$ we show the existence of a subset of $r$ points, denoted by $R'$, such that 
$$
    \forall s \in R',\quad \dist(s, S - s) ~\ge~ \Omega(1) \cdot \dist(s, S - R').
$$ 
Since $R' \subseteq T_1$ and $R^{*} \cap T_1 = \emptyset$, and given that $T_1$ consists of the points $s\in S$ with the smallest $\dist(s, S - s)$ values, it follows that 
\begin{equation}
    \forall s_1\in R', \forall s_2\in R^{*},\quad \dist(s_1, S - s_1) ~\le~ \dist(s_2, S - s_2).
\end{equation}
As a result, 
\begin{equation*}
    \sum_{s_2\in R^{*}} \dist^2(s_2, S - R^{*}) \ge 
    \sum_{s_2\in R^{*}} \dist^2(s_2, S - s_2) \ge
    \sum_{s_1\in R'} \dist^2(s_1, S - s_1) \ge \Omega(1)\cdot 
    \sum_{s_1\in R'} \dist^2(s_1, S - R').
\end{equation*}
This concludes that $R' \subseteq T_1$ is an $O(1)$-approximation to the restricted $k$-means problem, thereby proving \Cref{claim:restricted-intuition-1}.

\paragraph{Nearest Unselected Neighbors for Candidates (\Cref{claim:restricted-intuition-2}).}
The inequality $\cost(S, S - R) \le \cost(T, T - R)$ (equivalently, $\cost(R, S - R) \le \cost(R, T - R)$) follows directly from the fact that $T$ is a subset of $S$.
We then focus on the inequality $\cost(T, T - R) \le O(1)\cdot \cost(S, S - R)$.
Consider a fixed candidate solution $R' \subseteq T_1$ and suppose that its restricted $k$-means objective $\cost(T_1, T_1 - R') = \cost(R', T_1 - R')$ on the candidate set $T_1 \subseteq S$
is severely distorted, i.e., significantly larger than its objective $\cost(S, S - R') = \cost(R', S - R')$ on the full set $S$.
This distortion occurs only if, for some point $s \in R'$, the distance $\dist(s, T_1 - R')$ is significantly larger than $\dist(s, S - R')$.
This further means that the nearest neighbor of $s$ in $S - R'$ is not contained in $T_1 - R'$; such a neighbor must lie in $S - T_1$ since $R'\subseteq T_1 \subseteq S$.
Therefore, this issue can be directly fixed by adding, for every candidate point $s \in T_1$, its (approximate) nearest neighbor in $S - T_1$ to the coreset, which is precisely our second step.
This outlines the main idea behind the proof of \Cref{claim:restricted-intuition-2}.

\paragraph{Dynamic Implementations.} 
This coreset construction is well suited for the dynamic setting, since it only relies on maintaining/querying approximate nearest neighbors,
which is supported by our data structures introduced in \Cref{sec:ex-ab:range_query}.
Moreover, we note that all the above arguments hold for the case where $X = S$.
Therefore, we actually need an additional step that ``moves'' each point in $X$ to its (approximate) nearest neighbor in $S$, resulting in a weighted set $(S, w_S)$ where the weight $w_S(s)$ equal to the number of points moved to $s$.
To realize this moving step in the dynamic setting, we need to maintain the weight $w_S$, which is efficiently supported by our cluster weight estimation structure introduced in \Cref{sec:ex-ab:approx_assign}.

\subsection{\texorpdfstring{Augmented $k$-Means}{}}\label{sec:ex-ab:augmented-k-means-explain}

\begin{lemma}[Augmented $k$-Means] 
\label{ex-ab:lem:augmented-k-means}
    There exists a data structure that handles two dynamic sets $X, S\subseteq \Deld$, where updates are generated by an adaptive adversary. For every query with an integer $r \ge 1$, it runs in time $\tilde{O}(r\cdot 2^{\epsilon d })$ to return a set $A\subseteq \Deld$ of $O(r\cdot \poly(\epsilon^{-1})\cdot d\log \Delta)$ points such that the following holds w.h.p.: 
    \begin{align*}
        \cost(X, S + A) ~\le~ O(1)\cdot \min_{A^{\star}\subseteq \Deld: |A^{\star}| = r}  \cost(X, S + A^{\star}).
    \end{align*}
    Each update is handled in amortized time $\tilde{O}(2^{\epsilon d })$, and the algorithm succeeds w.h.p.
\end{lemma}

Our approach is to iteratively select and add points to $S$ until a good augmentation is achieved.
We borrow the idea from the classical $k$-means++ algorithm~\cite{ArthurV07}, which employs $D^2$-sampling -- i.e., sampling a point from $X$ with probability proportional to $\dist^2(x, S)$ -- to iteratively select centers and achieve an $O(\log k)$-approximation for the vanilla $k$-means problem.
A follow-up work~\cite{DBLP:conf/approx/AggarwalDK09} considers sampling $O(k)$ centers to achieve a bi-criteria $O(1)$-approximation for vanilla $k$-means, and we adapt their analysis to show that repeated $D^2$-sampling similarly yields a good approximate solution to the augmented $k$-means problem in our setting.
We employ an approximate $D^2$-sampling in \Cref{sec:ex-ab:approx_assign} to implement this idea in dynamic setting.

\section{Our Dynamic Algorithm: An Overview}
\label{ex-ab:sec:implementation}

In this section, we give an overview of our algorithm, which we obtain by implementing a significantly modified version of the algorithm from \cite{BCF24}
using the data structures from \Cref{ex-ab:sec:data-structures-main} as well as the main subroutines in \Cref{sec:ex-ab:restriced-augmented}. 
In this extended abstract, we provide our algorithm with $\tilde{O}(n^\epsilon)$ update time.
It is possible to \textbf{reduce the update time to $\tilde{O}(k^\epsilon)$} by sparsification techniques. The reader can find a thorough description of how to achieve $\tilde{O}(k^\epsilon)$ update time in \Cref{part:from-n-to-k}.

\paragraph{Roadmap.} We start with a review of the high-level algorithm in~\cite{BCF24} in \Cref{ex-ab:sec:alg:describe}.
This is useful for understanding our new algorithm which introduces various modifications to~\cite{BCF24}.
Then in \Cref{ex-ab:preprocessing-transformation}, we present some preliminaries and in \Cref{sec:our-alg:implementation}, we describe how to implement an epoch, which is a basic step, of our algorithm. The remaining subsections describe how to implement some specific tasks and subroutines.

\subsection{\texorpdfstring{Review of the Dynamic Algorithm of \cite{BCF24}}{}}
\label{ex-ab:sec:alg:describe}

In this section, we explain the algorithm of \cite{BCF24} and its main ingredients that work for general metric spaces.
Although their algorithm is primarily designed for $k$-median, it can be seamlessly extended to work for $k$-means as well. See \Cref{part:full-alg} for a thorough analysis.

\subsubsection{Robust Centers}

\begin{definition}[Robust Center, Definition 3.2 in the arxiv version of \cite{BCF24}]\label{ex-ab:def:robust}
    Assume $X \subseteq \Deld$.
    Let $(x_0,x_1,\ldots, x_t)$ be a sequence of $t+1$ points in $\Deld$, and let $B_i = \ball_X(x_i, 10^i)$ for each $i \in [0,t]$.
    We refer to $(x_0,x_1,\ldots, x_t)$ as a {\em $t$-robust sequence} w.r.t.\ $X$ iff for every $i \in [1, t]$:
    \begin{eqnarray*}
        x_{i-1} =
        \begin{cases}
            x_i \quad & \text{if} \ 
            \textsc{AverageCost}(B_i,x_i) \geq 10^i / 5; \\
            y_i \quad & \text{otherwise, where } y_i = \arg\min\limits_{y \in B_i + x_i} \cost(B_i,y).
        \end{cases}
    \end{eqnarray*}
     We say that a  point $x \in \Deld$ is {\em $t$-robust} w.r.t.\ $X$ iff there exists a $t$-robust sequence $(x_0,x_1,\ldots, x_t)$ w.r.t.\ $X$ such that $x_0 = x$.
\end{definition}

\begin{definition}[Robust Solution, Definition~3.5 in the arxiv version of \cite{BCF24}]
\label{ex-ab:def:robust-solution}
    We call a center set $S \subseteq \Deld$, \textit{robust} w.r.t.\ $X$ iff each center $u \in S$ satisfies the following
    \begin{equation}\label{ex-ab:cond:robust}
        u\  \text{is}\  t\text{-robust w.r.t.\ $X$, where} \ t\ \text{is the smallest integer satisfying} \ 
        10^t \geq \dist(u,S - u)/ 200.
    \end{equation}
\end{definition}

\subsubsection{\texorpdfstring{Description of The Algorithm of~\cite{BCF24}}{}}
\label{ex-ab:subsec:alg-describe}

Here, we only explain the version of the algorithm that has an exponential running time.
Our new implementation of their algorithm in high-dimensional Euclidean spaces, which has a small update time,
is discussed in later subsections.

The algorithm works in {\bf epochs};  each epoch lasts for some consecutive updates in $X$. Let $S \subseteq X$ denote the maintained solution (set of $k$ centers). We satisfy the following invariant.

\begin{invariant}
\label{ex-ab:inv:start}
At the start of an epoch, the set $S$ is robust and $\cost(X, S) \leq O(1) \cdot \OPT_k(X)$.
\end{invariant}

We now describe how the dynamic algorithm works in a given epoch, in four steps.

\medskip
\noindent {\bf Step 1: Determining the length of the epoch.} At the start of an epoch, the algorithm computes the maximum $\ell^{\star} \geq 0$ such that $\OPT_{k-\ell^{\star}}(X) \leq c \cdot \OPT_{k}(X)$,
and sets $\ell \leftarrow \lfloor \ell^{\star}/\Theta(c) \rfloor$,
where $c=O(1)$ is a large constant.
The epoch will last for the next $\ell+1$ updates.\footnote{Note that it is possible that $\ell = 0$.}
From now on, consider the superscript $t \in [0, \ell+1]$ to denote the status of some object after the algorithm has finished processing the $t^{th}$ update in the epoch.
For example, at the start of the epoch, $X = X^{(0)}$.

\medskip 
\noindent {\bf Step 2: Preprocessing at the start of the epoch.} Let $S_{\init} \leftarrow S$ be the solution maintained by the algorithm after it finished processing the last update in the previous epoch. Before handling the  very first update in the current epoch, the algorithm initializes the maintained solution by setting
\begin{equation}
\label{ex-ab:eq:init:epoch}
S^{(0)} \leftarrow \arg \min_{S' \subseteq S_{\init} \, : \, |S'| = k-\ell} \cost(X^{(0)}, S').
\end{equation}

\medskip 
\noindent {\bf Step 3: Handling the updates within the epoch.} Consider the $t^{th}$ update in the epoch, for $t \in [1, \ell+1]$.
The algorithm handles this update in a lazy manner, as follows. If the update involves the deletion of a point from $X$, then it does not change the maintained solution, and sets $S^{(t)} \leftarrow S^{(t-1)}$.
In contrast, if the update involves the insertion of a point $p$ into $X$, then it sets $S^{(t)} \leftarrow S^{(t-1)} + p$.

\medskip
\noindent {\bf Step 4: Post-processing at the end of the epoch.}
After the very last update in the epoch, the algorithm does some post-processing, and computes another set $S_{\final} \subseteq X$ of at most $k$ centers (i.e., $|S_{\final}| \leq k$) that satisfies Invariant~\ref{ex-ab:inv:start}.
Then the next epoch can be initiated with $S \leftarrow S_{\final}$ being the current solution.
The post-processing is done as follows.

The algorithm adds $O(\ell + 1)$ extra centers to the set $S_{\init}$, while minimizing the cost of the resulting solution w.r.t.\ $X^{(0)}$.
\begin{equation}
\label{ex-ab:eq:augment}
A^{\star} \leftarrow  \arg\min\limits_{\substack{A \subseteq \Deld : |A| \leq O(\ell+1)}} \cost(X^{(0)},S_{\init} + A), \text{ and } S^{\star} \leftarrow S_{\init} + A^{\star}.
\end{equation}
The algorithm then adds the newly inserted points within the epoch to the set of centers, so as to obtain the set $T^{\star}$.
Next, the algorithm identifies the subset $W^{\star} \subseteq T^{\star}$ of $k$ centers that minimizes the $k$-median objective w.r.t.\ $X^{(\ell+1)}$.
\begin{equation}
\label{ex-ab:eq:augment:1}
T^{\star} \leftarrow S^{\star} + \left( X^{(\ell+1)} - X^{(0)}\right), \text{ and }
W^\star \leftarrow \arg\min\limits_{\substack{W \subseteq T^{\star} \, : \, |W| = k}} \cost(X^{(\ell+1)},W).
\end{equation}
Finally, it calls a subroutine referred to as $\Robustify$  on $W^\star$ and lets $S_{\final}$ be the set of $k$ centers returned by this subroutine.
The goal of calling $\Robustify$ is to make the final solution $S_{\final}$, robust w.r.t.\ the current space $X=X^{(\ell+1)}$.
We will provide our implementation of this subroutine in \Cref{ex-ab:sec:robustify}.
Finally, before starting the next epoch, the algorithm sets $S \leftarrow S_{\final}$. 
\begin{equation}
\label{ex-ab:eq:augment:2}
S_{\final} \leftarrow \Robustify(W^{\star}), \ \text{and} \ S \leftarrow S_{\final}.
\end{equation}

It can be shown that the maintained set $S \subseteq \Deld$ is always of size at most $k$, and the solution $S = S_{\final}$ satisfies Invariant~\ref{ex-ab:inv:start} at the end of Step 4, that validates the analysis of the next epoch.

\subsection{Some Preliminaries}

\paragraph{Preprocessing.}\label{ex-ab:preprocessing-transformation} Throughout this section, we assume that the input space to our algorithm is $\Deld$, where the dimension $d$ is $O(\log n)$.
This assumption can be achieved by the well-known Johnson-Lindenstrauss transformation \cite{JL84}.
Assume the input space is $\R^m$.
We initialize a random matrix $A$ of size $m \times d$ where $d = O(\log n)$.
Whenever a point $x \in \R^m$ is inserted to the dynamic space, we first apply the transformation (compute $Ax \in \R^d$) and then round it to a discrete point in $[\Delta]^d$.
As a result, we can feed this new dynamic space into our algorithm.
Although our algorithm maintains a solution only in the new space,
using some standard techniques, we can transform it to maintain a solution in the original space with only a negligible loss in the approximation ratio and update time (see \Cref{remark:robust-adaptive}).
Since $d = O(\log n)$, we conclude that the running time performance of all the data structures in \Cref{ex-ab:sec:data-structures-main,sec:ex-ab:restriced-augmented} becomes $\tilde{O}(n^{O(\epsilon)})$.
Then, by scaling $\epsilon$, we assume that the running time of these data structures is $\tilde{O}(n^\epsilon)$.
For simplicity, we consider a parameter $\polyup$ which is a polynomial of $(1/\epsilon)$ and is an upper bound on the approximation guarantees of all data structures in \Cref{ex-ab:sec:data-structures-main}.
So, we consider the following throughout the entire \Cref{ex-ab:sec:implementation}.
\begin{align*}
    &d = O(\log n), \text{the input space:} \Deld, \text{data structures in \Cref{ex-ab:sec:data-structures-main} have} \\ &\text{approximation guarantee at most } \polyup \text{ and running time at most } \tilde{O}(n^\epsilon).
\end{align*}

First, we extend the building blocks of \Cref{ex-ab:def:robust} to work with approximate balls, because of the following reason: While using the fast data structures in \Cref{ex-ab:sec:data-structures-main}, for every $x \in \Deld$ and $r > 0 $, some of the points in the space that has a distance less than $r$ to $x$ might be missing in the approximate ball, as well as some point that has distance more than $r$ to $x$ might be mistakenly included in the ball.
This motivates us to modify the definition of robust centers, as follows.

\begin{definition}\label{ex-ab:def:robust-real}
    Let $t \geq 0$ be an integer and $(x_0,x_1,\ldots,x_t)$ be a sequence of $t+1$ points of $\Deld$.
    We call this sequence \textit{$t$-robust} w.r.t.~the current data set $X$, if there exists a sequence $(B_0,B_1, \cdots, B_t)$ of subsets of $X$ such that the following hold,
    \begin{enumerate}
        \item For each $0 \leq i \leq t$, we have $\ball(x_i,\polyup^{6i}) \subseteq B_i \subseteq \ball(x_i, \polyup^{6i+2})$.
        \item For each $1 \leq i \leq t$, at least one of the following hold,
        \begin{itemize}
            \item 
            Either 
            \begin{equation}\label{ex-ab:cond:robust1}
              \frac{\cost(B_i, x_i)}{w(B_i)} \geq \polyup^{12i-8} \ \text{and} \ x_{i-1} = x_i,
            \end{equation}
            \item 
            Or
            \begin{equation}\label{ex-ab:cond:robust2}
                \frac{\cost(B_i, x_i)}{w(B_i)} \leq \polyup^{12i-4} \ \text{and} \ \cost(B_i, x_{i-1}) \leq \min \{ (\polyup)^3 \cdot \OPT_1(B_i), \ \cost(B_i,x_i) \}. 
            \end{equation}
        \end{itemize}
    \end{enumerate}
    Note that both of these cases can happen simultaneously, and in the second case, $x_{i-1}$ can be equal to $x_i$ as well.
    Moreover, we call a point $x$, $t$-robust if there exists a $t$-robust sequence $(x_0,x_1,\ldots,x_t)$ such that $x = x_0$.
\end{definition}

\begin{definition}\label{ex-ab:def:robust-solution-real}
    We call a center set $S$, \textit{robust} w.r.t.~$X$ if and only if each center $u \in S$ satisfies the following
    \begin{equation}\label{ex-ab:cond:robust-real}
        u \text{ is at least } t\text{-robust, where } t \text{ is the smallest integer satisfying } \polyup^{6t} \geq \dist(u, S-u)/\polyup^{20}.
    \end{equation}
\end{definition}

These definitions are relaxations of \Cref{ex-ab:def:robust} and \Cref{ex-ab:def:robust-solution}, which enable us to work with approximate balls while keeping the main properties of robust center (with a small overhead on the approximation ratio).

\subsection{Implementation of an Epoch}\label{sec:our-alg:implementation}
We provide the implementation of our algorithm for an epoch lasting for $(\ell+1)$ updates.
At the beginning of the epoch, we satisfy the following invariant, which is a relaxed version of \Cref{ex-ab:inv:start}.
\begin{invariant}
\label{ex-ab:inv:start-real}
At the start of an epoch, the set $S$ is robust and $\cost(S,X) \leq (\polyup)^6 \cdot \OPT_k(X)$.
\end{invariant}

\noindent
We start by providing the implementation of Steps 1 to 4 of the framework of the algorithm by using data structures in \Cref{ex-ab:sec:data-structures-main}.
Then, we provide the $\Robustify$ subroutine in \Cref{ex-ab:sec:robustify}.

\subsubsection*{Implementing Step 1.} 

We find an estimation of $\ell^{\star}$ instead of the exact value.
We define the exact value of $\ell^\star$ to be the largest integer satisfying the following
\begin{equation} \label{ex-ab:eq:definition-of-ell-star}
   \OPT_{k-\ell^\star}(X^{(0)}) \leq \polyup^{88} \cdot \OPT_{k}(X^{(0)}). 
\end{equation}
Now, we proceed with the procedure of estimating $\ell^\star$.
For each $i \in [0, \log_2 k]$ define $s_i := 2^i$, and let $s_{-1} := 0$. We now run a {\bf for} loop, as described in \Cref{ex-ab:alg:find-ell-implementation}.

\medskip

\begin{algorithm}[ht]
\caption{\label{ex-ab:alg:find-ell-implementation}
Computing an estimate $\hat{\ell}$ of the value of $\ell^{\star}$.}
\begin{algorithmic}[1]
  \For{$i=0$ to $\log_2 k$}
    \State Using restricted $k$-means (\Cref{ex-ab:lem:restricted-k-means}), compute a subset $\hat{S}_i \subseteq S_{\init}$ of $(k-s_i)$ centers that is a $\polyup$-approximation to $\OPT_{k-s_i}^{S_{\init}}\left( X^{(0)}\right)$.
    \label{ex-ab:line:estimate:1} 

    \If{$\cost(X^{(0)}, \hat{S}_i) > \polyup^{96} \cdot \cost(X^{(0)}, S_{\init})$}
    \label{ex-ab:if:condition-inside-find-ell}
        \State\Return $\hat{\ell} := s_{i-1}$.
    \EndIf
    \EndFor
\end{algorithmic}
\end{algorithm}

After finding $\hat{\ell}$, we set the length of the epoch to be $\ell + 1$ where 
$ \ell \gets \left\lfloor \frac{\hat{\ell}}{36 \cdot \polyup^{102}} \right\rfloor$.

\subsubsection*{Implementing Step 2.}
Instead of finding the optimum set of $(k-\ell)$ centers within $S_{\init}$, we approximate it using restricted $k$-means (\Cref{ex-ab:lem:restricted-k-means}) by setting $S = S_{\init}$ and $r = \ell$.
We compute a set of $(k-\ell)$ centers $S^{(0)} \subseteq S_{\init}$ such that 
$ \cost(X^{(0)},S^{(0)}) \leq \polyup \cdot \OPT_{k-\ell}^{S_{\init}}(X^{(0)})   
$.

\subsubsection*{Implementing Step 3.}
Trivially, we can implement each of these updates in constant time.

\subsubsection*{Implementing Step 4.}
We approximate $A^\star$ with augmented $k$-means (\Cref{ex-ab:lem:augmented-k-means}) by setting $S = S_{\init}$, $X = X^{(0)}$ and $a = \polyup^{104} \cdot (\ell+1)$ to get $A \subseteq \Deld$ of size $\tilde{O}(a)$ such that
\begin{equation}\label{ex-ab:eq:augment-approx}
   \cost(X,S_{\init} + A) \leq \polyup \cdot \min\limits_{A^\star \subseteq \Deld: |A^\star| \leq a} \cost(X,S_{\init}+A^\star), \quad a = \polyup^{104} \cdot (\ell+1).
\end{equation}
Then, we set $S' = S_{\init} + A $ and $T' := S' + \left( X^{(\ell+1)} - X^{(0)} \right)$ (approximations of $S^\star$ and $T^\star$ respectively).
Next, we compute $W'$ (approximation of $W^\star$) using restricted $k$-means (\Cref{ex-ab:lem:restricted-k-means}) that satisfies
\begin{equation}
\label{ex-ab:eq:restricted-at-the-end-of-epoch}
 \cost(X^{(\ell+1)}, W') \leq \polyup \cdot \OPT_{k}^{T'}(X^{(\ell+1)}).
\end{equation}
Finally, we explain below how we implement the call to $\Robustify(W')$ (see \Cref{ex-ab:eq:augment:2}).

\subsection{Implementing Robustify}\label{ex-ab:sec:robustify}

The goal of this subroutine is to make the solution $W'$ at the end of the epoch, robust (see \Cref{ex-ab:def:robust-solution-real}) w.r.t.~the current dataset.
$\Robustify$ can be divided into two tasks.
First, identifying non-robust centers, i.e.~the centers that violate Condition (\ref{ex-ab:cond:robust-real}).
Next, making them robust w.r.t~the current data set.

\paragraph{Searching for Non-Robust Centers.}
In the algorithm, we only need to make sure that we do not miss any center violating Condition (\ref{ex-ab:cond:robust-real}).
In order to do this, for every center $u \in S$ in the main solution, we maintain an integer $t[u]$ indicating that the center $u$ is $t[u]$-robust.
We also maintain an approximation of $\dist(u,S-u)$ denoted by $\hat{\dist}(u,S-u)$ using \Cref{ex-ab:lem:nearest-neighbor-distance} (We use the notation $\hat{\dist}$ extensively throughout this section).
We consider the following condition.
\begin{equation}\label{ex-ab:condition-new}
    t[u] \geq t, \text{ where } t \text{ is the smallest integer satisfying } \polyup^{6t} \geq \hat{\dist}(u,S-u)/\polyup^{20}
\end{equation}
Since $\hat{\dist}(u,S-u) \geq \dist(u,S-u)$, it is obvious that if $u$ satisfies Condition (\ref{ex-ab:condition-new}), it definitely satisfies Condition (\ref{ex-ab:cond:robust-real}).
We will show how to maintain $t[u]$ and how to determine if $u$ violates Condition (\ref{ex-ab:condition-new}) in \Cref{ex-ab:sec:identify-robust}.

\paragraph{Making Centers Robust.}
After identifying a center $u$ that violates Condition (\ref{ex-ab:condition-new}), we find the smallest integer $t$ satisfying $\polyup^{6t} \geq \hat{\dist}(u, S-u) / \polyup^{14}$, swap $u$ with a close center $v$ which is $t$-robust through a call to subroutine $\MakeRbst(u)$, and set $t[v] := t$.
Note that here the denominator is $\polyup^{14}$ instead of $\polyup^{20}$, which is intentional and swaps $u$ with a center $v$ that is more robust than needed in Condition (\ref{ex-ab:cond:robust-real}).
This helps us to make sure that as long as $\dist(v, S-v)$ is not change significantly, the center $v$ remains robust enough.
We explain how $\MakeRbst(u)$ works in \Cref{ex-ab:sec:make-robust}.

\subsection{Searching for Non-Robust Centers}\label{ex-ab:sec:identify-robust}

We can split this task into two parts as follows.
\begin{enumerate}
    \item Identify all centers $u$ that might not be $t[u]$-robust anymore.\label{ex-ab:identify-part1}
    
    \item Identify a center $u$ that violates Condition (\ref{ex-ab:condition-new}).\label{ex-ab:identify-part2}
\end{enumerate}
At the end of part \ref{ex-ab:identify-part1}, we make all the identified centers robust (via a call to $\MakeRbst$), and update their $t[u]$ value.
Only after that, we follow with part \ref{ex-ab:identify-part2}.

We show that every center $u \in S$ will be $t[u]$-robust, i.e.~the detected centers in part \ref{ex-ab:identify-part1} become robust via a call to $\MakeRbst$, and a center $u$ that is not detected in part \ref{ex-ab:identify-part1} is still $t[u]$-robust.
This fact (\Cref{ex-ab:claim:remains-robust}) helps us to implement part \ref{ex-ab:identify-part2} later.

\begin{claim}\label{ex-ab:claim:remains-robust}
    Each center $u$ not detected in part \ref{ex-ab:identify-part1}, remains $t[u]$-robust w.r.t.~the final space $X^{(\ell+1)}$.
\end{claim}

The proof of this claim follows from \cite{BCF24}, but for the sake of completeness, we provide the proof in \Cref{part:full-alg} (see \Cref{sec:proof-of-claim-remains-robust}).

\paragraph{Partitioning $W'$.}
We consider a partitioning $W' = W_1 \cup W_2 \cup W_3$ as follows.
$W_1$ consists of the \textit{new} centers added to the solution, i.e.~$W_1 := W' - S_\init$.
$W_2$ consists of all \textit{contaminated} centers.
A center $u$ is called contaminated, if throughout the epoch, the ball of radius $ \polyup^{6t[u] + 3} $ around $u$ is changed, i.e.,
$$ \ball(u, \polyup^{6t[u]+2}) \cap \left( X^{(0)} \oplus X^{(\ell+1)}\right) \neq \emptyset . $$
$W_3$ is then the rest of the centers $W_3 := W' - (W_1 \cup W_2)$.

\subsubsection{Implementing Part \ref{ex-ab:identify-part1}}

According to \Cref{ex-ab:claim:remains-robust}, the only centers which we do not have any guarantee about their robustness are centers in $W_1 \cup W_2$.
Hence, we find $W_1$ and $W_2$ and consider centers $W_1 \cup W_2$ for part \ref{ex-ab:identify-part1}.
It is obvious how to find $W_1$ explicitly since $W_1 = W' - S_\init$.
We only need to keep track of the newly added centers during the epoch.
In order to find $W_2$, we provide a property of contaminated centers as follows.
\begin{claim}\label{ex-ab:claim:num-of-contamination}
    For each $x \in X^{(0)} \oplus X^{(\ell+1)}$ and each $i \in [0, \lceil \log_{\polyup^6} (\sqrt{d}\Delta) \rceil]$, there is at most one center $u \in S_\init$ such that $t[u] = i$ and $\dist(x,u) \leq \polyup^{6i + 4}$.
\end{claim}
This claim shows that every update in the data set can contaminate at most one center in each scale $i \in [0, \lceil \log_{\polyup^6} (\sqrt{d}\Delta) \rceil]$.
The proof of this claim follows from \cite{BCF24}, but for the sake of completeness, we provide the proof in \Cref{part:full-alg} (see \Cref{sec:proof-of-claim-num-of-contamination}).

\paragraph{Partitioning According to $t[u]$ Value.}
\Cref{ex-ab:claim:num-of-contamination} provides us with a way to find contaminated centers efficiently.
We maintain a partitioning of the main solution $S$ throughout the entire algorithm as follows.
$S = \cup_{i=0}^{\lceil \log_{\polyup^6} (\sqrt{d}\Delta) \rceil} S[i]$, where $S[i]$ consists of centers $u \in S$ with $t[u] = i$.
Whenever a new center $u$ is inserted into the solution, after computing $t[u]$, we insert $u$ into the appropriate part $S[t[u]]$.
Now, for each set $S[i]$, we invoke the approximate nearest neighbor oracle (ANN) that given $x \in \Deld$, returns a $u \in S[i]$ such that $\dist(x,u) \leq \polyup \cdot \dist(x, S[i])$.
In the following claim, we show that $S[i]$ has a special structure, such that the approximate nearest neighbor is the \textbf{only center} that can be contaminated by $x$ within $S[i]$ at the end of each epoch.
The proof of this claim is deferred to \Cref{ex-ab:sec:proof-of-claim-contaminte-only-one}.

\begin{claim}\label{ex-ab:claim:contaminate-only-one}
    For each $x \in X^{(0)} \oplus X^{(\ell+1)}$ and each $i \in [0, \lceil \log_{\polyup^6} (\sqrt{d}\Delta) \rceil]$, if $u \in S[i]$ satisfies $\dist(x,u) \leq \polyup \cdot \dist(x,S[i])$, then the only center in $S[i]$ that might be contaminated by $x$ is $u$. 
\end{claim}

This claim shows how to identify $W_2$.
For each $x \in X^{(0)} \oplus X^{(\ell+1)}$ and each $i \in [0, \lceil \log_{\polyup^6}(\sqrt{d}\Delta) \rceil]$, we call the nearest neighbor oracle on $S[i]$ to find $u \in S[i]$ satisfying $\dist(x,u) \leq \polyup \cdot \dist(x, S[i])$.
Then we check if $\dist(x,u) \leq \polyup^{6i+3} $.
If this happens, we see that $x$ contaminates $u$, otherwise $x$ does not contaminate any center in $S[i]$ according to \Cref{ex-ab:claim:contaminate-only-one}.
Hence, we can identify all centers in $W_2$.

\subsubsection{Implementing Part \ref{ex-ab:identify-part2}}
For this part, if we iterate over all centers in $S$ and check whether they satisfy Condition (\ref{ex-ab:condition-new}), the update time might be as large as $\Omega(k)$.
Since we make centers more robust when they violate Condition (\ref{ex-ab:condition-new}), we conclude that as long as the value of $\dist(u,S-u)$ is not increased significantly, they remain robust.
Hence, to identify if a center $u$ violates Condition (\ref{ex-ab:condition-new}), we only need to identify the centers $u$ such that the value of $\dist(u,S-u)$ is increased significantly.

\medskip
\noindent
\textbf{Required Data Structure.}
In order to do this, for every $i \in [0, \lceil \log_{\polyup^2} (\sqrt{d}\Delta) \rceil]$, we maintain the following data structure for parameter $\gamma = \polyup^{2i}$.

\begin{lemma}
\label{ex-ab:lem:bitwise-approx-NC}
Given any parameter $\gamma \geq 0$, there is a data structure $D_{\gamma}(S)$ that maintains a bit $b_{\gamma}(s, S) \in \{0, 1\}$ for every center $s \in S$, and after every update (insertion/deletion) to $S$, it reports all the centers $s$ that $b_{\gamma}(s, S)$ changes.
The following properties are satisfied, for every $s \in S$.
\begin{enumerate}[font = \bfseries]
    \item\label{ex-ab:lem:bitwise-approx-NC:1}
    If $\dist(s, S-s) \leq \gamma$, then $b_{\gamma}(s, S) = 1$.
    \item \label{ex-ab:lem:bitwise-approx-NC:2}
    If $\dist(s, S -s) > \poly(1/\epsilon) \cdot \gamma$, then $b_{\gamma}(s, S) = 0$.
\end{enumerate}
Each update in $S$ is handled in amortized $\neps$ time.
\end{lemma}

This lemma follows from combining the ANN and {\cite[Theorem 3.1]{BhattacharyaGJQ24}}.
We have already implemented ANN in \Cref{ex-ab:sec:data-structures-main}, and in
\cite{BhattacharyaGJQ24}, it is explained how we can implement the above data structure by using an ANN.

\medskip
\noindent
\textbf{Implementation Using Data Structure
$D_{\gamma}(S)$.}
Now, assume we have data structures
$D_{\polyup^{2i}}(S)$ of \Cref{ex-ab:lem:bitwise-approx-NC} (for parameter $\gamma := \polyup^{2i}$) for every $i \in [0, \lceil \log_{\polyup^2} (\sqrt{d}\Delta) \rceil]$.
According to the guarantee of this data structure, if the bit $b_{\polyup^{2i}}(u, S)$ has not changed for all values $i$, then $\dist(u,S-u)$ has not changed by more than a factor of $\polyup$.
This provides a way to implement part \ref{ex-ab:identify-part2} as follows.

After each change in $S$, we first identify all centers $u \in S$ such that for at least one $i$, the bit $b_{\polyup^{2i}}(u, S)$ has changed, and add all of these centers to a set $\calY$ that contains possible centers $u$ violating Condition (\ref{ex-ab:condition-new}).
Now, for each center $u \in \calY$, we find the smallest integer $t$ satisfying $\polyup^{6t} \geq \hat{\dist}(u,S-u)/\polyup^{20}$.
If $t[u] \geq t$, then $u$ is already $t$-robust (according to \Cref{ex-ab:claim:remains-robust}) and it satisfies Condition (\ref{ex-ab:condition-new}).
Hence, there is no need to call a $\MakeRbst$ on $u$.
Otherwise, (if $t[u] < t$), we call $\MakeRbst$ on $u$.
Note that it is also possible in this case that $u$ is actually satisfying Condition (\ref{ex-ab:cond:robust-real}), but we still call $\MakeRbst$ on $u$ since we should not miss any center violating the main Condition (\ref{ex-ab:cond:robust-real}).

\subsubsection{\texorpdfstring{Termination of $\Robustify$}{}}

After each change in the center set $S$, the distances of centers $u \in S$ to other centers might change.
This means that new centers might be inserted into $\calY$ and they might need a subsequent call to $\MakeRbst$.
Similar to \cite{BCF24}, it is possible to show that if a call $\MakeRbst$ is made on a center in $\Robustify$, it will continue to satisfy Conditions (\ref{ex-ab:condition-new}) and (\ref{ex-ab:cond:robust-real}) until the end of the call to the current $\Robustify$, and our procedure will not make another call on that center.

\begin{lemma}\label{ex-ab:lem:robustify-calls-once}
    Consider any call to \Robustify$(W')$, and suppose that it sets $w_0 \leftarrow \MakeRbst(w)$ during some iteration of the {\bf while} loop. Then in subsequent iterations of the {\bf while} loop in the same call to \Robustify$(W')$, we will {\em not} make any call to  \MakeRbst$(w_0)$.
\end{lemma}

The proof of this lemma follows from \cite{BCF24} (although we have new definitions), and we provide the proof in \Cref{part:full-alg} for the sake of completeness (see \Cref{sec:proof-of-lem-robustify-calls-once}).

\medskip
\noindent
We summarize our implementation of $\Robustify$ in \Cref{ex-ab:alg:modify:robustify}.
Throughout this procedure, after each insertion or deletion on $S$, all of the data structures ($\hat{\dist}(u,S-u)$, $D_{\polyup^{2i}}(S)$, and nearest neighbor oracle on $S[i]$) maintained related to $S$ are updated.

\begin{algorithm}[ht]
\caption{\label{ex-ab:alg:modify:robustify}
Implementation of the call to $\Robustify(W')$ at the end of an epoch.}
\begin{algorithmic}[1]
    \State
    $W_1 \gets W' \setminus S_{\init}$. \label{ex-ab:line:type1}\Comment{All new centers}
    \For{each $x \in X^{(0)} \oplus X^{(\ell+1)}$}
        \label{ex-ab:line-for-loop-points}
        \For{each $i \in [0, \lceil \log_{\polyup^6} (\sqrt{d}\Delta) \rceil]$}
            \State Let $ u \gets \hat{\dist}(x, S[i]-x)$. \;
            \If{$\dist(u,x) \leq \polyup^{6i+3}$}
                \State $W_2 \gets W_2 + u$.
                \Comment{Mark $u$ as a contaminated center}
            \EndIf
        \EndFor
    \EndFor
   \State $S \gets S_\init - (S_\init \setminus W') + (W' \setminus S_\init) $ \label{ex-ab:line:update-data-strucutres} \Comment{This line updates data structures}
   \For{each $u \in W_1 \cup W_2$}
        \State $\MakeRbst(u)$.
   \EndFor
    \While{$\calY \neq \emptyset$}\label{ex-ab:line:while-loop}\Comment{Before this line, many centers might have been inserted to $\calY$}
        \State $u \gets $pop$(\calY)$ \label{ex-ab:line:pop-robustify}\;
\State $t \gets $ Smallest integer satisfying $\polyup^{6t} \geq \hat{\dist}(u, S-u) / \polyup^{20}$ \;
        \If{$t[u] < t$}\label{ex-ab:if:condition-W3}
            \State $ \MakeRbst(u)$.\label{ex-ab:line:type3-call}
        \EndIf
    \EndWhile
  \State \Return $S_{\final} \gets S$ \label{ex-ab:line:end}.
\end{algorithmic}
\end{algorithm}

\subsection{Make-Robust}\label{ex-ab:sec:make-robust}

In this section, we introduce the subroutine $\MakeRbst$ that given $u \in S$, swaps this center with a close center $v$ which is robust (i.e.~satisfies Condition \Cref{ex-ab:condition-new} as well as \Cref{ex-ab:cond:robust-real}).
We start by computing the smallest integer $t$ satisfying $\polyup^{6t} \geq \hat{\dist}(u,S-u)/ \polyup^{14}$.
The goal is to follow the \Cref{ex-ab:def:robust-real}.
We define $x_t := u$ and by iterating over $j=t$ down to $0$, we find $x_{j-1}$ and $B_j$ satisfying the properties in \Cref{ex-ab:def:robust-real}.
We use \Cref{ex-ab:lem:one-median-on-a-ball}, and make a query on this data structure for $x=x_{j}$ and $\rho = \polyup^{6j}$.
It returns a center $c^\star \in \Deld$, together with numbers $\Hat{\cost}_{c^\star}$, $\Hat{\cost}_{x_j}$, and $b_{x_j}$.

\paragraph{Defining $B_j$.}
We let $B_{j} := B$, which is the existential ball $B$ in the guarantees of \Cref{ex-ab:lem:one-median-on-a-ball}.
Note that we do not need to know $B_j$ explicitly, and in \Cref{ex-ab:def:robust-real}, the existence of this set is sufficient. 
According to the guarantees in \Cref{ex-ab:lem:one-median-on-a-ball}, we have $b_{x_j} = w(B_j)$ and
$\ball(x_{j}, \polyup^{6j}) \subseteq B_j \subseteq \ball(x_{j}, \polyup \cdot \polyup^{6j}) $.
Since $\polyup \cdot \polyup^{6j} \leq \polyup^{6j+2}$, we conclude that $B_j = B$ satisfies the first condition in \Cref{ex-ab:def:robust-real} for $i=j$.
Now, we proceed with defining $x_{j-1}$ that satisfies the second condition.

\paragraph{Computing $x_{j-1}$.}
We compare $\Hat{\cost}_{x_j}/b_{x_j}$ and $\polyup^{6j-3}$.
There are two cases as follows.

\begin{enumerate}
\item $\Hat{\cost}_{x_j}/b_{x_j} \geq \polyup^{6j-3}$.\label{case-1} In this case, we let $x_{j-1} = x_{j}$, and show in \Cref{claim:make-robust-case1} that \Cref{ex-ab:cond:robust1} holds for $i=j$.

\item 
$\Hat{\cost}_{x_j}/b_{x_j} \leq \polyup^{6j-3}$.\label{case-2}
In this case, we compare $\Hat{\cost}_{x_j} / \polyup$ and $\Hat{\cost}_{c^\star}$.
If, $\Hat{\cost}_{x_j} / \polyup \leq \Hat{\cost}_{c^\star}$, we let $x_{j-1} = x_j$, otherwise, we let $x_{j-1} = c^\star$.
We show in \Cref{claim:make-robust-case2} that \Cref{ex-ab:cond:robust2} holds for $i = j$.
\end{enumerate}
As a result, in all cases, the second condition in \Cref{ex-ab:def:robust-real} holds for $i = j$.
Finally, it follows that the sequences 
$(x_0,x_1,\ldots, x_t)$ and
$(B_0,B_1,\ldots, B_t)$\footnote{We define  $B_0 := \ball(x_0, 1)$ at the end.} satisfy the conditions in \Cref{ex-ab:def:robust-real}.
Hence, $x_0$ is a $t$-robust center.
Below, we summarize the implementation of $\MakeRbst$ in \Cref{ex-ab:alg:make-robust}.

\begin{algorithm}[ht]
\caption{\label{ex-ab:alg:make-robust}
Implementation of a call to $\MakeRbst(u)$ at the end of an epoch.}
\begin{algorithmic}[1]
    \State $t \gets $ Smallest integer satisfying $\polyup^{6t} \geq \hat{\dist}(u, S-u) / \polyup^{14}$, \;
    $x_t \gets u$\;
    \For{$j = t$ down to $1$}
        \State Invoke \Cref{ex-ab:lem:one-median-on-a-ball} for $x = x_j$ and $\rho = \polyup^{6j}$ to get $c^\star \in \Deld$, $\Hat{\cost}_{x_j}$, $\Hat{\cost}_{c^\star}$ and $b_{x_j}$ \label{ex-ab:line:invoke-one-median-ball} \;
        \If{$\left( \Hat{\cost}_{x_j}/b_{x_j} \geq \polyup^{6j-3} \right)$ or $\left( \Hat{\cost}_{x_j}/\polyup \leq \Hat{\cost}_{c^\star} \right)$}
            \State $x_{j-1} \gets x_j$
        \Else
            \State $x_{j-1} \gets c^\star$
    \EndIf
  \EndFor
  \State $ v \gets x_0$
    \State $S \gets S - u + v$
    \State Save $t[v] := t$ together with $v$ 
\end{algorithmic}
\end{algorithm}

\section{\texorpdfstring{Analysis of Our Algorithm: Proof (Sketch) of \Cref{th:main}}{}}
\label{ex-ab:sec:alg-analysis}

In this section, we describe the analysis of our algorithm. We summarize the approximation, recourse, and update time guarantees of our algorithm with the following lemmas.

\begin{lemma}[Approximation Ratio]
    The approximation ratio of the algorithm is $\poly(1/\epsilon)$ w.h.p.
\end{lemma}

The analysis of the approximation ratio of the algorithm follows from the prior work \cite{BCF24} with deliberate changes in the parameters inside the proofs.
There are only two important considerations as follows.
First, our new definition of robust centers still satisfy the main properties of robust centers as in \cite{BCF24} (see \Cref{sec:main-properties-of-robust}).
Second, the analysis of \cite{BCF24} depends on the fact that the integrality gap of the standard LP-relaxation for $k$-median is constant. Since the same fact holds for $k$-means (see Theorem 1.1 in \cite{integrality-gap-k-means}), we can extend the analysis for our purpose.
For the sake of completeness, we provide a thorough analysis in \Cref{part:full-alg}.

\begin{lemma}[Recourse]\label{ex-ab:lem:main-recourse}
    The amortized recourse of the algorithm is at most $ \tilde{O}(1) $.
\end{lemma}

The recourse analysis of the algorithm also follows from the previous work \cite{BCF24} with deliberate changes.
There are only some considerations as follows.
First, in augmented $k$-means (\Cref{ex-ab:lem:augmented-k-means}), instead of finding a set $A$ of size exactly $a$, we find a set of size $\tilde{O}(a)$. 
This increases the recourse of the algorithm by some $\log$ factors, but the same analysis works as in \cite{BCF24}.
Second, we changed the Condition (\ref{ex-ab:cond:robust-real}) to a more restricted Condition (\ref{ex-ab:condition-new}), and we might call $\MakeRbst$ more often.
But, the same analysis as in \cite{BCF24} can be applied here to show that the total amortized recourse of $\Robustify$ is still $\tilde{O}(1)$.
For the sake of completeness, we provide a thorough analysis in \Cref{part:full-alg}.

\begin{lemma}[Update Time]\label{ex-ab:lem:main-update-time}
    The amortized update time of the algorithm is at most $ \tilde{O} (n^{\epsilon}) $.
\end{lemma}

We can \textbf{reduce the update time of our algorithm from $\tilde{O}(n^\epsilon)$ to $\tilde{O}(k^\epsilon)$} by sparsification techniques.
The reader can find a thorough description of how to achieve $\tilde{O}(k^\epsilon)$ update time in \Cref{part:from-n-to-k}.
The rest of this section is devoted to proving \Cref{ex-ab:lem:main-update-time}.

\subsection{Update Time Analysis}

\subsubsection{Maintaining the Background Data Structures}

Each update in the main data set $X$ and in the main center set $S$ is considered an update for all the background data structures that are maintained in our algorithm.
Consider a sequence of $T$ updates on $X$.
The number of updates on $S$ is bounded by $ \tilde{O}(T)$ according to \Cref{ex-ab:lem:main-recourse}.
Hence, the total number of updates on every background data structure is bounded by $\tilde{O}(T)$.

Since every background data structure has $\tilde{O}(n^\epsilon)$ update time (see \Cref{ex-ab:sec:data-structures-main}), the total time spent to maintain all of these data structures is at most $\tilde{O}(T \cdot n^\epsilon)$.
Hence, the amortized update time for maintaining all of the background data-structures is $\tilde{O}(n^\epsilon)$ (w.r.t.~the original update sequence on the main dataset $X$).

\subsubsection{Update Time of Each Step}
Now, we proceed with the update time of the main algorithm.

\paragraph{Step 1.}
The running time of \Cref{ex-ab:alg:find-ell-implementation} is dominated by the last call (say at iteration $i = i^\star$) to restricted $k$-means (\Cref{ex-ab:lem:restricted-k-means}) that takes $\tilde{O}(n^\epsilon \cdot s_{i^\star})$ time.
We have $s_{i^\star} = O(\hat{\ell}) = \tilde{O}(\ell+1)$.
Hence, the running time of this step is $\tilde{O}(n^\epsilon \cdot (\ell+1))$.

\paragraph{Step 2.}
We spend $\tilde{O}(n^\epsilon \cdot (\ell+1))$ time by calling restricted $k$-means (\Cref{ex-ab:lem:restricted-k-means}) to find the desired $S^{(0)} \subseteq S_\init$ of size $k-\ell$.

\paragraph{Step 3.}
This step is trivial, and each update can be done in $\tilde{O}(1)$ time.

\paragraph{Step 4.}
The call to augmented $k$-means (\Cref{ex-ab:lem:augmented-k-means}) for finding $A$ of size $\tilde{O}(a)$ (where $ a= \polyup^{104} \cdot (\ell+1)$) takes $\tilde{O}(n^\epsilon \cdot a) = \tilde{O}(n^\epsilon \cdot (\ell+1))$ time.
Computing $S' = S_{\init} + A$ and $ T' = S' + (X^{(\ell+1)} - X^{(0)})$ is trivial.
Then, the call to restricted $k$-means (\Cref{ex-ab:lem:restricted-k-means}) on $T'$ takes $\tilde{O}(n^\epsilon \cdot (|T'| - |W'|)) = \tilde{O}(n^\epsilon \cdot (\ell+1))$ time.
Finally, we call $\Robustify$ on $W'$.

\paragraph{Amortized Update Time.}
Apart from $\Robustify$, all of the other parts of the algorithm in this epoch, take at most $\tilde{O}(n^\epsilon (\ell+1))$ time.
As a result, the amortized update time over the length of the epoch is bounded by $\tilde{O}(n^\epsilon)$.

\subsubsection{\texorpdfstring{Running Time of $\Robustify$}{}}

Recall that the $\Robustify$ is divided into two tasks.
Searching for non-robust centers, and making them robust.
The search for non-robust centers is divided into two parts.
In part \ref{ex-ab:identify-part1}, we identify $W_1$ and $W_2$ (recall the partitioning $W' = W_1 \cup W_2 \cup W_3$).
Then, in part \ref{ex-ab:identify-part2}, we search for centers violating Condition (\ref{ex-ab:condition-new}).

\paragraph{Time Spent on Part \ref{ex-ab:identify-part1}.}
It is trivial that we can identify $W_1 = W' \setminus S_\init$ (new centers) in $\tilde{O}(|W_1|)$ time, which is $\tilde{O}(\ell+1) $.
According to the implementation in \Cref{ex-ab:alg:modify:robustify}, to find $W_2$, for each $x \in X^{(0)} \oplus X^{(\ell+1)}$ and each $i \in [0, \lceil \log_{\polyup^6} (\sqrt{d}\Delta) \rceil]$, we will call the approximate nearest neighbor oracle (ANN) on $S[i]$.
The oracle takes $\tilde{O}(n^\epsilon)$ time to answer, which concludes the total time spend to find $W_2$ is at most $\tilde{O}(n^\epsilon \cdot |X^{(0)} \oplus X^{(\ell+1)}| \cdot \lceil \log_{\polyup^6} (\sqrt{d}\Delta) \rceil) = \tilde{O}(n^\epsilon \cdot (\ell+1))$.
Hence, the amortized time spent on part \ref{ex-ab:identify-part1} for searching non-robust centers is $\tilde{O}(n^\epsilon)$.

\paragraph{Time Spent on Part \ref{ex-ab:identify-part2}.}
We need a global type of analysis for this part.
It is possible that a single call to $\Robustify$ takes a huge amount of time.
But, considering the entire sequence of $T$ updates, we show that the amortized update time of this part of the $\Robustify$ subroutine is at most $\tilde{O} (n^{\epsilon})$.
During these $T$ updates, $S$ will be updated by at most $\tilde{O}(T)$ centers.

Elements are being inserted into $\calY$ by data structures $D_{\polyup^{2i}}(S)$ (\Cref{ex-ab:lem:bitwise-approx-NC}) for each $i \in [0, \lceil \log_{\polyup^2} (\sqrt{d}\Delta) \rceil]$ after every update on the main data set $X$ or the main center set $S$.
Since this data structure has an amortized update time of $\tilde{O}(n^\epsilon)$, this \textbf{implicitly} means that the number of centers $u$ inserted into $\calY$ by data structure  $D_{\polyup^{2i}}(S)$  (i.e.~whose bit $b_{\polyup^{2i}}(u,S)$ has changed) is at most $\tilde{O}(T) \cdot \tilde{O}(n^\epsilon)$.
As a result, the total number of elements inserted into $\calY$ throughout these $T$ updates is at most $\tilde{O}((\lceil \log_{\polyup^2} (\sqrt{d}\Delta) \rceil + 1) \cdot T \cdot n^\epsilon)$.
In amortization, there are $\tilde{O}(n^\epsilon)$ centers inserted into $\calY$.

For each of these centers $u \in \calY$, during \Cref{ex-ab:line:pop-robustify} to \Cref{ex-ab:if:condition-W3}, the algorithm decides in $\tilde{O}(1)$ time whether or not a call to $\MakeRbst$ will be made on $u$.\footnote{Note that the number $\hat{\dist}(u,S-u)$ 
is maintained explicitly by the data structure in \Cref{ex-ab:lem:nearest-neighbor-distance}.}
As a result, the amortized time spend on part \ref{ex-ab:identify-part2} is
$\tilde{O}(n^\epsilon)$.

\paragraph{Time Spent on Making Centers Robust.}

According to the recourse analysis of the algorithm, the amortized number of calls to $\MakeRbst$ is bounded by $\tilde{O}(1)$.
We show in \Cref{ex-ab:sec:make-robust-time-analysis} that each call to $\MakeRbst$ takes $\tilde{O}(n^\epsilon)$ time.
Hence, the amortized time spent to make centers robust is bounded by $\tilde{O}(n^\epsilon)$.

\subsubsection{\texorpdfstring{Running Time of $\MakeRbst$}{}}
\label{ex-ab:sec:make-robust-time-analysis}

The main for loop in the implementation consists of at most $O(\log_{\polyup^2} (\sqrt{d}\Delta))$ many iterations.
Each iteration is done in $\tilde{O}(n^\epsilon)$ time according to the data structure (\Cref{ex-ab:lem:one-median-on-a-ball}) used in \Cref{ex-ab:line:invoke-one-median-ball} of \Cref{ex-ab:alg:make-robust}.
It is obvious that all other parts of the implementation take at most $\tilde{O}(1)$ time.
Hence, the total time spent to perform one call to $\MakeRbst$ is at most $\tilde{O}(n^\epsilon)$.

\subsection{Deferred Proofs}

\subsubsection{\texorpdfstring{Proof of \Cref{ex-ab:claim:contaminate-only-one}}{}}
\label{ex-ab:sec:proof-of-claim-contaminte-only-one}

We consider two cases.
\begin{enumerate}
    \item 
 $\dist(x,u) \leq \polyup^{6i + 4}$.
According to \Cref{ex-ab:claim:num-of-contamination}, there does not exists $v \in S_\init - u$ satisfying $t[v] = i$ and $\dist(x,v) \leq \polyup^{6i + 4}$. 
Since $\ball(v, \polyup^{6i+3}) \subseteq \ball(v, \polyup^{6i+4})$, we conclude that the distance between $x$ and $v$ is strictly greater than $\polyup^{6i + 3}$.
So, $x$ does not contaminate $v$.
Hence, there does not exist any $v \in S[i] - u$ contaminated by $x$, and the only center in $S[i]$ that might be contaminated by $x$ is $u$. 
\item
$\dist(x,u) > \polyup^{6i + 4}$.
In this case, we conclude that
$ \polyup^{6i + 4} < \dist(x,u) \leq \polyup \cdot \dist(x, S[i]) $.
Hence, the distance of any center $v$ in $S[i]$ to $x$ is strictly greater than $ \polyup^{6i + 4}/\polyup = \polyup^{6i + 3} $, which concludes there is no center $v \in S[i]$ contaminated by $x$.
\end{enumerate}

\subsubsection{\texorpdfstring{Correctness of $\MakeRbst$}{}}

\begin{claim}\label{claim:make-robust-case1}
    In Case \ref{case-1}, \Cref{ex-ab:cond:robust1} holds for $i=j$. 
\end{claim}
\begin{proof}
We have
$$ \polyup^{6j-3} \leq \Hat{\cost}_{x_j}/b_{x_j} \leq \polyup \cdot \cost(B_j, x_j)/w(B_j), $$
where the second inequality holds by the second guarantee in \Cref{ex-ab:lem:one-median-on-a-ball}.
It follows that 
$$\cost(B_j, x_j)/w(B_j) \geq \polyup^{6j-4} . $$ 
\end{proof}

\begin{claim}\label{claim:make-robust-case2}
In Case \ref{case-2}, \Cref{ex-ab:cond:robust2} holds for $i = j$.
\end{claim}

\begin{proof}
We have \begin{equation}\label{ex-ab:eq:try-to-satisfy-second-condition}
       \cost(B_j, x_j)/w(B_j) \leq \Hat{\cost}_{x_j}/b_{x_j} \leq \polyup^{6j-2}/\polyup \leq \polyup^{6j-2}, 
\end{equation}
where the first inequality holds by the second guarantee in \Cref{ex-ab:lem:one-median-on-a-ball}.
Now, we consider two cases.
\begin{enumerate}
    \item If $\Hat{\cost}_{x_j} / \polyup \leq \Hat{\cost}_{c^\star}$, we have $x_{j-1} = x_j$ and
    \begin{align*}
       \cost(B_j, x_{j-1}) = \cost(B_j, x_{j}) &\leq \Hat{\cost}_{x_j} \leq \polyup \cdot \Hat{\cost}_{c^\star} \\
       &\leq (\polyup)^2 \cdot \cost(B_j, c^\star) \leq (\polyup)^3 \cdot \OPT_1(B_j). 
    \end{align*}
    The first, third, and fourth inequalities hold by the second and third guarantees in \Cref{ex-ab:lem:one-median-on-a-ball}, respectively.
    Hence,
    $\cost(B_j, x_{j-1}) \leq \min \{(\polyup)^3 \cdot \OPT_1(B_j) ,\cost(B_j, x_{j})\}$.
    Combining with \Cref{ex-ab:eq:try-to-satisfy-second-condition}, we see that \Cref{ex-ab:cond:robust2} holds for $i = j$.

\item If $\Hat{\cost}_{x_j} / \polyup \geq \Hat{\cost}_{c^\star}$, we have $x_{j-1} = c^\star$ and
\begin{align*}
   \cost(B_j, x_{j-1}) = \cost(B_j,c^\star) &\leq \Hat{\cost}_{c^\star} \leq \Hat{\cost}_{x_j} / \polyup \\
   &\leq \polyup \cdot \cost(B_j, x_j)/\polyup = \cost(B_j,x_j), 
\end{align*}
where the first and third inequalities hold by the second guarantee in \Cref{ex-ab:lem:one-median-on-a-ball}.
Moreover,
$ \cost(B_j, x_{j-1}) = \cost(B_j, c^\star) \leq \polyup \cdot \OPT_1(B_j) $, by the third guarantee of \Cref{ex-ab:lem:one-median-on-a-ball}.
Combining with \Cref{ex-ab:eq:try-to-satisfy-second-condition}, we see that \Cref{ex-ab:cond:robust2} holds for $i=j$.
\end{enumerate}
\end{proof}

\newpage

\part{Data Structures in High Dimensional Euclidean Spaces}
\label{part:data-structures}

\paragraph{Roadmap.}
In this part of the paper, we provide several dynamic data structures that are useful in subsequent parts.
We begin by summarizing the formal statement of these data structures in \Cref{sec:data-structures-lemmas}.
As a crucial foundation of our data structures, we design a time-efficient consistent hashing scheme in \Cref{sec:consistent}.
Then, we elaborate on the construction and analysis of each data structure in \Cref{sec:range-query,sec:approximate-assignment,sec:restricted_kmeans,sec:augmented_kmeans}.

\paragraph{Notation.}
Throughout this part, we fix a small constant $0 < \epsilon < 1$, 
denote by $X \subseteq \Deld$ and $w: X \to \R_{\ge 0}$ a dynamic weighted point set and its weight function, respectively, and denote by $S \subseteq \Deld$ a dynamic (unweighted) center set.
Without loss of generality, we assume that both $X$ and $S$ are initially empty.
We use $\tO(\cdot)$ to hide the $\poly(d\epsilon^{-1}\log (\Delta))$ factors.

\section{Main Data Structures}
\label{sec:data-structures-lemmas}

In this section, we summarize the main data structures that will be used later in \Cref{part:full-alg}.
All of them maintain both the dynamic weighted point set $X \subseteq \Deld$ and the dynamic unweighted center set $S \subseteq \Deld$;
we thus omit these details from their statements for simplicity.

The following \Cref{lem:ANN-query,lem:bitwise-approx-NC,lem:nearest-neighbor-distance,lem:1-means} are established in \Cref{sec:range-query}, where we design a unified range query structure for handling a class of ``mergeable'' query problems.
Specifically, \Cref{lem:ANN-query,lem:bitwise-approx-NC,lem:nearest-neighbor-distance} address a set of \emph{approximate nearest neighbor} (ANN) problems, while \Cref{lem:1-means} concerns the $1$-means problem over approximate balls. 
Here, an ``approximate ball'' refers to a subset that is sandwiched between two balls with the same center and (up to a constant) roughly the same radii.

\begin{lemma}[ANN Oracle; see \Cref{lem:ANN-oracle}]
\label{lem:ANN-query}
There exists a data structure that, for every query $x \in \Deld$, returns a center $s \in S - x$ such that $\dist(x, s) \leq O(\epsilon^{-3/2}) \cdot \dist(x, S - x)$. 
The data structure works against an adaptive adversary and, with probability at least $1 - \frac{1}{\poly(\Delta^d)}$, has update and query time $\tO(2^{\epsilon d})$.
\end{lemma}

\begin{lemma}[ANN Indicators; see \Cref{lem:ANN-indicator}]
\label{lem:bitwise-approx-NC}
For any parameter $\gamma \geq 0$, there exists a data structure that maintains, for every $s \in S$, a bit $b_{\gamma}(s, S) \in \{0, 1\}$ with the following guarantee.
\begin{enumerate}[font = \bfseries]
    \item If $\dist(s, S - s) \leq \gamma$, then $b_{\gamma}(s, S) = 1$.
    \item 
    If $\dist(s, S - s) > \Omega(\epsilon^{-3/2}) \cdot \gamma$, then $b_{\gamma}(s, S) = 0$.
\end{enumerate}
After every update to $S$, the data structure reports all the centers $s$ that $b_{\gamma}(s, S)$ changes.
The data structure works against an adaptive adversary and, with probability at least $1 - \frac{1}{\poly(\Delta^d)}$, has amortized update time $\tO(2^{\epsilon d})$.
\end{lemma}
\begin{lemma}[ANN Distance; see \Cref{lem:ANN-distance}]
\label{lem:nearest-neighbor-distance}
The exists a data structure that maintains, for every $s\in S$, a value $\Hat{\dist}(s, S - s)$ such that 
\begin{align*}
    \dist(s, S - s) \le \Hat{\dist}(s, S - s) \le O(\epsilon^{-3/2})\cdot \dist(s, S - s).
\end{align*}
The data structure works against an adaptive adversary and, with probability at least $1 - \frac{1}{\poly(\Delta^d)}$, has amortized update time $\tO(2^{\epsilon d})$.
\end{lemma}

\begin{lemma}[$1$-Means over Approximate Balls; see \Cref{cor:1-means-on-ball}]
\label{lem:1-means}
There exists a data structure that, for every query point $x \in \Deld$ and radius $r \ge 0$, returns a real number $b_x \ge 0$, a center $c^{\star} \in \Deld$, and two estimates $\hat{\cost}_{c^{\star}}$ and $\hat{\cost}_x$, such that there exists a subset $B_x \subseteq X$ satisfying the following conditions.
\begin{enumerate}[font = \bfseries]
    \item\label{cor:1-means:1}
    $\ball(x, r) \cap X \subseteq B_x \subseteq \ball(x, O(\epsilon^{-3/2}) \cdot r) \cap X$ and $b_{x} = w(B_{x})$.
    
    \item\label{cor:1-means:2}
    $\cost(B_{x}, c) \leq \Hat{\cost}_{c} \leq O(1) \cdot \cost(B_{x}, c)$, for either $c \in \{c^{\star},\ x\}$.
    
    \item\label{cor:1-means:3}
    $\cost(B_{x}, c^{\star}) \leq O(1) \cdot \OPT_{1}(B_{x})$.
\end{enumerate}
The data structure works against an adaptive adversary and, with probability at least $1 - \frac{1}{\poly(\Delta^d)}$, has amortized update and query time $\tO(2^{\epsilon d})$.
\end{lemma}

The following \Cref{lem:restricted-k-means-main,lem:augmented-k-means-main} address two crucial subproblems -- restricted $k$-means and augmented $k$-means -- respectively, which play a central role in our dynamic algorithm described in \Cref{part:full-alg}. 
We present the data structures of \Cref{lem:restricted-k-means-main,lem:augmented-k-means-main} in \Cref{sec:restricted_kmeans,sec:augmented_kmeans}, respectively, both of which are built on a common data structure that (implicitly) maintains an \emph{approximate assignment} from the point set $X$ to the center set $S$, as described in \Cref{sec:approximate-assignment}.

\begin{lemma}[Restricted $k$-Means; see \Cref{lem:restricted-clustering}] 
\label{lem:restricted-k-means-main}
There exists a data structure that, on query of an integer $1 \le r\le |S|$, computes a size-$r$ subset $R \subseteq S$ in time $\tO(2^{\epsilon d} r +r^{1 + \epsilon})$, such that the following holds with probability $1 - \frac{1}{\poly(\Delta^d)}$:
\begin{align*}
    \cost(X, S - R)
    ~\le~ O(\epsilon^{-13})\cdot \OPT_{|S| - r}^{S}(X),
\end{align*}
where $\OPT_{|S| - r}^{S}(X)\eqdef \min_{S'\subseteq S:|S'| = |S| - r}\cost(X,S')$.
The data structure works against an adaptive adversary and, with probability at least $1 - \frac{1}{\poly(\Delta^d)}$, has amortized update time $\tO(2^{\epsilon d})$.
\end{lemma}

\begin{lemma}[Augmented $k$-Means; see \Cref{alg:augmented}]\label{lem:augmented-k-means-main}
There exists a data structure that, on query of an integer $a \ge 1$, computes a size-$\Theta(a \cdot \epsilon^{-6} \cdot d \log\Delta)$ subset $A \subseteq \Deld$ in time $\tilde{O}(2^{\epsilon d} \cdot a)$, such that the following holds with probability $1 - \frac{1}{\poly(\Delta^d)}$:
\begin{equation*}
    \cost(X, S + A)
    ~\le~ O(1)\cdot \min_{A' \subseteq \Deld: |A'| = a}  \cost(X, S + A').
\end{equation*}
The data structure works against an adaptive adversary and, with probability at least $1 - \frac{1}{\poly(\Delta^d)}$, has amortized update time $\tO(2^{\epsilon d})$.
\end{lemma}

\section{Efficient Consistent Hashing}
\label{sec:consistent}

In this section, we describe our construction of an efficient consistent hashing scheme. 
We begin by introducing the definition of efficient consistent hashing and then state our main result.

\begin{definition}[Efficient Consistent Hashing]
\label{def:efficient-consistent-hashing}
A mapping $\varphi:\Deld \to \Phi$ is an \emph{efficient $(\Gamma, \Lambda, \rho)$-hash}, for some codomain $\Phi$, if 
\begin{enumerate}[font = \bfseries]
    \item\label{def:consistent:closeness}
    (diameter) $\diam(\varphi^{-1}(z)) \le \rho$, for every hash value $z \in \varphi(\Deld)$.

    \item\label{def:consistent:consistency}
    (consistency) $|\varphi(\ball(x, \frac{\rho}{\Gamma}))| \le \Lambda$, for every point $x \in \Deld$.

    \item\label{def:consistent:efficiency}
    (efficiency) For every point $x \in [\Delta]^d$,    
    both the hash value $\varphi(x)$ and a size-$(\le \Lambda)$ subset $\Phi_{x} \subseteq \Phi$ satisfying $\varphi(\ball(x, \frac{\rho}{\Gamma})) \subseteq \Phi_{x} \cap \varphi([\Delta]^d) \subseteq \varphi(\ball(x, 2\rho))$ can be computed in time $\tO(\Lambda)$.
\end{enumerate}
\end{definition}

\begin{lemma}[Efficient Consistent Hashing]
\label{lem:efficient-consistent-hashing}
For any $\Gamma \ge 2\sqrt{2\pi}$, $\Lambda = 2^{{\Theta(d/\Gamma^{2/3})}} \cdot \poly(d \log \Delta)$, and $\rho > 0$,
there exists a random hash $\varphi$ on $[\Delta]^d$, such that
it can be sampled in $\poly(d \log \Delta)$ time,
and with probability at least $1 - \frac{1}{\poly(\Delta^d)}$,
$\varphi$ is an efficient $(\Gamma, \Lambda, \rho)$-hash.
\end{lemma}

\subsection{\texorpdfstring{Proof of \Cref{lem:efficient-consistent-hashing}}{}}

Our construction of an efficient consistent hash builds on a previous construction from~\cite{random-shift}, which provides a \emph{weaker} guarantee of consistency that holds in expectation. In addition, we leverage the special structure of their hash function -- based on \emph{randomly shifted grids} -- and show that the set of hash values over any given ball can be computed efficiently (see Item~\ref{lem:weakly-consistent:range} below). We summarize these results in the following lemma. 

\begin{lemma}[{Hashing with Weaker Consistency Guarantees \cite{random-shift}}]
\label{lem:weakly-consistent}
Given any $\Gamma \ge 2\sqrt{2\pi}$, $\Lambda = 2^{{\Theta(d/\Gamma^{2/3})}}\cdot \poly(d \log \Delta)$, and $\rho > 0$,
there exists a random hash $\Tilde{\varphi}:\Deld \to \Z^{d}$, such that
\begin{enumerate}[font = \bfseries]
    \item\label{lem:weakly-consistent:closeness}
    $\diam(\Tilde{\varphi}^{-1}(z)) \le \rho$, for every hash value $z \in \Z^{d}$.

    \item\label{lem:weakly-consistent:consistency}
    $\E\big[|\Tilde{\varphi}(\ball_{\R^{d}}(x, \frac{2\rho}{\Gamma} ))|\big] \le 2^{O(d/\Gamma^{2/3})}\cdot \poly(d)$, for every point $x \in \R^{d}$.

    \item\label{lem:weakly-consistent:value}
    Given a point $x \in [\Delta]^d$,
    the hash value $\Tilde{\varphi}(x)$ can be evaluated in time $\poly(d)$.

    \item\label{lem:weakly-consistent:range}
    Given a point $x\in \Deld$ and a radius $r \ge 0$, the set $\Tilde{\varphi}(\ball(x, r))$ can be computed in time $|\Tilde{\varphi}(\ball(x,r))|\cdot \poly(d)$.
\end{enumerate}
\end{lemma}
\begin{proof}
\Cref{lem:weakly-consistent:closeness,lem:weakly-consistent:consistency,lem:weakly-consistent:value} are explicitly stated in \cite{random-shift}. We now prove \Cref{lem:weakly-consistent:range}.

We review the construction from \cite{random-shift}: 
Let $v \sim \Unif([0, r / \sqrt{d}]^{d})$ be a uniformly random real vector. The hash function is defined as $\tilde{\varphi}: x \mapsto \lfloor \frac{x + v}{r / \sqrt{d}} \rfloor$ for $x \in [\Delta]^d$.
Herein, the integer vector $\lfloor z \rfloor \eqdef (\lfloor z_1 \rfloor, \dots, \lfloor z_d \rfloor) \in \Z^{d}$ coordinatewise rounds down a real vector $z = (z_1, \dots, z_d) \in \R^{d}$.

Consider such an (infinite) grid graph $G$: its vertex set is $\Z^{d}$, and every vertex $z \in \Z^{d}$ is adjacent to the $2d$ other vertices $z' \in \Z^{d}$ with $\|z - z'\| = 1$, i.e., exactly one coordinate differs by $1$.

The subgraph induced in $G$ by the subset $\Tilde{\varphi}(\ball(x, r)) = \{z \in \Z^{d}: \dist(x,\Tilde{\varphi}^{-1}(z)) \le r\}$ is connected, shown as follows.
For notational clarification, let $z^{x} = \Tilde{\varphi}(x)$ redenote the ``sink'' vertex.
Start from a different ``source'' vertex $z \in \Tilde{\varphi}(\ball(x, r)) \setminus \{z^{x}\}$\footnote{If otherwise $\Tilde{\varphi}(\ball(x, r)) \setminus \{z^{x}\} = \emptyset$, the singleton $\Tilde{\varphi}(\ball(x, r)) = \{z^{x}\} = \{\Tilde{\varphi}(x)\}$ is trivially connected.}
-- without loss of generality, $z_{i} < z^{x}_{i}$ for some coordinate $i \in [d]$ -- and consider the movement to one of its neighbors $z' = (z_{i} + 1, z_{-i})$.
By construction, the hash preimages $\Tilde{\varphi}^{-1}(z)$ and $\Tilde{\varphi}^{-1}(z')$ are hypercubes of the form
\begin{align*}
    \Tilde{\varphi}^{-1}(z)
    &\textstyle ~=~ \prod_{j \in [d]} \Tilde{\varphi}_{j}^{-1}(z)
    ~=~ \prod_{j \in [d]} \big[z_{j} \cdot r / \sqrt{d} - v_{j}, (z_{j} + 1) \cdot r / \sqrt{d} - v_{j}\big),\\
    \Tilde{\varphi}^{-1}(z')
    &\textstyle ~=~ \prod_{j \in [d]} \Tilde{\varphi}_{j}^{-1}(z)
    ~=~ \prod_{j \in [d]} \big[z'_{j} \cdot r / \sqrt{d} - v_{j}, (z'_{j} + 1) \cdot r / \sqrt{d} - v_{j}\big).
\end{align*}
Since $z_{i} < z^{x}_{i}$ and $z' = (z_{i} + 1, z_{-i})$, it is easy to verify that $\dist(x, \Tilde{\varphi}^{-1}(z')) \le \dist(x, \Tilde{\varphi}^{-1}(z)) \le r$, where the last step holds since $z \in \Tilde{\varphi}(\ball(x, r))$; we thus have $z' \in \Tilde{\varphi}(\ball(x, r))$.
Repeating such movement leads to a source-to-sink path $z$--$z'$--$\dots$--$z^{x}$ in which all vertices belong to $\Tilde{\varphi}(\ball(x, r))$.
Then, the arbitrariness of the source vertex $z \in \Tilde{\varphi}(\ball(x, r)) \setminus \{z^{x}\}$ implies the connectivity.

To find the entire subset $\Tilde{\varphi}(\ball(x, r))$, given the connectivity, we can start from $z^{x} = \varphi(x)$ and run a depth-first search; the running time is dominated by the product of two factors: 
\begin{itemize}
    \item The number of vertices $z \in \Z^{d}$ visited by the depth-first search. Such a vertex either belongs to $\Tilde{\varphi}(\ball(x, r))$ or is a neighbor thereof, thus at most $|\Tilde{\varphi}\left(\ball(x, r)\right)| \cdot (1 + 2d)$ many ones.
    
    \item The time to compute a distance $\dist(x, \Tilde{\varphi}^{-1}(z))$ and compare it with the radius $r$. This can be done in time $O(d)$, precisely because the hash preimage $\Tilde{\varphi}^{-1}(z) = \prod_{j \in [d]} \Tilde{\varphi}_{j}^{-1}(z)$ is a hypercube and thus $\dist(x, \Tilde{\varphi}^{-1}(z)) = (\sum_{j \in [d]} \dist^2(x_{j}, \Tilde{\varphi}_{j}^{-1}(z)))^{1 / 2}$.
\end{itemize}
In combination, the total running time is $|\Tilde{\varphi}\left(\ball(x, r)\right)| \cdot \poly(d)$. This finishes the proof.
\end{proof}

We now elaborate on our construction for achieving \Cref{lem:efficient-consistent-hashing}.

\subsection*{The Construction}

Our consistent hash function $\varphi$ builds on $C$ {\em independent} weakly consistent hash functions $\{\Tilde{\varphi}_{c}\}_{c \in [C]}$ (\Cref{lem:weakly-consistent}), with a {\em sufficiently large} integer $C = \Theta(d \log\Delta)$; for notational clarification, we refer to the indexes of $\{\Tilde{\varphi}_{c}\}_{c \in [C]}$ as {\em colors} $c \in [C]$.

Following \Cref{lem:weakly-consistent} and using the {\em sufficiently large} integer $\Lambda = 2^{{\Theta(d/\Gamma^{2/3})}} \cdot \poly(d \log\Delta)$ chosen in the statement of \Cref{lem:efficient-consistent-hashing}, every $\Tilde{\varphi}_{c}$ satisfies that:
\begin{itemize}
    \item $\diam(\Tilde{\varphi}_{c}^{-1}(z)) \le \rho$, for every hash value $z \in \Z^{d}$.
    \Comment{\Cref{lem:weakly-consistent:closeness} of \Cref{lem:weakly-consistent}}
    
    \item $\E\big[|\Tilde{\varphi}_{c}(\ball(x, \frac{2\rho}{\Gamma}))|\big] \le \frac{\Lambda}{2C}$, for every point $x \in \Deld$.
\Comment{\Cref{lem:weakly-consistent:consistency} of \Cref{lem:weakly-consistent}}
\end{itemize}

The following \Cref{claim:color-existence} measures the likelihood of a certain random event, which turns out to be the {\em success probability} of the construction of our consistent hash function $\varphi$.

\begin{claim}[Success Probability]
\label{claim:color-existence}
Over the randomness of the weakly consistent hash functions $\{\Tilde{\varphi}_{c}\}_{c \in [C]}$, the following holds with probability $1 - \frac{1}{\poly(\Delta^{d})}$.
\begin{align}
    & \forall x \in \Deld,\
    \exists c \in [C]:
    && |\Tilde{\varphi}_{c}(\ball(x, \tfrac{2\rho}{\Gamma}))|
    \le \tfrac{\Lambda}{C}.
    \label{eq:color-existence}
\end{align}
\end{claim}

\begin{proof}
Let $\mathcal{E}_{x, c}$ denote the event that \Cref{eq:color-existence} holds for a specific pair $(x, c) \in \Deld \times [C]$.
By \Cref{lem:weakly-consistent:consistency} of \Cref{lem:weakly-consistent} and Markov's inequality, the event $\mathcal{E}_{x, c}$ occurs with probability\\
$\Pr[\mathcal{E}_{x, c}]
= \Pr\big[|\Tilde{\varphi}_{c}(\ball(x, \frac{2\rho}{\Gamma}))| \le \frac{\Lambda}{C}\big]
\ge \Pr\big[|\Tilde{\varphi}_{c}(\ball(x, \frac{2\rho}{\Gamma}))| \le 2\cdot \E\big[|\Tilde{\varphi}_{c}(\ball(x, \frac{2\rho}{\Gamma} ))|\big]\big]
\ge \frac{1}{2}$.\\
Following the independence of $\{\Tilde{\varphi}_{c}\}_{c \in [C]}$ and the union bound, when $C = \Theta(d \log\Delta)$ is large enough, \Cref{eq:color-existence} holds with probability $\Pr[\wedge_{x \in \Deld} \vee_{c \in [C]} \mathcal{E}_{x, c}] \ge 1 - \Delta^{d} \cdot \frac{1}{2^{C}} = 1 - \frac{1}{\poly(\Delta^{d})}$.

This finishes the proof.
\end{proof}

Since \Cref{eq:color-existence} holds with high enough probability, we would safely assume so hereafter.
Thus, every point $x \in \Deld$ has at least one color $c \in [C]$ satisfying \Cref{eq:color-existence}; we designate the {\em color of this point} $c_{x} \in [C]$ as, say, the smallest such $c \in [C]$.
(We need not determine the colors $\{c_{x}\}_{x \in \Deld}$ at this moment -- just their existence is sufficient.)

Based on the {\em weakly consistent} hash functions $\{\Tilde{\varphi}_{c}\}_{c \in [C]}$, we define our {\em consistent} hash function $\varphi: \Deld \to \Phi$ as follows, where the output space $\Phi \eqdef [C] \times \Z^{d}$; following \Cref{lem:weakly-consistent}, the entire construction takes time $C \cdot O(d) = \poly(d \log\Delta)$.
\begin{align}
    \label{eq:def-our-hashing}
    & \varphi(x) ~\eqdef~ (c_{x}, \Tilde{\varphi}_{c_{x}}(x)),
    && \forall x \in \Deld.
\end{align}

\subsection*{The Performance Guarantees}

Below, we show that the constructed $\varphi$ is an efficient $(\Gamma, \Lambda, \rho)$-hash (\Cref{def:efficient-consistent-hashing}), conditioned on the high-probability event given in \Cref{eq:color-existence}.

First, the diameter bound and the consistency properties (\Cref{def:consistent:closeness,def:consistent:consistency} of \Cref{def:efficient-consistent-hashing}) follow from \Cref{claim:consistent-hashing:closeness,claim:consistent-hashing:consistency}, respectively.

\begin{claim}[Diameter]
\label{claim:consistent-hashing:closeness}
$\diam(\varphi^{-1}(c, z)) \le \rho$, for every hash value $(c, z) \in \varphi(\Deld)$.
\end{claim}

\begin{proof}
By construction, our consistent hash function $\varphi(x) = (c_{x}, \Tilde{\varphi}_{c_{x}}(x))$ satisfies $\diam(\varphi^{-1}(c, z)) \le \diam(\Tilde{\varphi}_{c}^{-1}(z)) \le \rho$, where the last step applies \Cref{lem:weakly-consistent:closeness} of \Cref{lem:weakly-consistent}.
This finishes the proof.
\end{proof}

\begin{claim}[Consistency]
\label{claim:consistent-hashing:consistency}
$|\varphi(\ball(x, \frac{\rho}{\Gamma}))| \le \Lambda$, for every point $x \in \Deld$.
\end{claim}

\begin{proof}
By construction, our consistent hash function $\varphi(x) = (c_{x}, \Tilde{\varphi}_{c_{x}}(x))$ induces a natural partition of the input space $\Deld = \cup_{c \in [C]} G_{c}$, where $G_{c} \eqdef \{x \in \Deld \mid c_{x} = c\}$, $\forall c \in [C]$.
In this manner, we have $|\varphi(\ball(x, \frac{\rho}{\Gamma}))| \le \sum_{c \in [C]} |\varphi(\ball(x, \frac{\rho}{\Gamma}) \cap G_{c})|$, so it suffices to show that
\begin{align*}
    & \forall c \in [C]:
    && |\varphi(\ball(x, \tfrac{\rho}{\Gamma}) \cap G_{c})| \le \tfrac{\Lambda}{C}.
\end{align*}
Consider a specific color $c \in [C]$ and a specific point $y \in \ball(x, \frac{\rho}{\Gamma}) \cap G_{c}$.\footnote{If otherwise $\ball(x, \frac{\rho}{\Gamma}) \cap G_{c} = \emptyset$, we trivially have $|\varphi(\ball(x, \frac{\rho}{\Gamma}) \cap G_{c})| = 0 \le \frac{\Lambda}{C}$.}
Note that:\\
(i)~$\varphi(\ball(x, \frac{\rho}{\Gamma}) \cap G_{c}) \subseteq 
\{(c,z): z\in \Tilde{\varphi}_{c}(\ball(x, \frac{\rho}{\Gamma}))\}$, by construction.\\
(ii)~$\Tilde{\varphi}_{c}(\ball(x, \frac{\rho}{\Gamma})) \subseteq \Tilde{\varphi}_{c}(\ball(y, \frac{2\rho}{\Gamma}))$, since $y \in \ball(x, \frac{\rho}{\Gamma})$ and thus $\ball(x, \frac{\rho}{\Gamma}) \subseteq \ball(y, \frac{2\rho}{\Gamma})$.\\
(iii)~$|\Tilde{\varphi}_{c}(\ball(y, \frac{2\rho}{\Gamma}))| \le \frac{\Lambda}{C}$, following \Cref{claim:color-existence}, \Cref{eq:def-our-hashing}, and that $c_{y} = c \iff y \in G_{c}$.\\
In combination, we directly conclude with $|\varphi(\ball(x, \frac{\rho}{\Gamma}) \cap G_{c})| \le \frac{\Lambda}{C}$, as desired.

This finishes the proof.
\end{proof}

Second, the efficiency guarantee (\Cref{def:consistent:efficiency} of \Cref{def:efficient-consistent-hashing}) directly follows from a combination of \Cref{claim:consistent-hashing:value,claim:consistent-hashing:image}.

\begin{claim}[Hash Value Query]
\label{claim:consistent-hashing:value}
Given a point $x \in \Deld$, the hash value $\varphi(x)$ can be computed in time $\tO(\Lambda)$.
\end{claim}

\begin{proof}
By definition (\Cref{eq:def-our-hashing}), we have $\varphi(x) = (c_{x}, \Tilde{\varphi}_{c_{x}}(x))$.

First, we can find the color $c_{x} \in [C]$ -- the smallest $c \in [C]$ satisfying \Cref{eq:color-existence}, and we have assumed its existence -- as follows.
Run such a for-loop $c = 1, 2, \dots, C$:
\begin{quote}
    Try the depth-first search for \Cref{lem:weakly-consistent:range} of \Cref{lem:weakly-consistent} to compute $\Tilde{\varphi}_{c}(\ball(x, \frac{2\rho}{\Gamma}))$ -- it starts with the empty set $\emptyset$ and adds elements in $\Tilde{\varphi}_{c}(\ball(x, \frac{2\rho}{\Gamma}))$ one by one -- but stop it once we are aware of violation $|\Tilde{\varphi}_{c}(\ball(x, \frac{2\rho}{\Gamma}))| > \frac{\Lambda}{C}$ of \Cref{eq:color-existence}.\\
    If it answers $\Tilde{\varphi}_{c}(\ball(x, \frac{2\rho}{\Gamma}))$ before stop -- namely satisfaction $|\Tilde{\varphi}_{c}(\ball(x, \frac{2\rho}{\Gamma}))| \le \frac{\Lambda}{C}$ of \Cref{eq:color-existence} -- return $c_{x} \gets c$ as the color (and terminate the for-loop $c = 1, 2, \dots, C$).
\end{quote}
Following \Cref{lem:weakly-consistent} and the stop condition ``$|\Tilde{\varphi}_{c}(\ball(x, \frac{2\rho}{\Gamma}))| > \frac{\Lambda}{C}$'', the depth-first search for every single $c = 1, 2, \dots, C$ stops in time $\frac{\Lambda}{C} \cdot \poly(d)$, thus the total running time $\Lambda \cdot \poly(d) = \tO(\Lambda)$.

Next, provided the above color $c_{x} \in [C]$, we can simply compute $\Tilde{\varphi}_{c_{x}}(x) \in \Z^{d}$ using \Cref{lem:weakly-consistent} in time $O(d) = \tO(1)$.
This finishes the proof.
\end{proof}

\begin{claim}[Range Query]
\label{claim:consistent-hashing:image}
\begin{flushleft}
Given a point $x \in \Deld$, a size-$(\le \Lambda)$ subset $\Phi_{x} \subseteq \Phi$ such that $\varphi(\ball(x, \tfrac{\rho}{\Gamma}))
\subseteq \Phi_{x} \cap \varphi(\Deld)
\subseteq \varphi(\ball(x, 2\rho))$ can be computed in time $\tO(\Lambda)$.
\end{flushleft}
\end{claim}

\begin{proof}
The algorithm for constructing the subset $\Phi_{x}$ is very similar to the one for \Cref{claim:consistent-hashing:value}.
Start with the empty set $\Phi_{x} \gets \emptyset$ and run such a for-loop $c = 1, 2, \dots, C$:
\begin{quote}
    Try the depth-first search for \Cref{lem:weakly-consistent:range} of \Cref{lem:weakly-consistent} to compute $\Tilde{\varphi}_{c}(\ball(x, \frac{\rho}{\Gamma}))$ -- it starts with the empty set $\emptyset$ and adds elements in $\Tilde{\varphi}_{c}(\ball(x, \frac{\rho}{\Gamma}))$ one by one -- but stop it once we are aware of $|\Tilde{\varphi}_{c}(\ball(x, \frac{\rho}{\Gamma}))| > \frac{\Lambda}{C}$.\\
    If it answers $\Tilde{\varphi}_{c}(\ball(x, \frac{\rho}{\Gamma}))$ before stop -- namely $|\Tilde{\varphi}_{c}(\ball(x, \frac{\rho}{\Gamma}))| \le \frac{\Lambda}{C}$ -- augment the subset $\Phi_{x} \gets \Phi_{x} \cup \{(c, z)\}_{z \in \Tilde{\varphi}_{c}(\ball(x, \frac{\rho}{\Gamma}))}$ (and move on to the next iteration of the for-loop $c = 1, 2, \dots, C$).    
\end{quote}
Following \Cref{lem:weakly-consistent} and the stop condition ``$|\Tilde{\varphi}_{c}(\ball(x, \frac{\rho}{\Gamma}))| > \frac{\Lambda}{C}$'', every iteration $c = 1, 2, \dots, C$ will stop in time $\frac{\Lambda}{C} \cdot \poly(d)$ and augment the subset $\Phi_{x}$ by at most $\frac{\Lambda}{C}$, thus the total running time $\Lambda \cdot \poly(d) = \tO(\Lambda)$ and the total size $|\Phi_{x}| \le \Lambda$.

It remains to establish that $\varphi(\ball(x, \tfrac{\rho}{\Gamma}))
\subseteq \Phi_{x} \cap \varphi(\Deld)
\subseteq \varphi(\ball(x, 2\rho))$.

First, consider a specific hash value $(c,z) \in \varphi(\ball(x, \frac{\rho}{\Gamma}))$ -- there exists a point $y \in \ball(x, \frac{\rho}{\Gamma})$ such that $\varphi(y) = (c, z) \implies c_{y} = c$. Note that:\\
(i)~$|\Tilde{\varphi}_{c}(\ball(x, \tfrac{\rho}{\Gamma}))|
\le |\Tilde{\varphi}_{c}(\ball(y, \tfrac{2\rho}{\Gamma}))|
\le \tfrac{\Lambda}{C}$, given $\ball(x, \tfrac{\rho}{\Gamma}) \subseteq \ball(y, \tfrac{2\rho}{\Gamma})$ and \Cref{eq:color-existence}.\\
(ii)~$z \in \Tilde{\varphi}_{c}(\ball(x, \frac{\rho}{\Gamma}))$, given our construction of the hash function $\varphi$ (\Cref{eq:color-existence,eq:def-our-hashing}).\\
Thus, our subset $\Phi_{x}$ by construction will include the considered hash value $\Phi_{x} \ni (c, z)$.\\
Then, the arbitrariness of $(c, z) \in \varphi(\ball(x, \frac{\rho}{\Gamma}))$ implies that $\varphi(\ball(x, \frac{\rho}{\Gamma})) \subseteq \Phi_{x} \cap \varphi(\Deld)$.

Second, consider a specific hash value $(c, z) \in \Phi_{x} \cap \varphi(\Deld)$ -- there exists a point $y \in \Deld$ such that $\varphi(y) = (c, z) \implies y \in \Tilde{\varphi}_{c}^{-1}(z)$ (\Cref{eq:def-our-hashing}); we deduce from the triangle inequality that
\begin{align*}
    \dist(x, y)
    ~\le~ \dist(x, \Tilde{\varphi}_{c}^{-1}(z)) + \diam( \Tilde{\varphi}_{c}^{-1}(z))
    ~\le~ \tfrac{\rho}{\Gamma} + \rho
    ~\le~ 2\rho.
\end{align*}
Here, the second step applies $\dist(x,\Tilde{\varphi}_{c}^{-1}(z)) \le \frac{\rho}{\Gamma} \impliedby z \in \Tilde{\varphi}_{c}(\ball(x, \frac{\rho}{\Gamma}))$ (by our construction of the subset $\Phi_{x}$) and $\diam( \Tilde{\varphi}_{c}^{-1}(z)) \le \rho$ (\Cref{lem:weakly-consistent:closeness} of \Cref{lem:weakly-consistent}).\\
Then, the arbitrariness of $(c,z)\in\Phi_{x} \cap \varphi(\Deld)$ implies that $\Phi_{x} \cap \varphi(\Deld) \subseteq \varphi(\ball(x, 2\rho))$.

This finishes the proof.
\end{proof}

 \section{Approximate Range Query Structure}
\label{sec:range-query}

We present a range query data structure in \Cref{lem:range-query}.
A query problem can be formalized as a pair $(U, f)$, where $U$ is the \emph{candidate solution space}, and $f$ is a mapping from (input) a weighted point set $X \subseteq \Deld$ to (output) a set of valid solutions $f(X) \subseteq U$.

\begin{lemma}[Range Query Structure]
\label{lem:range-query}
For a query problem $(U, f)$ and a parameter $r > 0$, assume access to an efficient $(\Gamma, \Lambda, r\cdot \Gamma)$-hash, and suppose there exists a data structure $\mathcal{A}$ that maintains a dynamic weighted set $X \subseteq \Deld$, and returns a solution $Y \in f(X)$ upon query.

Then, there exists a data structure that maintains a dynamic weighted set $X \subseteq \Deld$, and for every query point $x \in \Deld$, returns  $Y_1, \dots, Y_t\in U$ for some $1\le t \le \Lambda$, with the following guarantees: 
There exists a collection of disjoint subsets $P_1,\dots,P_{t}\subseteq X$ such that 
\begin{itemize}
    \item $\forall i\in [t]$, $Y_i\in f(P_i)$, and

    \item $\ball(x, r) \cap X \subseteq \bigcup_{i\in [t]} P_i \subseteq \ball(x, 3\Gamma \cdot r) \cap X$.
\end{itemize}
The update time is that of $\calA$ plus $\tO(\Lambda)$, and the query time is $\tO(\Lambda)$ times that of $\calA$. The data structure works against an adaptive adversary if $\mathcal{A}$ does.
\end{lemma}

We provide the proof of \Cref{lem:range-query} in \Cref{sec:proof-range-query}. 

\Cref{lem:range-query} is particularly useful for \emph{mergeable} query problems, which means that for any two disjoint weighted point sets $X_1, X_2 \subseteq \Deld$ and their respective valid solutions $Y_1 \in f(X_1)$ and $Y_2 \in f(X_2)$, one can compute from $\{Y_1, Y_2\}$ a valid solution $Y \in f(X_1 + X_2)$ for the union.
For such problems, the output $\{Y_1, \dots, Y_t\}$ returned by \Cref{lem:range-query} can be integrated into a single solution $Y \in f(\bigcup_{i \in [t]} P_i)$, where $\bigcup_{i \in [t]} P_i$ is ``sandwiched'' between the ball $\ball(x, r)$ (the given query range) and a slightly larger ball $\ball(x, O(\Gamma) \cdot r)$.
Specifically, in \Cref{sec:range-query-ANN}, as a warm-up, we show how \Cref{lem:range-query} can be used to derive the fundamental approximate nearest neighbor data structures;
in \Cref{sec:range-query-1-means}, we design a data structure for the $1$-means problem on any approximate ball, which serves as an important component of our dynamic $k$-means algorithm.

\subsection{\texorpdfstring{Proof of \Cref{lem:range-query}}{}}
\label{sec:proof-range-query}

Let $\varphi: \Deld \to \Phi$ denote the given efficient $(\Gamma, \Lambda, r \cdot \Gamma)$-hash.
For every $z \in \Phi$, let $X_z \eqdef X \cap \varphi^{-1}(z)$.
The data structure maintains the partition $\{X_{z}\}_{z \in \varphi(X)}$ induced by the hash function $\varphi$ and, for every $z \in \varphi(X)$, an instance $\mathcal{A}_z$ that handles dynamic updates to $X_{z}$.

Upon each insertion (resp. deletion) of a point $x \in \Deld$, we compute its hash value $z_x \eqdef \varphi(x)$ in $\tilde{O}(\Lambda)$ time (\Cref{def:efficient-consistent-hashing}). We then update the partition via $X_{z_x} \gets X_{z_x} + x$ (resp. $X_{z_x} \gets X_{z_x} - x$), and feed the update to $\mathcal{A}_{z_x}$ (creating it if it does not exist).
Clearly, this procedure correctly maintains the data structure, and it involves one update operation of $\calA$ with an additional $\tO(\Lambda)$ time.

Now, we explain how we answer queries. For a given $x\in \Deld$, we compute a set $\Phi_{x}$ such that 
\begin{align}
    \varphi(\ball(x, r))
    ~\subseteq~
    \Phi_{x} \cap \varphi(\Deld)
    ~\subseteq~ 
    \varphi(\ball(x, 2\Gamma\cdot r)).
    \label{eq:approximate-ball}
\end{align}
Since $\varphi$ is an efficient $(\Gamma, \Lambda, r \cdot \Gamma)$-hash (\Cref{def:efficient-consistent-hashing}), the computation of $\Phi_x$ takes $\tilde{O}(\Lambda)$ time and we have $|\Phi_x| \le \Lambda$.
Since we explicitly maintain the partition $\{X_z\}_{z \in \varphi(X)}$, we can remove all $z$ with $X_z = \emptyset$ from $\Phi_x$ in $\tilde{O}(\Lambda)$ time.
Henceforth, we assume $X_z \neq \emptyset$ (and thus $\mathcal{A}_z$ exists) for every $z \in \Phi_x$.
We query $\mathcal{A}_z$ to obtain $Y_z \in f(X_z)$ for every $z \in \Phi_x$, and return the collection $\{Y_z\}_{z \in \Phi_x}$.
The query time consists of $\tO(\Lambda)$ queries to $\calA$ and $\tO(\Lambda)$ overhead, resulting in a total of $\tO(\Lambda)$ times the query time of $\calA$.

Clearly, the sets $\{X_z\}_{z \in \Phi_x}$ are pairwise disjoint. 
It remains to show that $\bigcup_{z \in \Phi_x} X_z$ forms an approximate ball such that 
\begin{align}
\label{eq:approximate-ball-X}
    \ball(x, r) \cap X
    \subseteq \bigcup_{z\in \Phi_{x}} X_{z}
    \subseteq \ball(x,3\Gamma\cdot r) \cap X,
\end{align}
thereby ensuring the correctness of $\{Y_z\}_{z \in \Phi_x}$.

To this end, consider an arbitrary point $y \in \ball(x, r) \cap X$.  
By \Cref{eq:approximate-ball}, we have $\varphi(y) \in \Phi_x$, which implies that $y \in \bigcup_{z \in \Phi_x} (X \cap \varphi^{-1}(z)) = \bigcup_{z \in \Phi_x} X_z$.  
The arbitrariness of $y \in \ball(x, r) \cap X$ implies that $\ball(x, r) \cap X \subseteq \bigcup_{z \in \Phi_x} X_z$.

On the other hand, consider a specific $y \in X_{z}$ for some $z \in \Phi_x$.  
By \Cref{eq:approximate-ball}, there exists $y' \in \ball(x, 2\Gamma\cdot r)$ such that $\varphi(y) = \varphi(y')$.  
Since $\varphi$ is $(\Gamma, \Lambda, r\cdot \Gamma)$-consistent, we have $\dist(y, y') \le r\cdot \Gamma$ (see \Cref{def:efficient-consistent-hashing}).
By the triangle inequality,
\begin{align*}
    \dist(x, y) \le \dist(x, y') + \dist(y', y) \le 3\Gamma\cdot r \implies y\in \ball(x, 3\Gamma\cdot r) \cap X.
\end{align*}
The arbitrariness of $y\in \bigcup_{z\in \Phi_{x}} X_{z}$ implies $\bigcup_{z\in \Phi_{x}} X_z \subseteq \ball(x, 3\Gamma \cdot r) \cap X$.

Together, the above discussions imply \Cref{eq:approximate-ball-X} and thus complete the proof.

\subsection{Application: Approximate Nearest Neighbor Structure}
\label{sec:range-query-ANN}

As a warm-up application of \Cref{lem:range-query}, we describe how to design an \emph{approximate nearest neighbor (ANN) oracle} -- a data structure that, given a dynamic center set $S \subseteq \Deld$, returns for each query point $x$ a (distinct) approximate nearest center $s \in S - x$.
We note that similar guarantee of ANN has also been obtained in e.g.,~\cite{DBLP:conf/nips/0001DFGHJL24}, via more specific techniques albeit it may potentially yield better tradeoff beyond consistent hashing.

Without loss of generality, we assume hereafter that $S$ contains at least two centers, which guarantees the existence of a distinct nearest neighbor for any query point. Also, the case $|S| \le 1$ is trivial.

\begin{lemma}[ANN Oracle]
\label{lem:ANN-oracle}
For any $0 < \epsilon < 1$, there exists a data structure that maintains a dynamic center set $S \subseteq \Deld$, and for every query $x\in \Deld$, returns a center $s\in S - x$ such that $\dist(x,s) \le O(\epsilon^{-3/2})\cdot \dist(x, S - x)$. 
The data structure works against an adaptive adversary and, with probability at least $1 - \frac{1}{\poly(\Delta^d)}$, has update and query time $\tilde{O}(2^{\epsilon d})$.
\end{lemma}

\begin{proof}
Let $L \eqdef \lceil \log(\sqrt{d} \Delta) \rceil$. 
For each $0 \le i \le L$, we sample a $(\Gamma, \Lambda, 2^{i}\cdot \Gamma)$-hash $\varphi_i$ from \Cref{lem:efficient-consistent-hashing}, where $\Gamma = O(\epsilon^{-3/2})$ and $\Lambda = \tO(2^{\epsilon d})$.
By a union bound, with probability at least $1 - \frac{L + 1}{\poly(\Delta^d)} = 1 - \frac{1}{\poly(\Delta^d)}$, all hash functions $\varphi_0, \dots, \varphi_L$ are successfully sampled. 

Consider a simple query problem $(U, f)$, where $U = \Deld \cup \{\perp\}$, and for every $P \subseteq \Deld$, we define $f(P) = \{\{x\} : x \in P\}$ if $P \neq \emptyset$, and $f(P) = \{\perp\}$ if $P = \emptyset$; that is, we return an arbitrary point from the set $P$ if it is non-empty.
One can easily design a data structure that maintains a dynamic set $S$ and reports $a \in f(S)$ with both update and query time $\tO(1)$, while working against an adaptive adversary.  
Let $\mathcal{A}$ denote this data structure.
For every $i \in [0, L]$, by plugging the $(\Gamma, \Lambda, 2^{i}\cdot \Gamma)$-hash $\varphi_i$ and $\mathcal{A}$ into \Cref{lem:range-query}, we obtain a data structure $\mathcal{D}_i$ with update time $\tO(\Lambda)$. Our ANN data structure for \Cref{lem:ANN-oracle} simply maintains $\{\calD_{i}\}_{i\in [0,L]}$, so the update time is $(L+1)\cdot \tO(\Lambda) = \tO(2^{\epsilon d })$.

We then explain how to use $\{\mathcal{D}_i\}_{i \in [0, L]}$ to answer a query. First, we consider a query for a non-center point $x \in \Deld - S$.
In this case, for every $i \in [0, L]$, we query $\mathcal{D}_i$ and obtain $a_{i,1}, \dots, a_{i,t} \in U$ for some $1 \le t \le \Lambda$. 
Since this query problem $(U, f)$ is clearly mergeable, we can compute a solution $a_i$ from $\{a_{i,j}\}_{j \in [t]}$ in $O(t) = O(\Lambda)$ time, such that, by \Cref{lem:range-query}, $a_i \in f(B_i)$ for some $B_i \subseteq S$ satisfying that $\ball(x, 2^{i}) \cap S \subseteq B_i \subseteq \ball(x, 3\Gamma \cdot 2^{i}) \cap S$.
Let $i^* \in [0, L]$ denote the smallest index such that $a_{i^*} \ne \perp$, and return $a_{i^*}$ as the answer.
Such an index $i^{*}$ must exist, since $2^L \ge \sqrt{d} \cdot \Delta$ implies $S \subseteq B_L$. Hence, $a_L \in f(S)\neq \{\perp\}$.

We then prove the correctness:
Notice that $a_{i^*} = f(B_{i^*}) \neq \perp$ is equivalent to $B_{i^*} \neq \emptyset$ and $a_{i^*} \in B_{i^*}$. Since $B_{i^*} \subseteq \ball(x, 3\Gamma \cdot 2^{i^*}) \cap S$, it follows that $\dist(x, a_{i^*}) \leq 3\Gamma \cdot 2^{i^*}$.
On the other hand, if $i^* = 1$, then $\dist(x, S) \ge 1$ (since $x \notin S$), and therefore $\dist(x, a_{i^*}) \le 3\Gamma \cdot \dist(x, S)$.
If $i^* > 1$, then $a_{i^* - 1} = \perp$, which implies that $\ball(x, 2^{i^* - 1}) \cap S \subseteq B_{i^* - 1} = \emptyset$. Hence, $\dist(x, S) \ge 2^{i^* - 1}$, and therefore $\dist(x, a_{i^*}) \le 6\Gamma \cdot \dist(x, S)$. Since $\Gamma = O(\epsilon^{-3/2})$, it follows that $a_{i^{*}}$ is a correct answer for the query $x\in \Deld - S$.

The query time is primarily due to querying $\mathcal{D}_i$ and merging the output for all $i \in [0, L]$, totaling $(L + 1) \cdot \tO(\Lambda) = \tO(2^{\epsilon d})$.

Finally, for a query where the point $x \in S$ is a center, we temporarily remove $x$ from $S$, perform the query as above, and then reinsert $x$ into $S$. The total query time remains $\tO(2^{\epsilon d})$.
\end{proof}

An ANN oracle can be used to derive many useful data structures.
By plugging \Cref{lem:ANN-oracle} into the framework of~\cite{BhattacharyaGJQ24}, we obtain the following \emph{ANN indicator} data structure.

\begin{lemma}[ANN Indicators]
\label{lem:ANN-indicator}
For any $\gamma \geq 0$ and $0<\epsilon < 1$, there exists a data structure that maintains a dynamic set $S \subseteq \Deld$ and for every $s \in S$, maintains a bit $b_{\gamma}(s, S) \in \{0, 1\}$ with the following guarantee.
\begin{enumerate}[font = \bfseries]
    \item If $\dist(s, S - s) \leq \gamma$, then $b_{\gamma}(s, S) = 1$.
    \item 
    If $\dist(s, S - s) > \epsilon^{-3/2} \cdot \gamma$, then $b_{\gamma}(s, S) = 0$.
\end{enumerate}
After every update to $S$, the data structure reports all the centers $s$ that $b_{\gamma}(s, S)$ changes.
The data structure works against an adaptive adversary and, with probability at least $1 - \frac{1}{\poly(\Delta^d)}$, has amortized update time $\tO(2^{\epsilon d})$.
\end{lemma}
\begin{proof}
Follows from \Cref{lem:ANN-oracle} and {\cite[Theorem 3.1]{BhattacharyaGJQ24}}
\end{proof}

Furthermore, based on the ANN indicators maintained in \Cref{lem:bitwise-approx-NC}, we derive the following result, which is a data structure that maintains the approximate distance from each center in $S$ to its nearest neighbor.

\begin{lemma}[ANN Distance]
\label{lem:ANN-distance}
For any $0 < \epsilon < 1$,
there exists a data structure that maintains a dynamic set $S\subseteq \Deld$ and for every $s\in S$, maintains a value $\Hat{\dist}(s, S - s)$ such that 
\begin{align*}
    \dist(s, S - s) \le \Hat{\dist}(s, S - s) \le O(\epsilon^{-3/2})\cdot \dist(s, S - s).
\end{align*}
The data structure works against an adaptive adversary and, with probability at least $1 - \frac{1}{\poly(\Delta^d)}$, has amortized update time $\tO(2^{\epsilon d})$.
\end{lemma}
\begin{proof}
The data structure for \Cref{lem:ANN-distance} builds upon the one used to maintain the nearest neighbor indicator in \Cref{lem:ANN-indicator}.

Without ambiguity, we simply write $b_{\ell}^{s} = b_{2^{\ell}}(s, S)$ throughout this proof. Let $L\eqdef \lceil \log(\sqrt{d}\cdot \Delta)\rceil$.

For every scale $\ell \in [L]$, we build a data structure $D_{\ell}(S)$ from \Cref{lem:bitwise-approx-NC}, which maintains a bit $b_{\ell}^{s}$ for every center $s \in S$ such that:
\begin{itemize}
    \item If $\dist(s, S - s) \leq 2^{\ell}$, then $b_{\ell}^{s} = 1$. 
    (If $b_{\ell}^{s} = 0$, then $\dist(s, S  - s) > 2^{\ell}$.)
    
    \item If $\dist(s, S - s) > \epsilon^{-3/2} \cdot 2^{\ell}$, then $b_{\ell}^{s} = 0$.    
    (If $b_{\ell}^{s} = 1$, then $\dist(s, S - s) \le  \epsilon^{-3/2} \cdot 2^{\ell}$.)
\end{itemize}
We define $\Hat{\dist}(s, S -s) \eqdef  \epsilon^{-3/2} \cdot 2^{\ell_{s}}$ for every center $s \in S$, where $\ell_{s} \eqdef \min \{\ell \in [L] : b_{\ell}^{s} = 1\}$ is the ``threshold scale''; note that $\ell_{s} \in [L]$ must be well-defined, since $2^{L} = 2^{\lceil \log(\sqrt{d} \Delta) \rceil} \ge \sqrt{d} \Delta = \diam(\Deld)$ and thus $b_{L}^{s} = 1$.

First, we show the correctness of $\{\Hat{\dist}(s, S - s)\}_{s\in S}$. Consider a specific center $s \in S$. By construction, we have $b_{\ell_{s}}^{s} = 1 \implies \dist(s, S - s) \le  \epsilon^{-3/2} \cdot 2^{\ell_{s}} = \Hat{\dist}(s, S -s)$.
When $\ell_{s} = 1$, we have $\dist(s, S - s) \ge 1 = \frac{\Hat{\dist}(s, S -s)}{2 \epsilon^{-3/2}}$, since two distinct points in the discrete space $\Deld$ have distance $\ge 1$.
When $\ell_{s} \ne 1$, we have $b_{\ell_{s} - 1}^{s} = 0 \implies \dist(s, S - s) >  2^{\ell_{s} - 1} = \frac{\Hat{\dist}(s, S -s)}{2 \epsilon^{-3/2}}$.
In either case, we always have $\dist(s, S - s)
\le \Hat{\dist}(s, S -s)
\le 2 \epsilon^{-3/2} \cdot \dist(s, S - s)$.

Second, we consider the update time. Clearly, the values $\{\Hat{\dist}(s, S -s)\}_{s \in S}$ can be maintained simultaneously with $\{b_{\ell}^{s}\}_{s \in S}$ for all scales $\ell \in [L]$.
Specifically, for every update to the center set $S$, \Cref{lem:bitwise-approx-NC} reports all changes of the indicators $\{b_{\ell}^{s}\}_{s \in S, \ell \in [L]}$.
Regarding a reported change $b_{\ell}^{s}$ (say), we can simultaneously recompute the threshold scale $\ell_{s} = \min \{\ell \in [L] : b_{\ell}^{s} = 1\}$ and update the value $\Hat{\dist}(s, S -s) =  \epsilon^{-3/2} \cdot 2^{\ell_{s}}$. 
This step takes time $O(L) = \tO(1)$. 
Hence, the update time for updating $\{\Hat{\dist}(s, S -s)\}_{s \in S}$ is asymptotically equal to updating all $\{b_{\ell}^{s}\}_{s \in S}$ (up to an $\tO(1)$ factor), which overall takes
$O(L) \cdot \tO(\Lambda) = \tO(\Lambda)$.

This finishes the proof.
\end{proof}

\subsection{\texorpdfstring{Application: $1$-Means Range Query}{}}
\label{sec:range-query-1-means}

Here, we use \Cref{lem:range-query} to devise a data structure that supports estimating both the $1$-means objective and its solution on any given approximate ball.
This data structure is obtained by plugging a dynamic coreset algorithm from~\cite{HenzingerK20} (for the entire set $X$) into \Cref{lem:range-query}; hence, for each query, we first obtain a small-sized coreset and then compute the objective and/or solution based on it.

\begin{lemma}[$1$-Means over Approximate Balls]
\label{cor:1-means-on-ball}
For any $0 < \epsilon < 1$, there exists a data structure that maintains a dynamic weighted dataset $X \subseteq [\Delta]^d$, and for every query point $x \in \Deld$ and radius $r \ge 0$, returns a real number $b_x \ge 0$, a center $c^{*} \in \Deld$, and two estimates $\hat{\cost}_{c^{*}}$ and $\hat{\cost}_x$, such that there exists a subset $B_x \subseteq X$ satisfying the following conditions.
\begin{enumerate}[font = \bfseries]
    \item\label{cor:1-means-on-ball:1}
    $\ball(x, r) \cap X \subseteq B_x \subseteq \ball(x, O(\epsilon^{-3/2}) \cdot r) \cap X$ and $b_{x} = w(B_{x})$.
    
    \item\label{cor:1-means-on-ball:2}
    $\cost(B_{x}, c) \leq \Hat{\cost}_{c} \leq O(1) \cdot \cost(B_{x}, c)$, for either $c \in \{c^{*},\ x\}$.
    
    \item\label{cor:1-means-on-ball:3}
    $\cost(B_{x}, c^{*}) \leq O(1) \cdot \OPT_{1}(B_{x})$.
\end{enumerate}
The data structure works against an adaptive adversary and, with probability at least $1 - \frac{1}{\poly(\Delta^d)}$, has amortized update and query time $\tilde{O}(2^{\epsilon d})$.
\end{lemma}

\begin{proof}
We first consider two extreme cases of the query radius $r$.
In the case of a small radius $0 \le r < 1$, we have $\ball(x, r)= \{x\}$ (recall that any two distinct points in the discrete space $[\Delta]^d$ have distance at least $1$), and we simply consider the point subset $B_x = \{x\} \cap X$.
Either subcase -- $B_x = \{x\}$ or $B_x = \emptyset$ -- is trivial.
In the case of a large radius $r > \sqrt{d} \Delta = \diam([\Delta]^d)$, we have $\ball(x, r) \cap X = X$, and we simply consider the entire point set $B_x = X$.
Thus, this extreme case $r > \sqrt{d} \Delta$ reduces to the boundary case $r = \sqrt{d} \Delta$.

Hereafter, we assume a bounded query radius $1 \le r \le \sqrt{d} \Delta$.
Let $L \eqdef \lceil \log(\sqrt{d} \Delta) \rceil$.
For each $i\in [L+1]$, we sample a $(\Gamma, \Lambda, 2^{i}\cdot \Gamma)$-hash $\varphi_i$ from \Cref{lem:efficient-consistent-hashing}, where $\Gamma = O(\epsilon^{-3/2})$ and $\Lambda = \tO(2^{\epsilon d})$.
By a union bound, with probability at least $1 - \frac{L + 1}{\poly(\Delta^d)} = 1 - \frac{1}{\poly(\Delta^d)}$, all hash functions $\varphi_0, \dots, \varphi_L$ are successfully sampled.

Fix an integer $i \in [L + 1]$. We design a sub-data structure to handle queries with $2^{i - 1} \le r < 2^{i}$.
Since there are at most $L + 1$ such sub-data structures, the overall update time for \Cref{cor:1-means-on-ball} is $L + 1 = \tO(1)$ times that of a single sub-data structure, and hence asymptotically the same.

The following claim provides a data structure that (1) maintains the total weight of the dataset (a trivial task), and (2) constructs a \emph{coreset} for the entire dataset, following the approach of~\cite{HenzingerK20}.\footnote{\cite{HenzingerK20} provides a general $(1+\epsilon)$-coreset construction for $k$-means. In this paper, we only need a constant-factor coreset for $1$-means (e.g., $k = 1$, $\epsilon = 0.5$), so we state this special case.}

\begin{claim}
\label{claim:previous-dynamic-coreset}
There exists a data structure that maintains a dynamic weighted set $X \subseteq [\Delta]^d$, and, upon query, returns the total weight of $X$ and a weighted set $Y\subseteq \Deld$ such that the following guarantee holds with probability at least $1 - \frac{1}{\poly(\Delta^d)}$.
\begin{align*}
    \forall c\in \Deld, \qquad \cost(Y, c) = \Theta(\cost(X,c)).
\end{align*}
The data structure has update and query time $\poly(d \log \Delta)$, and it works against an adaptive adversary.
\end{claim}

We plug the $(\Gamma,\Lambda,2^{i} \cdot \Gamma)$-hash $\varphi_i$ and \Cref{claim:previous-dynamic-coreset} into \Cref{lem:range-query} and obtain a data structure $\calD_{i}$, which has an update time of $\tO(\Lambda\cdot \poly(d\log\Delta)) = \tO(2^{\epsilon d})$.

We now explain how to use $\calD_i$ to answer a query for a point $x \in \Deld$ and radius $r \in [2^{i-1}, 2^i)$.
We query $\calD_i$ with $x$, and it returns, in time $\tO(\Lambda\cdot \poly(d\log\Delta)) = \tO(2^{\epsilon d})$, weights $w_1, \dots, w_t$ and weighted sets $Y_1, \dots, Y_t$, for some $1 \le t \le \Lambda = \tO(2^{\epsilon d})$. Let $Y \eqdef \bigcup_{j \in [t]} Y_j$. By construction, the size of $Y$ is bounded by the query time, which is $\tO(2^{\epsilon d})$. 
Finally, we perform the following computations:
\begin{itemize}
    \item $b_x \gets \sum_{j \in [t]} w_j$. This step takes $\tO(t) = \tO(2^{\epsilon d})$ time.

    \item $c^* \gets$ the $1$-means center of $Y$.
    This step takes $\tO(|Y|) = \tO(2^{\epsilon d})$ time.
    
    \item $\widehat{\cost}_{c^*} \gets \cost(Y, c^*)$ and $\widehat{\cost}_x \gets \cost(Y, x)$.
    This step takes $\tO(|Y|) = \tO(2^{\epsilon d})$ time.
\end{itemize}
Overall, the query time is $\tO(2^{\epsilon d})$.
To see the correctness, by \Cref{lem:range-query}, there exists disjoint sets $P_1, \dots, P_t\subseteq X$ such that
\begin{itemize}
    \item $\ball(x, 2^{i}) \cap X \subseteq \bigcup_{j \in [t]} P_j \subseteq \ball(x, O(\epsilon^{-3/2}) \cdot 2^{i}) \cap X$, which, combined with the fact that $r \in [2^{i-1}, 2^i)$, implies that $\ball(x, r) \cap X \subseteq \bigcup_{j \in [t]} P_j \subseteq \ball(x, O(\epsilon^{-3/2}) \cdot r) \cap X$;
    
    \item for every $j \in [t]$, $w_{j} = w(P_j)$; and

    \item for every $j\in [t]$,  the following holds with probability at least $1 - \frac{1}{\poly(\Delta^d)}$.\footnote{While \Cref{lem:range-query} does not explicitly consider success probability, its analysis can be easily adapted to handle data structures that succeed with high probability.}
    \begin{align}
    \label{eq:coreset-subset}
        \forall c\in \Deld, \qquad \cost(Y_{j},c) = \Theta(\cost(P_{j},c)).
    \end{align}
\end{itemize}
Let $B_x \eqdef \bigcup_{j\in[t]} P_j$.
\Cref{cor:1-means-on-ball:1} then follows immediately.
Moreover, by a union bound, with probability at least $1 - \frac{t}{\poly(\Delta^d)} = 1 - \frac{1}{\poly(\Delta^d)}$, \Cref{eq:coreset-subset} holds for all $j \in [t]$ simultaneously.
Therefore, it is easy to verify that
\begin{align}
    \forall c \in \Deld, \qquad \cost(Y, c) = \Theta(\cost(B_x, c)).
\end{align}
Hence, the $1$-means center $c^*$ on $Y$ is guaranteed to be an $O(1)$-approximation to the $1$-means center on $B_x$ (i.e., \Cref{cor:1-means-on-ball:2} holds), and the values $\widehat{\cost}_{c^*}$ and $\widehat{\cost}_x$ are also $O(1)$-approximations to $\cost(B_x, c^*)$ and $\cost(B_x, x)$, respectively (i.e., \Cref{cor:1-means-on-ball:3} holds).

This finishes the proof.
\end{proof}

\section{Approximate Assignment Structure}
\label{sec:approximate-assignment}

In this section, we design a data structure that \emph{implicitly} maintains an approximate assignment from a point set $X \subseteq \Deld$ to a center set $S \subseteq \Deld$, which supports dynamic updates to both sets.
We first describe the data structure in \Cref{sec:approx-assignment-description} and analyze its update time in \Cref{sec:approx-assignment-time}.
We then present two key applications -- maintaining center weights and supporting $D^2$-sampling -- in \Cref{sec:approx-assignment-weight} and \Cref{sec:approx-assignment:2}, respectively. These are critical building blocks for the restricted $k$-means subroutine in \Cref{sec:restricted_kmeans} and the augmented $k$-means subroutine in \Cref{sec:augmented_kmeans}.

\subsection{Data Structure Description}
\label{sec:approx-assignment-description}

In this part, we describe our data structure for maintaining an approximate assignment. 
We first present static and dynamic objects that are \emph{explicitly maintained} by our data structure (\Cref{def:approx-assignment-DS}). 
Then, we show how these explicit objects define an \emph{implicit partition} of the non-center points $X - S$ (\Cref{claim:H-partition}). Built on top of this implicit partition, we define an \emph{implicit assignment} $\sigma: X - S \to S$ (\Cref{claim:implicit-assignment}).

\subsection*{Explicit Objects}
We summarize all the explicitly maintained objects below.
\begin{definition}[Data Structure]
\label{def:approx-assignment-DS}
Firstly, the data structure explicitly maintains the following static objects, which remain unchanged under updates to the point set $X$ and the center set $S$.
\begin{enumerate}[font = \bfseries]
    \item\label{def:approx-assignment-DS:parameter}
    Parameters $\Gamma \eqdef O(\epsilon^{-3/2})$, $\Lambda \eqdef \tO(2^{\epsilon d})$, $m\eqdef \lceil \log(\sqrt{d}\cdot \Delta\cdot \Gamma) \rceil$, and $\rho_{i} \eqdef \frac{1}{2} \cdot (3\Gamma)^{i}$, $\forall i\in [0,m]$.
    
    \item\label{def:approx-assignment-DS:hash}
    Hash functions $\{\varphi_{i}\}_{i \in [0, m]}$; we assume without loss of generality that $\varphi_{0}: \Deld \to \Phi$ is an injective function, which is trivially an efficient $(\Gamma, \Lambda, \rho_{0})$-hash, and every other $\varphi_{i}: \Deld \to \Phi$ is an efficient $(\Gamma, \Lambda, \rho_{i})$-hash (\Cref{def:efficient-consistent-hashing})
\end{enumerate}
Secondly, defining (low-level buckets) $L \eqdef [m] \times \Phi \times \Phi$ and (high-level buckets) $H \eqdef [m] \times \Phi$, the data structure maintains the following dynamic objects.\footnote{The sizes of $L$ and $H$ are $\poly(\Delta^d)$, so a full representation of the dynamic objects in \Cref{def:approx-assignment-DS:preimage,def:approx-assignment-DS:llb,def:approx-assignment-DS:close-center,def:approx-assignment-DS:hlb} would be inefficient. Instead, we store only the ``supported'' objects, i.e., those non-empty sets.
}
\begin{enumerate}[font = \bfseries]
\setcounter{enumi}{2}
    \item\label{def:approx-assignment-DS:preimage} $X(i, z)\eqdef X \cap \varphi_{i}^{-1}(z)$, for every high-level bucket $(i, z) \in H$.

    \item\label{def:approx-assignment-DS:llb}
    $\set(i, z, z') \eqdef X \cap \varphi_{i}^{-1}(z) \cap \varphi_{i - 1}^{-1}(z')$, for every low-level bucket $(i,z,z') \in L$.
    
    \item\label{def:approx-assignment-DS:close-center} $S(i, z)$, a subset of \emph{close centers} for every high-level bucket $(i, z) \in H$ that, if $z \in \varphi_{i}(\Deld)$, then\footnote{For any hash value $z \notin \varphi_i(\Deld)$, we impose no requirement on the corresponding center subset $S(i,z)$.}
    \begin{align}
        \big\{s\in S: \dist(s,\varphi_{i}^{-1}(z)) \le \tfrac{\rho_{i}}{\Gamma}\big\}
        ~\subseteq~ S(i, z)
        ~\subseteq~ \big\{s\in S: \dist(s,\varphi_{i}^{-1}(z)) \le 2\rho_{i}\big\}.
        \label{eq:bucket-center}
    \end{align}

    \item\label{def:approx-assignment-DS:hlb} $f(i, z) \eqdef \{z' \in \Phi : S(i - 1, z') = \emptyset \land \set(i, z, z') \ne \emptyset\}$, for every high-level bucket $(i,z)\in H$.
    
    \item\label{def:approx-assignment-DS:partition} $H' \eqdef \{(i, z) \in H : S(i, z) \ne \emptyset\}$.
    
    \item\label{def:approx-assignment-DS:assignment} $\sigma'(i, z) \in S(i, z)$, an arbitrary close center for every high-level bucket $(i, z) \in H'$.
    \end{enumerate}
\end{definition}
\begin{lemma}[Initialization]\label{lem:approx-assignment-init}
    The data structure of \Cref{def:approx-assignment-DS} can be initialized in time $\tO(1)$ and with success probability $1 - \frac{1}{\poly(\Delta^{d})}$.
\end{lemma}

\begin{proof}
Initially, both $X$ and $S$ are empty $= \emptyset$, so the initialization only needs to compute the $m + 1$ hash functions $\{\varphi_{i}\}_{i \in [0, m]}$.
By \Cref{lem:efficient-consistent-hashing}, every $\varphi_{i}$ can be constructed in $\tO(1)$ time with a success probability of $1 - \frac{1}{\poly(\Delta^d)}$. Thus, the total initialization time is $(m+1) \cdot \tO(1) = \tO(1)$.  
    Moreover, taking a union bound, the overall success probability is $1 - \frac{m + 1}{\poly(\Delta^d)} = 1 - \frac{1}{\poly(\Delta^d)}$. \end{proof}

\subsection*{Implicit Partition and Assignment}

Finally, we show how the explicit objects (\Cref{def:approx-assignment-DS}) implicitly define a partition of the non-center points $X - S$ (\Cref{claim:H-partition}). 
This partition exhibits an approximate \emph{equidistant} structure, where points in each part are roughly equally distant from the center set (\Cref{claim:equidistant}).
We further define an \emph{implicit approximate assignment} based on this partition (\Cref{claim:implicit-assignment}).

\paragraph{Prerequisites: Properties of Close Centers.} 
Before defining the implicit partition and the implicit assignment, we establish some useful properties of the close centers $\{S(i, z)\}_{(i, z) \in H}$ as prerequisites.
We introduce the following notations for convenience.
\begin{align}
    & \forall i\in [0,m],\ \forall x\in X:
    && S(i,x) \eqdef S(i, \varphi_{i}(x))
    \label{eq:point-close-center}\\
    &\forall (i, z) \in H:
    && \set(i, z) \eqdef \bigcup_{z' \in f(i, z)} \set(i, z, z')
    \label{eq:high-level-bucket}
\end{align}

\begin{claim}
\label{claim:bucket-point}
For every point $x \in X$ and every high-level bucket $(i,z) \in H'$, if $x \in \set(i, z)$, then $z = \varphi_{i}(x)$, $S(i, x) \neq \emptyset$ and $S(i - 1, x) = \emptyset$.
\end{claim}

\begin{proof}
Consider a specific pair of $x \in X$ and $(i, z) \in H'$ such that $x \in \set(i, z)$; we must have $x \in \set(i, z, z')$ for some hash value $z' \in f(i, z)$ (\Cref{eq:high-level-bucket}). Then,
\begin{itemize}
    \item $S(i,z) \neq \emptyset$, given the definition of $H'$ (\Cref{def:approx-assignment-DS:partition} of \Cref{def:approx-assignment-DS}) and that $(i, z) \in H' \subseteq H$;
    
    \item $S(i - 1, z') = \emptyset$, given the definition of $f(i, z)$ (\Cref{def:approx-assignment-DS:hlb} of \Cref{def:approx-assignment-DS}) and that $z' \in f(i, z)$;
    
    \item $\varphi_{i}(x) = z$ and $\varphi_{i - 1}(x) = z'$, given the definition of $\set(i, z, z')$ (\Cref{def:approx-assignment-DS:llb} of \Cref{def:approx-assignment-DS}) and that $x \in \set(i, z, z')$.
\end{itemize}
Putting everything together, we conclude that $z = \varphi_{i}(x)$, $S(i, x) \neq \emptyset$ and $S(i - 1, x) = \emptyset$.
\end{proof}

\begin{claim}
\label{claim:hierarchical-close-centers}
    For every $x\in X$ and every $i \in [m]$, it holds that $S(i-1,x) \subseteq S(i, x)$.
\end{claim}

\begin{proof}
Let $z_{i}\eqdef \varphi_{i}(x)$ and $z_{i-1} \eqdef \varphi_{i-1}(x)$. Consider a specific center $s\in S(i - 1, x)$. By \Cref{eq:bucket-center}, it holds that $\dist(s, \varphi_{i - 1}^{-1}(z_{i - 1})) \le 2\rho_{i-1}$. Therefore,
\begin{align*}
    \dist(s, x)
    & ~\le~ \dist(s, \varphi_{i - 1}^{-1}(z_{i - 1})) + \diam(\varphi_{i - 1}^{-1}(z_{i - 1}))
    \tag{Triangle inequality}\\
    & ~\le~ 2\rho_{i-1} + \rho_{i-1} \tag{\Cref{def:approx-assignment-DS:hash} of \Cref{def:approx-assignment-DS}}\\
    & ~=~ \rho_{i} / \Gamma \tag{\Cref{def:approx-assignment-DS:parameter} of \Cref{def:approx-assignment-DS}}.
\end{align*}
Then, $\dist(s, \varphi_{i}^{-1}(z_{i})) \le \dist(s,x) \le \rho_{i} / \Gamma \implies s\in S(i,z_{i}) = S(i,x)$ by \Cref{eq:bucket-center,eq:point-close-center}. The arbitrariness of $s \in S(i - 1, x)$ implies that $S(i-1,x) \subseteq S(i, x)$.
This finishes the proof.
\end{proof}

\paragraph{Implicit Partition.}
Now we are ready to define the implicit partition and the implicit assignment, using the explicit objects maintained in \Cref{def:approx-assignment-DS}.

\begin{claim}[Partition]
\label{claim:H-partition}
$\{\set(i, z)\}_{(i, z) \in H'}$ is a partition of $X - S$.
\end{claim}

\begin{proof}
Consider a specific point $x\in X$ and let $z_{i} \eqdef \varphi_{i}(x)$ for every $i \in [0, m]$.

In case of a center point $x \in S$, we must have $\dist(x, \varphi_{0}^{-1}(z_{0})) = \dist(x, x) = 0 \le \frac{\rho_{0}}{\Gamma}$, which implies $x \in S(0, z_{0})$ (\Cref{eq:bucket-center}) and $S(i, z_{i}) \neq \emptyset$ for every $i \in [0, m]$ (\Cref{claim:hierarchical-close-centers}). Then, it follows from \Cref{claim:bucket-point} that $x \notin \set(i, z)$ for all $(i, z) \in H'$.

In case of a non-center point $x \in X - S$, since $x \notin S$ and $\varphi_{0}$ is an injective function, we have $\dist(s, \varphi_{0}^{-1}(z_{0})) = \dist(s, x) \ge 1 \ge 2\rho_{0}$ for every $s \in S$ and thus $S(0, z_{0}) = \emptyset$ (\Cref{eq:bucket-center}).
Also, since $(3\Gamma)^L / \Gamma \ge \sqrt{d} \Delta = \diam(\Deld)$, we have $S(L, z_L) = S \ne \emptyset$ (\Cref{eq:bucket-center}).
Then, by \Cref{claim:hierarchical-close-centers}, there exists a unique index $1 \le j < L$ such that $S(i, x) = \emptyset$ when $i < j$ and $S(i, x) \ne \emptyset$ when $i \ge j$.
By the definition of $H'$ (\Cref{def:approx-assignment-DS:partition} of \Cref{def:approx-assignment-DS}), it follows that $(i, z_{i}) \in H'$ for every $i \ge j$. Then, by \Cref{claim:bucket-point}, we conclude that $(j, z_j) \in H'$ is the unique pair such that $x \in \set(j, z_j)$.

Since $x \in X$ is chosen arbitrarily, the proof is complete.
\end{proof}

\begin{claim}[Equidistant]
\label{claim:equidistant}
For every $(i,z)\in H'$ and every $x\in \set(i,z)$, it holds that 
\begin{align*}
\frac{1}{2\Gamma} \cdot (3\Gamma)^{i - 1}
    ~\le~ \dist(x, S)
    ~\le~ \dist(x, \sigma'(i,z))
    ~\le~ \frac{3}{2} \cdot (3\Gamma)^{i}.
\end{align*}
\end{claim}

\begin{proof}
Consider a specific high-level bucket $(i, z) \in H'$ and a specific point $x \in \set(i, z)$; we have $\varphi_{i}(x) = z$ and $S(i, z) \neq \emptyset$ (\Cref{claim:bucket-point}) and $x \in \set(i, z, z')$ for some hash value $z' \in f(i, z) \subseteq \Phi$ (\Cref{eq:high-level-bucket}). Thus, $\varphi_{i - 1}(x) = z'$ (\Cref{def:approx-assignment-DS:llb} of \Cref{def:approx-assignment-DS}) and $z' \in f(i, z) \implies S(i - 1, z') = \emptyset$ (\Cref{def:approx-assignment-DS:hlb} of \Cref{def:approx-assignment-DS}).

Since $S(i, z) \ne \emptyset$  and $z = \varphi_{i}(x)$, by \Cref{eq:bucket-center} and the fact that $\sigma'(i, z) \in S(i, z)$ (\Cref{def:approx-assignment-DS:assignment} of \Cref{def:approx-assignment-DS}), we know
\begin{align*}
    \dist(\sigma'(i, z), \varphi_{i}^{-1}(z))
    ~\le~ 2\rho_{i}
    ~=~ (3\Gamma)^{i}.
\end{align*}
Together with the fact $\diam(\varphi_{i}^{-1}(z)) \le \rho_{i} = \frac{1}{2} \cdot (3\Gamma)^{i}$ (\Cref{def:consistent:closeness} of \Cref{def:efficient-consistent-hashing}), the distance from $x$ to $S$ can be bounded as:
\begin{align*}
    \dist(x, S) 
    & ~\le~ \dist(x, \sigma'(i,z)) \\
    & ~\le~ \dist(\sigma'(i,z), \varphi_{i}^{-1}(z)) + \diam(\varphi_{i}^{-1}(z))
    \tag{$\varphi_{i}(x) = z$}\\
    & ~\le~ (3\Gamma)^{i} + \frac{1}{2} \cdot (3\Gamma)^{i} \\
    & ~=~ \frac{3}{2} \cdot (3\Gamma)^{i}.
\end{align*}
On the other hand, since $S(i - 1, z') = \emptyset$ and $z' = \varphi_{i - 1}(x)$, by \Cref{eq:bucket-center} again, we have
\[
    \min_{s \in S} \dist(s, \varphi_{i - 1}^{-1}(z'))
    ~=~ \frac{\rho_{i}}{\Gamma}
    ~>~ \frac{1}{2\Gamma} \cdot (3\Gamma)^{i - 1}.
\]
Therefore,
\begin{align*}
    \dist(x, S) 
    ~=~ \min_{s \in S} \dist(x, s)
    ~\ge~ \min_{s \in S} \dist(s, \varphi_{i - 1}^{-1}(z')) 
    ~>~ \frac{1}{2\Gamma} \cdot (3\Gamma)^{i - 1}.
\end{align*}
Overall, we have 
\begin{align*}
    \frac{1}{2\Gamma} \cdot (3\Gamma)^{i - 1}
    ~\le~ \dist(x, S)
    ~\le~ \dist(x, \sigma'(i, z))
    ~\le~ \frac{3}{2} \cdot (3\Gamma)^{i}.
\end{align*}
The arbitrariness of $(i, z) \in H'$ and $x\in \set(i, z)$ finishes the proof.
\end{proof}

\paragraph{Implicit Assignment.}

The following \Cref{claim:implicit-assignment} provides an implicit scheme for maintaining an approximate assignment, which follows directly from \Cref{claim:H-partition,claim:equidistant}.

\begin{claim}[Assignment]
\label{claim:implicit-assignment}
    For every $x \in X - S$, define $\sigma(x) \eqdef \sigma'(H'(x))$, where $H'(x) = (i_x, z_x) \in H$ denotes the unique pair such that $x \in \set(i_x, z_x)$.
    The assignment $\sigma :X - S \to S$ satisfies that 
    \begin{align*}
        & \forall x\in X - S,
        && \dist(x, \sigma(x))
        \le (3\Gamma)^{2} \cdot \dist(x,S).
    \end{align*}
\end{claim}

\subsection{Update Time Analysis}
\label{sec:approx-assignment-time}
In this section, we analyze the update time for the data structure of \Cref{def:approx-assignment-DS}.

\begin{lemma}[Update Time]
\label{lem:approx-assignment-DS:time}
For each update to either the point set $X$ or the center set $S$, the data structure of \Cref{def:approx-assignment-DS} can be maintained deterministically in $\tO(2^{O(\epsilon d)})$ time.
\end{lemma}

Since the updates are handled deterministically, the data structure naturally works against adaptive adversaries.

In the rest of this section, we would establish \Cref{lem:approx-assignment-DS:time} by addressing all considered dynamic explicit objects one by one.
Firstly, the preimages $\{X(i, z)\}_{(i, z) \in H}$ (\Cref{def:approx-assignment-DS:preimage} of \Cref{def:approx-assignment-DS}) are affected only by updates to the point set $X$, and not by updates to the center set $S$, and can be maintained in a straightforward manner.

\begin{claim}[Maintenance of Preimages]
\label{claim:maintain-preimage}
For every update (insertion or deletion) to the point set $X$, the preimages $\{X(i, z)\}_{(i, z) \in H}$ can be maintained in time $\tO(\Lambda)$.
\end{claim}

\begin{proof}
Given an insertion $X \gets X + x$ (resp.\ a deletion $X \gets X - x$), we compute all hash values $z_{i} = \varphi_{i}(x)$ for $i \in [0, m]$ in total time $(m + 1) \cdot \tO(\Lambda)$ (\Cref{def:efficient-consistent-hashing}) and augment $X(i, z_{i}) \gets X(i, z_{i}) + x$ (resp.\ diminish $X(i, z_{i}) \gets X(i, z_{i}) - x$) all preimages for $i \in [0, m]$ in total time $(m + 1) \cdot \tO(1)$.

The total update time is $(m + 1) \cdot \tO(\Lambda) + (m + 1) \cdot \tO(1) = \tO(\Lambda)$.
\end{proof}

Secondly, the low-level buckets $\{\set(i, z, z')\}_{(i,z,z')\in L}$ (\Cref{def:approx-assignment-DS:llb} of \Cref{def:approx-assignment-DS}) can likewise be maintained straightforwardly.

\begin{claim}[Maintenance of Low-Level Buckets]
\label{claim:maintain-low-level}
\begin{flushleft}
For every update (insertion or deletion) to the point set $X$, the low-level buckets $\{\set(\ell, z, z')\}_{(\ell,z, z') \in L}$ can be maintained in time $\tO(\Lambda)$.
\end{flushleft}
\end{claim}

\begin{proof}
Given an insertion $X \gets X + x$ (resp.\ a deletion $X \gets X - x$), we compute all hash values $z_{i} = \varphi_{i}(x)$ for $i \in [0, m]$ in total time $(m + 1) \cdot \tO(\Lambda)$ (\Cref{def:efficient-consistent-hashing}) and augment $\set(i, z_{i}, z_{i - 1}) \gets \set(i, z_{i}, z_{i - 1}) + x$ (resp.\ diminish $\set(i, z_{i}, z_{i - 1}) \gets \set(i, z_{i}, z_{i - 1}) - x$) all low-level buckets $(i, z_{i}, z_{i - 1})$ for $i \in [1, m]$ in total time $\tO(m)$.
The total update time is $(m + 1) \cdot \tO(\Lambda) + (m + 1) \cdot \tO(1) = \tO(\Lambda)$.
\end{proof}

Thirdly, the maintenance of the center subsets $\{S(i, z)\}_{(i, z) \in H}$ (\Cref{def:approx-assignment-DS:close-center} of \Cref{def:approx-assignment-DS}) only responds to updates to the center set $S$ and does not depend on updates to the point set $X$.

\begin{claim}[Maintenance of Close Centers]
\label{claim:maintain-close-center}
\begin{flushleft}
For every update (insertion or deletion) to the center set $S$, the center subsets $\{S(i, z)\}_{(i, z) \in H}$ can be maintained in time $\tO(\Lambda)$.
\end{flushleft}
\end{claim}

\begin{proof}
We consider an insertion or an deletion separately.

To insert a new center $s \in \Deld - S$, we perform the following for every $i \in [0, m]$: compute a subset $\Phi_{i} \subseteq \Phi$ (\Cref{def:consistent:efficiency} of \Cref{def:efficient-consistent-hashing}) such that
\begin{align*}
    \varphi_{i}(\ball(s, \tfrac{\rho_{i}}{\Gamma}))
    ~\subseteq~ \Phi_{i} \cap \varphi_{i}(\Deld)
    ~\subseteq~ \varphi_{i}(\ball(s, 2\rho_{i})).
\end{align*}
Then, augment $S(i, z) \gets S(i, z) + s$, for every $z \in \Phi_{i}$. Since every $\Phi_{i}$ can be computed in $\tO(\Lambda)$ time and has size at most $\Lambda$ (\Cref{def:consistent:efficiency} of \Cref{def:efficient-consistent-hashing}), the entire insertion procedure runs in total time $(m + 1) \cdot \tO(\Lambda) = \tO(\Lambda)$.

To verify correctness, consider a specific pair $(i, z) \in H$ such that $\dist(s, \varphi_{i}^{-1}(z)) \le \frac{\rho_{i}}{\Gamma} \iff z \in \varphi_{i}(\ball(s, \frac{\rho_{i}}{\Gamma})) \subseteq \Phi_{i}$.
The above procedure inserts $s$ into $S(i, z)$. The arbitrariness of $(i, z) \in H$ implies that
$\{s \in S: \dist(s, \varphi_{i}^{-1}(z)) \le \frac{\rho_{i}}{\Gamma} \} \subseteq S(i, z)$, for every $(i, z) \in H$, after the insertion.
Further, consider a specific pair $(i, z) \in H$ with $z\in \varphi_{i}(\Deld)$ such that $s$ is inserted into $S(i, z)$ by the above procedure. 
Then $z \in \Phi_{i}  \cap \varphi_{i}(\Deld) \subseteq \varphi_{i}(\ball(s, 2\rho_{i})) \implies \dist(s, \varphi_{i}^{-1}(z)) \le 2 \cdot\rho_{i}$.
Therefore, the arbitrariness of $(i, z) \in H$ ensures that
$S(i, z) \subseteq \{s \in S: \dist(s, \varphi_{i}^{-1}(z)) \le 2\rho_{i}\}$, for every $(i, z) \in H$, after the insertion. Overall, the above procedure maintains \Cref{def:approx-assignment-DS:close-center} of \Cref{def:approx-assignment-DS}.

To delete a current center $s \in S$, we perform the following for every $i \in [0, m]$: recompute the size-$(\le \Lambda)$ subset $\Phi_{i} \subseteq \Phi$ in $\tO(\Lambda)$ time (\Cref{def:consistent:efficiency} of \Cref{def:efficient-consistent-hashing}) -- the same as the one when $s$ was inserted into $S$ -- and diminish $S(i, z) \gets S(i, z) - s$, for every $z \in \Phi_{i}$.
Clearly, after this procedure, $s$ is no longer contained in any $S(i, z)$, so \Cref{eq:bucket-center} is maintained, and this deletion procedure runs in total time $(m + 1) \cdot \tO(\Lambda) = \tO(\Lambda)$.
\end{proof}

Fourthly, in alignment with maintenance of the low-level buckets $\{\set(i, z, z')\}_{(i, z, z') \in L}$ and the close center sets $\{S(i, z)\}_{(i, z) \in H}$,  we can synchronously maintain the high-level bucket descriptions $f(i, z) = \{(i, z, z') \in L : S(i - 1, z') = \emptyset,\ \set(i, z, z') \neq \emptyset\}$ for all $(i, z) \in H$ (\Cref{def:approx-assignment-DS:hlb} of \Cref{def:approx-assignment-DS}), as summarized in the following claim.

\begin{claim}[High-Level Buckets Description]
\label{claim:high-level-bucket}
\begin{flushleft}
The following guarantees hold.
\begin{enumerate}[font = \bfseries]
    \item\label{claim:high-level-bucket:1}
    When a single low-level bucket $\set(i, z, z')$ for some $(i, z, z') \in L$ is updated, the high-level buckets $\{f(i, z)\}_{(i, z) \in H}$ can be synchronously maintained in time $\tO(1)$.

    \item\label{claim:high-level-bucket:2}
    When a single close center set $S(i, z)$ for some $(i, z) \in H$ is updated, the high-level buckets $\{f(i, z)\}_{(i, z) \in H}$ can be synchronously maintained in time $\tO(\Lambda)$.
\end{enumerate}
\end{flushleft}
\end{claim}

\begin{proof}
We prove the two items separately.

\Cref{claim:high-level-bucket:1}:
Consider an update to $\set(i, z, z')$ for some $(i, z, z') \in L$. By definition (\Cref{def:approx-assignment-DS:hlb} of \Cref{def:approx-assignment-DS}), this update only affects $f(i, z)$.
Specifically, if $S(i - 1, z') \neq \emptyset$ and $\set(i, z, z')$ changes from non-empty to empty, then $z'$ should be removed from $f(i, z)$; conversely, if $S(i - 1, z') \neq \emptyset$ and $\set(i, z, z')$ changes from empty to non-empty, then $z'$ should be added to $f(i, z)$.
Clearly, either update can be performed in $\tO(1)$ time. 
In all other cases, $f(i, z)$ remains unchanged.

\Cref{claim:high-level-bucket:2}:
Consider an update to a single close center set.
If it is $S(m, z')$ for some $z' \in \Phi$ that is updated, then all $f(i, z)$ remain unchanged (\Cref{def:approx-assignment-DS:hlb} of \Cref{def:approx-assignment-DS}).

Otherwise, it is $S(i - 1, z')$ for some $i \in [1, m]$ and some $z' \in \Phi$ that is updated. Regarding a hash value $z \in \Phi$ with $\set(i, z, z') \ne \emptyset$ (\Cref{def:approx-assignment-DS:hlb} of \Cref{def:approx-assignment-DS}):\\
If $S(i - 1, z')$ changes from empty to non-empty, then $z'$ should be deleted from $f(i, z)$.\\
If $S(i - 1, z')$ changes from non-empty to empty, then $z'$ should be added to $f(i, z)$.\\
In all other cases, $f(i, z)$ remains unchanged.\\
Hence, it remains to identify all $z \in \Phi$ with $\set(i, z, z') \neq \emptyset$, which we refer to as \emph{to-update} hash values.
If $X(i - 1, z') = X \cap \varphi_{i - 1}^{-1}(z') = \emptyset$, then $\set(i, z, z') = X \cap \varphi_{i}^{-1}(z) \cap \varphi_{i - 1}^{-1}(z') = \emptyset$ for every $z \in \Phi$, so there is no to-update hash values.
Otherwise, we pick a specific point $x \in X(i - 1, z')$ and compute a size-$(\le \Lambda)$ subset $\Phi_{x} \subseteq \Phi$ with $\varphi_{i}(\ball(x, \frac{1}{\Gamma} \cdot (3\Gamma)^{i})) \subseteq \Phi_{x}$ in time $\tO(\Lambda)$ (\Cref{def:consistent:efficiency} of \Cref{def:efficient-consistent-hashing}).
Then, to identify a to-update hash value $z \in \Phi_{x}$ with $\set(i, z, z') \neq \emptyset$, we can consider an arbitrary point $y \in \set(i, z, z')$; given that $\varphi_{i - 1}(y) = z' = \varphi_{i - 1}(x)$ (\Cref{def:approx-assignment-DS:llb} of \Cref{def:approx-assignment-DS}), we can deduce $\dist(x, y) \le (3\Gamma)^{i - 1}$ (\Cref{def:consistent:closeness} of \Cref{def:efficient-consistent-hashing}) and thus
$$
    z
    ~=~ \varphi_{i}(y)
    ~\in~ \varphi_{i}(\ball(x, (3\Gamma)^{i - 1}))
    ~\subseteq~ 
    \varphi_{i}(\ball(x,\tfrac{1}{\Gamma}\cdot (3\Gamma)^{i}))
    ~\subseteq~
    \Phi_{x}.
$$
Namely, to identify all to-update hash values $z \in \Phi_{x}$ with $\set(i, z, z') \neq \emptyset$, it suffices to enumerate all hash values $z \in \Phi_{x}$ and check the conditions $\set(i, z, z') \neq \emptyset$, which runs in time $\tO(\Lambda)$.

This completes the proof.
\end{proof}

Since each update to the point set $X$ requires $\tO(\Lambda)$ time to maintain the low-level buckets $\{\set(i, z, z')\}_{(i, z, z') \in L}$ (\Cref{claim:maintain-low-level}), and each update to the center set $S$ takes $\tO(\Lambda)$ time to maintain the center subsets $\{S(i, z)\}_{(i, z) \in H}$ (\Cref{claim:maintain-close-center}), by \Cref{claim:high-level-bucket}, maintaining $\{f(i, z)\}_{(i, z) \in H}$ under each update to either $X$ or $S$ takes at most $\tO(\Lambda^2)$ time.

Fifthly, in alignment with maintenance of the close center sets $\{S(i, z)\}_{(i, z) \in H}$, we can synchronously maintain $H' = \{(i, z) \in H : S(i, z) \neq \emptyset\}$ and $\{\sigma'(i, z)\}_{(i, z) \in H'}$ (\Cref{def:approx-assignment-DS:partition,def:approx-assignment-DS:assignment} of \Cref{def:approx-assignment-DS}).

\begin{claim}[Partition and Assignment]
\label{claim:maintain-partition}
When a single close center set $S(i, z)$ for some $(i, z) \in H$ is updated, both $H'$ and  $\{\sigma'(i, z)\}_{(i, z) \in H'}$ can be maintained in time $\tO(1)$.
\end{claim}

\begin{proof}
When a single close center set $S(i, z)$ for some $(i, z) \in H$ is updated, we proceed as follows:  
\begin{itemize}
    \item If $S(i, z)$ changes from empty to non-empty, insert $(i, z)$ to $H'$ and set $\sigma'(i, z)$ to an arbitrary point in $S(i, z)$.
    
    \item If $S(i, z)$ changes from non-empty to empty, delete $(i, z)$ from $H'$ and discard $\sigma'(i, z)$.
    
    \item Otherwise, if the update to $S(i, z)$ deletes $\sigma'(i, z)$, set $\sigma'(i, z)$ to another point in $S(i, z)$.
\end{itemize}
Clearly, this procedure correctly maintains both $H'$ and $\{\sigma'(i, z)\}_{(i, z) \in H'}$, and takes $\tO(1)$ time.
\end{proof}

\begin{proof}[Proof of \Cref{lem:approx-assignment-DS:time}]
The total update time for the data structure of \Cref{def:approx-assignment-DS} follows by combining the results from \Cref{claim:maintain-preimage,claim:maintain-low-level,claim:maintain-close-center,claim:high-level-bucket,claim:maintain-partition}. Among these, the most expensive one is maintaining the mappings $\{f(i, z)\}_{(i, z) \in H}$, which requires $\tO(\Lambda^2)$ time. Therefore, the overall update time for the data structure is $\tO(\Lambda^2)$.
\end{proof}

\subsection{\texorpdfstring{Application: Maintaining Cluster Weight}{}}
\label{sec:approx-assignment-weight}

In this part, we consider the dynamic point set $X \subseteq \Deld$ to be weighted, with a weight function $w : X \to \mathbb{R}_{\ge 0}$. Thus, each update to $X$ corresponds to the insertion or deletion of a weighted point.

We extend the data structure of \Cref{def:approx-assignment-DS} to support explicit maintenance of \emph{center weights} for the dynamic center set $S \subseteq \Deld$ (\Cref{lem:maintain-weights}). These center weights correspond to the cluster sizes in some an approximate clustering.

\begin{lemma}[Size of Clusters]
\label{lem:maintain-weights}
There exists a data structure for handling a dynamic weighted point set $X \subseteq \Deld$ and a dynamic center set $S \subseteq \Deld$.
It maintains, for every $s \in S$, a weight $w_S(s) \ge 0$ such that for some point-to-center assignment $\sigma : X \to S$, 
\begin{enumerate}[font = \bfseries]
    \item\label{lem:maintain-weights:1} $\dist(x, \sigma(x)) \leq {O(\epsilon^{-3})} \cdot \dist(x, S)$, for every point $x \in X$;

    \item\label{lem:maintain-weights:2} $w_S(s) = w(\sigma^{-1}(s))$, for every center $s \in S$.
\end{enumerate}
The data structure works against an adaptive adversary and, with probability at least $1 - \frac{1}{\poly(\Delta^d)}$, has update time $\tilde{O}(2^{\epsilon d})$.
\end{lemma}

Before presenting our construction for \Cref{lem:maintain-weights}, we state the following corollary, which follows as a consequence of \Cref{lem:maintain-weights,lem:ANN-distance}, and will be useful in the analysis of restricted $k$-means (\Cref{sec:restricted_kmeans}).

\begin{corollary}[Importance of Centers]
\label{cor:order-approx-cluster-cost}
There exists a data structure that maintains an ordering $c_1, c_2, \ldots, c_{|S|}$ of centers $S$ such that for each $1 \leq i \leq |S|-1$, the following holds,
$$
    w_S(c_{i}) \cdot \Hat{\dist}^{2}(c_{i}, S - c_{i}) \leq w_S(c_{i+1}) \cdot \Hat{\dist}^{2}(c_{i+1}, S - c_{i + 1}),
$$
where $w_{S}$ is the weight function maintained by \Cref{lem:maintain-weights}, and for every $i\in [|S|]$, the value $\Hat{\dist}(c_{i}, S - c_{i})$ satisfies that 
\begin{align*}
    \dist(c_{i}, S - c_{i}) \le \Hat{\dist}(c_{i}, S - c_{i}) \le O(\epsilon^{-3/2})\cdot \dist(c_{i}, S - c_{i}).
\end{align*}
The data structure works against an adaptive adversary and, with probability at least $1 - \frac{1}{\poly(\Delta^d)}$, has amortized update time $\tilde{O}(2^{\epsilon d})$.
\end{corollary}

\begin{proof}
Since both the values $\{\widehat{\dist}(s, S - s)\}_{s \in S}$ and the weights $\{w_S(s)\}_{s \in S}$ are explicitly maintained by \Cref{lem:ANN-distance,lem:maintain-weights}, respectively, the desirable ordering $c_1, c_2, \ldots, c_{|S|}$ of the centers $S$ of centers $S$ can be maintained synchronously, with only a $\tO(1)$-factor increase in update time.
\end{proof}

We now turn to the construction of the data structure for \Cref{lem:maintain-weights}.

\subsection*{Data Structure for \Cref{lem:maintain-weights}.}
Built on top of the data structure of \Cref{def:approx-assignment-DS}, we additionally maintain the following explicit objects:
\begin{itemize}
    \item $w_L(i,z,z')\eqdef w(\set(i,z,z')) = \sum_{x\in \set(i,z,z')} w(x)$, for every $(i,z,z')\in L$.
    \item $w_H(i,z)\eqdef \sum_{z'\in f(i,z)} w_L(i,z,z') = \sum_{z'\in f(i,z)} w(\set(i,z,z'))$, for every $(i,z)\in H$.  
    \item $w_S(s) \eqdef \sum_{(i,z)\in H':\sigma'(i,z) = s} w_H(i,z) + w(s)$,\footnote{
        If $s\notin X$, we define $w(s) = 0$.
    }
    for every center $s\in S$.
\end{itemize}
By \Cref{lem:approx-assignment-init}, the data structure of \Cref{def:approx-assignment-DS} can be initialized with success probability $1 - \frac{1}{\poly(\Delta^d)}$; we henceforth assume that this initialization succeeds.

\paragraph{Update Time.}
Firstly, the values $\{w_L(i, z, z')\}_{(i, z, z') \in L}$ can be maintained synchronously with $\{\set(i, z, z')\}_{(i, z, z') \in L}$ with only $\tO(1)$ overhead: whenever a point $x$ is inserted into (resp.\ deleted from) $\set(i, z, z')$ for some $(i, z, z') \in L$, we increase (resp.\ decrease) $w_L(i, z, z')$ by $w(x)$.

Secondly, the values $\{w_H(i,z)\}_{(i,z)\in H}$ can be maintained synchronously with  $\{f(i,z)\}_{(i,z)\in H}$ and $\{w_L(i, z, z')\}_{(i, z, z') \in L}$ with only $\tO(1)$ overhead: whenever a value $z'$ is inserted into (resp.\ deleted from) $f(i, z)$ for some $(i, z) \in H$, we increase (resp.\ decrease) $w_H(i,z)$ by $w_L(i,z,z')$; whenever $w_L(i, z, z')$ is updated for some $(i, z, z') \in L$, we also update $w_H(i, z)$ accordingly if $z' \in f(i, z)$.

Finally, the weights $\{w_S(s)\}_{s\in S}$ can be maintained synchronously with $H'$ and $\{\sigma'(i,z)\}_{(i,z)\in H}$ with only $\tO(1)$ overhead: once a pair $(i,z)$ is inserted into (resp.\ deleted from) $H'$, we increase (resp.\ decrease) $w_S(\sigma'(i,z))$ by $w_H(i,z)$; once $\sigma'(i, z)$ changes from, say, $s$ to $s'$, we update $w_S(s)\gets w_S(s) - w_H(i,z)$ and $w_S(s')\gets w_S(s') + w_H(i,z)$.

Overall, the update time of the data structure for \Cref{lem:maintain-weights} remains asymptotically the same as that of the data structure in \Cref{def:approx-assignment-DS}, which is $\tO(\Lambda^{2}) = \tO(2^{O(\epsilon d)})$ by \Cref{lem:approx-assignment-DS:time}. Rescaling $\epsilon$ yields the update time $\tO(2^{\epsilon d})$ as in \Cref{lem:maintain-weights}.

\paragraph{Correctness.}
We construct an assignment from $X$ to $S$ as follows: Let $\sigma: X - S \to S$ be the assignment as defined in \Cref{claim:implicit-assignment}; we naturally extend this definition by setting $\sigma(x) \eqdef x$ for every $x \in X \cap S$.
Clearly, this construction gives us a full assignment from $X$ to $S$.
We then show that the weights $\{w_S(s)\}_{s \in S}$ and this assignment $\sigma \colon X \to S$ satisfy both \Cref{lem:maintain-weights:1,lem:maintain-weights:2} of \Cref{lem:maintain-weights}.

\Cref{lem:maintain-weights:1}:
For a center point $x \in X \cap S$, we have $\sigma(x) = x\in S$ by definition, so \Cref{lem:maintain-weights:1} follows directly.
For a non-center point $x \in X - S$, we have $\sigma(x) = \sigma'(H'(x))$ by definition, so \Cref{lem:maintain-weights:1} follows directly from \Cref{claim:implicit-assignment} and that $\Gamma = O(\epsilon^{-3/2})$.

\Cref{lem:maintain-weights:2}:
Observe that for every $(i, z) \in H$, the low-level buckets $\{\set(i, z, z')\}_{z' \in \Phi}$ are disjoint (\Cref{def:approx-assignment-DS:llb} of \Cref{def:approx-assignment-DS}). Therefore,
\begin{align*}
    w_H(i, z)
= \sum_{z'\in f(i, z)} w(\set(i, z, z'))
    = w\Big(\bigcup_{z' \in f(i, z)}\set(i, z, z')\Big)
    = w(\set(i, z)).
\end{align*}
where the last step follows from \Cref{eq:high-level-bucket}.
Now, for a specific center $s \in S$, we can write 
\begin{align*}
    w_S(s)
    & ~~=~ \sum_{(i,z)\in H':\sigma'(i,z) = s} w_H(s) + w(s)\\
    & ~=~ \sum_{(i,z)\in H':\sigma'(i,z)= s} w(\set(i,z)) + w(s)\\
    & ~=~ \sum_{(i,z)\in H': \sigma'(i,z) = s} \sum_{x\in \set(i,z)} w(s) + w(s).
\end{align*}
Since $\{\set(i, z)\}_{(i, z) \in H'}$ is a partition of $X - S$ (\Cref{claim:H-partition}) and $\sigma(x) = \sigma'(i, z)$ for every $(i, z) \in H'$ and every $x \in \set(i, z)$ (\Cref{claim:implicit-assignment}). It then follows that 
\begin{align*}
    w_S(s) ~=~ \sum_{x\in X - S: \sigma(x) = s} w(x) + w(s) = \sum_{x\in X:\sigma(x) = s} w(x).
\end{align*}
This finishes the proof of \Cref{lem:maintain-weights:2}.

\subsection{\texorpdfstring{Application: $D^{2}$-Sampling}{}}
\label{sec:approx-assignment:2}

In this part, we extend the data structure of \Cref{def:approx-assignment-DS} to support $D^2$-sampling over the dynamic weighted point set $X \subseteq \Deld$ with respect to the dynamic center set $S \subseteq \Deld$.
\begin{corollary}[$D^2$-Sampling]\label{lem:D2samlping}
There exists a data structure for handling a dynamic weighted point set $X \subseteq \Deld$ and a dynamic center set $S \subseteq \Deld$.
At any time, it can generate independent random points from $X$ such that each random point $y\in X$ is generated in time $\tO(1)$ and satisfies that
\begin{align}
    \label{eq:sampling-prob}
    & \forall x \in X:
    && \Pr[y = x] ~\ge~ \Omega(\epsilon^{-6}) \cdot \frac{w(x) \cdot \dist^{2}(x, S)}{\cost(X, S)}.
\end{align} 
The data structure works against an adaptive adversary and, with probability at least $1 - \frac{1}{\poly(\Delta^d)}$, has update time $\tilde{O}(2^{\epsilon d})$.
\end{corollary}

\begin{proof}
By \Cref{lem:approx-assignment-init}, the data structure of \Cref{def:approx-assignment-DS} can be initialized with probability $1 - \frac{1}{\poly(\Delta^d)}$; we henceforth assume that this initialization succeeds.

Built on top of the data structure of \Cref{def:approx-assignment-DS}, the data structure for \Cref{lem:D2samlping} additionally maintain the following explicit objects:
\begin{itemize}
    \item $w_L(i,z,z')\eqdef w(\set(i,z,z')) = \sum_{x\in \set(i,z,z')} w(x)$, for every $(i,z,z')\in L$.
    
    \item $w_H(i,z)\eqdef (3\Gamma)^{2i}\cdot   \sum_{z'\in f(i,z)} w_L(i,z,z') = \sum_{z'\in f(i,z)} w(\set(i,z,z'))$, for every $(i,z)\in H$.  
\end{itemize}
Similar to the analysis for \Cref{lem:maintain-weights}, the data structure for \Cref{lem:D2samlping} has an update time of $\tO(2^{\epsilon d})$. 
The remainder of this proof is devoted to implementing the sampling procedure, which proceeds in three steps.
\begin{enumerate}
    \item First, we sample a high-level bucket $(i,z)$ from $H'$ with probability proportional to $w_H(i,z)$.
    
    Since both the support $H'$ and the weights $\{w_H(i,z)\}_{(i,z) \in H'}$ are explicitly maintained, this can be implemented by storing the weighted set $H'$ in, e.g., a binary search tree, allowing for $\tO(1)$-time sampling and incurring only $\tO(1)$ overhead in the update time.
    
    \item Second, we sample a value $z'$ from $f(i,z)$ with probability proportional to $w_L(i,z,z')$.  
    Again, since both the support $f(i,z)$ and the weights $\{w_L(i,z,z')\}_{z'\in f(i,z)}$ are explicitly maintained, this step also takes $\tO(1)$ time.
    
    \item Third, we sample a point $y$ from $\set(i,z,z')$ with probability proportional to $w(y)$.  
    Again, since the weighted set $\set(i,z,z')$ is explicitly maintained (\Cref{def:approx-assignment-DS:llb} of \Cref{def:approx-assignment-DS}), this step can also be implemented in $\tO(1)$ time.
\end{enumerate}
Hence, the overall sampling time is $\tO(1)$.

It remains to show the correctness of the procedure.
Clearly, the final sampled point $y$ satisfies
\begin{align}
    \label{eq:sample-range}
    y ~\in~ \bigcup_{(i, z) \in H'} \bigcup_{z' \in f(i, z)} \set(i, z, z') ~=~ \bigcup_{(i, z) \in H'} \set(i, z) ~=~ X - S.
\end{align}
Consider a specific point $x \in X$.
If $x \in S$ is a center, then by \Cref{eq:sample-range}, we have $\Pr[y = x] = 0$. Since $\dist(x, S) = 0$, \Cref{eq:sampling-prob} trivially holds in this case.
If $x \notin S$ is not a center, then there exists a unique pair $(i_x, z_x) \in H'$ such that $x\in \set(i_x, z_x) = \bigcup_{z'\in f(i_x, z_x)} \set(i, z, z')$, since the collection $\{\set(i, z)\}_{(i, z) \in H'}$ forms a partition of $X - S$ (\Cref{claim:H-partition}).
Moreover, since the low-level buckets $\{\set(i_x, z_x, z')\}_{z' \in f(i_x, z_x)}$ are disjoint (\Cref{def:approx-assignment-DS:llb} of \Cref{def:approx-assignment-DS}), there exists a unique value $z_x' \in f(i_x, z_x)$ such that $x \in \set(i_x, z_x, z_x')$.
Thus, the point $x$ is sampled with probability
\begin{align*}
    \Pr[y = x]
    & ~=~ \frac{w_H(i_x, z_x)}{\sum_{(i, z) \in H'} w_H(i, z)} \cdot \frac{w_L(i_x, z_x, z'_x)}{\sum_{z' \in f(i_x, z_x)} w_L(i_x, z_x, z')} \cdot \frac{w(x)}{w(\set(i_x, z_x, z'_x))}\\
    & ~=~ \frac{(3\Gamma)^{2i_x} \cdot w(x)}{\sum_{(i, z) \in H'} w_H(i, z)}.
\end{align*}
By \Cref{claim:equidistant}, the numerator $(3\Gamma)^{2i_x} \ge \Omega(1) \cdot \dist^{2}(x, S)$, and the denominator
\begin{align*}
    \sum_{(i, z) \in H'} w_H(i, z)
    & ~=~ 
    \sum_{(i, z) \in H'} (3\Gamma)^{2i} \cdot \sum_{z'\in f(i, z)} w(\set(i, z, z'))\\
    & ~=~\sum_{(i, z) \in H'} w(\set(i, z)) \cdot (3\Gamma)^{2i}\\
    & ~=~\sum_{(i, z) \in H'} \sum_{x \in \set(i, z)} w(x) \cdot (3\Gamma)^{2i}\\
    & ~\le~\sum_{(i, z) \in H'} \sum_{x \in \set(i, z)} O(\Gamma^{4}) \cdot w(x) \cdot \dist^{2}(x, S) \tag{\Cref{claim:equidistant}}\\
    & ~=~ O(\Gamma^{4}) \cdot \sum_{x \in X - S} w(x) \cdot \dist^{2}(x, S) \tag{\Cref{claim:H-partition}}\\
    & ~=~ O(\Gamma^{4}) \cdot \cost(X, S).
\end{align*}
Therefore, we conclude that 
\begin{align*}
    \Pr[y = x]
    ~\ge~ \Omega(\Gamma^{4}) \cdot \frac{w(x) \cdot \dist^{2}(x, S)}{\cost(X, S)}
    ~=~ \Omega(\epsilon^{6}) \cdot \frac{w(x) \cdot \dist^{2}(x, S)}{\cost(X, S)}.
\end{align*}
This finishes the proof.
\end{proof}
 \newcommand{\Rout}{\Hat{R}}
\newcommand{\Ropt}{R^{*}}
\newcommand{\Rmid}{\Hat{R}^{*}}
\newcommand{\ropt}{r^{*}}
\newcommand{\rmid}{{\Hat{r}}^*}
\section{\texorpdfstring{Restricted $k$-Means}{}}
\label{sec:restricted_kmeans}

In this section, we consider the \emph{restricted $k$-means problem}:
Given a weighted point set $X \subseteq \Deld$, a center set $S \subseteq \Deld$, and an integer $r \ge 1$ as input, the goal is to select a size-$r$ subset $R \subseteq S$ that minimizes the $k$-means objective $\cost(X, S - R)$.

Assuming that the weighted point set $X$ and the center set $S$ are maintained by the data structures previously discussed (specifically \Cref{cor:order-approx-cluster-cost,lem:ANN-oracle,lem:maintain-weights}),
we design a subroutine that computes an approximate restricted $k$-means solution in sublinear time.

\begin{lemma}[Subroutine for Restricted $k$-Means]
\label{lem:restricted-clustering}
Given access to the data structures from \Cref{cor:order-approx-cluster-cost,lem:ANN-oracle,lem:maintain-weights}, there exists a subroutine that, on input an integer $1 \le r\le |S|$, computes a size-$r$ subset $R \subseteq S$ in time $\tO(2^{\epsilon d} r +r^{1 + \epsilon})$, such that the following holds with probability $1 - \frac{1}{\poly(\Delta^d)}$:
\begin{align*}
    \cost(X, S - R)
    ~\le~ O(\epsilon^{-13}) \cdot \OPT_{|S| - r}^{S}(X),
\end{align*}
where $\OPT_{|S| - r}^{S}(X)\eqdef \min_{S'\subseteq S:|S'| = |S| - r}\cost(X,S')$.
\end{lemma}

Recall that all data structures from \Cref{cor:order-approx-cluster-cost,lem:ANN-oracle,lem:maintain-weights} support dynamic updates to both $X$ and $S$ against an adaptive adversary, each with update time $\tO(2^{\epsilon d})$. 
By packing them into a single data structure together with the subroutine in \Cref{lem:restricted-clustering}, we obtain \Cref{lem:restricted-k-means-main} stated in \Cref{sec:data-structures-lemmas}.

\begin{proof}[Proof of \Cref{lem:restricted-clustering}]
    
We describe the subroutine in \Cref{alg:restricted-description}. We then analyze its running time and correctness, which together establish \Cref{lem:restricted-clustering}.

\begin{algorithm}[h!]
\caption{\label{alg:restricted-description}
Approximating restricted $k$-means.}
\begin{algorithmic}[1]
    \State Let $T_1 \gets \{c_{i}\}_{i \in [6r]} \subseteq S$ be the first $6r$ centers in the ordering maintained by \Cref{cor:order-approx-cluster-cost}. 
    \label{restricted:1}
    
    \State Let $T_2 \gets \{s_{i}\}_{i \in [6r]} \subseteq S - T_{1}$ be such that $\forall i \in [6r]$, $\dist(c_{i}, s_{i}) \le O(\epsilon^{-3/2}) \cdot \dist(c_{i}, S - T_1)$. 
    \label{restricted:2}

    \State Compute a size-$r$ subset $\Rout\subseteq T \eqdef T_1 + T_2$ such that \Comment{$T\subseteq S$.} 
    \begin{align*}
        \sum_{x\in T} w_S(x)\cdot \dist^{2}(x, T - \Rout) \le O(\epsilon^{-1})\cdot \min_{R\subseteq T: |R| = r} \sum_{x\in T} w_S(x) \cdot \dist^{2}(x, T - R),
    \end{align*}
    where $w_S:S\to \R_{\ge 0}$ is the weight function maintained by \Cref{lem:maintain-weights}.
    \label{restricted:3}
    
    \State \Return $\Rout$.
\end{algorithmic}
\end{algorithm}

\paragraph{Running Time.} 
We analyze the running time of \Cref{alg:restricted-description} line by line.

\Cref{restricted:1} takes trivially $\tO(r)$ time, since the entire ordering $c_1, \dots, c_{|S|}$ is explicitly maintained by \Cref{cor:order-approx-cluster-cost}.

\Cref{restricted:2} can be implemented in time $\tO(r \cdot 2^{\epsilon d})$: We temporarily delete all centers in $T_1$ from $S$ and update the data structure from \Cref{lem:ANN-oracle}. By querying this data structure for every center $c_{i} \in T_1$, we can obtain $s_{i}$ in time $\tO(2^{\epsilon d})$, thus constructing $T_2 = \{s_{i}\}_{i \in [6r]}$ in time $\tO(r \cdot 2^{\epsilon d})$. Finally, we insert all centers from $T_1$ back into $S$ and update the data structure from \Cref{lem:ANN-oracle}.

\Cref{restricted:3} can be implemented in time $\tO(r^{1 + \epsilon})$: Since the weight function $w_S: S \to \R_{\ge 0}$ is explicitly maintained by \Cref{lem:maintain-weights}, we can construct the weighted set $T = T_1 \cup T_2$, where every point $s \in T$ has weight $w_S(s)$, in time $\tO(r)$. 
Then, to compute $\Rout$, it suffices to run an $(|T| - r)$-means algorithm on the weighted set $T$, with centers selected from $T$, and obtain a solution $\Hat{S} \subseteq T$ with $|\Hat{S}| = |T| - r$. Then, we return $\Hat{R} \gets T - \Hat{S}$ as the desired output; for this, we can apply the static algorithm from \cite{abs-2504-03513}, which takes $\tO(r^{1 + \epsilon})$ time.

Notice that each step has a success probability (either from the data structures or the static algorithm) of $1 - \frac{1}{\poly(\Delta^d)}$; therefore, by the union bound, all steps succeed simultaneously with high probability, resulting in a total running time of $\tO(2^{\epsilon d} \cdot r + r^{1 + \epsilon})$.

\paragraph{Correctness Analysis.}
It remains to establish the correctness of \Cref{alg:restricted-description}. Specifically, we aim to show that the subset $\Rout \subseteq S$ computed in \Cref{restricted:3} satisfies
\begin{align}
    \label{eq:restricted-correctness}
    \cost(X, S - \Rout)
    ~\le~ O(\epsilon^{-13})\cdot \OPT_{|S| - r}^{S}(X).
\end{align}

We first have the following claim, which establishes a sufficient condition for \Cref{eq:restricted-correctness}. Intuitively, this condition states that the optimal value of the restricted $k$-means objective on the weighted set $T$ is at most a $\poly(1/\epsilon)$ factor larger than that on the weighted set $S\supseteq T$.

\begin{claim}\label{claim:restricted-sufficient}
    \Cref{eq:restricted-correctness} holds if 
    \begin{align*}
\min_{R\subseteq T: |R| = r} \sum_{s\in T} w_S(s) \cdot \dist^{2}(s, T - R)
        ~\le~ 
        O(\epsilon^{-6})\cdot \min_{R\subseteq S: |R| = r} \sum_{s\in S} w_S(s)\cdot \dist^{2}(s, S - R).
    \end{align*}
\end{claim}

We defer the proof of \Cref{claim:restricted-sufficient} to \Cref{sec:proof-restricted-sufficient}. 
\Cref{claim:restricted-sufficient} allows us to focus on the restricted problem for the weighted sets $S$ and $T$. From now on, for any center set $C \subseteq \Deld$, we denote the clustering objective of $C$ on $S$ and $T$ by $\cost(S, C) = \sum_{s \in S} w_S(s) \cdot \dist^{2}(s, C)$ and $\cost(T, C) = \sum_{s \in T} w_S(s) \cdot \dist^{2}(s, C)$, respectively. Hence, the sufficient condition stated in \Cref{claim:restricted-sufficient} becomes
\begin{align}
    \label{eq:restricted-correctness-sufficient}
    \min_{R \subseteq T: |R| = r} \cost(T, T - R) 
    ~\le~ 
    O(\epsilon^{-6}) \cdot \min_{R \subseteq S: |R| = r} \cost(S, S - R).
\end{align}
Let $\Ropt \eqdef \mathrm{argmin}_{R \subseteq S: |R| = r} \cost(S, S - R)$
denote the optimal solution to the restricted $k$-means on $S$. The following \Cref{claim:restricted-a-solution} shows the existence of an approximate solution $\Rmid$ by selecting points solely from $T_1 \subseteq T$. Moreover, \Cref{claim:restricted-every-solution} guarantees that this solution yields a comparable cost when evaluated on both $T$ and $S$.

\begin{claim}
    \label{claim:restricted-a-solution}
    There exists a size-$r$ subset $\Rmid\subseteq T_1$ such that 
    \begin{align*}
        \cost(S, S - \Rmid)
        ~\le~
        O(\epsilon^{-3})\cdot \cost(S, S - \Ropt).
    \end{align*}
\end{claim}
\begin{claim}
    \label{claim:restricted-every-solution}
    For any size-$r$ set $R\subseteq T_1$, it holds that 
    \begin{align*}
        \cost(T, T - R) 
        ~\le~
        O(\epsilon^{-3})\cdot \cost(S, S - R).
    \end{align*}
\end{claim}

We defer the proof of \Cref{claim:restricted-a-solution,claim:restricted-every-solution} to \Cref{sec:proof-restricted-a-solution,sec:proof-restricted-every-solution}, respectively.

We now show how these two claims together imply \Cref{eq:restricted-correctness-sufficient}. Let $\Rmid \subseteq T_1$ be a size-$r$ subset that satisfies $\cost(S, S - \Rmid) \leq O(\epsilon^{-3}) \cdot \cost(S, S - \Ropt)$ (\Cref{claim:restricted-a-solution}). Then, \Cref{eq:restricted-correctness-sufficient} follows from the derivation below.
\begin{align*}
    \text{LHS of \Cref{eq:restricted-correctness-sufficient}} 
    ~&=~ 
    \min_{R\subseteq T: |R| = r} \cost(T, T - R) \\
    &\le~ \cost(T, T - \Rmid)
    \tag{$\Rmid\subseteq T_1\subseteq T$, $|\Rmid| = r$}\\
    &\le~
    O(\epsilon^{-3}) \cdot \cost(S, S - \Rmid)
    \tag{\Cref{claim:restricted-every-solution}}\\
    &\le~ 
    O(\epsilon^{-6})\cdot \cost(S, S - \Ropt) 
    \tag{\Cref{claim:restricted-a-solution}}\\
    &=~
    O(\epsilon^{-6})\cdot \min_{R\subseteq S: |R| = r} \cost(S, S - R)
    \tag{Definition of $\Ropt$}\\
    &=~ 
    \text{RHS of \Cref{eq:restricted-correctness-sufficient}}.
\end{align*}
By \Cref{claim:restricted-sufficient}, this completes the correctness analysis of \Cref{alg:restricted-description}.
\end{proof}

\subsection{\texorpdfstring{Proof of \Cref{claim:restricted-sufficient}}{}}
\label{sec:proof-restricted-sufficient}

\begin{restate}[\Cref{claim:restricted-sufficient}]
\Cref{eq:restricted-correctness} (i.e., $\cost(X, S - \Rout) \le O(\epsilon^{-13}) \cdot \OPT_{|S| - r}^{S}(X)$) holds if 
\begin{align}
    \label{eq:restricted-correctness-sufficient-restate}
    \min_{R\subseteq T: |R| = r} \sum_{s\in T} w_S(s)\cdot \dist(s, T - R)
    ~\le~ 
    O(\epsilon^{-6})\cdot \min_{R\subseteq S: |R| = r} \sum_{s\in S} w_S(s)\cdot \dist^{2}(s, S - R).
\end{align}
\end{restate}
\begin{proof}

Recall the property of the weight function $w_S$ from \Cref{lem:maintain-weights}:
There exists an assignment $\sigma: X \to S$ such that 
(i) $\forall x \in X$, $\dist(x, \sigma(x)) \le O(\epsilon^{-3}) \cdot \dist(x, S)$; and 
(ii) $\forall s \in S$, $w_S(s) = w(\sigma^{-1}(s))$. Therefore,
\begin{align*}
    & \cost(X, S - \Rout) \\
    & ~=~ \sum_{x\in X} w(x) \cdot \dist^{2}(x, S - \Rout)\\
    & ~\le~ \sum_{x\in X} w(x) \cdot (\dist(x, \sigma(x)) + \dist(\sigma(x), S - \Rout))^{2} 
    \tag{triangle inequality}\\
    & ~\le~ \sum_{x\in X} w(x)\cdot 2\cdot (\dist^{2}(x, \sigma(x)) + \dist^{2}(\sigma(x), S - \Rout)) 
    \tag{AM-GM inequality}\\
    & ~\le~ O(\epsilon^{-6})\cdot \cost(X, S) + 2\cdot \sum_{s\in S} w_S(x)\cdot \dist^{2}(s, S - \Rout)
    \tag{Properties (i) and (ii)}\\
    & ~\le~ O(\epsilon^{-6})\cdot \OPT_{|S| - r}^{S}(X) + 2 \cdot \cost(S, S - \Rout).
\end{align*}
The last step applies $\cost(X, S) \le \OPT_{|S| - r}^{S}(X)$ and $\cost(S, S - \Rout) = \sum_{s\in S} w_S(x)\cdot \dist^{2}(s, S - \Rout)$. Since every center $s \in S - \Rout$ has a zero distance to $S - \Rout$, we have

\begin{align*}
    \cost(S, S - \Rout)
    & ~=~ \sum_{s\in \Rout} w_S(s)\cdot \dist^{2}(s, S - \Rout)\\
    & ~\le~ \sum_{s\in \Rout} w_S(s)\cdot \dist^{2}(s, T - \Rout)
    \tag{$T\subseteq S$}\\
    & ~\le~ \sum_{s\in T} w_S(s)\cdot \dist^{2}(s, T - \Rout)
    \tag{$\Rout\subseteq T$}\\
    & ~\le~ O(\epsilon^{-1})\cdot \min_{R\subseteq T: |R| = r} \cost(T, T - R)
    \tag{\Cref{restricted:3} of \Cref{alg:restricted-description}}
\end{align*}
Let $\Ropt_{X} \subseteq S$ be the optimal solution to the restricted $k$-means on $X$, i.e., $|\Ropt_{X}| = r$ and $\cost(X, S - \Ropt_{X}) = \OPT_{|S| - r}^{S}(X)$.
It follows that
\begin{align*}
    \cost(S, S - \Rout)
    & ~\le~ O(\epsilon^{-1}) \cdot \min_{R \subseteq T: |R| = r} \cost(T, T - R)\\
    & ~\le~ O(\epsilon^{-7}) \cdot \min_{R \subseteq S: |R| = r} \cost(S, S - R)
    \tag{\Cref{eq:restricted-correctness-sufficient-restate}}\\
    & ~\le~ O(\epsilon^{-7}) \cdot \cost(S, S - \Ropt_{X})
    \tag{$\Ropt_{X} \subseteq S$, $|\Ropt_{X}| = r$}\\
    & ~=~ O(\epsilon^{-7}) \cdot \sum_{s\in S} w_{S}(s) \cdot \dist^{2}(s, S - \Ropt_{X})\\
    & ~=~ O(\epsilon^{-7}) \cdot \sum_{s\in S} \sum_{x \in \sigma^{-1}(s)} w(x) \cdot \dist^{2}(s, S - \Ropt_{X})\\
    & ~\le~ O(\epsilon^{-7}) \cdot \sum_{s \in S} \sum_{x \in \sigma^{-1}(s)} w(x) \cdot 2 \cdot (\dist^{2}(x, s) + \dist^{2}(x, S - \Ropt_{X}))\\
    & ~=~ O(\epsilon^{-7}) \cdot \sum_{x\in X} w(x) \cdot 2 \cdot (\dist^{2}(x, \sigma(x)) + \dist^{2}(x, S - \Ropt_{X}))\\ 
    & ~\le~ O(\epsilon^{-13}) \cdot \cost(X, S) + O(\epsilon^{-7}) \cdot \cost(X, S - \Ropt_{X})
    \tag{Property~(i)}\\
    & ~\le~ O(\epsilon^{-13}) \cdot \OPT_{|S| - r}^{S}(X).
\end{align*}
Putting everything together, we obtain \Cref{eq:restricted-correctness}:
\begin{align*}
    \cost(X, S - \Rout)
    & ~\le~ O(\epsilon^{-6})\cdot \OPT_{|S| - r}^{S}(X) + O(\epsilon^{-13}) \cdot \OPT_{|S| - r}^{S}(X)\\
    & ~\le~ O(\epsilon^{-13}) \cdot \OPT_{|S| - r}^{S}(X).
\end{align*}
This finishes the proof.
\end{proof}

\subsection{\texorpdfstring{Proof of \Cref{claim:restricted-a-solution}}{}}
\label{sec:proof-restricted-a-solution}

\begin{restate}[\Cref{claim:restricted-a-solution}]
    There exists a size-$r$ subset $\Rmid\subseteq T_1$ such that 
    \begin{align*}
        \cost(S, S - \Rmid)
        ~\le~
        O(\epsilon^{-3})\cdot \cost(S, S - \Ropt).
    \end{align*}
\end{restate}

\paragraph{The Construction of $\Rmid$.} Recall the definition of $\Ropt = \argmin_{R\subseteq S: |R| = r} \cost(S, S - R)$.
To construct the desired approximate solution $\Rmid \in T_1$, we need the following technical claim, whose proof is deferred to \Cref{sec:proof-restricted-specific-replacement}. 

\begin{claim}
    \label{claim:restricted-specific-replacement}
    There exists a size-$r$ subset $T'_{1}\subseteq T_1$ such that the following guarantees hold:
    \begin{enumerate}[font = \bfseries]
        \item\label{restrict-good-subset-p1} $T'_{1}\cap \Ropt = \emptyset$.
        \item\label{restrict-good-subset-p3} For every $s\in \Ropt$, there exists $s'\in S - (\Ropt + T'_{1})$ such that $\dist(s,s') = \dist(s, S - \Ropt)$.
        \item \label{restrict-good-subset-p2} For every $s\in T'_{1}$, there exists $s'\in S - (\Ropt + T'_{1})$ such that $\dist(s,s')\le 2\cdot \dist(s, S - s)$.
    \end{enumerate}
\end{claim}
Our construction of $\Rmid$ based on the subset $T'_{1} \subseteq T_1$ from \Cref{claim:restricted-specific-replacement} proceeds in the following three steps.

\begin{enumerate}
    \item Let $\Ropt_{1} \eqdef \Ropt \cap T_{1}$ and $\Ropt_{2} \eqdef \Ropt - T_{1}$.
    \item Since $|\Ropt_{2}| \le |\Ropt| = r$, we can pick an arbitrary size-$|\Ropt_{2}|$ subset of $T'_{1}$, denoted by $\Rmid_{2}$.
    \item Define $\Rmid \eqdef \Ropt_{1} + \Rmid_{2}$.
\end{enumerate}

\paragraph{Correctness of $\Rmid$.}
Since $\Ropt_{1} \subseteq T_1$ and $\Rmid_{2} \subseteq T'_{1}\subseteq T_{1}$, we have $\Rmid = \Ropt_{1} + \Rmid_{2} \subseteq T_{1}$. 
By \Cref{restrict-good-subset-p1} of \Cref{claim:restricted-specific-replacement}, we have $\Rmid_{2} \cap \Ropt_{1} = \emptyset$.
Since $|\Rmid_{2}| = |\Ropt_{2}|$, it then follows that $|\Rmid| = |\Ropt_{1}| + |\Rmid_{2}| = |\Ropt_{1}| + |\Ropt_{2}| = r$. Hence,
it remains to show that 
\begin{align}
    \label{eq:correctness-Rmid}
    \cost(S, S - \Rmid)
    ~\le~ O(\epsilon^{-3}) \cdot \cost(S, S - \Ropt).
\end{align}
To this end, observe that 
\begin{align*}
    \cost(S, S - \Rmid) ~&=~ \sum_{s\in S} w_S(s)\cdot \dist^{2}(s, S - \Rmid)\\
    &= \sum_{s\in \Rmid} w_S(s)\cdot \dist^{2}(s, S - \Rmid)\\
    &= \sum_{s\in \Ropt_{1}} w_S(s)\cdot \dist^{2}(s, S - \Rmid) + \sum_{s\in \Rmid_{2}} w_S(s)\cdot \dist^{2}(s, S - \Rmid).
\end{align*}
Similarly, we can obtain
\begin{align*}
    \cost(S, S - \Ropt) ~=~ \sum_{s\in \Ropt_{1}} w_S(s) \cdot \dist^{2}(s, S - \Ropt) + \sum_{s\in \Ropt_{2}} w_S(s) \cdot \dist^{2}(s, S - \Ropt).
\end{align*}
As a result, we can further divide the proof of \Cref{eq:correctness-Rmid} into the following two parts:
\begin{align}
    &\sum_{s\in \Ropt_{1}} w_S(s)\cdot \dist^{2}(s, S - \Rmid) ~\le~ \sum_{x\in \Ropt_{1}} w_S(s)\cdot \dist^{2}(s, S - \Ropt),\label{eq:correctness-Rmid-1}\\
    &\sum_{s\in \Rmid_{2}} w_S(s)\cdot \dist^{2}(S, S - \Rmid) ~\le~ O(\epsilon^{-3})\cdot \sum_{s\in \Ropt_{2}} w_S(s) \cdot \dist^{2}(s, S - \Ropt). \label{eq:correctness-Rmid-2}
\end{align}

\Cref{eq:correctness-Rmid-1}:
Consider a specific center $s \in \Ropt_{1}$.
By \Cref{restrict-good-subset-p3} of \Cref{claim:restricted-specific-replacement}, there exists $s'\in S - (\Ropt + T'_{1})$ such that $\dist(s,s')  = \dist(s, S - \Ropt)$. 
Since $\Rmid = \Ropt_{1} + \Rmid_{2} \subseteq \Ropt + T'_{1}$, it follows that $s'\in S - \Rmid$. Therefore,
\begin{equation*}
    \dist(s, S - \Rmid)
    ~\le~ \dist(s, s')
    ~=~ \dist(s, S - \Ropt).
\end{equation*}
The arbitrariness of $s\in \Ropt_{1}$ implies \Cref{eq:correctness-Rmid-1}.

\Cref{eq:correctness-Rmid-2}: 
Let $t \eqdef |\Ropt_{2}| = |\Rmid_{2}|$, and write $\Ropt_{2} = \{\ropt_{1},\dots, \ropt_{t}\}$, $\Rmid_{2} = \{\rmid_{1}, \dots, \rmid_{t}\}$. 
Recall that $T_1 = \{c_1,\dots, c_{6r}\}\subseteq S$ (\Cref{restricted:1} of \Cref{alg:restricted-description}) is the first $6r$ centers in the ordering maintained by \Cref{cor:order-approx-cluster-cost}. 
Since $\Rmid_{2} \subseteq T'_{1} \subseteq T_{1}$ and $\Ropt_{2} = \Ropt - T_{1} \implies \Ropt_{2} \cap T_{1} = \emptyset$, it follows that
\begin{align*}
    & \forall i \in [t]:
    && w_{S}(\rmid_{i})\cdot \Hat{\dist}^{2}(\rmid_{i}, S - \rmid_{i}) ~\le~ w_{S}(\ropt_{i}) \cdot \Hat{\dist}^{2}(\ropt_{i}, S - \ropt_{i}).
\end{align*}
Now consider a fixed $i \in [t]$. We can deduce from \Cref{cor:order-approx-cluster-cost} that 
\begin{align*}
    w_{S}(\rmid_{i}) \cdot \dist^{2}(\rmid_{i}, S - \rmid_{i})
    & ~\le~ w_{S}(\rmid_{i}) \cdot \Hat{\dist}^{2}(\rmid_{i}, S - \rmid_{i})\\ 
    & ~\le~ w_{S}(\ropt_{i}) \cdot \Hat{\dist}^{2}(\ropt_{i}, S - \ropt_{i})\\
    & ~\le~ O(\epsilon^{-3}) \cdot w_{S}(\ropt_{i}) \cdot \dist^{2}(\ropt_{i}, S - \ropt_{i})\\
    & ~\le~ O(\epsilon^{-3}) \cdot w_{S}(\ropt_{i}) \cdot \dist^{2}(\ropt_{i}, S - \Ropt).
\end{align*}
Since $\rmid_{i} \in \Rmid_{2} \subseteq T'_{1}$, by \Cref{restrict-good-subset-p2} of \Cref{claim:restricted-specific-replacement}, there exists $s'\in S - (\Ropt + T'_{1})$ such that $\dist(\rmid_{i}, s') \le 2\cdot \dist(\rmid_{i}, S - \rmid_{i})$. Recall that $\Rmid = \Ropt_{1} + \Rmid_{2} \subseteq \Ropt + T'_{1}$. It follows that $s'\in S - \Rmid$ and thus 
\begin{align*}
    w_{S}(\rmid_{i}) \cdot \dist^{2}(\rmid_{i}, S - \Rmid)
    & ~\le~ w_{S}(\rmid_{i}) \cdot \dist^{2}(\rmid_{i}, s')\\
    & ~\le~ 4 \cdot w_{S}(\rmid_{i}) \cdot \dist^{2}(\rmid_{i}, S - \rmid_{i})\\
    & ~\le~ O(\epsilon^{-3}) \cdot w_{S}(\ropt_{i}) \cdot \dist^{2}(\ropt_{i}, S - \Ropt).
\end{align*}
The arbitrariness of $i \in [t]$ implies \Cref{eq:correctness-Rmid-2}.

This finishes the proof of \Cref{claim:restricted-a-solution}.

\subsection{\texorpdfstring{Proof of \Cref{claim:restricted-specific-replacement}}{}}
\label{sec:proof-restricted-specific-replacement}

\begin{restate}[\Cref{claim:restricted-specific-replacement}]
There exists a size-$r$ subset $T'_{1}\subseteq T_1$ such that the following guarantees hold:
\begin{enumerate}[font = \bfseries]
    \item\label{restrict-good-subset-p1-re}
    $T'_{1} \cap \Ropt = \emptyset$.
    
    \item\label{restrict-good-subset-p3-re}
    For every $s\in \Ropt$, there exists $s' \in S - (\Ropt + T'_{1})$ such that $\dist(s,s') = \dist(s, S - \Ropt)$.
    
    \item\label{restrict-good-subset-p2-re}
    For every $s\in T'_{1}$, there exists $s' \in S - (\Ropt + T'_{1})$ such that $\dist(s,s') \le 2 \cdot \dist(s, S - s)$.
\end{enumerate}
\end{restate}

\paragraph{\Cref{restrict-good-subset-p1-re,restrict-good-subset-p3-re}.}
By removing the following centers from $T_1$, we can ensure that \Cref{restrict-good-subset-p1-re,restrict-good-subset-p3-re} hold for every remaining subset $T'_{1} \subseteq T_1$:
\begin{itemize} 
    \item Let $T_{1} \gets T_{1} - \Ropt$. 
    It is then clear that \Cref{restrict-good-subset-p1-re}, i.e., $T'_{1} \cap \Ropt = \emptyset$, holds for any subset $T'_{1} \subseteq T_1$.
    
    \item For every $s \in \Ropt$, select an arbitrary $s' \in S - \Ropt$ such that $\dist(s, s') = \dist(s, S - \Ropt)$, and let $T_1 \gets T_1 - s'$.
    This ensures that $s' \in S - (\Ropt + T_1)$, which implies \Cref{restrict-good-subset-p3-re}, i.e., $s' \in S - (\Ropt + T'_{1})$ for any subset $T'_{1} \subseteq T_1$.
\end{itemize}
Clearly, the above process will remove at most $2 \cdot |\Ropt| = 2r$ centers from $T_1$, thus leaving at least $6r - 2r = 4r$ centers still in $T_1$ (\Cref{restricted:1} of \Cref{alg:restricted-description}).

\paragraph{\Cref{restrict-good-subset-p2-re}.}
Based on the above discussion, we can assume that $|T_1| \geq 4r$ and that \Cref{restrict-good-subset-p1-re,restrict-good-subset-p3-re} hold for every subset $T'_{1} \subseteq T_1$.
We now show the existence of a size-$r$ subset $T'_{1} \subseteq T_1$ that also satisfies \Cref{restrict-good-subset-p2-re}. 
To this end, let us consider the exact nearest neighbor assignment $g: T_1 \to S$, i.e., $\dist(s, g(s)) = \dist(s, S - s)$ for every $s \in T_1$.
We then proceed with a case study.

\paragraph{Case 1: $|g(T_{1})| \ge 3r$.}
We construct a desired subset $T'_{1} \subseteq T_1$ as follows.
First, consider the set $T'_{1} = \bigcup_{y\in g(T_{1}) - \Ropt} g^{-1}(y)$.
This set has size
\begin{align*}
    |T'_{1}|
    ~=~ \sum_{y \in g(T_{1}) - \Ropt} |g^{-1}(y)|
    ~\ge~ |g(T_{1}) - \Ropt|
    ~\ge~ |g(T_{1})| - |\Ropt|
    ~\ge~ 3r - r
    ~=~ 2r.
\end{align*}
Further, for every $s \in T'_{1}$, it holds that $g(s) \in g(T_{1}) - \Ropt = S - \Ropt$ and $\dist(s, g(s)) = \dist(s, S - s)$. 
We then apply the following claim from~\cite{Thorup04}.

\begin{claim}[{\cite[Lemma 21]{Thorup04}}]
\label{claim:select}
Given a set $A$, suppose that every element $a \in A$ has at most one successor $s(a)\in A$. Then, we can mark at least $|A|/2$ elements such that, if $a$ is marked, then $s(a)$ is unmarked.
\end{claim}

Applying \Cref{claim:select}, we can mark at least $|T'_{1}| / 2 \ge r$ centers in $T'_{1}$ such that, for every marked $s\in T'_{1}$,  $g(s)$ is unmarked; we then remove all unmarked centers from $T'_{1}$. Afterward, we still have $|T'_{1}| \geq r$, and for every $s \in T'_{1}$, $g(s) \in S - (\Ropt + T'_{1})$; namely, there exists $s' \in S - (\Ropt + T'_{1})$ such that $\dist(s, s') = \dist(s, g(s)) = \dist(s, S - s) \le 2 \cdot \dist(s, S - s)$.
Therefore, any size-$r$ subset of $T'_{1}$ satisfies \Cref{restrict-good-subset-p2-re}.

\paragraph{Case 2: $|g(T_{1})| < 3r$.} 
Let $t \eqdef |g(T_{1})| < 3r$ and write $g(T_{1}) = \{g_{1},\dots,g_t\}$. We construct a desired subset $T'_{1} \subseteq T_{1}$ as follows: First, we partition the subset $T_{1}$ into groups
$\{g^{-1}(g_{i})\}_{i \in [t]}$.
For every $i \in [t]$, we define $s'_{i} \eqdef \argmin_{s \in g^{-1}(g_{i})} \dist(s, g_{i})$.
Then, the subset $T'_{1}$ is defined as 
\begin{align}
    \label{eq:T1-case2}
    T'_{1} ~\eqdef~ \bigcup_{i \in [t]} (g^{-1}(g_{i}) - s'_{i}). 
\end{align}
Clearly, we have $|T'_{1}| \ge |T_1| - t > 4r - 3r = r$.
Now consider a specific center $s \in T'_{1}$; we have $g(s) = g_{i}$ for some $i \in [t]$. We will show that $s'_{i} \in S - (\Ropt + T'_{1})$ and $\dist(s, s'_{i}) \leq 2 \cdot \dist(s, S - s)$, thus proving \Cref{restrict-good-subset-p2-re} for this $s \in T'_{1}$.

Since $g(s) = g_{i}$, it follows that $\dist(s, g_{i}) = \dist(s, S - s)$. Additionally, by the definition of $s'_{i}$, we have $\dist(s'_{i}, g_{i}) \leq \dist(s, g_{i})$. Hence, by applying the triangle inequality, we have
\begin{align*}
    \dist(s, s'_{i}) ~\le~ \dist(s, g_{i}) + \dist(s'_{i}, g_{i}) ~\le~ 2\cdot \dist(s, g_{i}) ~=~ 2\cdot \dist(s, S - s).
\end{align*} 
Since $s'_{i} \in T_1\subseteq S$, it remains to show that $s'_{i} \notin \Ropt + T'_{1}$: Recall that $T_1$ already satisfies (\Cref{restrict-good-subset-p1-re}) $T_1 \cap \Ropt = \emptyset$, and thus $s'_{i} \in T_1 \implies s'_{i} \notin \Ropt$. Moreover, by definition (\Cref{eq:T1-case2}), $s'_{i} \notin T'_{1}$.
In conclusion, \Cref{restrict-good-subset-p2-re} holds for this $s \in T'_{1}$. The arbitrariness of $s \in T'_{1}$ implies \Cref{restrict-good-subset-p2-re}.

This completes the proof of \Cref{claim:restricted-specific-replacement}

\subsection{\texorpdfstring{Proof of \Cref{claim:restricted-every-solution}}{}}
\label{sec:proof-restricted-every-solution}

\begin{restate}[\Cref{claim:restricted-every-solution}]
For any size-$r$ set $R\subseteq T_1$, it holds that 
\begin{align*}
    \cost(T, T - R) 
    ~\le~ O(\epsilon^{-3})\cdot \cost(S, S - R).
\end{align*}
\end{restate}

\begin{proof}
Consider a specific size-$r$ subset $R \subseteq T_{1}$.
Since $T_1 \cap T_{2} = \emptyset$ (\Cref{restricted:1,restricted:2} of \Cref{alg:restricted-description}), we have $T - R = (T_{1} - R) + T_2$ and $S - R = (S - T_{1}) + (T_{1} - R)$. 
Therefore, we have
\begin{align*}
    \cost(T, T - R)
    & ~=~ \sum_{s\in R} w_S(s) \cdot \dist^{2}(s, T - R)\\
    & ~=~ \sum_{s\in R} w_S(s) \cdot \min\{ \dist^{2}(s, T_{1} - R),\ \dist^{2}(s, T_{2}) \}.
\end{align*}
Now, for every center $s \in R \subseteq T_{1}$, by \Cref{restricted:2} of \Cref{alg:restricted-description}, there exists $s' \in T_{2}$ such that $\dist(s, T_{2}) \le \dist(s,s') \le O(\epsilon^{-3/2})\cdot \dist(s, S - T_{1})$. Hence, it follows that
\begin{align*}
    \cost(T, T - R)
    & ~=~ \sum_{s \in R} w_S(s) \cdot \min\{ \dist^{2}(s, T_{1} - R), \dist^{2}(s, T_{2}) \}\\
    & ~\le~ \sum_{s\in R} w_S(s) \cdot \min\{ \dist^{2}(s, T_{1} - R), O(\epsilon^{-3}) \cdot \dist^{2}(s, S - T_{1}) \}\\
    & ~\le~ O(\epsilon^{-3}) \cdot \sum_{s \in R} w_S(s) \cdot \dist^{2}(s, S - R)\\
    & ~=~ O(\epsilon^{-3}) \cdot \cost(S, S - R).
\end{align*}
This finishes the proof.
\end{proof}
 \section{\texorpdfstring{Augmented $k$-Means}{}}
\label{sec:augmented_kmeans}

In this section, we consider the \emph{augmented $k$-means problem}:
Given a weighted point set $X \subseteq \Deld$, a center set $S \subseteq \Deld$ and a parameter $a\ge 1$ as input, the goal is to compute a size-$a$ subset $A \subseteq \Deld$ that minimizes the $k$-means objective $\cost(X, S + A)$.

\begin{lemma}[Augmented $k$-Means]\label{alg:augmented}
Given access to the data structure of \Cref{lem:D2samlping}, there exists a subroutine that, on input integer $a\ge 1$, computes a size-$\Theta(a \cdot \epsilon^{-6} \cdot d \log\Delta)$ subset $A \subseteq \Deld$ in time $\tilde{O}(2^{\epsilon d} \cdot a)$, such that the following holds with probability $1 - \frac{1}{\poly(\Delta^d)}$:
\begin{align*}
    \cost(X,S + A) \leq O(1) \cdot 
    \min_{A' \subseteq \Deld, |A| = a}  \cost(X, S + A'). 
\end{align*}
\end{lemma}
Recall that the data structure of \Cref{lem:D2samlping} support dynamic updates to both $X$ and $S$ against an adaptive adversary, with update time $\tO(2^{\epsilon d})$. Hence, \Cref{alg:augmented} directly implies \Cref{lem:augmented-k-means-main}, as stated in \Cref{sec:data-structures-lemmas}.

Without loss of generality, we assume throughout this section that $|X| > \Omega(a\cdot \epsilon^{-6}\cdot \log n)$ is sufficiently large, as otherwise the algorithm could simply return the entire dataset $X$ as the output.

We present the subroutine for \Cref{alg:augmented} in \Cref{alg:augmented-description}, and establish its running time and correctness in \Cref{lem:augment-running-time} and \Cref{lem:augment-correctness}, respectively. Together, these results complete the proof of \Cref{alg:augmented}.

\begin{algorithm}[h!]
\caption{\label{alg:augmented-description}
Approximating augmented $k$-Means.}
\begin{algorithmic}[1]
    \State $A\gets \emptyset$
    \Comment{The set of augmented centers}
    
    \ForAll{$i \in [a + 1]$}
    \label{alg:augment-for}
        \State Sample $t \eqdef \Theta(\epsilon^{-6}\cdot d\log \Delta)$ i.i.d.\ points $\{x_{i}\}_{i \in [t]}$ from $X$, such that $\Pr[x_{i} = y] \ge \Omega(\epsilon^{6}) \cdot \frac{w(y)\cdot \dist^{2}(y, S + A)}{\cost(X, S + A)}$ for every $y \in X$.
        \label{alg:augment-sample}
        
        \State $A \gets A + \{x_{i}\}_{i\in [t]}$ \label{alg:augment-add}
    \EndFor
    \Return $A$\;
\end{algorithmic}
\end{algorithm}

\begin{claim}[Fast Implementation of \Cref{alg:augmented-description}]
    \label{lem:augment-running-time}
    Given access to the data structure of \Cref{lem:D2samlping}, we can implement \Cref{alg:augmented-description} in time $\tO(n^{\epsilon}\cdot a)$.
\end{claim}
\begin{proof}
    We modify \Cref{alg:augmented-description} as follows: in each of the $a+1$ iterations, instead of only adding the points $\{x_{i}\}_{i\in [t]}$ (sampled in \Cref{alg:augment-sample}) to the set of augmented centers $A$, we also add them to the center set $S$ and update the data structure of \Cref{lem:D2samlping}. This update takes a total of $\tO(2^{\epsilon d})$ time per iteration.
With this modification, \Cref{alg:augment-sample} becomes equivalent to sampling points $x \in X$ with probability at least $\Omega(\epsilon^{6}) \cdot \frac{w(x)\cdot \dist(x, S)}{\cost(X, S)}$. By querying the data structure of \Cref{lem:D2samlping} $t$ times, we obtain $t$ independent samples in total time $\tO(2^{\epsilon d})\cdot t = \tO(2^{\epsilon d})$.
    Hence, the total running time for the entire {\bf for} loop is $\tO(2^{\epsilon d}\cdot a)$. 
After the {\bf for} loop completes, we perform an additional step: 
    we remove the augmented centers $A$ from $S$ and update the data structure from \Cref{lem:D2samlping} accordingly.
    This step takes $|A| \cdot \tO(2^{\epsilon d}) = \tO(2^{\epsilon d} \cdot a)$ time.
    
    Therefore, the overall running time of this implementation is $\tO(2^{\epsilon d})$.
\end{proof}

Clearly, \Cref{alg:augmented-description} returns a set $A$ of size $\tO(a\cdot \epsilon^{-6}\cdot d\log\Delta)$.
We now state the following claim regarding the correctness of $A$.

\begin{claim}[Correctness of \Cref{alg:augmented-description}]
    \label{lem:augment-correctness}
    Let $A$ denote the (random) output of \Cref{alg:augmented-description}. 
    Then, with probability at least $1 - \frac{1}{\poly(\Delta^d)}$, 
    \begin{align*}
        \cost(X, S + A) \le 16\cdot \min_{A'\subseteq\Deld:|A'| = a} \cost(X, S + A').
    \end{align*}
\end{claim}

The proof of \Cref{lem:augment-correctness} follows almost directly from~\cite[Theorem 1]{DBLP:conf/approx/AggarwalDK09}, where a similar sampling procedure is analyzed to obtain a bi-criteria approximation to the vanilla clustering problem. For completeness, we provide a proof in \Cref{sec:augment-proof}.

\subsection{\texorpdfstring{Proof of \Cref{lem:augment-correctness}}{}}
\label{sec:augment-proof}

\subsection*{Effect of $D^{2}$-sampling on reducing bad clusters}

A crucial step in the analysis is to understand the behavior of the (distorted) $D^{2}$-sampling procedure, as described in \Cref{alg:augment-sample} of \Cref{alg:augmented-description} -- namely, sampling a random point $x \in X$ such that 
\begin{align}
    \label{eq:Dsquared-sampling}
    \forall y\in X,\quad \Pr[x = y] \ge \Omega(\epsilon^{6})\cdot \frac{w(y)\cdot \dist^{2}(y, S + A)}{\cost(X,S + A)}.
\end{align}
To this end, we define a potential function for any center set $A$ that augments $S$: the number of {\em bad clusters}. This quantity will play a key role in the proof.

Let $A^{*} = \{c_{1},\dots,c_a\}$ denote the optimal set of $a$ centers that minimizes $\cost(X, S + A^{*})$.
Let
\begin{align}
    \label{eq:A0}
    X_0\eqdef \{x \in X: \dist(x, S) \le \dist(x, A^{*}) \},
\end{align}
and let $\{X_{1},\dots,X_a\}$ be the clustering of $X - X_0$ with respect to $A^{*} = \{c_{1},\dots, c_a\}$. Then, 
\begin{align}
    \label{eq:Aj}
    \forall j\in [a], \forall x \in X_{j},\quad \dist(x, c_{j}) = \dist(x, S + A^{*}).
\end{align}
The bad clusters with respect to $A$, denoted by $\calB(A)$, are those whose clustering objective under $S + A$ remains large. Formally:
\begin{align}
    \label{eq:bad-cluster}
    \calB(A) \eqdef \left\{X_{j} : \cost(X_{j}, S + A) > 8\cdot \cost(X_{j}, S + A^{*}) \right\}.
\end{align}
We also define the good clusters as $\calG(A)\eqdef \{X_{1},\dots, X_a\} \setminus \calB(A)$.
Throughout the proof, we slightly abuse notation by writing $j \in \calB(A)$ to refer to the index $j$ such that $X_j \in \calB(A)$.

We establish the following claim regarding the bad clusters.

\begin{claim}\label{claim:bad-cluster}
    For any center sets $A, A'$ with $A\subseteq A'$, it holds that $|\calB(A)|\le a$ and $\calB(A)\supseteq \calB(A')$.
\end{claim}
\begin{proof}
    The bound $|\calB(A)|\le a$ follows directly from the definition of bad clusters (see \Cref{eq:bad-cluster}). 
    To see that $\calB(A) \supseteq \calB(A')$, consider any $X_j \in \calB(A')$. Since $A \subseteq A'$, we have
    $\cost(X_j, S + A) \ge \cost(X_j, S + A') > 8\cdot \cost(X_j, S + A^{*})$, and hence $X_j \in \calB(A)$ as well. 
\end{proof}

The following claim shows that if all clusters are good, then the augmented center set $S + A$ achieves a good approximation:
\begin{claim}
    Let $A$ be a center set. If $\calB(A) = \emptyset$, then $\cost(X, S + A) \le 8\cdot \cost(X, S + A^{*})$.
\end{claim}
\begin{proof}
    By \Cref{eq:A0}, we have $\cost(X_0, S+ A)\le \cost(X_0, S) = \cost(X_0, S + A^{*})$. 
    Since all clusters are good, by \Cref{eq:bad-cluster}, we have 
    $\cost(X_{j}, S + A)\le 8\cdot \cost(X_{j}, S + A^{*})$ for all $j\in [a]$. Therefore,
    \begin{align*}
        \cost(X, S + A) ~&=~ \cost(X_0, S + A) + \sum_{j = 1}^a \cost(X_{j}, S + A)\\
        &\le~ \cost(X_0, S + A^{*}) + \sum_{j = 1}^a 8\cdot \cost(X_{j}, S + A^{*})\\
        &\le~ 8\cdot \cost(X, S + A^{*}).
    \end{align*}
    The claim follows.
\end{proof}

Next, we investigate the behavior of $D^{2}$-sampling, and show its effectiveness in reducing the number of bad clusters, as stated below.

\begin{claim}[Effect of $D^{2}$-Sampling on Reducing Bad Clusters]
\label{claim:reduce-bad-clusters}
\begin{flushleft}
Suppose that a center set $A$ satisfies $\cost(X, S + A) > 16 \cdot \cost(X, S + A^{*})$, then a random point $x \in X$ whose distribution satisfies \Cref{eq:Dsquared-sampling} leads to $|\calB(A + x)| \le |\calB(A)| - 1$  with probability at least $\Omega(\epsilon^{6})$.
\end{flushleft}
\end{claim}

\begin{proof}
For every bad cluster $X_{j}\in \calB(A)$ -- its center $c_{j}$ satisfies \Cref{eq:Aj} -- we define
\begin{align}
    \widehat{X}_{j}
    ~\eqdef~ \left\{y\in X_{j}:\ \dist^{2}(y,c_{j}) \le \frac{2}{w(X_{j})}\cdot \cost(X_{j}, c_{j})\right\}.
    \label{eq:reduce-bad-clusters:1}
\end{align}
A random point $x \in X$ sampled via $D^{2}$-sampling (\Cref{alg:augment-sample} of \Cref{alg:augmented-description}) turns out to satisfy:
\begin{enumerate}[label=(\alph*)]
    \item\label{item:Bj-prob}
    $\Pr[\text{$x \in \widehat{X}_{j}$ for some $j \in \calB(A)$}] \ge \Omega(\epsilon^{6})$.
    
    \item\label{item:Bj-cost}
    If $x \in \widehat{X}_{j}$ for some $j \in \calB(A)$, then $\cost(X_{j}, x) \le 8\cdot \cost(X_{j}, S + A^{*})$. 
\end{enumerate}
We first leverage \Cref{item:Bj-prob,item:Bj-cost} to establish \Cref{claim:reduce-bad-clusters}:
Suppose that the random point $x \in X$ satisfies $x \in \widehat{X}_{j}$ for some $j \in \calB(A)$ -- this occurs with probability at least $\Omega(\epsilon^{6})$ (\Cref{item:Bj-prob}) -- then $\cost(X_{j}, S + A + x) \le \cost(X_{j}, x) \le 8 \cdot \cost(X_{j}, S + A^{*})$ (\Cref{item:Bj-cost}), so every such bad cluster $j \in \calB(A)$ regarding $A$ becomes a good cluster $j \notin \calB(A + x)$ regarding $A + x$ (\Cref{eq:bad-cluster}).
And since $\calB(A + x) \subseteq \calB(A)$ (\Cref{claim:bad-cluster}), we conclude with $|\calB(A + x)| \le |\calB(A)| - 1$.

\vspace{.1in}
\noindent
{\bf \Cref{item:Bj-prob}.}
According to the $D^{2}$-sampling procedure (\Cref{alg:augment-sample} of \Cref{alg:augmented-description}), we have
\begin{align*}
    \Pr[\text{$x \in \widehat{X}_{j}$ for some $j \in \calB(A)$}]
    & ~\ge~ \Omega(\epsilon^{6})\cdot  \frac{\sum_{j \in \calB(A)} \cost(\widehat{X}_{j}, S + A)}{\cost(X, S + A)}\\
    & ~=~ \Omega(\epsilon^{6}) \cdot \underbrace{\frac{\sum_{j \in \calB(A)} \cost(X_{j}, S + A)}{\cost(X, S + A)}\Bigg.}_{(\dagger)}
    \cdot \underbrace{\frac{\sum_{j \in \calB(A)} \cost(\widehat{X}_{j}, S + A) }{\sum_{j \in \calB(A)} \cost(X_{j}, S + A)}\Bigg.}_{(\ddagger)}.
\end{align*}
It remains to show that $(\dagger) = \Omega(1)$ and $(\ddagger) = \Omega(1)$.

For the term $(\dagger)$, we deduce that
\begin{align*}
    (\dagger)
    & ~=~ 1 - \frac{\cost(X_0, S + A) + \sum_{j \ne 0 \land j \notin \calB(A)} \cost(X_{j}, S + A)}{\cost(X, S + A)}\\
    & ~\ge~ 1 - \frac{\cost(X_0, S + A^{*}) + \sum_{j \ne 0 \land j \notin \calB(A)} 8 \cdot \cost(X_{j}, S + A^{*})}{16\cdot \cost(X, S + A^{*})}\\
    & ~\ge~ \frac{1}{2}.
\end{align*}
Here the second step applies a combination of
$\cost(X_0, S + A^{*}) = \cost(X_0, S) \ge \cost(X_0, S + A)$ (\Cref{eq:A0}),
$\cost(X_{j}, S + A) \le 8 \cdot \cost(X_{j}, S + A^{*})$ when $j \ne 0 \land j \notin \calB(A)$ (\Cref{eq:bad-cluster}), and
$\cost(X, S + A) > 16 \cdot \cost(X, S + A^{*})$ (the premise of \Cref{claim:reduce-bad-clusters}).

For the term $(\ddagger)$, let us consider a specific bad cluster $j \in \calB(A)$.
Let $\alpha \eqdef \sqrt{w(X_{j})} \cdot \dist(c_{j}, S + A)$ and $\beta \eqdef \sqrt{\cost(X_{j}, c_{j})}$ for brevity.
We can upper-bound $\cost(X_{j}, S + A)$ using these parameters.
\begin{align*}
    \cost(X_{j}, S + A)
    & ~\le~ \sum_{y \in X_{j}} w(y)\cdot \big(\dist(y, c_{j}) + \dist(c_{j}, S + A)\big)^{2}
    \tag{triangle inequality}\\
    & ~\le~ \sum_{y \in X_{j}} w(y)\cdot\big(2 \cdot \dist^{2}(y, c_{j}) + 2 \cdot \dist^{2}(c_{j}, S + A)\big)
    \tag{AM-GM inequality}\\
    & ~=~ 2\beta^{2} + 2\alpha^{2}.
\end{align*}
This together with $\cost(X_{j}, S + A) \ge 8 \cdot \cost(X_{j}, c_{j}) = 8\beta^{2}$ (\Cref{eq:bad-cluster}) also gives $\alpha \ge \sqrt{3} \beta$.

Also, we can lower-bound $\cost(\widehat{X}_{j}, S + A)$ analogously.
First, following the triangle inequality and \Cref{eq:reduce-bad-clusters:1}, we have $\dist(y, S + A) \ge \dist(c_{j}, S + A) - \dist(y,c_{j}) \ge (\alpha - \sqrt{2}\beta) / \sqrt{w(X_{j})}$, for every point $y \in \widehat{X}_{j}$.
Second, following \Cref{eq:reduce-bad-clusters:1} and a simple average argument, we have $w(X_{j} - \widehat{X}_{j}) \le \frac{w(X_{j})}{2} \le w(\widehat{X}_{j})$.
Altogether, we deduce that
\begin{align*}
    \cost(\widehat{X}_{j}, S + A)
    ~=~ \sum_{y \in \widehat{X}_{j}} w(y)\cdot \dist^{2}(y, S + A)
    ~\ge~ \frac{1}{2} \cdot (\alpha - \sqrt{2}\beta )^{2}.
\end{align*}

These two observations in combination with arbitrariness of $j \in \calB(A)$ implies that
\begin{align*}
    (\ddagger)
    ~\ge~ \frac{\sum_{j \in \calB(A)} \cost(\widehat{X}_{j}, S + A)}{\sum_{j \in \calB(A)} \cost(X_{j}, S + A)}
    ~\ge~ \frac{\left(\alpha - \sqrt{2}\beta\right)^{2}}{4\beta^{2} + 4\alpha^{2}}
    ~\ge~ \frac{(\sqrt{3} - \sqrt{2})^{2}}{4 + 4 \cdot (\sqrt{3})^{2}}
    ~=~ \frac{5 - 2\sqrt{6}}{16}.
\end{align*}
Here the last step follows from elementary algebra; specifically, the equality holds when $\alpha = \sqrt{3}\beta$.

\vspace{.1in}
\noindent
{\bf \Cref{item:Bj-cost}.}
If $x \in \widehat{X}_{j}$ for some $j \in \calB(A)$, then the triangle and AM-GM inequalities together give
\begin{align*}
    \cost(X_{j}, x)
    & ~\le~ \sum_{y\in X_{j}} w(y)\cdot \big(2 \cdot \dist^{2}(y, c_{j}) + 2 \cdot \dist^{2}(x ,c_{j})\big)\\
    & ~\le~ 2\cdot \cost(X_{j}, c_{j}) + 4\cdot\cost(X_{j}, c_{j})
    \tag{\Cref{eq:reduce-bad-clusters:1}}\\
& ~=~ 6\cdot \cost(X_{j}, S + A^{*}).
    \tag{\Cref{eq:Aj}}
\end{align*}
This finishes the proof of \Cref{claim:reduce-bad-clusters}.
\end{proof}

\begin{proof}[Proof of \Cref{lem:augment-correctness}]
For every $i \in [a + 1]$, let $A_{i}$ denote the center set $A$ after the $i$-th iteration of the {\bf for} loop (\Cref{alg:augment-for,alg:augment-sample,alg:augment-add}) in \Cref{alg:augmented-description}. Also, let $A_0\eqdef\emptyset$.
We define the following two events for every $i \in [a + 1]$:
\begin{itemize}
    \item {\bf Event $\calE_{i}$:} $\cost(X, S + A_{i-1}) \le 16\cdot \cost(X, S + A^{*})$.
    \item {\bf Event $\calF_{i}$:} $|\calB(A_{i})| \le |\calB(A_{i-1})| - 1$.
\end{itemize}
We prove \Cref{lem:augment-correctness} in two steps.

\paragraph{Step 1: $\land_{i \in [a + 1]} (\calE_{i} \lor \calF_{i}) \implies$  \Cref{lem:augment-correctness}.}
Conditional on $\land_{i \in [a + 1]} (\calE_{i} \lor \calF_{i})$, at least one event $\calE_{i}$ some $i\in [a + 1]$ must occur; otherwise, we immediately see a contradiction $|\calB(A_{a + 1})|
\le |\calB(A_{0})| - (a + 1)
\le a - (a + 1) = -1$ to the nonnegativity of $|\calB(A_{a + 1})|$.

Consider an arbitrary index $i\in [a+1]$ such that event $\calE_{i}$ occurs.
\Cref{alg:augmented-description} guarantees that $A_{i-1} \subseteq A_{a + 1}$, and hence
$$
\cost(X, S + A_{a + 1})\le \cost(X, S + A_{i-1}) \le 16\cdot \cost(X, S + A^{*}).
$$
Since $A_{a+1}$ is the output of \Cref{alg:augmented-description}, this completes the proof of \Cref{lem:augment-correctness}.

\paragraph{Step 2: $\Pr[\land_{i \in [a + 1]} (\calE_{i} \lor \calF_{i})]\ge 1 - \frac{1}{\poly(\Delta^d)}$.}
It remains to lower-bound the probability of this desirable case:
\begin{align}
\label{eq:pr-Ai-Bi}
    \Pr\big[ \land_{i \in [a + 1]} (\calE_{i} \lor \calF_{i}) \big]
    = 1 - \Pr\big[\lor_{i \in [a + 1]} (\neg \calE_{i} \land \neg \calF_{i}) \big]
    \ge 1 - \sum_{i \in [a+1]} \Pr\big[ \neg \calE_{i} \land \neg \calF_{i} \big],
\end{align}
where the last step applies the union bound. Then, consider a fixed $i\in [a+1]$. We have
\begin{align*}
    \Pr\big[ \neg \calE_{i} \land \neg \calF_{i} \big]
    ~\le~ \Pr\big[ \neg \calF_{i} ~\big\vert~ \neg \calE_{i} \big]
    ~=~\Pr\big[|\calB(A_{i})| > |\calB(A_{i-1})| - 1 ~\big\vert~ \neg \calE_{i}\big].
\end{align*}
Recall that $A_{i} = A_{i - 1} + \{x_{1},\dots,x_t\}$, where $x_{1},\dots, x_t$ are the independent points sampled in the $i$-th iteration by \Cref{alg:augment-sample}. According to \Cref{claim:reduce-bad-clusters}, we have 
\begin{align*}
    \forall j\in [t],\quad \Pr\big[|\calB(A_{i-1} + x_{j})| > |\calB(A_{i-1})| - 1 ~\big\vert~ \neg \calE_{i}\big] \le 1 - \Omega(\epsilon^6).
\end{align*}
Let $t = O(\epsilon^{-6}\cdot d\log(\Delta))$ be sufficiently large, we have 
\begin{align*}
    \Pr\big[\forall j\in [t], |\calB(A_{i-1} + x_{j})| > |\calB(A_{i-1})| - 1 ~\big\vert~ \neg \calE_{i}\big]
    \le \left(1 - \Omega(\epsilon^6)\right)^t \le \frac{1}{\poly(\Delta^d)}.
\end{align*}
For every $j\in [t]$, by \Cref{claim:bad-cluster}, we have $\calB(A_{i-1} + x_{j}) \supseteq \calB(A_{i})$. Hence, if $|\calB(A_{i})| > |\calB(A_{i-1})| - 1$, then it must be that $\forall j\in [t], |\calB(A_{i-1} + x_{j})| > |\calB(A_{i-1})| - 1$. As a result,
\begin{align*}
    \Pr[\neg \calE_{i} \land \neg \calF_{i}]
    ~&\le~ \Pr\bigg[|\calB(A_{i})| > |\calB(A_{i-1})| - 1 ~\bigg\vert~ \neg \calE_{i}\bigg]\\
    &\le~ \Pr\bigg[\forall j\in [t], |\calB(A_{i-1} + x_{j})| > |\calB(A_{i-1})| - 1 ~\bigg\vert~ \neg \calE_{i}\bigg] \\
    &\le~ \frac{1}{\poly(\Delta^d)}.
\end{align*}
Finally, we apply the union bound over all $i \in [a+1]$, plug the result into \Cref{eq:pr-Ai-Bi}, and obtain $\Pr[\land_{i \in [a + 1]} (\calE_{i} \lor \calF_{i})] \ge 1 - \frac{1}{\poly(\Delta^d)}$.

This finishes the proof.
\end{proof}

\newpage
\part{Our Dynamic Algorithm}
\label{part:full-alg}

\paragraph{Roadmap.}
In this part of the paper, we provide our algorithm, its implementation, and its analysis.
We start by explaining the framework of \cite{BCF24} in \Cref{sec:alg:describe}, {\em with some major modifications that are necessary to ensure that we can implement it using our geometric data structures} (see \Cref{part:data-structures}).
Next, in \Cref{sec:implementation}, we show how to implement this modified framework in Euclidean spaces  with  an update time of $\tilde{O}(n^\epsilon)$.
Finally, in \Cref{sec:alg-analysis,sec:approx-analysis,sec:recourse-analysis,sec:update-time-analysis-full}, we provide the analysis of the algorithm.

\begin{theorem}
\label{full-part-theorem:main}
For every sufficiently small $\epsilon > 0$, there is a randomized dynamic algorithm for Euclidean $k$-means  with $\text{poly}(1/\epsilon)$-approximation ratio, $\tilde O(n^\epsilon)$ amortized update time and $\tilde O(1)$ amortized recourse. If the length of the update sequence is polynomially bounded by $n$ (the maximum number of input points at any time), then the approximation, update time and recourse guarantees of the algorithm hold with probability $1-1/\text{poly}(n)$ across the entire sequence of updates.
\end{theorem}

Later, in \Cref{part:from-n-to-k}, we \textbf{improve the update time of the algorithm to $\tilde{O}(k^\epsilon)$} by considering \Cref{full-part-theorem:main} as a black box.
We do this in a way that the probability of the guarantees of the algorithm remains $1-1/\poly(n)$. 
Hence, our final result is as follows.

\begin{theorem}
\label{full-part-theorem:main2}
For every sufficiently small $\epsilon > 0$, there is a randomized dynamic algorithm for Euclidean $k$-means  with $\text{poly}(1/\epsilon)$-approximation ratio, $\tilde O(k^\epsilon)$ amortized update time and $\tilde O(1)$ amortized recourse. If the length of the update sequence is polynomially bounded by $n$ (the maximum number of input points at any time), then the approximation, update time and recourse guarantees of the algorithm hold with probability $1-1/\text{poly}(n)$ across the entire sequence of updates.
\end{theorem}

\paragraph{Notation.}
From this point forward, we assume that $ \polyup $ is a polynomial on $1/\epsilon$ which is an upper bound for all of the polynomials in \Cref{sec:data-structures-lemmas}.
As a result, the guarantee of every one of these data structures and subroutines holds for the fixed polynomial $\polyup$.
If we fix the value of $\epsilon$, we can first compute the value of $\polyup$ and then the parameters of the algorithm are tuned w.r.t.~this value.
We have a widely used parameter $\lambda$ in the algorithm that is set to $\lambda = (\polyup)^2$.  
\begin{align}\label{eq:value-of-polyup-lambda}
 \polyup & \text{: A polynomial of $1/\epsilon$ as an upper bound on all the approximation guarantees} \nonumber \\
 & \text{in main lemmas in Sections \ref{sec:data-structures-lemmas}}, \nonumber \\ 
 \lambda & = (\polyup)^2 .
\end{align}
Throughout the rest of the paper, we use notation $\tilde{O}$ to hide $\poly(d\epsilon^{-1}\log (n\Delta))$.
$d$ is the dimension of the space. $n$ is the maximum size of the space, and $[0,\Delta]$ is the range of the coordinates of the points.
Note that the aspect ratio of the space would be $\sqrt{d}\Delta$ since we only consider points in the grid $\Deld$.

\paragraph{Preprocessing.}
\label{sec:full-version-preprocessing}
Throughout \Cref{sec:implementation,sec:alg-analysis,sec:approx-analysis,sec:recourse-analysis,sec:update-time-analysis-full}, we assume that the dimension of the input space $d$ is $O(\log n)$. Hence, the running time performance of all the data structures in \Cref{part:data-structures} becomes $\tilde{O}(n^\epsilon)$.
The assumption $d=O(\log n)$ can be achieved by the well-known Johnson-Lindenstrauss  transform \cite{JL84}.
Assume that the input space is $\R^m$.
We can initialize a random matrix $A$ of size $m \times d$ where $d = O(\log n)$, and whenever we have an insertion of a point $x \in \R^m$, we apply the transformation (compute $Ax \in \R^d$) and then round it to a discrete point in $[\Delta]^d$.
Then, we consider feeding this new dynamic space into our algorithm.
Note that according to this transformation, we will have an additive $\tilde{O}(m)$ in the update time of the algorithm, but we omit to write this down from now on since it only appears because of this preprocessing step, and the rest of the algorithm will not depend on that.

\begin{remark}\label{remark:robust-adaptive}

    Our algorithms in \Cref{full-part-theorem:main,full-part-theorem:main2} work against an oblivious adversary,
    due to the use of the Johnson--Lindenstrauss transform~\cite{JL84} and dynamic sparsification.
    Without these steps, our algorithm works against an adaptive adversary, but has an update time of $2^{O(\epsilon d)} \cdot \tO(n^{\epsilon})$, where the $2^{O(\epsilon d)}$ factor comes from our data structures in \Cref{sec:data-structures-lemmas}.

    Another issue that may arise from using the Johnson-Lindenstrauss transform is that, after dimension reduction via $A$, our algorithm maintains a set of centers $S$ in the reduced space rather than in the original space, and computing the inverse $A^{-1}s$ is often inefficient for an arbitrary $s$.
    Our solution is that, at the cost of a $O(1)$-factor loss in approximation, we can ensure that the maintained centers lie among points that have been inserted into the dataset (possibly deleted). This is achieved by maintaining, for each center, an approximate nearest neighbor among inserted points, which increases the update time by only $\tO(n^{\epsilon})$ via our data structure in \Cref{sec:data-structures-lemmas}.
    Since we can also maintain the pairs $(x, Ax)$ for every inserted point $x$ in the original space, we can efficiently maintain the inverse image $S'$ of $S$ under $A$, which is a $\poly(1/\epsilon)$-approximation in the original space.

\end{remark}

\section{\texorpdfstring{The Algorithmic Framework of \cite{BCF24}}{}}
\label{sec:alg:describe}

\paragraph{Roadmap.}
In this section, we explain how the algorithmic framework of \cite{BCF24} works, with some significant modifications that are necessary for our purpose.
We start by defining the modified version of the main tool used in the algorithm in \Cref{subsec:robust-sol}, which is called the robust solution.
Then, we continue with the description of the algorithm in \Cref{sec:alg-describe}.

\begin{remark}
Although the algorithm of \cite{BCF24} is primarily designed for $k$-median, it also works for $k$-means.
The only consideration is that the analysis of \cite{BCF24} depends on the fact that the integrality gap of the LP-relaxation for $k$-median is constant.
Since this fact is also correct for $k$-means (see Theorem 1.3 in \cite{integrality-gap-k-means}), we can extend the analysis of \cite{BCF24} for $k$-means.
The constant factors in the analysis might change, and here we provide a thorough analysis of the algorithm specifically for \textbf{Euclidean $k$-means}.    
\end{remark}

\subsection{Robust Solution}\label{subsec:robust-sol}

Before we explain the algorithm of \cite{BCF24}, we start with the necessary definitions for describing the algorithm.
Since our goal is to achieve an update time of $\neps$, we slightly changed some of the structural definitions of \cite{BCF24} that enable us to work with approximate balls around points instead of exact balls.

\begin{definition}\label{def:robust}
    Let $t \geq 0$ be an integer and $(x_0,x_1,\ldots,x_t)$ be a sequence of $t+1$ points of $\Deld$.
    We call this sequence \textit{$t$-robust} w.r.t.~the current data set $X$, if there exists a sequence $(B_0,B_1, \cdots, B_t)$ of subsets of $X$ such that the following hold,
    \begin{enumerate}
\item For each $0 \leq i \leq t$, we have $\ball(x_i,\lambda^{3i}) \subseteq B_i \subseteq \ball(x_i, \lambda^{3i+1})$.
        \item For each $1 \leq i \leq t$, at least one of the following hold,
        \begin{itemize}
            \item 
            Either 
            \begin{equation}\label{cond:robust1}
              \frac{\cost(B_i, x_i)}{w(B_i)} \geq \lambda^{6i-4} \ \text{and} \ x_{i-1} = x_i,
            \end{equation}
            \item 
            Or
            \begin{equation}\label{cond:robust2}
                \frac{\cost(B_i, x_i)}{w(B_i)} \leq \lambda^{6i-2} \ \text{and} \ \cost(B_i, x_{i-1}) \leq \min \{ (\polyup)^3 \cdot \OPT_1(B_i), \ \cost(B_i,x_i) \}. 
            \end{equation}
        \end{itemize}
    \end{enumerate}
    Note that both of these cases can happen simultaneously, and in the second case, $x_{i-1}$ can be equal to $x_i$ as well.
    Moreover, we call a point $x$, $t$-robust if there exists a $t$-robust sequence $(x_0,x_1,\ldots,x_t)$ such that $x = x_0$.
\end{definition}

\begin{definition}\label{def:robust-solution}
    We call a center set $S$, \textit{robust} w.r.t.~$X$ if and only if each center $u \in S$ satisfies the following
    \begin{equation}\label{cond:robust}
        u \text{ is at least } t\text{-robust, where } t \text{ is the smallest integer satisfying } \lambda^{3t} \geq \dist(u, S-u)/\lambda^{10}.
    \end{equation}
\end{definition}

\subsection{Description of The Algorithm}
\label{sec:alg-describe}

In this section, we describe the algorithm framework of \cite{BCF24}.
For simplicity.
Later in \Cref{sec:implementation}, we provide our algorithm using this framework in Euclidean spaces that has a small update time of $\neps$.
Note that since we are considering new \Cref{def:robust} and we will implement subroutines of the algorithm approximately, the parameters of the algorithm here change significantly.
The constant values in different parts of the algorithm are extremely tied together.
For simplicity, in this section, we use big $O$ notation to describe their algorithm, which makes it easier to read.
We will provide the exact parameters in \Cref{sec:implementation} and the proper proofs in \Cref{sec:alg-analysis}.

\medskip
\noindent
The algorithm works in {\bf epochs};  each epoch lasts for some consecutive updates in $X$. Let $S \subseteq X$ denote the maintained solution (set of $k$ centers). We satisfy the following invariant.

\begin{invariant}
\label{inv:start}
At the start of an epoch, the set $S$ is robust and $\cost(S,X) \leq O(1) \cdot \OPT_k(X)$.
\end{invariant}

We now describe how the dynamic algorithm works in a given epoch, in four steps.

\medskip
\noindent {\bf Step 1: Determining the length of the epoch.} At the start of an epoch, the algorithm computes the maximum $\ell^{\star} \geq 0$ such that $\OPT_{k-\ell^{\star}}(X) \leq c \cdot \OPT_{k}(X)$, and sets $\ell \leftarrow \lfloor \ell^{\star}/( \Theta(c)) \rfloor$, where $c=O(1)$ is a large constant.
The epoch will last for the next $\ell+1$ updates.\footnote{Note that it is very well possible that $\ell = 0$.}
From now on, consider the superscript $t \in [0, \ell+1]$ to denote the status of some object after the algorithm has finished processing the $t^{th}$ update in the epoch.
For example, at the start of the epoch, $X = X^{(0)}$.

\medskip 
\noindent {\bf Step 2: Preprocessing at the start of the epoch.} Let $S_{\init} \leftarrow S$ be the solution maintained by the algorithm after it finished processing the last update in the previous epoch. Before handling the  very first update in the current epoch, the algorithm initializes the maintained solution by setting
\begin{equation}
\label{eq:init:epoch}
S^{(0)} \leftarrow \arg \min_{S' \subseteq S_{\init} \, : \, |S'| = k-\ell} \cost(X^{(0)}, S').
\end{equation}

\medskip 
\noindent {\bf Step 3: Handling the updates within the epoch.} Consider the $t^{th}$ update in the epoch, for $t \in [1, \ell+1]$.
The algorithm handles this update in a lazy manner, as follows. If the update involves the deletion of a point from $X$, then it does not change the maintained solution, and sets $S^{(t)} \leftarrow S^{(t-1)}$.
In contrast, if the update involves the insertion of a point $p$ into $X$, then it sets $S^{(t)} \leftarrow S^{(t-1)} + p$.

\medskip
\noindent {\bf Step 4: Post-processing at the end of the epoch.}
After the very last update in the epoch, the algorithm does some post-processing, and computes another set $S_{\final} \subseteq X$ of at most $k$ centers (i.e., $|S_{\final}| \leq k$) that satisfies Invariant~\ref{inv:start}.
Then the next epoch can be initiated with $S \leftarrow S_{\final}$ being the current solution.
The post-processing is done as follows.

The algorithm adds $O(\ell + 1)$ extra centers to the set $S_{\init}$, while minimizing the cost of the resulting solution w.r.t.~$X^{(0)}$.
\begin{equation}
\label{eq:augment}
A^{\star} \leftarrow  \arg\min\limits_{\substack{A \subseteq \Deld  : |A| \leq O(\ell+1)}} \cost(X^{(0)},S_{\init} + A), \text{ and } S^{\star} \leftarrow S_{\init} + A^{\star}.
\end{equation}
The algorithm then adds the newly inserted points within the epoch to the set of centers, so as to obtain the set $T^{\star}$.
Next, the algorithm identifies the subset $W^{\star} \subseteq T^{\star}$ of $k$ centers that minimizes the $k$-means objective w.r.t.~$X^{(\ell+1)}$.
\begin{equation}
\label{eq:augment:1}
T^{\star} \leftarrow S^{\star} + \left( X^{(\ell+1)} - X^{(0)}\right), \text{ and }
W^\star \leftarrow \arg\min\limits_{\substack{W \subseteq T^{\star} \, : \, |W| = k}} \cost(X^{(\ell+1)},W).
\end{equation}
Finally, it calls a subroutine referred to as $\Robustify$  on $W^\star$ and lets $S_{\final}$ be the set of $k$ centers returned by this subroutine.
The goal of calling \Robustify is to make the final solution $S_{\final}$, robust w.r.t.~the current space $X=X^{(\ell+1)}$.
We will provide our implementation of this subroutine in \Cref{sec:robustify}.
Finally, before starting the next epoch, the algorithm sets $S \leftarrow S_{\final}$. 
\begin{equation}
\label{eq:augment:2}
S_{\final} \leftarrow \Robustify(W^{\star}), \ \text{and} \ S \leftarrow S_{\final}.
\end{equation}

It is shown that the maintained set $S \subseteq \Deld$ is always of size at most $k$, and at the final step, the solution $S = S_{\final}$ satisfies Invariant~\ref{inv:start} at the end of Step 4, that validates the analysis of the next epoch.

\section{Our Dynamic Algorithm}
\label{sec:implementation}

\paragraph{Roadmap.}
In this section, we provide a complete description of our dynamic algorithm.  In \Cref{sec:implement-steps}, we implement different steps of the framework from \Cref{sec:alg:describe}. 
Then, we proceed with the implementation of the very last subroutine of Step 4 of the algorithm in \Cref{sec:robustify}, which is called $\Robustify$.

\subsection{Implementing Steps of an Epoch}\label{sec:implement-steps}

In this section, we provide the implementation of different steps of the algorithmic framework (see \Cref{sec:alg:describe}) except the very last subroutine $\Robustify$.
We provide the implementation for a fixed epoch that lasts for $(\ell+1)$ updates.

\subsubsection*{Implementing Step 1.} 

We find an estimation of $\ell^{\star}$ instead of the exact value.
We define the exact value of $\ell^\star$ to be the largest integer satisfying the following
\begin{equation} \label{eq:definition-of-ell-star}
   \OPT_{k-\ell^\star}(X^{(0)}) \leq \lambda^{44} \cdot \OPT_{k}(X^{(0)}). 
\end{equation}
Now, we proceed with the procedure of estimating $\ell^\star$.
For each $i \in [0, \log_2 k]$ define $s_i := 2^i$, and let $s_{-1} := 0$. We now run a {\bf for} loop, as described below.

\medskip

\begin{algorithm}[ht]
\caption{\label{alg:find-ell-implementation}
Computing an estimate $\hat{\ell}$ of the value of $\ell^{\star}$.}
\begin{algorithmic}[1]
  \For{$i=0$ to $\log_2 k$}
    Using restricted $k$-means (\Cref{lem:restricted-k-means-main}), in $\tilde{O}(n^\epsilon s_i)$ time compute a subset $\hat{S}_i \subseteq S_{\init}$ of $(k-s_i)$ centers,
that is a $\polyup$-approximation to $\OPT_{k-s_i}^{S_{\init}}\left( X^{(0)}\right)$.\label{line:estimate:1} 

    \If{$\cost(X^{(0)}, \hat{S}_i) > (\polyup)^8 \cdot \lambda^{44} \cdot \cost(X^{(0)}, S_{\init})$}
    \label{if:condition-inside-find-ell}
        \State\Return $\hat{\ell} := s_{i-1}$.
    \EndIf
    \EndFor
\end{algorithmic}
\end{algorithm}

After finding $\hat{\ell}$, we set the length of the epoch to be $\ell + 1$ where 
$$ \ell \gets \left\lfloor \frac{\hat{\ell}}{36 \cdot \lambda^{51}} \right\rfloor . $$
With some extra calculations, with the same argument as in \cite{BCF24}, we can show that $\ell+1 = \tilde{\Omega}(\hat{\ell}+1) = \tilde{\Omega}(\ell^\star+1)$ and 
$$ \frac{\OPT_{k-{\ell}}(X^{(0)})}{\lambda^{51}} \leq \OPT_{k}(X^{(0)}) \leq 12 \cdot \OPT_{k+{\ell}}(X^{(0)}). $$
The running time of this {\bf for} loop  is $\tilde{O}\left(n^\epsilon \cdot \sum_{j=0}^{i^\star}  s_j\right) = \tilde{O}(n^\epsilon \cdot s_{i^\star})$, where $i^{\star}$ is the index s.t.~$s_{i^{\star}-1} = \hat{\ell}$. Thus, we have $s_{i^\star} = O(\hat{\ell}) = O(\ell + 1)$, and hence we can implement Step 1  in $\tilde{O}(n^\epsilon \cdot (\ell+1))$ time.

\subsubsection*{Implementing Step 2.}
Instead of finding the optimum set of $(k-\ell)$ centers within $S_{\init}$, we approximate it using restricted $k$-means (\Cref{lem:restricted-k-means-main}).
We compute a set of $(k-\ell)$ centers $S^{(0)} \subseteq S_{\init}$ such that 
\begin{equation}\label{eq:cost-S0-approx}
 \cost(X^{(0)},S^{(0)}) \leq \polyup \cdot \OPT_{k-\ell}^{S_{\init}}(X^{(0)}).   
\end{equation}
Note that the running time for this step is also $\tilde{O}(n^\epsilon \cdot (\ell+1))$.

\subsubsection*{Implementing Step 3.}
Trivially, we can implement each of these updates in constant time.

\subsubsection*{Implementing Step 4.}
We first need to add $ O(\ell+1)$ centers to $S_{\init}$, while minimizing the cost of the solution w.r.t.~$X^{(0)}$.
Instead of finding $A^\star$ (see \Cref{eq:augment}), we approximate it by augmented $k$-means (\Cref{lem:augmented-k-means-main}).
By setting $S = S_{\init}$, $X = X^{(0)}$ and $s = \lambda^{52} \cdot (\ell+1)$ in augmented $k$-means (\Cref{lem:augmented-k-means-main}), we can compute $A$ of size at most $\tilde{O}(s)$, in $\tilde{O}(n^\epsilon \cdot (\ell+1))$ time such that
\begin{equation}\label{eq:augment-approx}
   \cost(X,S_{\init} + A) \leq \polyup \cdot \min\limits_{A^\star \subseteq \Deld: |A^\star| \leq s} \cost(X,S_{\init}+A^\star), \quad s = \lambda^{52} \cdot (\ell+1)
\end{equation}
At this stage, we compute $S' = S_{\init}+A$ and  $T' := S' + \left( X^{(\ell+1)} - X^{(0)} \right)$ which is an approximation of $T^\star$ as in \Cref{eq:augment:1}, in only $\tilde{O}(\ell+1)$ time.
Next, we compute $W'$ which is an approximation of  $W^\star$ (see \Cref{eq:augment:1}) using restricted $k$-means (\Cref{lem:restricted-k-means-main}), which again takes $\tilde{O}(n^\epsilon \cdot (\ell+1))$ time, and satisfies
\begin{equation}
\label{eq:restricted-at-the-end-of-epoch}
 \cost(X^{(\ell+1)}, W') \leq \polyup \cdot \OPT_{k}^{T'}(X^{(\ell+1)}).
\end{equation}
Finally, we explain below how we implement the call to $\Robustify(W')$ (see \Cref{eq:augment:2}).

\subsection{Implementing Robustify}\label{sec:robustify}

In this section, we provide the implementation of the subroutine called $\Robustify$, which is the very last subroutine called in an epoch.
The main purpose of this subroutine is to make the solution $W'$ at the end of the epoch, robust (see \Cref{def:robust-solution}) w.r.t.~the current dataset.
\paragraph{Notation.} Throughout this section, we refer to $S$ as the main center set maintained in our algorithm and refer to $W'$ as the center set computed in Step 4 of the epoch.
So, if we say we maintain a data structure for $S$, we mean that we maintain it in the entire algorithm not a specific epoch.

\paragraph{Outline.}
$\Robustify$ can be divided into two parts.
First, identifying the centers that might violate Condition (\ref{cond:robust}).
Next, making those centers robust w.r.t~the current space.

In our implementation, for each center $u \in S$ in our main solution, we always maintain an integer $t[u]$ that indicates the center $u$ is $t[u]$-robust, i.e.~we want to maintain the variant that $t[u]$ is always greater than or equal to the smallest integer $t$ satisfying $\lambda^{3t} \geq \dist(u, S-u) / \lambda^{10}$ (see Condition (\ref{cond:robust})).
Hence, to check whether a center $u$ satisfies Condition (\ref{cond:robust}), it suffices to check whether $t[u] \geq t$.\footnote{Note that there might be a center $u$ such that $t[u] < t$, but $u$ is still $t$-robust.}
More precisely, we maintain an approximation of $\dist(u,S-u)$ using \Cref{lem:nearest-neighbor-distance} and find the smallest integer $t$ satisfying
$\lambda^{3t} \geq \hat{\dist}(u,S-u)/\lambda^{10}$.

After identifying a center $u$, where $t[u] < t$, 
we use the $\polyup$ approximation of $\dist(u,S-u)$ (using \Cref{lem:nearest-neighbor-distance}), and find the smallest integer $t$ satisfying $\lambda^{3t} \geq \hat{\dist}(u, S-u) / \lambda^{7}$ and swap $u$ with a near center $v$ which is $t$-robust through a call to subroutine $\MakeRbst(u)$ introduced in \Cref{sec:make-robust}.
Note that here the denominator is $\lambda^{7}$ instead of $\lambda^{10}$, which is intentional and make $u$ to be more robust that what we need in Condition (\ref{cond:robust}).

In \Cref{sec:identify-robust}, we explain how to identify the centers that might violate Condition (\ref{cond:robust}). Then in \Cref{sec:make-robust}, we explain how to swap a center $u$ with a near center $v$ which is $t$-robust for a given $t \geq 0$.

\subsubsection{Identifying Non-Robust Centers}\label{sec:identify-robust}

At the end of each epoch, we partition the center set $W'$ into three parts $W' = W_1 \cup W_2 \cup W_3$.
$W_1$ consists of the new centers added to the solution, i.e. $W_1 := W' - S_\init$.
We call a center $u$ \textit{contaminated}, if throughout the epoch, the ball of radius $ \lambda^{3t[u] + 1.5} $ around $u$ is changed, i.e.
$$ \ball(u, \lambda^{3t[u]+1.5}) \cap \left( X^{(0)} \oplus X^{(\ell+1)}\right) \neq \emptyset . $$
$W_2$ consists of all \textit{contaminated} centers.
$W_3$ then is the rest of the centers $W_3 := W' - (W_1 \cup W_2)$.
We will show the following claim.
The proof of this claim is deferred to \Cref{sec:proof-of-claim-remains-robust}.

\begin{claim}\label{claim:remains-robust}
    Each center $u \in W_3$ remains $t[u]$-robust w.r.t.~the final space $X^{(\ell+1)}$.
\end{claim}

The above claim shows that if a center $u$ is in $W_3$, we have the guarantee that $u$ is still $t[u]$-robust.
The only centers that we do not know about their robustness are centers in $W_1 \cup W_2$.

As a result, we do the following.
First, we call $\MakeRbst$ on all centers $u \in W_1 \cup W_2$ and set their $t[u]$ value.
After that, every center $u$ in our current solution is $t[u]$-robust, and in order to figure out if they violate Condition (\ref{cond:robust}), as described before, it is sufficient to check if $t[u] \geq t$, where $t$ is the smallest integer satisfying $\lambda^{3t} \geq \hat{\dist}(u, W'-u) / \lambda^{10}$.

\paragraph{Bottleneck.}
Since we need an amortized update time of $\tilde{O}(n^\epsilon)$, it is not efficient to iterate over all centers $u$, find $t$ and check if $t[u] \geq t$.
Instead, we need a new way to figure out if a center $u$ satisfies $t[u] \geq t$.
It is also non-trivial how to find the contaminated centers very fast.

We will use \Cref{lem:nearest-neighbor-distance} and \Cref{lem:bitwise-approx-NC} to identify a superset of centers that contains all centers that might violate Condition (\ref{cond:robust}).
The size of this superset is small enough such that the amortized recourse of our algorithm remains $\tilde{O}(1)$.

\paragraph{Identifying $W_1$.}
It is obvious how to find $W_1$ explicitly since $W_1 = W' - S_\init$. We only need to keep track of the newly added centers during the epoch which can be done in $\tilde{O}(W' - S_\init)$ time.
According to the implementation, we have $W' - S_\init = \tilde{O}(\ell+1)$, which concludes we can find $W_1$ in a total of $\tilde{O}(1)$ time in amortization.

\paragraph{Identifying $W_2$.}

We will show that each update within an epoch can contaminate at most $O(\log (\sqrt{d}\Delta))$ many centers.
More precisely, we show something stronger as follows.
We defer the proof of this claim to \Cref{sec:proof-of-claim-num-of-contamination}.
\begin{claim}\label{claim:num-of-contamination}
    For each $x \in X^{(0)} \oplus X^{(\ell+1)}$ and each $i \in [0, \lceil \log_{\lambda^3} (\sqrt{d}\Delta) \rceil]$, there is at most one center $u \in S_\init$ such that $t[u] = i$ and $\dist(x,u) \leq \lambda^{3i + 2}$.
\end{claim}

This claim provides us with a way to find contaminated centers efficiently.
We split our center set $S$ into $ \lceil \log_{\lambda^3} (\sqrt{d}\Delta) \rceil + 1$ parts.
For each $i \in [0, \lceil \log_{\lambda^3} (\sqrt{d}\Delta) \rceil]$, we keep track of centers $u \in S$ with $t[u] = i$ and maintain a subset $S[i] \subseteq S$ that contains all centers $u \in S$ satisfying $t[u] = i$.
For each set $S[i]$, we invoke the nearest neighbor oracle (\Cref{lem:ANN-query}) that given $x \in \Deld$, returns a $u \in S[i]$ such that $\dist(x,u) \leq \polyup \cdot \dist(x, S[i])$.
We will show that $S[i]$ has a special structure, and this approximate nearest neighbor is the only center that can be contaminated by $x$ within $S[i]$ at the end of each epoch.
More precisely, if $x$ contaminates a center $v$, then $v$ must be an exact nearest neighbor and there does not exist any other $\polyup$-approximate nearest neighbor within $S[i]$ (all other centers in $S[i]$ have distance more than $\polyup \cdot \dist(x, S[i])$ to $x$).
So, we show the following claim.
The proof of this claim is deferred to \Cref{sec:proof-of-claim-contaminte-only-one}.

\begin{claim}\label{claim:contaminate-only-one}
    For each $x \in X^{(0)} \oplus X^{(\ell+1)}$ and each $i \in [0, \lceil \log_{\lambda^3} (\sqrt{d}\Delta) \rceil]$, if $u \in S[i]$ satisfies $\dist(x,u) \leq \polyup \cdot \dist(x,S[i])$, then the only center in $S[i]$ that might be contaminated by $x$ is $u$. 
\end{claim}

This claim shows how to identify $W_2$.
For each $x \in X^{(0)} \oplus X^{(\ell+1)}$ and each $i \in [0, \lceil \log_{\lambda^3} (\sqrt{d}\Delta) \rceil]$, we call the nearest neighbor oracle on $S[i]$ to find $u \in S[i]$ satisfying $\dist(x,u) \leq \polyup \cdot \dist(x, S[i])$.
Then we check if $\dist(x,u) \leq \lambda^{3i+1.5} $.
If this happens, we see that $x$ contaminates $u$, otherwise $x$ does not contaminate any center in $S[i]$ according to \Cref{claim:contaminate-only-one}.
Hence, we can identify all centers in $W_2$.

\paragraph{Make Centers in $W_1$ and $W_2$ Robust.}
Until now, we found $W_1$ and $W_2$ explicitly.
We made a call to $\MakeRbst$ on each center in $W_1$ and $W_2$ to make them robust.
As a result, at this moment in time, we have a set of centers and each $u$ has an integer $t[u]$ that indicates $u$ is $t[u]$-robust w.r.t.~the current dataset.
This is because we called a $\MakeRbst$ explicitly on centers in $W_1$ and $W_2$, and for the rest of the center we use \Cref{claim:remains-robust}.
Now, using these integers, we proceed with an efficient way to identify centers violating Condition (\ref{cond:robust}).

\paragraph{Identifying Other Centers Violating Condition (\ref{cond:robust}).}
We introduce a way to identify the centers that might violate Condition (\ref{cond:robust}).
The goal is to not miss any single center violating Condition (\ref{cond:robust}).
So, it is possible that our procedure reports a center $u$ that does not violate Condition (\ref{cond:robust}), but if a center $u$ violates Condition (\ref{cond:robust}), our procedure definitely reports that center.
Note that after each change in the center set $S$, the distances of centers $u \in S$ to other centers might change.
As a result, a call to $\MakeRbst$ on $u$ might make another center $v \in S$ to violate Condition (\ref{cond:robust}).
We show that although the distances between centers will change after each change in $S$, in the current call to $\Robustify$, we will not call $\MakeRbst$ on a center twice, which means if we call $\MakeRbst$ on a center, it will continue to satisfy (\ref{cond:robust}) until the end of the call to the current $\Robustify$, and our procedure will not make another call on that center.

\begin{lemma}\label{lem:robustify-calls-once}
    Consider any call to \Robustify$(W')$, and suppose that it sets $w_0 \leftarrow \MakeRbst(w)$ during some iteration of the {\bf while} loop. Then in subsequent iterations of the {\bf while} loop in the same call to \Robustify$(W')$, we will {\em not} make any call to  \MakeRbst$(w_0)$.
\end{lemma}

We provide the proof of this lemma in \Cref{sec:proof-of-lem-robustify-calls-once}.

Now, we proceed with the explanation of the subroutine.
For every $i \in [0, \lceil \log_{\lambda} (\sqrt{d}\Delta) \rceil]$,
we maintain the data structure $D_{\lambda^i}(S)$ of \Cref{lem:bitwise-approx-NC} for parameter $\gamma := \lambda^i$.
Now, after each change in $S$, we do the following.
First, we identify all centers $u \in S$ such that for at least one $i$, the bit $b_{\lambda^i}(u, S)$ has changed, and add all of these centers to a maintained set $\calY$ referred to as \textit{yellow} set.
This set will contain centers that might not satisfy Condition (\ref{cond:robust}) and there is a possibility that we call $\MakeRbst$ on these centers.

Now, for each center $u$ in yellow set, we do the following.
We find the smallest integer $t$ satisfying $\lambda^{3t} \geq \hat{\dist}(u,S-u)/\lambda^{10}$.
Then, check whether $t[u] \geq t$.
If $t[u] \geq t$, we are sure that $u$ is already $t$-robust and it satisfies Condition (\ref{cond:robust}) since
$$ \lambda^{3t[u]} \geq \lambda^{3t} \geq  \hat{\dist}(u,S-u)/\lambda^{10} \geq \dist(u,S-u)/\lambda^{10}. $$
Hence, there is no need to call a $\MakeRbst$ on $u$.
Otherwise (if $t[u] < t$), we call $\MakeRbst$ on $u$ that find the smallest integer $t'$ satisfying $\lambda^{3t'} \geq \hat{\dist}(u,S-u)/\lambda^{7}$ \footnote{Note the discrepancy between $\lambda^7$ here and $\lambda^{10}$ in Condition (\ref{cond:robust})} and swaps $u$ with a close center which is $t'$-robust.
Note that it is also possible in this case that $u$ is satisfying Condition (\ref{cond:robust}), but we call $\MakeRbst$ on $u$.
We will show that the number of these calls is $\tilde{O}(n)$ throughout the entire algorithm.
The main reason is that if this case happens, then the distance between $u$ and other centers must have increased by a multiplicative factor, and this can happen at most $O(\log (\sqrt{d}\Delta))$ many times in a row before $u$ is contaminated or become removed from the center set.

We summarize our implementation of $\Robustify$ in \Cref{alg:modify:robustify}.
Throughout this procedure, after each insertion or deletion on $S$, all of the data structures maintained related to $S$ are updated.
For instance, $D_{\lambda^i}(S)$ might add some centers to $\calY$ whose bits have changed in this data structure.
Note that we do not need to update these data structures in the previous steps of our algorithm.
We can stay until the end of the epoch and update them all together.

\begin{algorithm}[ht]
\caption{\label{alg:modify:robustify}
Implementation of the call to $\Robustify(W')$ at the end of an epoch.}
\begin{algorithmic}[1]
    \State
    $W_1 \gets W' \setminus S_{\init}$\label{line:type1} \Comment{All new centers}
    \For{each $x \in X^{(0)} \oplus X^{(\ell+1)}$}
        \label{line-for-loop-points}
        \For{each $i \in [0, \lceil \log_{\lambda^3} (\sqrt{d}\Delta) \rceil]$}
            \State Let $ u \gets \nn(x, S[i])$ \;
            \If{$\dist(u,x) \leq \lambda^{3i+1.5}$}{
                $W_2 \gets W_2 + u$ \;
                \Comment{Mark $u$ as a contaminated center}
            }
            \EndIf
        \EndFor
    \EndFor
    \State
   $S \gets S_\init - (S_\init \setminus W') + (W' \setminus S_\init) $ \label{line:update-data-strucutres} 
\For{each $u \in W_1 \cup W_2$}
        \State $\MakeRbst(u)$. 
\EndFor
    \While{$\calY \neq \emptyset$}
    \label{line:while-loop}
        \State $u \gets $pop$(\calY)$ \label{line:pop-robustify}\;
\State $t \gets $ Smallest integer satisfying $\lambda^{3t} \geq \hat{\dist}(u, S-u) / \lambda^{10}$ \;
\If{$t[u] \geq t$}\label{if:condition-W3}
            \State go back to \Cref{line:while-loop}
        \Else
            \State $ \MakeRbst(u)$.\label{line:type3-call} 
\EndIf
    \EndWhile
    \State
  \Return $S_{\final} \gets S$ \label{line:end}.
\end{algorithmic}
\end{algorithm}

\subsubsection{Make-Robust}\label{sec:make-robust}

In this section, we introduce the subroutine that swaps a center with a near center which is robust.
Assume we are given $u \in S$.
First, we compute the smallest integer $t$ satisfying $\lambda^{3t} \geq \hat{\dist}(u,S-u)/ \lambda^7$.
Now, we want to find a $t$-robust sequence $(x_0,x_1,\ldots, x_t)$ such that $x_t = u$.
Then swap $u$ in the solution with $x_0$ which is $t$-robust.
To do this, we follows the \Cref{def:robust}.
Starting from $j = t$, at each iteration, we find $x_j$ such that it satisfies the properties in \Cref{def:robust}.
Initially, $x_t = u$.
Now assume that we have found $x_t, x_{t-1}, \ldots, x_{j}$ and we want to find $x_{j-1}$ such that the second condition in \Cref{def:robust} holds for all $ i = j $.

We use the data structure of \Cref{lem:1-means}.
We make a query on this data structure for $x=x_{j}$ and $r = \lambda^{3j}$.
Hence, we get $c^* \in \Deld$, $\Hat{\cost}_{x_j}$, $\Hat{\cost}_{c^*}$, and $b_{x_j}$ explicitly such that there exists a ball $B_{x_j} \subseteq X$ with the specified guarantees.
We let $B_{j} = B_{x_j}$ in \Cref{def:robust}. Note that we do not need to know $B_j$ explicitly, and in \Cref{def:robust}, the existence of this set is sufficient. 
According to the first guarantee in \Cref{lem:1-means}, we have
$\ball(x_{j}, \lambda^{3j}) \leq B_j \leq \ball(x_{j}, \polyup \cdot \lambda^{3j}) $.
Since $\lambda = (\polyup)^2$, we have $\polyup \cdot \lambda^{3j} \leq \lambda^{3j+1}$, and we conclude that $B_{x_j}=B_j$ satisfies the first condition in \Cref{def:robust} for $i=j$.
Now, we proceed with the second condition.

We compare $\Hat{\cost}_{x_j}/{b_{x_j}}$ and $\lambda^{3j-1}/\polyup$.
There are two cases as follows.
\paragraph{Case 1.}
$\Hat{\cost}_{x_j}/b_{x_j} \geq \lambda^{3j-1}/\polyup$.
In this case, we let $x_{j-1} = x_{j}$, and show that \Cref{cond:robust1} holds for $i=j$.
We have
$$ \lambda^{3j-1}/\polyup \leq \Hat{\cost}_{x_j} / b_{x_j} \leq \polyup \cdot \cost(B_j, x_j)/w(B_j), $$
where the second inequality holds by the second and the fifth guarantees in \Cref{lem:1-means}.
Since $\lambda=(\polyup)^2$, it follows that 
$\cost(B_j, x_j)/w(B_j) \geq \lambda^{3j-2}$.
\Cref{cond:robust1} then holds since we set $x_{j-1} = x_j$. 

\paragraph{Case 2.}
$\Hat{\cost}_{x_j}/b_{x_j} \leq \lambda^{3j-1}/\polyup$.
In this case, we conclude
\begin{equation}\label{eq:try-to-satisfy-second-condition}
   \cost(B_j, x_j)/w(B_j) \leq \Hat{\cost}_{x_j}/b_{x_j} \leq \lambda^{3j-1}/\polyup \leq \lambda^{3j-1}, 
\end{equation}
where the first inequality holds by the second and the fifth guarantees in \Cref{lem:1-means}.
Now, we compare $\Hat{\cost}_{x_j} / \polyup$ and $\Hat{\cost}_{c^*}$. We have two cases as follows.
\begin{itemize}
    \item $\Hat{\cost}_{x_j} / \polyup \leq \Hat{\cost}_{c^*}$.
    In this case, we let $x_{j-1} = x_j$.
    We have
    \begin{align*}
       \cost(B_j, x_{j-1}) = \cost(B_j, x_{j}) &\leq \Hat{\cost}_{x_j} \leq \polyup \cdot \Hat{\cost}_{c^*} \\
       &\leq (\polyup)^2 \cdot \cost(B_j, c^*) \leq (\polyup)^3 \cdot \OPT_1(B_j). 
    \end{align*}
    The first, third, and fourth inequalities are followed by the second, third, and fourth guarantees in \Cref{lem:1-means} respectively.
    Since $x_{j-1} = x_j$, we obviously have
    $ \cost(B_j, x_{j-1}) \leq \cost(B_j, x_{j}) $.
    Together with \Cref{eq:try-to-satisfy-second-condition}, we see that \Cref{cond:robust2} holds in this case.
    
    \item $\Hat{\cost}_{x_j} / \polyup \geq \Hat{\cost}_{c^*}$.
    In this case, we let $x_{j-1} = y$.
    Similarly,
    \begin{align*}
       \cost(B_j, x_{j-1}) = \cost(B_j,c^*) &\leq \Hat{\cost}_{c^*} \leq \Hat{\cost}_{x_j} / \polyup \\
       &\leq \polyup \cdot \cost(B_j, x_j)/\polyup = \cost(B_j,x_j), 
    \end{align*}
    where the first, and the third inequalities are followed by the third and the second guarantees in \Cref{lem:1-means}.
    We also have
    $$ \cost(B_j, x_{j-1}) = \cost(B_j, c^*) \leq \polyup \cdot \OPT_1(B_j). $$
    Together with \Cref{eq:try-to-satisfy-second-condition}, we see that \Cref{cond:robust2} holds in this case.
\end{itemize}

As a result, in all cases, the second condition in \Cref{def:robust} holds for $j = i$.
Finally, it follows that the sequences 
$(x_0,x_1,\ldots, x_t)$ and
$(B_0,B_1,\ldots, B_t)$\footnote{we define  $B_0 := \ball(x_0, 1)$ at the end.} satisfy the conditions in \Cref{def:robust}.
Hence, $x_0$ is a $t$-robust center.
Below, we summarize the implementation of $\MakeRbst$ in \Cref{alg:make-robust}.

\begin{algorithm}[ht]
\caption{\label{alg:make-robust}
Implementation of a call to $\MakeRbst(u)$ at the end of an epoch.}
\begin{algorithmic}[1]
    \State $t \gets $ Smallest integer satisfying $\lambda^{3t} \geq \hat{\dist}(u, S-u) / \lambda^{7}$ \;
    \State $x_t \gets u$\;
    \For{$j = t$ down to $1$}
        \State Invoke \Cref{lem:1-means} for $x = x_j$ and $r = \lambda^{3j}$ to get $c^* \in \Deld$, $\Hat{\cost}_{x_j}$, $\Hat{\cost}_{c^*}$ and $b_{x_j}$ \label{line:invoke-one-median-ball} \;
        \If{$\left( \Hat{\cost}_{x_j}/b_{x_j} \geq \lambda^{3j-1}/\polyup \right)$ or $\left( \Hat{\cost}_{x_j}/\polyup \leq \Hat{\cost}_{c^*} \right)$}
            \State $x_{j-1} \gets x_j$
        \Else
\State $x_{j-1} \gets c^*$
\EndIf
  \EndFor
  \State $ u_0 \gets x_0$ \;
    \State $S \gets S - u + u_0$ \;
\State Save $t[u_0] \gets t$ together with $u_0$  \;
\end{algorithmic}
\end{algorithm}

\section{\texorpdfstring{Analysis of Our Dynamic Algorithm: Proof of \Cref{full-part-theorem:main}}{}}
\label{sec:alg-analysis}

In this section, we provide the analysis of our algorithm.
We start by providing some useful lemmas in \Cref{sec:key-lemmas-full} that are used to analyze our algorithm.
We defer the proof of these main lemmas in \Cref{sec:deferred-full}.
We provide the approximation ratio analysis in \Cref{sec:approx-analysis}.
Then, we proceed with the recourse analysis in \Cref{sec:recourse-analysis}.
Finally, we provide the update time analysis in \Cref{sec:update-time-analysis-full}.

 \subsection{Key Lemmas}\label{sec:key-lemmas-full}

In this section, we provide the statements of the main lemmas that are used to analyze the algorithm.

\subsubsection{Structure of Robust Centers}\label{sec:main-properties-of-robust}

\begin{lemma}\label{lem:robust-property-1}
    If $(x_0,x_1,\ldots,x_t)$ is a $t$-robust sequence, then for each $1 \leq j \leq t$, the following hold,
    \begin{enumerate}
        \item $\dist(x_{j-1}, x_j) \leq 2 \cdot \lambda^{3j-1}$,
        \item $B_{j-1} \subseteq B_j$,
        \item $\dist(x_0, x_j) \leq 4 \cdot \lambda^{3j-1}$.
    \end{enumerate}
\end{lemma}

\begin{lemma}\label{lem:robust-property-2}
    Let $(x_0,x_1,\ldots , x_t)$ be a $t$-robust sequence w.r.t.~$X \subseteq \Deld$.
    Then, for every $0 \leq i \leq t$ and every  $B_i \subseteq U \subseteq X $, we have
    $\cost(U, x_0) \leq 4 \cdot \cost(U, x_i)$.
\end{lemma}

\begin{lemma}\label{lem:cost-after-robustify}
    If $S$ is the output of $\Robustify(W)$, then
    $\cost(X, S) \leq 4 \cdot \cost(X, W)$.
\end{lemma}

\subsubsection{Well-Separated Pairs}

\begin{definition}\label{def:well-sep}
    Assume $S,T \subseteq \Deld$ are two center sets.
    We call $(u,v) \in S \times T$ a well-separated pair if and only if
    $$  \dist(u,v) \leq \frac{1}{\lambda^{20}} \cdot \min \{ \dist(u,S-u) , \dist(v,T-v)\}. $$
\end{definition}

\begin{lemma}\label{lem:cost-well-sep-pairs}
    Consider any two sets of centers $S, V \subseteq \Deld$ such that $S$ is robust.
    Then, for every well-separated pair $(u,v) \in S \times V$, we have 
    $$\cost(C_v, u) \leq 4\cdot(\polyup)^3 \cdot  \cost(C_v, v), \footnote{Recall $\polyup$ from \Cref{eq:value-of-polyup-lambda}.} $$
    where $C_v$ is the cluster of $v$ w.r.t.~the center set $V$, i.e. $C_v = \{x \in X \mid \dist(x,v)=\dist(x,V) \}$.
\end{lemma}

\begin{lemma}[$k$-Means Version of Lemma 4.2 \cite{BCF24}]\label{lem:num-well-sep}
      Let $r \geq 0$ and $m \in [0, k]$. Consider any two sets of centers $S, V \subseteq X$ such that $|S| = k$ and $|V| = k+r$. 
    If the number of well-separated pairs w.r.t.~$(S, V)$ is $k - m$, then there exists a subset $\tilde{S} \subseteq S$ of size at most $k - \lfloor (m- r) / 4 \rfloor$ such that
    $$ \cost(X,\tilde{S}) \leq 81\lambda^{40} \cdot \left( \cost(X,S) + \cost(X,V) \right) . $$
\end{lemma}

\subsubsection{Other Key Lemmas}

\begin{lemma}[Projection Lemma]\label{lem:projection-lemma}
    Consider any set of centers $C \subseteq \Deld$ of size $|C| \geq k$, where $k$ is a positive integer. Then we have $\OPT_{k}^{C}(X) \leq  2 \cdot \cost(X,C) + 8 \cdot \OPT_{k}(X)$.
\end{lemma}

\begin{lemma}[Lazy-Updates Lemma]\label{lem:lazy-updates}
    Consider any two sets of input points $X, X' \subseteq \Deld$ such that $|X \oplus X'| \leq s$. Then for every $k \geq 1$, we have $\OPT_{k+s}(X') \leq \OPT_{k}(X)$.
\end{lemma}

\begin{lemma}[$k$-Means Version of Double-Sided Stability (Lemma 4.1~\cite{BCF24})]\label{lem:double-sided-stability} Consider any $r \in [0, k-1]$ and any $\eta \geq 1$. If 
    $\OPT_{k-r}(X) \leq \eta \cdot  \OPT_k(X)$, then we must have 
    $\OPT_{k}(X) \leq 12 \cdot \OPT_{k + \lfloor r / (36\eta) \rfloor }(X)$. 
\end{lemma}

\section{Approximation Ratio Analysis of Our Algorithm}\label{sec:approx-analysis}

In this section, we provide the approximation ratio analysis of our algorithm as follows.

\begin{lemma}
    The approximation ratio of our algorithm is $\poly(1/\epsilon)$ w.h.p.
\end{lemma}

\paragraph{Roadmap.}
In \Cref{sec:approx-analysis-during-epoch-full}, we show an upper bound on the cost of the solution during the epoch, given that the cost of the initial solution in the epoch is bounded.
Then in \Cref{sec:approx-analysis-end-of-the-epoch-full}, we show that the cost of the initial solution at any epoch is bounded.

\subsection{Approximation Ratio During the Epoch}\label{sec:approx-analysis-during-epoch-full}

In this section, we show that the cost of the solution during any epoch is bounded by $\poly(1/\epsilon)$.

\paragraph{Assumption.}
Consider an epoch of length $\ell+1$.
We assume that the initial solution in the epoch satisfies 
\begin{equation}\label{eq:initial-cost-of-S-init}
    \cost(X^{(0)}, S_\init) \leq (\polyup)^6 \cdot \OPT_{k}(X^{(0)}).
\end{equation}
We will show that throughout the epoch the solution is always a $\lambda^{53}$ approximation w.r.t.~the current dataset.
According to the definition of $ \hat{\ell} = s_{i-1} $ in Step 1 of our algorithm (see \Cref{alg:find-ell-implementation}, \Cref{if:condition-inside-find-ell}), we infer that
$$ \OPT_{k-\hat{\ell}}(X^{(0)}) \leq \cost(X^{(0)}, \hat{S}_{i-1}) \leq  (\polyup)^8 \cdot \lambda^{44} \cdot \cost(X^{(0)}, S_{\init}). $$
Hence, we conclude by \Cref{eq:initial-cost-of-S-init} that
$$ \OPT_{k-\hat{\ell}}(X^{(0)}) \leq (\polyup)^8 \cdot \lambda^{44} \cdot \cost(X^{(0)}, S_{\init}) \leq (\polyup)^8 \cdot \lambda^{44} \cdot (\polyup)^6 \cdot \OPT_{k}(X^{(0)}) \leq \lambda^{51} \cdot \OPT_{k}(X^{(0)}), $$
where the last inequality follows since $\lambda = (\polyup)^2$.
According to Double-Sided Stability \Cref{lem:double-sided-stability} for $X=X^{(0)}$, $r = \hat{\ell}$, and $\eta = \lambda^{51}$, we conclude that
$$ \OPT_{k}(X^{(0)}) \leq 12 \cdot \OPT_{k+ \lfloor \hat{\ell}/ (36 \cdot \lambda^{51}) \rfloor}(X^{(0)}). $$
Finally, since $ \ell = \lfloor \hat{\ell}/ (36 \cdot \lambda^{51}) \rfloor \leq \hat{\ell} $, we have
\begin{equation}\label{eq:stability-equation}
   \frac{\OPT_{k-\ell}(X^{(0)})}{\lambda^{51}} \leq \OPT_{k}(X^{(0)}) \leq 12 \cdot \OPT_{k+ \ell}(X^{(0)}). 
\end{equation}
According to the definition of $S^{(0)}$ in Step 2 of our algorithm, we have
\begin{align}
   \cost(X^{(0)},S^{(0)}) &\leq \polyup \cdot \OPT_{k-\ell}^{S_{\init}}(X^{(0)}) \nonumber \\ 
   &\leq \polyup \cdot( 2 \cdot \cost(X^{(0)},S_\init) +
   8 \cdot \OPT_{k-\ell}(X^{(0)}) ) \nonumber\\
   &\leq \polyup \cdot( 2 \cdot (\polyup)^6 \cdot \OPT_k(X^{(0)}) + 8 \cdot \lambda^{51} \cdot \OPT_{k}(X^{(0)}) ) \nonumber\\
   &\leq \lambda^{52} \cdot \OPT_{k}(X^{(0)})  \nonumber\\
   &\leq 12 \cdot \lambda^{52} \cdot \OPT_{k+\ell}(X^{(0)}). \label{eq:S0-compare-to-opt-kl}
\end{align}
The first inequality follows from \Cref{eq:cost-S0-approx}, the second inequality follows from Projection \Cref{lem:projection-lemma}, the third inequality
follows from \Cref{eq:initial-cost-of-S-init} and \Cref{eq:stability-equation}, the fourth inequality follows from the fact that $\lambda = (\polyup)^2$ is a large number, and the last inequality
follows from \Cref{eq:stability-equation}.

For each $t \in [0, \ell]$, we have $\left| X^{(t)} \oplus X^{(0)} \right| \leq t \leq \ell$, which concludes
\begin{equation}
\label{eq:part1:approx:1}
\OPT_{k+\ell}\left( X^{(0)} \right) \leq \OPT_k\left( X^{(t)} \right),
\end{equation}
As a result, for each $t \in [0, \ell]$, it follows that
\begin{align*}
   \cost(X^{(t)},S^{(t)}) &\leq \cost(X^{(0)},S^{(0)}) 
 \leq 12 \cdot \lambda^{52} \cdot \OPT_{k+\ell}(X^{(0)}) \\ 
 &\leq 12 \cdot \lambda^{52} \cdot \OPT_{k}(X^{(t)}) \leq \lambda^{53} \cdot \OPT_{k}(X^{(t)}).
\end{align*}
The first inequality follows from $S^{(t)} = S^{(0)} + (X^{(t)} - X^{(0)})$, the second inequality follows from \Cref{eq:S0-compare-to-opt-kl}, the third inequality follows from \Cref{eq:part1:approx:1}, and the last one follows since $\lambda $ is a large number.
This shows that at all times within an epoch, the set $S^{(t)}$ maintained by our algorithm remains a valid $\lambda^{53}$-approximate solution to the $k$-median problem on the current dataset $X^{(t)}$.
It remains to show that our algorithm successfully restores \Cref{eq:initial-cost-of-S-init} when the epoch terminates after the $(\ell+1)^{th}$ update. We do the analysis of this part in the following section.

\subsection{Approximation Ratio At the End of the Epoch}\label{sec:approx-analysis-end-of-the-epoch-full}

In this section, we show that the cost of the final solution at the end of any epoch (which is the initial solution for the next epoch) is bounded by $\poly(1/\epsilon)$ as follows.

\begin{lemma}
\label{lm:restore:inv}
At the end of the epoch, the set $S = S_{\final}$ satisfies 
$$\cost(X^{(\ell+1)}, S) \leq (\polyup)^6 \cdot \OPT_{k}(X^{(\ell+1)}). $$
\end{lemma}

The claim below bounds the cost of the solution $S'$ w.r.t.~the dataset $X^{(0)}$.

\begin{claim}
\label{cl:approx:key}
In step 4 of our algorithm, we have $\cost(X^{(0)},S') \leq 4 \cdot (\polyup)^4
 \cdot \OPT_{k+\ell+1}\left(X^{(0)}\right)$.
\end{claim}

Before proving Claim~\ref{cl:approx:key}, we explain how it implies Lemma~\ref{lm:restore:inv}. Towards this end, note that:
\begin{eqnarray}
\cost(X^{(\ell+1)}, S_{\final}) & \leq & 4 \cdot \cost(X^{(\ell+1)}, W^{'}) \label{eq:derive:1} \\
& \leq & 4 \polyup \cdot \OPT_{k}^{T'}(X^{(\ell+1)}) \label{eq:derive:2} \\
& \leq &
 4 \polyup \cdot \left(2 \cdot \cost(X^{(\ell+1)}, T') + 8 \cdot \OPT_{k}(X^{(\ell+1)})\right) \label{eq:derive:3}\\
 & \leq &
 4 \polyup \cdot \left(2\cdot \cost(X^{(0)}, S') + 8 \cdot \OPT_{k}(X^{(\ell+1)})\right) \label{eq:derive:4}\\
 &\leq & 4 \polyup \cdot \left( 8(\polyup)^4 \cdot \OPT_{k+\ell+1}( X^{(0)} ) + 8 \cdot \OPT_{k}(X^{(\ell+1)})\right) \label{eq:derive:5}\\
 &\leq & 4 \polyup \cdot \left( 8(\polyup)^4 \cdot \OPT_{k}( X^{(\ell+1)} ) + 8 \cdot \OPT_{k}(X^{(\ell+1)})\right) \label{eq:derive:6}\\
& \leq & (\polyup)^6 \cdot \OPT_k\left( X^{(\ell+1)}\right). \label{eq:derive:7}
\end{eqnarray}
In the above derivation, the first step~(\ref{eq:derive:1}) follows from \Cref{lem:cost-after-robustify}.
The second step~(\ref{eq:derive:2}) follows from \Cref{eq:restricted-at-the-end-of-epoch}.
The third step~(\ref{eq:derive:3}) follows from Projection \Cref{lem:projection-lemma}.
The fourth step~(\ref{eq:derive:4}) follows from the fact that $ T' = S' + ( X^{(\ell+1)} - X^{(0)})$.
The fifth step~(\ref{eq:derive:5}) follows from Claim~\ref{cl:approx:key}.
The fourth step~(\ref{eq:derive:6}) follows from the Lazy-Updates \Cref{lem:lazy-updates} and the observation that $\left| X^{(\ell+1)} \oplus X^{(0)} \right| \leq \ell+1$.
The last step~(\ref{eq:derive:7}) follows from the fact that $\polyup$ is a large number.

\subsubsection{Proof of Claim~\ref{cl:approx:key}}
Let $V \subseteq \Deld$ be an optimal $(k+\ell+1)$-median solution for the dataset $X^{(0)}$, i.e., $\left| V \right| = k+\ell+1$ and $\cost(X^{(0)},V) = \OPT_{k+\ell+1}\left( X^{(0)} \right)$.
Let $m \in [0, k]$ be the unique integer such that there are $(k-m)$ well-separated pairs w.r.t.~$\left(S_{\init}, T\right)$.
Let $\{ (u_1, v_1), (u_2, v_2), \ldots, (u_{k-m}, v_{k-m})\} \subseteq  S_{\init} \times V$ denote the collection of $(k-m)$ well-separated pairs w.r.t.~$\left( S_{\init}, V \right)$.
Define the set $\tilde{A} := V \setminus \{v_1, \ldots, v_{k-m}\}$.
It is easy to verify that: 
\begin{equation}
\label{eq:size}
\left| \tilde{A} \right| = \left| V \right| - (k-m) = m+\ell+1.
\end{equation}

\begin{claim}
\label{cl:approx:key:10}
We have $ m \leq 300 \cdot \lambda^{51} \cdot (\ell+1) $.
\end{claim}

\begin{claim}
\label{cl:approx:key:11}
We have $\cost(X^{(0)}, S_{\init}+\tilde{A}) \leq 4 \cdot (\polyup)^3 \cdot \OPT_{k+\ell+1}\left(X^{(0)}\right)$.
\end{claim}

By~(\ref{eq:size}), Claim~\ref{cl:approx:key:10} and Claim~\ref{cl:approx:key:11}, there exists a set $\tilde{A} \subseteq \Deld$ of $m+\ell+1 \leq  (300 \cdot \lambda^{51}+1) \cdot (\ell+1)\leq \lambda^{52} \cdot (\ell+1) = s $ centers such that $\cost(X^{(0)},S_{\init}+\tilde{A}) \leq 4(\polyup)^3 \cdot \OPT_{k+\ell+1}\left(X^{(0)}\right)$. 
Accordingly, from~(\ref{eq:augment-approx}), we get 
\begin{align*}
   \cost(X^{(0)}, S') &\leq  \polyup \cdot \min\limits_{A \subseteq \Deld: |A| \leq s} \cost(X^{(0)},S_\init + A) \\
   &\leq \polyup \cdot \cost(X^{(0)},S_{\init}+\tilde{A}) \\
   &\leq 4 \cdot (\polyup)^4 \cdot \OPT_{k+\ell+1}\left(X^{(0)}\right), 
\end{align*}
which implies Claim~\ref{cl:approx:key}.

It now remains to prove Claim~\ref{cl:approx:key:10} and Claim~\ref{cl:approx:key:11}.

\subsubsection*{Proof of Claim~\ref{cl:approx:key:10}}
We apply \Cref{lem:num-well-sep}, by setting $r = \ell+1$, $S = S_{\init}$, $X = X^{(0)}$.
This implies the existence of a set $\tilde{S} \subseteq S_{\init}$ of at most $(k - b)$ centers, with $b = \lfloor (m- \ell-1)/4 \rfloor$, such that 
\begin{eqnarray*}
\cost(X^{(0)},\tilde{S}) & \leq &  81 \lambda^{40} \cdot \left( \cost(X^{(0)},S_{\init}) + \cost(X^{(0)},V) \right) \\
& \leq & 81 \lambda^{40} \cdot \left( (\polyup)^6 \cdot \OPT_k\left( X^{(0)} \right) + \OPT_{k+\ell+1}\left(X^{(0)} \right) \right)  \leq  \lambda^{44} \cdot \OPT_k\left( X^{(0)}\right).
\end{eqnarray*}
In the above derivation, the second inequality follows from the initial assumption at the beginning of the epoch (\ref{eq:initial-cost-of-S-init}), and the last inequality holds because $\OPT_{k+\ell+1}\left(X^{(0)}\right) \leq \OPT_k(X^{(0)})$ and $\lambda = (\polyup)^2$ is a large number (see \Cref{eq:value-of-polyup-lambda}).
Since $\OPT_{k-b}\left( X^{(0)}\right) \leq \cost(X^{(0)},\tilde{S})$, we get
\begin{equation}
\label{eq:approx:key:1}
\OPT_{k-b}\left( X^{(0)} \right) \leq \lambda^{44} \cdot \OPT_k\left(X^{(0)} \right).
\end{equation}
Recall the way $\ell, \hat{\ell}$ and $ \ell^{\star}$ are defined in Step 1 of our algorithm and our implementation (\Cref{alg:find-ell-implementation}).
From~(\ref{eq:approx:key:1}), it follows that $b \leq \ell^{\star}$.
According to the definition of $\hat{\ell} = s_{i-1}$ in \Cref{alg:find-ell-implementation}, for $s_i = 2 s_{i-1} = 2\hat{\ell}$, from (\Cref{if:condition-inside-find-ell}), we have
$$ (\polyup)^8 \cdot \lambda^{44} \cdot \cost(X^{(0)}, S_{\init}) < \cost(X^{(0)}, \hat{S}_i) \leq \polyup \cdot \OPT^{S_\init}_{k - 2\hat{\ell}}(X^{(0)}) $$
We also have,
\begin{align*}
    \OPT^{S_\init}_{k - 2\hat{\ell}}(X^{(0)}) &\leq 2 \cdot \cost(X^{(0)}, S_\init) + 8 \cdot \OPT_{k - 2\hat{\ell}}(X^{(0)}) \\
    & \leq 2\cdot (\polyup)^6 \cdot \OPT_{k}(X^{(0)}) + 8 \cdot \OPT_{k - 2\hat{\ell}}(X^{(0)}) \\
    & \leq (\polyup)^7 \cdot \OPT_{k - 2\hat{\ell}}(X^{(0)}) \\
\end{align*}
The first inequality follows from \Cref{lem:projection-lemma}, the second inequality follows from \Cref{eq:initial-cost-of-S-init}, and the last inequality follows since $\polyup$ is a large number.
Hence,
$$ (\polyup)^8 \cdot \lambda^{44} \cdot \OPT_k(X^{(0)}) <  \polyup \cdot \OPT^{S_\init}_{k - 2\hat{\ell}}(X^{(0)}) \leq (\polyup)^8 \cdot \OPT_{k - 2\hat{\ell}}(X^{(0)}), $$
which concludes
$$ \lambda^{44} \cdot \OPT_k(X^{(0)}) < \OPT_{k - 2\hat{\ell}}(X^{(0)}). $$
This shows that $2\hat{\ell} \geq \ell^\star$ by the definition of $\ell^\star$.

Now, we have $b \geq \frac{m- \ell-1}{4} - 1 = \frac{m-\ell-5}{4}$ and $\ell \geq \frac{\hat{\ell}}{36 \cdot \lambda^{51}} - 1$, we get
$$ \frac{m- \ell-5}{4} \leq b \leq \ell^\star \leq 2 \hat{\ell} \leq 2 \cdot 36 \cdot \lambda^{51} \cdot (\ell + 1) $$ 
and hence
$$ m \leq 288 \cdot \lambda^{51} \cdot (\ell+1) + \ell + 5 \leq 300 \cdot \lambda^{51} \cdot (\ell+1). $$
This concludes the proof of the claim.

\subsubsection*{Proof of Claim~\ref{cl:approx:key:11}}
For each $v_i \in V$, denote the cluster of $v_i \in V$ w.r.t.~the center set $V$ by $C_{v_i}$, i.e. $C_{v_i} = \{ x \in X^{(0)} \mid \dist(x,v_i) = \dist(x, V) \}$.
We define assignment $\sigma : X^{(0)} \rightarrow S_{\init} + \tilde{A}$, as follows.
\begin{itemize}
\item If $p \in C_{v_i}$ for some $i \in [1, k-m]$, then $\sigma(p) := u_i$. 
\item Otherwise,  $\sigma(p) := \pi_{V}(p)$.
\end{itemize}
In words, for every well-separated pair $(u_i, v_i) \in S_{\init} \times V$ all the points in the cluster of $v_i$ get reassigned to the center $u_i$, and the assignment of all other points remain unchanged (note that their assigned centers are present in $ S_{\init} + \tilde{A}$ as well as $V$).
Now, recall that at the beginning of the epoch, the set of centers $S_{\init}$ is robust (see Invariant~\ref{inv:start}).
Hence, by applying Lemma~\ref{lem:cost-well-sep-pairs}, we infer that 
$$\cost(X^{(0)},S_{\init}+\tilde{A}) \leq \sum_{x \in X^{(0)}} \dist(x, \sigma(x)) \leq 4\cdot  (\polyup)^3 \cdot \cost(X^{(0)},V). $$
The claim then follows since 
$\cost(X^{(0)},V) = \OPT_{k+\ell+1} \left(X^{(0)}\right)$ by definition.

\section{Recourse Analysis of Our Algorithm}\label{sec:recourse-analysis}

In this section, we provide the recourse analysis of our algorithm as follows.

\begin{lemma}\label{lem:main-recourse}
    The amortized recourse of the algorithm is at most $ \poly(d \cdot \log (n\Delta) /\epsilon) $, where $d$ is the dimension of the space, $n$ is the upper bound on the number of the points in the data-set, and $[0,\Delta]$ is the range of the coordinates of the points.
\end{lemma}

\paragraph{Roadmap.}
We start by showing the bound on the recourse throughout an epoch except the call to $\Robustify$ in \Cref{sec:recourse-before-robustify-full}.
Then, we proceed by analysis of the recourse overhead of $\Robustify$ subroutine in \Cref{sec:recourse-of-robustify-full}.

\subsection{\texorpdfstring{Recourse Before Calling $\Robustify$}{}}
\label{sec:recourse-before-robustify-full}

Fix an epoch (say) $\mathcal{E}$ of length $(\ell+1)$.
The total recourse incurred by our algorithm during this epoch is
\begin{equation}
\label{eq:recourse:1}
R_{\mathcal{E}} \leq \left| S_{\init} \oplus S^{(0)} \right| + \left( \sum_{t=1}^{\ell+1}  \left| S^{(t)} \oplus S^{(t-1)} \right| \right) + \left| S^{(\ell+1)} \oplus S_{\final} \right|.
\end{equation}
$S^{(0)}$ is obtained by deleting $\ell$ centers from $S_{\init}$ in Step 2, and hence we have $\left| S_{\init} \oplus S^{(0)} \right| = \ell$.
During Step 3, we incur a worst-case recourse of at most one per update.
Specifically, we have $\left| S^{(t)} \oplus S^{(t-1)} \right| \leq 1$ for all $t \in [1, \ell+1]$, and hence $\left|S^{(\ell+1)} \oplus S_{\init} \right| \leq \left|S^{(\ell+1)} \oplus S^{(0)} \right| + \left|S^{(0)} \oplus S_{\init} \right| \leq (\ell+1) + \ell = 2 \ell+1$. Thus, from~(\ref{eq:recourse:1}) we get:
\begin{eqnarray}
R_{\mathcal{E}} & \leq & \ell + (\ell+1) + \left| S^{(\ell+1)} \oplus S_{\final} \right| \nonumber  \\
& \leq & (2\ell+1) + \left|S^{(\ell+1)} \oplus S_{\init} \right| + \left|S_{\init} \oplus S_{\final} \right| \nonumber \\
& \leq & (4\ell+2) + \left|S_{\init} \oplus W' \right| + \left|W' \oplus S_{\final} \right| \label{eq:recourse:3}
\end{eqnarray}

Next, we bound $\left|S_{\init} \oplus W' \right|$.
According to Step 4 of our algorithm, we infer that 
\begin{align*}
   \left| S_{\init} \oplus W' \right| &\leq \left| S_{\init} \oplus S' \right| + \left| S' \oplus T' \right| + \left| T' \oplus W' \right|
\end{align*}
According to \Cref{eq:augment-approx}, $S' \supseteq S_\init$ satisfies $\left| S_{\init} \oplus S' \right| \leq \tilde{O}(\ell+1)$.
The $\left| S' \oplus T' \right|$ term is bounded by $\ell+1$ since $T' = S' + (X^{(\ell+1)} - X^{(0)})$.
Since $\left| S' \right| = k + \tilde{O}(\ell+1)$, and $W'$ is a subset of $T'$ of size $k$, we get $\left| T' \oplus W' \right| \leq  \tilde{O}(\ell+1)$.
This concludes
\begin{equation}\label{eq:W-prime-S-init}
    \left| S_{\init} \oplus W' \right| \leq \tilde{O} (\ell+1).
\end{equation}
Since the epoch lasts for $\ell+1$ updates, the $\left|S_{\init} \oplus W' \right|$ and $4\ell+2$ terms in the right hand side of~(\ref{eq:recourse:3}) contributes an amortized recourse of $\tilde{O}(1)$.
Moreover, the term $\left|W' \oplus S_{\final} \right|$ is proportional to the number of  $\MakeRbst$ calls made while computing $S_{\final} \leftarrow \Robustify\left( W' \right)$.
In \Cref{lm:makerobustcalls}, we show that the amortized number of calls made to the $\MakeRbst$ subroutine throughout the entire sequence of updates is at most $\tilde{O}(1)$, which proves \Cref{lem:main-recourse}.

\begin{lemma}
\label{lm:makerobustcalls}
The algorithm makes at most $\poly(1/\epsilon) \cdot \log n \cdot \log^2 (\sqrt{d}\Delta)$ many calls to $\MakeRbst$, amortized over the entire sequence of updates (spanning multiple epochs).
\end{lemma}

\subsection{\texorpdfstring{Proof of \Cref{lm:makerobustcalls}}{}}\label{sec:recourse-of-robustify-full}

At the end of the epoch,  we set $S_{\final} \leftarrow \Robustify\left(W'\right)$
Recall the subroutine $\Robustify$, and specifically, the partition $W' = W_1 \cup W_2 \cup W_3$.

We can partition the calls to $\MakeRbst$ that are made by $\Robustify\left(W'\right)$ into the following three types.

\medskip
\noindent {\bf Type I.} A call to $\MakeRbst(w)$ for some $w \in W_1 $.
The set $W_1 = W' \setminus S_\init $, has size at most $\tilde{O} (\ell+1) $ by \Cref{eq:W-prime-S-init}.
Since the epoch lasts for $(\ell+1)$ updates, the amortized number of Type I calls to $\MakeRbst$, per update, is $\tilde{O}(1)$.

\medskip
\noindent {\bf Type II.} A call to $\MakeRbst(w)$ for some $w \in W_2$ that is contaminated at the end of the epoch.
We can bound the number of these calls by $(\ell+1) \cdot \log (\sqrt{d}\Delta)$.
This bound comes from \Cref{claim:contaminate-only-one} which concludes each $x \in X^{(\ell+1)} \oplus X^{(0)}$ can contaminate at most $\lceil\log_{\lambda^3} (\sqrt{d}\Delta)\rceil + 1 \leq \log (\sqrt{d}\Delta)$ centers.
The amortized number of such Type II $\MakeRbst$ calls, per update, is then at most $\log (\sqrt{d}\Delta)$.

\medskip
\noindent {\bf Type III.} A call to $\MakeRbst(w)$ for some $w \in W_3$.
Note that according to \Cref{lem:robustify-calls-once}, after a call $u^\new \gets \MakeRbst(u)$ for some $u \in W_1 \cup W_2$, we will not make a call to $\MakeRbst$ on $u^\new$ in the current $\Robustify$.
Hence, our classification of calls to $\MakeRbst$ is comprehensive.

To bound the amortized number of Type III calls, we need to invoke a more ``global'' argument, that spans across multiple epochs. Consider a maximal ``chain'' of $j$ many Type III calls (possibly spanning across multiple different epochs),  in increasing order of time: 
\begin{eqnarray*}
&& w_1 \leftarrow \MakeRbst(w_0),  w_2 \leftarrow \MakeRbst(w_1), \\ 
&& \ldots, w_j \leftarrow \MakeRbst(w_{j-1}).
\end{eqnarray*}
Note that the calls in the above chain can be interspersed with other Type I, Type II, or Type III calls that are {\em not} part of the chain. 

\begin{claim}\label{claim:length-of-chain}
    The length of this chain is at most $\log (\sqrt{d}\Delta)$.
\end{claim}

\begin{proof}
Consider one Type III call $w_i \leftarrow \MakeRbst(w_{i-1})$.
This call is made on \Cref{line:type3-call} in \Cref{alg:modify:robustify}, which means that $t[w_{i-1}]$ satisfies $t[w_{i-1}] < t$ where $t$ is the smallest integer satisfying
$\lambda^{3t} \geq \hat{\dist}(w_{i-1}, W'-w_{i-1})/\lambda^{10}$.
Hence,
$$ \lambda^{3t[w_{i-1}]} < \hat{\dist}(w_{i-1}, W'-w_{i-1})/\lambda^{10} \leq \polyup \cdot \dist(w_{i-1}, W' - w_{i-1}) / \lambda^{10} \leq \dist(w_{i-1}, W' - w_{i-1}) / \lambda^{9}. $$
The second inequality follows by the guarantee of $\hat{\dist}$ in \Cref{lem:nearest-neighbor-distance} and the third inequality follows since $\lambda = (\polyup)^2$ is a large number.
Now, in the call $\MakeRbst(w_{i-1})$, we set $t[w_i]$ to be the smallest integer satisfying $\lambda^{3t[w_i]} \geq \hat{\dist}(w_{i-1}, W'-w_{i-1})/\lambda^{7}$.
It follows that,
\begin{align*}
   \lambda^{3t[w_i]} \geq \hat{\dist}(w_{i-1}, W'-w_{i-1})/\lambda^{7} &\geq \dist(w_{i-1}, W' - w_{i-1})/\lambda^{7} \\ &> \dist(w_{i-1}, W' - w_{i-1}) / \lambda^{9} > \lambda^{3t[w_{i-1}]}, 
\end{align*}
which concludes $t[w_i] > t[w_{i-1}]$.

As a result, $t[w_0] < t[w_1] < \ldots < t[w_j]$.
Since the aspect ratio of the space is $(\sqrt{d}\Delta)$, all of the integers $t[w_i]$ are between $0$ and $\log_{\lambda^3} (\sqrt{d}\Delta) < \log (\sqrt{d}\Delta)$ which concludes $j \leq \log (\sqrt{d}\Delta)$.
\end{proof}

For the chain to start in the first place, we must have had a Type I or Type II call to $\MakeRbst$ which returned the center $w_0$.
We can thus ``charge'' the length (total number of Type III calls) in this chain to the Type I or Type II call that triggered it (by returning the center $w_0$).
In other words, the total number of Type III calls ever made is at most $\log (\sqrt{d}\Delta)$ times the total number of Type I plus Type II calls. Since the amortized number of Type I and Type II calls per update is at most $ \tilde{O}(1)$, the amortized number of Type III calls per update is also at most $ \tilde{O}(1)$.
Recall that $(\sqrt{d}\Delta) = \sqrt{d}\Delta$ is the aspect ratio of the space, and we hide $\poly(1/\epsilon), \poly\log(n\Delta)$ and $\poly(d)$ in the $\tilde{O}$ notation. Hence, the proof of \Cref{lm:makerobustcalls} is complete.

\section{Update Time Analysis}\label{sec:update-time-analysis-full}
In this section, we analyze the update time of our algorithm as follows.

\begin{lemma}\label{lem:main-update-time}
    The amortized update time of the algorithm is at most $ n^{\epsilon} \cdot \poly(d \cdot \log(n\Delta)/\epsilon)$.
\end{lemma}

As discussed in different Steps of our algorithm in \Cref{sec:implementation}, all of the subroutines, except $\Robustify$, will spend at most $\tilde{O}(n^\epsilon \cdot (\ell+1))$ time.

In Step 1, the running time is dominated by the last call to restricted $k$-means (\Cref{lem:restricted-k-means-main}) that takes $\tilde{O}(n^\epsilon \cdot s_{i^\star})$ time, where $s_{i^\star} = O(\hat{\ell}) = O(\ell+1)$.
Step 2 is similar to Step 1, where we spend $\tilde{O}(n^\epsilon \cdot (\ell+1))$ time by calling restricted $k$-means to find the desired $S^{(0)} \subseteq S_\init$ of size $k-\ell$.
Step 3 is trivial and each update can be done in $\tilde{O}(1)$ time (apart from updating the underlying data structures).
In Step 4, we call the augmented $k$-means (\Cref{alg:augmented}) to find the desire $S' = S_\init + A$ of size $k + \lambda^{52} \cdot (\ell+1)$.
This subroutine spends $\tilde{O}(n^\epsilon \cdot \lambda^{52} \cdot (\ell+1)) = \tilde{O}(n^\epsilon \cdot (\ell+1))$ time.
Computing $T' = S' + (X^{(\ell+1)} - X^{(0)})$ is trivial.
We will call restricted $k$-means on $T'$ to get $W' \subseteq T'$ of size $k$.
This takes $\tilde{O}(|T'| - |W'|) = \tilde{O}(n^\epsilon \cdot (\ell+1))$ time as well.
Finally, we call $\Robustify$ on $W'$.

Apart from $\Robustify$ and updating the background data structures, all of the other parts of the algorithm in this epoch, take at most $\tilde{O}(n^\epsilon \cdot (\ell+1))$ time, which can be amortized on the length of the epoch.

\subsection{\texorpdfstring{Running Time of $\Robustify$}{}}

Recall that at the end of each epoch, $W'$ is partitioned into three sets $W' = W_1 \cup W_2 \cup W_3$.

\paragraph{Time Spent on $W_1 \cup W_2$.}
At the end of an epoch of length $\ell+1$, Since $W_1 = W' \setminus S_\init$, it is obvious that we can find $W_1$ in a total of $ |S' \setminus S_\init| = \tilde{O} (\ell+1) $.
According to the subroutine in \Cref{alg:modify:robustify}, to find $W_2$, for each $x \in X^{(0)} \oplus X^{(\ell+1)}$ and each $i \in [0, \lceil \log_{\lambda^3} (\sqrt{d}\Delta) \rceil]$, we will call the nearest neighbor oracle $\nn(x,S[i])$ (\Cref{lem:ANN-query}).
The oracle takes $\tilde{O}(n^\epsilon)$ time, which concludes the total time spend to find $W_2$ is at most $\tilde{O}(n^\epsilon \cdot |X^{(0)} \oplus X^{(\ell+1)}| \cdot \lceil \log_{\lambda^3} (\sqrt{d}\Delta) \rceil) = \tilde{O}(n^\epsilon \cdot (\ell+1))$.
Next, for each $u \in W_1 \cup W_2$ a call to $\MakeRbst(u)$ is made that takes $\tilde{O}(n^\epsilon)$ according to \Cref{sec:make-robust-time-analysis}.
In total, the time spend to handle centers in $W_1$ and $W_2$ (Lines 1 to 9 in \Cref{alg:modify:robustify}) is at most $\tilde{O}( n^\epsilon \cdot (\ell+1))$, which can be amortized over the length of the epoch.

\paragraph{Time Spent on $W_3$.}
We need a global type of analysis for this part.
It is possible that $\Robustify$ takes a huge amount of time for handling points in $W_3$.
But, considering the entire sequence of updates, we show that the amortized update time of this part of the $\Robustify$ subroutine is at most $\tilde{O} (n^{\epsilon})$.

According to our algorithm, for each center $u$ in the yellow set $\calY$, in $\tilde{O}(1)$ time (\Cref{line:pop-robustify} to \Cref{if:condition-W3})\footnote{Note that the number $\hat{\dist}(u,S-u)$ 
is maintained explicitly by the data structure in \Cref{lem:nearest-neighbor-distance} 
.}, our algorithm decides whether or not a call to $\MakeRbst$ will be made on $u$.
Each call to $\MakeRbst$ also takes $\tilde{O}(n^\epsilon)$ time according to \Cref{sec:make-robust-time-analysis}.
As a result, if the number of elements added to $\calY$ is $N$ and the number of calls to $\MakeRbst$ is $M$ throughout the entire algorithm, we conclude that the running time of handling $W_3$ is at most $ N \cdot \tilde{O}(1) + M \cdot \tilde{O}(n^\epsilon)$.
According to the recourse analysis of our algorithm, the amortized number of calls to $\MakeRbst$ is bounded by $M \leq \tilde{O}(1)$.

Elements are added to $\calY$ by data structures $D_{\lambda^i}(S)$ for each $i \in [0, \lceil \log_\lambda (\sqrt{d}\Delta) \rceil]$.
Each update on the main center set $S$ (insertion or deletion) is considered an update for this data structure.
This data structure has an amortized update time of $\tilde{O}(n^\epsilon)$, which implicitly means that the number of centers $u$ whose bit $b_{\lambda^i}(u,S)$ has changed is at most $\tilde{O}(n^\epsilon)$ in amortization (w.r.t.~the sequence of updates on $S$).
Hence, for each $i$, the number of elements added to $\calY$ by the data structure $D_{\lambda^i}(S)$ is at most $\tilde{O}(M \cdot n^\epsilon)$.
As a result, the total number of elements added to $\calY$ throughout the entire algorithm is at most $N \leq \tilde{O}((\lceil \log_{\lambda} (\sqrt{d}\Delta) \rceil + 1) \cdot M \cdot n^\epsilon) = \tilde{O}( n^\epsilon)$

Finally, the amortized update time of $\Robustify$ spent on handling $W_3$ throughout the entire algorithm is $\tilde{O}(n^\epsilon)$.

\subsection{Update Time of Background Data Structures}
Each update in the main data set $X$ and in the main center set $S$ is considered an update for the background data structures that are maintained in our algorithm.
Throughout a sequence of $T$ updates on $X$, the number of updates on $S$ is bounded by $\tilde{O}(T)$ according to the recourse analysis of our algorithm (see \Cref{lem:main-recourse}).
Since each of the data structures maintained in the background has $\tilde{O}(n^\epsilon)$ update time (see \Cref{sec:data-structures-lemmas}), the total time spent to update these data structures during these $T + \tilde{O}(T)$ updates on either $X$ or $S$, is at most $\tilde{O}(T \cdot n^\epsilon)$.
Hence, the amortized update time of all the background data structures is $\tilde{O}(n^\epsilon)$ in amortization (w.r.t.~the original update sequence on the main dataset $X$).

\subsection{\texorpdfstring{Running Time of $\MakeRbst$}{}}
\label{sec:make-robust-time-analysis}

The main for loop in the implementation consists of at most $O(\log_\lambda (\sqrt{d}\Delta))$ many iterations.
Each iteration is done in $\tilde{O}(n^\epsilon)$ time according to the data structure (\Cref{lem:1-means}) used in \Cref{line:invoke-one-median-ball} of the implementation.
It is obvious that all other parts of the algorithm take at most $\tilde{O}(1)$ time.
Hence, the total time spent to perform one call to $\MakeRbst$ is at most $\tilde{O}(n^\epsilon)$.

\section{\texorpdfstring{Improving the Update Time to $\tilde O(k^\epsilon)$}{}}
\label{part:from-n-to-k}

\paragraph{Roadmap.}
In this section, we show how to improve the update time of our algorithm from $\tilde O(n^\epsilon)$ to $\tilde O(k^\epsilon)$. We do this by using \emph{dynamic sparsification} as a black box, in a similar manner to previous dynamic clustering algorithms \cite{BhattacharyaCGLP24, BCF24}. In \cite{BCF24}, the authors give an algorithm that maintains a $O(1)$-approximation with $\tilde O(1)$ recourse and $\tilde O(n)$ update time. Using the dynamic $k$-means algorithm of \cite{nips/BhattacharyaCLP23} as a dynamic \emph{sparsifier} in a black-box manner, they can then improve the update time of their algorithm to $\tilde O(k)$, while only incurring $O(1)$ overhead in the approximation and $\tilde O(1)$ overhead in the recourse.
We speed up the update time of our algorithm from $\tilde O(n^\epsilon)$ to $\tilde O(k^\epsilon)$ using a completely analogous approach, except that we use the sparsifier of \cite{esa/TourHS24}, which is more well-suited to Euclidean spaces.

\paragraph{Our Sparsified Dynamic Algorithm.}
The following theorem summarises the relevant properties of the dynamic algorithm of \cite{esa/TourHS24}, which we call $\Sparsifier$.\footnote{The dynamic algorithm of \cite{esa/TourHS24} dynamically maintains a \emph{coreset} of the input space. For the sake of simplicity, we do not define the notion of a coreset and instead focus on the fact that it can be used as a sparsifier.} We use this algorithm as a black-box.

\begin{theorem}[\cite{esa/TourHS24}]\label{thm:dyn-sparsifier}
    There is an algorithm $\Sparsifier$ that, given a dynamic weighted $d$-dimensional Euclidean space $X$, maintains a weighted subspace $U \subseteq X$ of size $\tilde O(k)$ with $\tilde O(1)$ amortized update time and $\tilde O(1)$ amortized recourse. Furthermore, after each update, the following holds with probability at least $1 - 1/\poly(n)$, 
    \begin{itemize}
        \item Any $\beta$-approximation to the $k$-means problem on the weighted subspace $U$ is also a $O(\beta)$-approximation to the $k$-means problem on the space $X$.
    \end{itemize}
\end{theorem}

The algorithm \Cref{thm:dyn-sparsifier} shows that we can efficiently maintain a dynamic sparsifier of the space $X$ with high probability. The following theorem summarizes our main algorithm, which we call $\ALG$.

\begin{theorem}\label{thm:main-pre-sparsifier}
    There is an algorithm $\ALG$ that, given a dynamic weighted $d$-dimensional Euclidean space of size $n$ and a sufficiently small parameter $\epsilon > 0$, maintains a $\poly(1/\epsilon)$-approximate solution to the $k$-means problem with $\tilde O(n^\epsilon)$ amortized update time and $\tilde O(1)$ amortized recourse. The approximation guarantee holds with probability at least $1 - 1 / \poly(n)$ after each update.
\end{theorem}

We now show how to combine these algorithms in order to improve the running time of our algorithm from $\tilde O(n^\epsilon)$ to $\tilde O(k^\epsilon)$.\footnote{We note that the main challenge here is doing this while ensuring that the success probability of our algorithm remains $1 - 1/\poly(n)$, since the naive approach leads to a success probability that is $1 - 1/\poly(k)$, which can be significantly smaller.}

\subsection{Algorithm Description}

We initialize a copy of the algorithm $\Sparsifier$ and feed it the dynamic space $X$ as input. In response, the algorithm $\Sparsifier$ maintains a weighted subspace $U \subseteq X$, with the properties outlined in \Cref{thm:dyn-sparsifier}. We also initialize one `primary' copy of the algorithm $\ALG$, which we denote by $\ALG^\star$, along with $L = O(\log n)$ `verification' copies of $\ALG$, which we denote by $\ALG_1',\dots, \ALG_L'$. Each of these algorithms is fed the dynamic space $U$ as input. In respose, the algorithms $\ALG^\star, \ALG_1',\dots,\ALG_L'$ maintain solutions $S^\star, S_1',\dots,S_L'$ to the $k$-means problem on the space $U$.
Thus, whenever the space $X$ is updated by some point insertion or deletion, the algorithm $\Sparsifier$ updates the subspace $U$. In turn, the algorithms  $\ALG^\star, \ALG_1',\dots,\ALG_L'$ then update the solutions $S^\star, S_1',\dots,S_L'$.

At any point in time, the output of our algorithm is the solution $S^\star$. In order to make sure that the solution $S^\star$ has a good approximation guarantee, we make the following modification to this algorithm.

\paragraph{Detecting When $S^\star$ Expires.} Our algorithm maintains an estimate $\mathcal E := \min_{i \in [L]} \cost(U, S_i')$ of $\OPT_k(U)$. Let $\alpha \leq \poly(1/\epsilon)$ denote an upper bound on the approximation ratio of $\ALG^\star$.
After handling each update, we check if $\cost(U, S^\star) > \alpha \cdot \mathcal E$. If this is the case, we completely recompute the primary instance $\ALG^\star$ from scratch \emph{using fresh randomness}, which reconstructs the solution $S^\star$. We repeat this process until $\cost(U, S^\star) \leq \alpha \cdot \mathcal E$.

\subsection{Analysis}

We now show that this modified version of the algorithm, which uses sparsification, has an update time of $\tilde O(k^\epsilon)$ without sacrificing any other guarantees. In particular, we prove the following theorem.

\begin{theorem}\label{thm:main-post-sparsifier}
    There is an algorithm that, given a dynamic weighted $d$-dimensional Euclidean space of size at most $n$, a polynomially bounded sequence of updates to the space $X$, and a sufficiently small parameter $\epsilon > 0$, maintains a $\poly(1/\epsilon)$-approximate solution to the $k$-means problem with $\tilde O(k^\epsilon)$ amortized update time and $\tilde O(1)$ amortized recourse. The approximation, update time, and recourse guarantees hold with probability at least $1 - 1 / \poly(n)$ across the entire sequence of updates.
\end{theorem}

\subsubsection{Approximation Ratio Analysis}

We first show that, at any point in time, we have that $\cost(X, S^\star) \leq O(\alpha^2) \cdot \OPT_k(X)$ with probability at least $1 - 1/\poly(n)$.
We begin with the following lemma, which shows that $\mathcal E$ is an $\alpha$-approximation to $\OPT_k(U)$ with high probability.

\begin{lemma}\label{lem:sparsifier:1}
    After each update, we have $\mathcal E \leq \alpha \cdot \OPT_k(U)$ with probability at least $1 - 1/\poly(n)$.
\end{lemma}

\begin{proof}
    By \Cref{thm:main-pre-sparsifier}, at each point in time, we have that $\cost(U, S_i') \leq \alpha \cdot \OPT_k(U)$ with probability at least $1 - 1/\poly(k) \geq 1/2$ for each $i \in [L]$. Since each of the algorithms $\ALG_1',\dots,\ALG_L'$ uses independent randomness, the probability that $\cost(U, S_i') > \alpha \cdot \OPT_k(U)$ for all $i \in [L]$ is at most $(1/2)^L = 1/\poly(n)$.
\end{proof}

Since our algorithm maintains that $\cost(U, S^\star) \leq \alpha \cdot \mathcal E$ after each update, it follows from \Cref{lem:sparsifier:1} that $\cost(U, S^\star) \leq \alpha^2 \cdot \OPT_k(U)$ with probability at least $1 - 1/\poly(n)$ after each update.
By \Cref{thm:dyn-sparsifier}, we know that if $\cost(U, S^\star) \leq \alpha^2 \cdot \OPT_k(U)$, then $\cost(X, S^\star) \leq O(\alpha^2) \cdot \OPT_k(X)$ with probability at least $1 - 1/\poly(n)$. Thus, our approximation guarantee follows from union bounding over these two events.

\subsubsection{Recourse and Update Time Analysis}

Consider a sequence of $T = O(\poly(n))$ updates $\sigma_1,\dots, \sigma_T$ defining the dynamic space $X$. For each $i \in [T]$, let $R_i$ denote the recourse caused by the update $\sigma_i$, i.e.~the total number of points that are either inserted or deleted from $S^\star$ by our algorithm while handling the update, and let $R := \sum_{i = 1}^T R_i$. In other words, $R$ is the total recourse of the solution $S^\star$ caused by handling the updates $\sigma_1,\dots,\sigma_T$. Thus, $R/T$ is the amortized recourse. Similarly, for each $i \in [T]$, let $r_i$ denote the indicator for the event that the algorithm restarts the algorithm $\ALG^\star$ using fresh randomness after handling the update $\sigma_i$, and let $r := \sum_{i = 1}^T r_i$.

\begin{lemma}\label{lem:sparsifier:2}
    We have that $R = \tilde O(T) + O(kr)$.
\end{lemma}

\begin{proof}
    Since the algorithms $\Sparsifier$ and $\ALG$ both have an amortized recourse of $\tilde O(1)$, it follows that the algorithm obtained by composing them also has an amortized recourse of $\tilde O(1)$. Thus, ignoring the recourse caused by resetting the algorithm $\ALG^\star$, the algorithm has an amortized recourse of $\tilde O(1)$, and thus the total recourse across $T$ updates is at most $\tilde O(T)$. If the algorithm resets the algorithm $\ALG^\star$, it incurs an additional recourse of $O(k)$ since the solution might change completely. Thus, we get an additive $O(kr)$ term in the recourse of our algorithm to account for the resets.
\end{proof}

Now, let $p(n) = \poly(n) = \Omega(n)$ be the polynomial such that the approximation guarantee of $\ALG$ holds with probability at least $1 - 1/p(n)$ after each update, as described in \Cref{thm:main-pre-sparsifier}. Consider some subsequence of updates $\sigma_{\ell + 1},\dots,\sigma_{\ell + q}$, where $q \leq p(k)/2$. We denote $\sigma_{\ell + j}$ and $r_{\ell + j}$ by $\sigma'_j$ and $r'_j$ respectively, and let $r' := \sum_{j = 1}^q r'_j$.

\begin{lemma}\label{lem:sparsifier:3}
    We have that $\Pr[r' \geq \Omega (\log n)] \leq 1/\poly(n)$.
\end{lemma}

\begin{proof}
    We begin with the following simple claim, which bounds the probability of our algorithm resetting $\ALG^\star$ after an update.
\begin{claim}\label{claim:sparsifier:0}
    For each $j \in [q]$, we have that $\Pr [r'_j = 1 ] \leq 1/p(k)$.
\end{claim}

\begin{proof}
    If the space $U$ has size less than $k$ after handling the update $\sigma'_j$, then $S^\star$ is the optimum solution, and we will not restart $\ALG^\star$. Otherwise, by \Cref{thm:main-pre-sparsifier}, we know that $S^\star$ is an $\alpha$-approximation to $\OPT_k(U)$ with probability at least $1 - 1/p(|U|) \geq 1 - 1/p(k)$, which implies that $\cost(U, S^\star) \leq \alpha \cdot \OPT_k(U) \leq \alpha \cdot \mathcal E$. Thus, we have that $\cost(U, S^\star) > \alpha \cdot \mathcal E$, and hence $r'_j = 1$, with probability at most $1/p(k)$.
\end{proof}

\begin{claim}\label{claim:sparsifier:1}
    For each $i \in [k]$, we have that $\Pr [r' \geq i \mid r' \geq i - 1 ] \leq 1/2$.
\end{claim}

\begin{proof}
    We refer to an update $\sigma_j'$ as \emph{bad} if $r_{j}' = 1$, i.e.~if we restart the algorithm $\ALG^\star$ with fresh randomness at the end of the update.
    Suppose that $r' \geq i - 1$ and let $\sigma_{j_1}',\dots,\sigma_{j_{i-1}}'$ be the first $i - 1$ bad updates in $\sigma'_1,\dots,\sigma'_q$. Since we restart $\ALG^\star$ after update $\sigma'_{j_{i-1}}$, the random bits used by $\ALG^\star$ that effect whether any of the updates $\sigma_{j_{i-1} + 1}',\dots ,\sigma_q'$ are bad are \emph{independent} of the random bits used by $\ALG^\star$ to handle the previous updates. Thus, conditioned on this event, each of the updates $\sigma_{j_{i-1} + 1}',\dots ,\sigma_q'$ is bad with probability at most $1 / p(k)$ by \Cref{claim:sparsifier:0}. Taking a union bound, we get that none of them is bad with probability at least $1 - q/p(k) \geq 1/2$.
\end{proof}
    By applying \Cref{claim:sparsifier:1}, we get that
    $$ \Pr[r' \geq \Omega (\log n)] = \Pr [r' \geq 0] \cdot \prod_{i = 1}^{\Omega (\log n)} \Pr [r' \geq i \mid r' \geq i -1] \leq \prod_{i = 1}^{\Omega (\log n)} \frac{1}{2} = 2^{-\Omega(\log n)} = \frac{1}{\poly(n)}. $$
\end{proof}

\begin{corollary}\label{cor:sparsifier:1}
    We have that $r = O(T \log n / k)$ with probability at least $1 - 1/\poly(n)$.
\end{corollary}

\begin{proof}
    By splitting up the sequence of updates $\sigma_1,\dots,\sigma_T$ into $O(T/p(k))$ many subsequences of length $p(k)/2$, we can see that our algorithm resets $\ALG^\star$ at most $O(\log n)$ many times during while handling each of these subsequences with probability at least $1 - 1/\poly(n)$. Applying a union bound, it follows that $r = O(\log n) \cdot O(T / p(k)) \leq O(T \log n / k)$ with probability at least $1 - 1/\poly(n)$.
\end{proof}

It follows from \Cref{lem:sparsifier:2,cor:sparsifier:1} that $R = \tilde O(T) + k \cdot \tilde O(T/k) = \tilde O(T)$ with probability at least $1 - 1/\poly(n)$. Thus, the amortized recourse of our algorithm is $\tilde O(1)$ with high probability.

\paragraph{Update Time Analysis.}
Let $r^\star$ denote the \emph{total} number of times that the algorithm resets $\ALG^\star$ across the sequence of updates $\sigma_1,\dots,\sigma_T$. Note that we might have $r^\star > r$ if the algorithm resets $\ALG^\star$ multiple times between two updates. The following lemma shows that this can not happen too many times.

\begin{claim}\label{claim:sparsifier:4}
    We have that $r^\star = O(r \log n)$ with probability at least $1 - 1/\poly(n)$.
\end{claim}

\begin{proof}
    For each $i \in [T]$, let $r^\star_i$ denote the total number of times that our algorithm resets $\ALG^\star$ after handling the update $\sigma_i$. Note that $r_i \geq r_i$ and $r^\star = \sum_{i = 1}^T r_i^\star$. Now, suppose that we have restarted $\ALG^\star$ $j$ times so far after handling the update $\sigma_i$, i.e.~that $r^\star_i \geq j$. Since the algorithm uses fresh randomness, the probability that we restart it again (i.e.~that $r_i^\star \geq j + 1$) is at most $1/p(k) \leq 1/2$. Thus, it follows that
    $$ \Pr[r^\star_i \geq \Omega (\log n) \mid r^\star_i \geq 1] = \prod_{j = 1}^{\Omega (\log n)} \Pr [r^\star_i \geq j + 1 \mid r^\star_i \geq j] \leq  2^{-\Omega(\log n)} =\frac{1}{\poly(n)}. $$
    Since $r_i^\star = 0$ if $r_i = 0$, it follows from a union bound that $r^\star = O(r \log n)$ with probability at least $1 - 1/\poly(n)$.
\end{proof}

\begin{lemma}\label{lem:sparsifier:time}
    The total update time of our algorithm is $\tilde O(k^\epsilon) \cdot (T + kr^\star)$.
\end{lemma}

\begin{proof}
    The total time taken to update our algorithm is the sum of (i) the time taken to maintain $\Sparsifier$, (ii) the time taken to maintain our $O(\log n)$ copies of $\ALG$, and (iii) the time taken to reset $\ALG^\star$. Since $\Sparsifier$ has an amortized update time of $\tilde O(1)$, this can be maintained in time $\tilde O(1) \cdot T$. Since $\Sparsifier$ has an amortized recourse of $\tilde O(1)$ and $\ALG^\star$ has an amortized update time of $\tilde O(k^\epsilon)$, the total time taken to update the $O(\log n)$ copies of $\ALG$ maintained by our algorithm is $O(\log n) \cdot \tilde O(T) \cdot \tilde O(k^\epsilon) = \tilde O(k^\epsilon) \cdot T$. Finally, each time we reset $\ALG^\star$, it takes us $\tilde O(k^{1 + \epsilon})$ time to reconstruct the data structure using fresh randomness. Thus, the total time spent resetting $\ALG^\star$ is $\tilde O(k^{1 + \epsilon}) \cdot r^\star$.
\end{proof}

It follows from \Cref{lem:sparsifier:time,claim:sparsifier:4,cor:sparsifier:1} that the total update time of our algorithm is $\tilde O(k^\epsilon) \cdot T$ with probability at least $1 - 1/\poly(n)$. Thus, the amortized update time of our algorithm is $\tilde O(k^\epsilon)$ with high probability.

 \section{Deferred Proofs}\label{sec:deferred-full}

\subsection{\texorpdfstring{Proof of \Cref{lem:robust-property-1}}{}}

According to the second property of a $t$-robust sequence, we have two cases.
If \Cref{cond:robust1} holds, then $x_{j-1} = x_j$ and it is trivial that $\dist(x_{j-1}, x_j) \leq 2 \cdot \lambda^{3j-1}$.
Now, assume \Cref{cond:robust2} holds.
We have
\begin{align*}
    \dist(x_j,x_{j-1})^2 
    &= \frac{1}{w(B_j)}\cdot \sum_{x \in B_j} w(x) \cdot \dist(x_j,x_{j-1})^2 \\
    &\leq
    \frac{1}{w(B_j)}\cdot  \sum_{x \in B_j} w(x) \cdot  (2 \cdot \dist(x_j,x)^2 + 2\cdot \dist(x_{j-1},x)^2) \\
    &=
    \frac{1}{w(B_j)}\cdot \left( 2 \cdot \cost(B_j, x_j) + 2 \cdot \cost(B_j, x_{j-1}) \right) \\
    &\leq
    \frac{4}{w(B_j)} \cdot \cost(B_j, x_j)
    \leq 4 \cdot \lambda^{6j-2}.
\end{align*}
This concludes that $\dist(x_j,x_{j-1}) \leq 2 \cdot \lambda^{3j-1}$ which proves the first property.

Now, assume $x \in B_{j-1}$ is arbitrary.
We show that $x \in B_j$, which concludes $B_{j-1} \subseteq B_j$.
Since $B_{j-1} \subseteq \ball(x_j, \lambda^{3(j-1)+1})$, we have $\dist(x, x_j) \leq \lambda^{3j-2}$.
Hence,
$$ \dist(x,x_{j}) \leq \dist(x,x_{j-1}) + \dist(x_{j-1},x_j) \leq \lambda^{3j-2} + 2 \cdot \lambda^{3j-1} \leq \lambda^{3j}, $$
where the last inequality follows from the fact that we assumed $\lambda$ is a large constant.
Finally, since $\ball(x_j, \lambda^{3j}) \subseteq B_j$, we conclude that $x \in B_j$.

Property 3 follows from
\begin{equation*}
    \dist(x_0,x_j) \leq \sum_{i=1}^j \dist(x_{i-1},x_i) \leq 2 \cdot \sum_{i=1}^j \lambda^{3i-1} \leq 2 \cdot \frac{\lambda^{3(j+1)}-1}{\lambda(\lambda^3 - 1)} \leq 4 \cdot \lambda^{3j-1},
\end{equation*}
where the last inequality follows from $\lambda$ being a large constant.

\subsection{\texorpdfstring{Proof of \Cref{lem:robust-property-2}}{}}

Since $(x_0,x_1,\ldots , x_t)$ is $t$-robust, for every $1 \leq j \leq t$, we know
\begin{equation}\label{eq:pj-1-better-than-pj}
    \cost(B_j,x_j) \geq \cost(B_{j}, x_{j-1}).
\end{equation}
Assume $y \in X \setminus B_j$.
We have
$$ \dist(y,x_j) \geq \lambda^{3j} \geq (\lambda / 4) \cdot  \dist(x_0,x_j) \geq \dist(x_0,x_j). $$
The first inequality follows from $\ball(x_j, \lambda^{3j}) \subseteq B_j$, the second inequality follows from \Cref{lem:robust-property-1}, and the last one follows from $\lambda$ being a large constant.
Hence,
$$ 2 \cdot  \dist(y,x_j)
= \dist(y,x_j) + \dist(y,x_j)
\geq \dist(y,x_j) + \dist(x_0,x_j)
\geq \dist(y,x_0),$$
which means
\begin{equation}\label{eq:dyxj-compare-to-dyx0}
    \dist(y,x_j)^2 \geq \frac{1}{4}\cdot \dist(y,x_0)^2.
\end{equation}
Finally, Since $ B_1 \subseteq \cdots \subseteq B_i \subseteq U$ (by \Cref{lem:robust-property-1}) we can apply \Cref{eq:pj-1-better-than-pj} and \Cref{eq:dyxj-compare-to-dyx0} repeatedly to get
\begin{align*}
    \cost(U, x_i) &= \cost(B_i,x_i) + \cost(U \setminus B_i, x_i) \\
    &\geq
    \cost(B_i, x_{i-1}) + \frac{1}{4}\cdot  \cost(U \setminus B_i,x_0) \\
    &=
    \cost(B_{i-1}, x_{i-1}) + \cost(B_i \setminus B_{i-1}, x_{i-1}) + \frac{1}{4} \cdot \cost(U \setminus B_i, x_0) \\
    &\geq
    \cost(B_{i-1},x_{i-2}) + \frac{1}{4} \cdot \cost(B_i \setminus B_{i-1},x_0) + \frac{1}{4}\cdot  \cost(U \setminus B_i,x_0) \\
    &=
    \cost(B_{i-1}, x_{i-2}) + \frac{1}{4}\cdot  \cost(U \setminus B_{i-1}, x_0) \\
    & \ \ \! \vdots \\
    &\geq
    \cost(B_{1}, x_{0}) + \frac{1}{4} \cdot \cost(U \setminus B_{1}, x_0)
    \geq
    \frac{1}{4}\cdot\cost(U, x_0) .
\end{align*}

\subsection{\texorpdfstring{Proof of \Cref{lem:cost-after-robustify}}{}}

Assume $S \gets \Robustify(W)$ is the center set returned by calling $\Robustify$ on $W$.
According to \Cref{lem:robustify-calls-once}, $\Robustify$ calls $\MakeRbst$ on each center at most once.
Let $r$ be the number of centers in $W$ for which a call to $\MakeRbst$ has happened.
Assume
$W = \{w_1,\ldots, w_k\}$ is the ordering by the time that $\MakeRbst$ is called on the centers $w_1,w_2,\ldots,w_r$ and the last elements are ordered arbitrarily.
Assume also that $S =  \{w'_1,\ldots, w'_k\}$ where $w_i'$ is obtained by the call \MakeRbst \ on $w_i$ for each $1 \leq i\leq r$ and $w_i' = w_i$ for each $r+1 \leq i \leq k$.

For every $1 \leq j \leq r$, integer $t[w_j']$ is the smallest integer satisfying
$$ \lambda^{3t[w_j']} \geq \hat{\dist}(w_j,\{ w_1',\ldots,w_{j-1}',w_{j+1},\ldots,w_k \})/\lambda^7. $$
According to the guarantee of $\hat{\dist}$ in \Cref{lem:nearest-neighbor-distance}, we conclude
$$ \lambda^{3t[w_j']} < \hat{\dist}(w_j,\{ w_1',\ldots,w_{j-1}',w_{j+1},\ldots,w_k \})/\lambda^4 \leq \dist(w_j,\{ w_1',\ldots,w_{j-1}',w_{j+1},\ldots,w_k \})/\lambda^4. $$
Hence,
\begin{equation}\label{eq:proof-cost-robustify-dwjwiprime}
    \lambda^{3t[w_j']} \leq \dist(w_j,w_i')/\lambda^4 \quad \forall i < j
\end{equation}
and
\begin{equation}\label{eq:proof-cost-robustify-dwjwi}
    \lambda^{3t[w_j']} \leq \dist(w_j,w_i)/\lambda^4 \quad \forall i > j.
\end{equation}
Now, for every $1 \leq i < j \leq r$, we conclude
\begin{eqnarray*}
    \lambda^4\cdot \lambda^{3t[w_j']} &\leq& \dist(w_i',w_j) \leq  \dist(w_i,w_j) + \dist(w_i,w_i') \\ &\leq& \dist(w_i,w_j) + 4 \cdot \lambda^{3t[w_i']-1} \leq (1 + 4/\lambda^5) \cdot \dist(w_i,w_j).
\end{eqnarray*}
The first inequality holds by \Cref{eq:proof-cost-robustify-dwjwiprime}, the third one holds by \Cref{lem:robust-property-1}  (note that $w_i' = $ \MakeRbst$(w_i)$) and the last one holds by \Cref{eq:proof-cost-robustify-dwjwi} (by exchanging $i$ and $j$ in this equation since $i < j$).
Since $\lambda$ is a large number, we conclude that for each $i < j$, the following holds
$$ \dist(w_i,w_j) > \lambda^4/(1 + 4/\lambda^5) \cdot \lambda^{3t[w_j']} > 2 \cdot \lambda^{3t[w_j']}. $$
The same inequality holds for each $i>j$ by \Cref{eq:proof-cost-robustify-dwjwi} and $\lambda$ being large.
As a result,
$\dist(w_i,w_j) > 2 \cdot \lambda^{3t[w_j']}$ for all $i \neq j$.
This concludes that all the points around $w_j$ of distance at most $\lambda^{3t[w_j']}$ belong to the cluster of center $w_j$ in solution $W$.
Denote by $C_{w_j}$ the cluster of $w_j$ in solution $W$, i.e., $C_{w_j} := \{ x \in X \mid \dist(x,w_j) = \dist(x, W) \}$  (breaking the ties arbitrarily).
We conclude
$$ \ball(w_j, \lambda^{3t[w_j']}) \subseteq C_{w_j}, $$
for every $1 \leq j \leq r$.
It follows from \Cref{lem:robust-property-2} that
\begin{equation*}
   \cost(C_{w_j},w_j') \leq 4 \cdot \cost(C_{w_j},w_j) \quad \forall 1 \leq j \leq r. 
\end{equation*}
Finally, since $C_{w_j}$ for $1 \leq j \leq k$ form a partition of $X$, we conclude
\begin{eqnarray*}
    \cost(X, S) &=& 
    \sum_{j=1}^k \cost(C_{w_j}, S)
    \leq
    \sum_{j=1}^k \cost(C_{w_j},w_j') \\
    &=&
    \sum_{j=1}^r \cost(C_{w_j},w_j')  + \sum_{j=r+1}^k \cost(C_{w_j}, w_j') \\
    &\leq&
    4 \cdot \sum_{j=1}^r \cost(C_{w_j},w_j) + \sum_{j=r+1}^k \cost(C_{w_j},w_j) \\
    &\leq&
    4 \cdot \sum_{j=1}^k \cost(C_{w_j},w_j)
    =4 \cdot \cost(X,W).
\end{eqnarray*}

\subsection{\texorpdfstring{Proof of \Cref{lem:cost-well-sep-pairs}}{}}

Denote the set of points in the cluster of $v \in V$ in the center set $V$ by $C_v$, i.e.
$C_v = \{ x \in X \mid \dist(x,v) = \dist(x,V) \}$.
If $\dist(u,v)=0$, the statement is trivial.
Now assume $\dist(u,v) \geq 1$.
Since $S$ is robust, then $u$ is $t$-robust where $t$ is the smallest integer such that
$ \lambda^{3t} \geq \dist(u,S-u)/\lambda^{10} \geq (\lambda^{20}/\lambda^{10}) \cdot \dist(u,v) = \lambda^{10} \cdot \dist(u,v) $.
The second inequality is because $(u,v)$ is well-separated pair.
Assume $t^*$ is the integer satisfying
\begin{equation}\label{eq:definition-tstar}
   \lambda^{10} \cdot \dist(u,v) \leq \lambda^{3t^*} < \lambda^{13} \cdot \dist(u,v). 
\end{equation}
Note that $t^* \geq 1$ (since we assumed $\dist(u,v) \geq 1$). We have $t^* \leq t$ and $u$ is $t$-robust, then it is also $t^*$-robust which means there is a $t^*$-robust sequence 
$(x_0,x_1,\ldots,x_{t^*})$ such that $x_0=u$.
Consider approximate balls $B_i$ around $x_i$ in the definition of the robust sequence.

\begin{claim}\label{claim:Bi-subseteq-Cv}
    For each $0 \leq i \leq t^*$, we have $B_i \subseteq C_v$.
\end{claim}

\begin{proof}
    Assume $y \in B_i$. Then,
    \begin{eqnarray*}
        \dist(y,v) &\leq& \dist(y,x_i) + \dist(x_i,u) + \dist(u,v)
        = \dist(y,x_i) + \dist(x_i,x_0) + \dist(u,v) \\
        &\leq& \lambda^{3i+1} + 4 \cdot \lambda^{3i-1} + \lambda^{3t^* - 10}
        \leq (2\lambda) \cdot \lambda^{3t^*}.
    \end{eqnarray*}
    Second inequality follows from $y \in B_i \subseteq \ball(x_i,\lambda^{3i+1})$, \Cref{lem:robust-property-1} and \Cref{eq:definition-tstar}. The last inequality follows from $\lambda$ being a large constant.
    Now, for every $v' \in V - v$ we have
    $$ \dist(v,v') \geq \dist(v,V-v) \geq \lambda^{20} \cdot \dist(u,v) > \lambda^{3t^*+7} \geq (4\lambda) \cdot \lambda^{3t^*} \geq 2 \cdot \dist(y,v), $$
    where the second inequality is by $(u,v)$ being well-separated, the third one is by \Cref{eq:definition-tstar}, and the fourth one is since $\lambda$ is a large constant.
    This concludes $\dist(y,v') \geq \dist(v,v') - \dist(y,v) > \dist(y,v)$ which means $v$ is the closest center to $y$ in $T$. So, $y \in C_v$. Since $y \in B_i$ was arbitrary, we conclude that $B_i \subseteq C_v$.
\end{proof}

Now, according to the definition of a $t$-robust sequence, we have two cases as follows.

\paragraph{Case 1.} 
\Cref{cond:robust1} holds for $i= t^{*}$. So,
$\cost(B_{t^*}, x_{t^*})/w(B_{t^*}) \geq \lambda^{6t^*-4}$ and $x_{t^*-1} = x_{t^*}$. We conclude
\begin{eqnarray*}
    \dist(v,x_{t^*}) &=& \dist(v,x_{t^*-1})
    \leq \dist(v,x_0) + \dist(x_0,x_{t^*-1})
    = \dist(v,u) + \dist(x_0,x_{t^*-1}) \\
    &\leq& \lambda^{3t^*-10} + 4 \cdot \lambda^{3(t^*-1)-1}
    \leq \lambda^{3t^*-3}.
\end{eqnarray*}
The second inequality holds by \Cref{lem:robust-property-1} and \Cref{eq:definition-tstar}, and the last one holds since $\lambda$ is a large constant.
So,
\begin{equation}\label{eq:d-v-ptstar}
    \dist(v,x_{t^*}) \leq \lambda^{3t^*-3}.
\end{equation}
For every $y \in C_v - B_{t^*}$, we have
$ \dist(y,v) \geq \dist(x_{t^*},y) - \dist(v,x_{t^*}) \geq \lambda^{3t^*} - \dist(v,x_{t^*}) \geq \dist(v, x_{t^*}) $.
The second inequality follows by $y \in C_v-B_{t^*} \subseteq C_v - \ball(x_{t^*}, \lambda^{3t^*})$, and the last inequality follows by \Cref{eq:d-v-ptstar} and $\lambda$ being a large constant.
So, 
$ \dist(y,v) \geq (\dist(y,v) + \dist(v,x_{t^*}))/2 \geq \dist(y,x_{t^*})/2 $,
which implies
\begin{equation}\label{eq:cost-v-outside}
    \cost(C_v - B_{t^*},v) \geq \frac{\cost(C_v - B_{t^*},x_{t^*})}{4}.
\end{equation}
We also have
\begin{equation}\label{eq:avregcost-geq-2d}
   \cost(B_{t^*}, x_{t^*}) /w(B_{t^*}) \geq \lambda^{6t^*-4}  \geq 4 \cdot \dist(v,x_{t^*})^2, 
\end{equation}
by the assumption of this case, \Cref{eq:d-v-ptstar}, and $\lambda$ begin a large constant.
Now, we conclude 
\begin{eqnarray*}
    \cost(C_v,v) &=& \cost(B_{t^*},v) + \cost(C_v-B_{t^*},v) \\
    &=& w(B_{t^*}) \cdot \cost(B_{t^*},v)/w(B_{t^*}) + \cost(C_v-B_{t^*},v) \\
    &\geq& w(B_{t^*}) \cdot (\cost(B_{t^*},x_{t^*})/(2w(B_{t^*})) - \dist(v,x_{t^*})^2) + \cost(C_v-B_{t^*},v) \\
    &\geq& w(B_{t^*}) \cdot \frac{\cost(B_{t^*},x_{t^*})}{4\cdot w(B_{t^*})} + \cost(C_v-B_{t^*},v) \\
    &=& \frac{\cost(B_{t^*},x_{t^*})}{4} + \cost(C_v-B_{t^*},v) \\
    &\geq& \frac{\cost(B_{t^*},x_{t^*})}{4} + \frac{\cost(C_v - B_{t^*},x_{t^*})}{4}
    = \frac{\cost(C_v,x_{t^*})}{4}.
\end{eqnarray*}
The first equality holds by $B_{t^*} \subseteq C_v$ (\Cref{claim:Bi-subseteq-Cv}).
The first inequality holds by $\dist(y,v)^2 \geq \frac{1}{2} \cdot \dist(y,x_{t^*})^2 - \dist(v, x_{t^*})^2$ for all $y \in B_{t^*}$, the second inequality holds by \Cref{eq:avregcost-geq-2d} and the last inequality holds by $B_{t^*} \subseteq C_v$ (\Cref{claim:Bi-subseteq-Cv}) and \Cref{eq:cost-v-outside}.
Finally, since $B_{t^*} \subseteq C_v$ (by \Cref{claim:Bi-subseteq-Cv}), we can apply \Cref{lem:robust-property-2} to get
$$ \cost(C_v,u) = \cost(C_v,x_0) \leq 4 \cdot \cost(C_v,x_{t^*}) \leq 16 \cdot \cost(C_v,v) \leq 4\cdot(\polyup)^3 \cdot \cost(C_v,v). $$

\paragraph{Case 2.}
\Cref{cond:robust2} holds for $i= t^{*}$. So, $$\cost(B_{t^*},x_{t^*})/w(B_{t^*}) \leq \lambda^{6t^*-2} \text{ and } 
\cost(B_{t^*}, x_{t^*-1}) \leq (\polyup)^3 \cdot \OPT_1(B_{t^*}). $$
In this case, for every $y \in C_v - B_{t^*}$,
$$ \dist(y,x_{t^*-1}) \geq \dist(y,x_{t^*}) - \dist(x_{t^*},x_{t^*-1}) \geq \lambda^{3t^*} - \lambda^{3t^*-1} \geq \lambda^{3t^*-1}. $$
The second inequality holds by $q \in C_v - B_{t^*} \subseteq C_v - \ball(x_{t^*}, \lambda^{3t^*})$ and \Cref{lem:robust-property-1}, and the last one holds by $\lambda$ being a large constant.
So,
\begin{equation}\label{eq:dqpt-geq-10t}
    \dist(y,x_{t^*-1}) \geq \lambda^{3t^*-1}.
\end{equation}
We also have
\begin{eqnarray*}
    \dist(v,x_{t^*-1}) &\leq& \dist(v,u) + \dist(u,x_{t^*-1}) \\
    &=& \dist(v,u) + \dist(x_0,x_{t^*-1})
    \leq \lambda^{3t^*-10} + 4 \cdot \lambda^{3(t^*-1)-1}
    \leq \lambda^{3t^*-3}.
\end{eqnarray*}
The second inequality holds by \Cref{eq:definition-tstar} and \Cref{lem:robust-property-1}, and the last one holds by $\lambda$ being a large constant.
So,
\begin{equation}\label{eq:dvpt-leq-10t}
    \dist(v,x_{t^*-1}) \leq \lambda^{3t^*-3}.
\end{equation}
Combining \Cref{eq:dqpt-geq-10t}, \Cref{eq:dvpt-leq-10t} and $\lambda$ begin a large constant, we have
$$ \dist(v,x_{t^*-1}) \leq \lambda^{3t^*-3} \leq \frac{1}{2} \cdot \lambda^{3t^*-1} \leq \frac{1}{2} \cdot \dist(y,x_{t^*-1}), $$
which concludes
$$  \frac{\dist(y,x_{t^*-1}) }{\dist(y,v)} \leq \frac{\dist(y,x_{t^*-1}) }{\dist(y,x_{t^*-1}) - \dist(x_{t^*-1},v)} \leq \frac{\dist(y,x_{t^*-1}) }{\dist(y,x_{t^*-1}) - \frac{1}{2}\cdot \dist(y,x_{t^*-1})} = 2. $$
So, we have
$\dist(y,x_{t^*-1}) \leq 2 \cdot \dist(y,v) $.
Using this inequality for all  $y \in C_v - B_{t^*}$, we conclude
\begin{equation}\label{eq:cost-ptstarminusone-compare-to-v}
    \cost(C_v-B_{t^*},x_{t^*-1}) \leq 4 \cdot \cost(C_v-B_{t^*},v).
\end{equation}
We also have
\begin{equation}\label{eq:cost-tstar-1-is-les-than-2-times-cost-of-v}
   \cost(B_{t^*},x_{t^*-1})
    \leq (\polyup)^3 \cdot \OPT_{1}(B_{t^*})  
    \leq (\polyup)^3 \cdot \cost(B_{t^*},v). 
\end{equation}
The first inequality holds by the assumption of this case.
Hence,
\begin{eqnarray*}
    \cost(C_v,x_{t^*-1}) &=& \cost(B_{t^*},x_{t^*-1}) + \cost(C_v - B_{t^*},x_{t^*-1}) \\
    &\leq& (\polyup)^3 \cdot \cost(B_{t^*},v) + 4 \cdot  \cost(C_v - B_{t^*},v)
    \leq (\polyup)^3  \cdot \cost(C_v,v).
\end{eqnarray*}
The equality holds since $B_{t^*} \subseteq C_v$ (by \Cref{claim:Bi-subseteq-Cv}) and 
the first inequality holds by \Cref{eq:cost-ptstarminusone-compare-to-v} and \Cref{eq:cost-tstar-1-is-les-than-2-times-cost-of-v}.
Finally, applying \Cref{lem:robust-property-2} implies
$$ \cost(C_v,u) = \cost(C_v,x_0) \leq 4 \cdot \cost(C_v,x_{t^*-1}) \leq 4 \cdot (\polyup)^3 \cdot \cost(C_v, v) . $$
In both cases, we showed that  $\cost(C_v,u) \leq 4 \cdot (\polyup)^3 \cdot  \cost(C_v,v)$, which completes the proof.

\subsection{\texorpdfstring{Proof of \Cref{lem:num-well-sep}}{}}

Consider the standard LP relaxation for the weighted $k$-means problem.
We consider the set of potential centers to open as $S$, and the goal is to open at most $k - (m- r) / 4 $ many centers, i.e., find $\tilde{S} \subseteq S$.
\begin{align*}
    \min & \sum_{p \in X} \sum_{u \in S}  w(p) \cdot \dist(u,p)^2 \cdot x_{up} & \\
    \text{s.t.}  & \quad x_{up} \leq y_{u}  & \forall u \in S, p \in X 
    \\
    &\sum_{u \in S} x_{up} \geq 1  & \forall p \in X \\
    &\sum_{u\in S} y_u \leq k - (m- r) / 4 & \\
    &x_{up}, y_u \geq 0 & \forall u\in S, p \in X
\end{align*}
Now, we explain how to construct a fractional solution for this LP.

\medskip
\noindent
\textbf{Fractional Opening of Centers.}
Consider the projection $\pi_S: V \rightarrow S$ function.
Assume $S = S_I + S_F$ is a partition of $S$ where $S_I$ contains those centers $u \in S$ satisfying at least one of the following conditions:
\begin{itemize}
    \item $u$ forms a well-separated pair with one center in $V$.
    \item $|\pi^{-1}_S(u)| \geq 2$.
\end{itemize}
For every $u\in S_I$, set $y_u=1$ and for every $u\in S_F$ set $y_u = 1/2$. 
First, we show that
$$\sum_{u\in S} y_u \leq  k - (m- r) / 4.$$
 
Each center $u \in S$ that forms a well-separated pair with a center $v\in V$ has $|\pi^{-1}_S(u)| \geq 1$ since $u$ must be the closest center to $v$ in $S$.
Since the number of well-separated pairs is $k-m$, we have
\begin{equation*}
    k+r \geq |V| = \sum_{u\in S} |\pi^{-1}(u)| \geq \sum_{u\in S_I} | \pi^{-1}(u)| \geq k- m   + 2\cdot\left(|S_I| - (k-m)\right).
\end{equation*}
Hence,
$$|S_I| \leq \frac{k+r + (k-m)}{2} = k - \frac{m-r}{2}. $$
Finally, we conclude
\begin{eqnarray*}
   \sum_{u\in S} y_u &=& \sum_{u\in S_I} y_u + \sum_{u\in S_F} y_u \\
   &\leq&  k - \frac{m-r}{2} + \frac{1}{2} \cdot \frac{m-r}{2} \\
   &=&
   k - \frac{m-r}{4}.
\end{eqnarray*}

\medskip
\noindent
\textbf{Fractional Assignment of Points.}     
For every $p\in X$, assume $ v_p = \pi_{V}(p) $ is the closest center to $p$ in $V$ and $u_p = \pi_S(v_p)$ is the closest center in $S$ to $v_p$.
We have three cases:
\begin{itemize}

\item If $y_{u_p} = 1$, then set $x_{u_pp} = 1$. The cost of this assignment would be $w(p) \cdot \dist(u_p, p)^2$.

\item If $y_{u_p} = 1/2$ and there is a center $u'_p\in S  - u_p$ such that $\dist(u_p, u'_p) \leq \lambda^{20}\cdot  \dist(u_p, v_p)$, then set $x_{u_pp} = x_{u_p' p} = 1/2$. Note that $u_p' \neq u_p$ which means point $p$ is assigned one unit to centers.
The cost of this assignment would be
\begin{eqnarray*}
    & & \frac{1}{2} w(p)\left(\dist(p,u_p)^2 + \dist(p,u_p')^2\right) \\ 
    &\leq&  \frac{1}{2} w(p)\left(\dist(p,u_p)^2 + 2 \cdot \dist(p,u_p)^2 + 2\cdot \dist(u'_p,u_p)^2\right)  \\
    &\leq& w(p)\left(\frac{3}{2}\cdot \dist(p,u_p)^2 + \lambda^{40}\cdot \dist(u_p,v_p)^2\right). 
\end{eqnarray*}
In the first inequality, we have used the fact that $\dist(x,y)^2 \leq 2 \cdot \dist(x,z)^2 + 2\cdot \dist(z,y)^2$ for all $x,y,z \in \R^d $.

\item If $y_{u_p} = 1/2$ and the previous case does not hold, then since $(u_p, v_p)$ is not a well-separated pair, there is a center $v'_p \in V - v_p$ such that $\dist(v_p, v'_p) \leq \lambda^{20} \cdot \dist(u_p, v_p)$. Let $u_p' = \pi_S(v'_p)$ and  set $x_{u_p p} = x_{u'_p p} = 1/2$. First, we show that $u'_p \neq u_p$. Since $y_{u_p} = 1/2$, we have $u_p \in S_F$ which concludes $|\pi^{-1}_S(u_p)| \leq 1$.
We also know that $\pi_S(v_p) = u_p$.
So, $v_p$ is the only center in $V$ mapped to $u_p$ which implies
$\pi_S(v'_p) \neq u_p$ or $u_p' \neq u_p$ (note that $v_p' \neq v_p$). So, point $p$ is assigned one unit to centers.
The cost of this assignment would be
\begin{eqnarray*}
    & & \frac{1}{2} w(p) \left(\dist(p, u_p)^2 + \dist(p, u'_p)^2\right) 
    \\ &\leq&
    \frac{1}{2} w(p) \left( \dist(p, u_p)^2 + (\dist(p, u_p) + \dist(u_p, v'_p)  + \dist(v'_p, u'_p))^2\right) \\
    &\leq&
    \frac{1}{2} w(p) \left( \dist(p, u_p)^2 + (\dist(p, u_p) + \dist(u_p, v'_p)  + \dist(v'_p, u_p))^2\right) \\
    &\leq&
    \frac{1}{2} w(p) \left( \dist(p, u_p)^2 + 2 \cdot \dist(p, u_p)^2 + 8 \cdot \dist(u_p, v'_p)^2 \right) \\
     &=&
    w(p) \left( \frac{3}{2} \cdot \dist(p, u_p)^2  + 4 \cdot \dist(u_p,v_p')^2\right) \\
    &\leq&
    w(p) \left( \frac{3}{2}\cdot \dist(p, u_p)^2  + 4\cdot (\dist(u_p, v_p) + \dist(v_p,v_p'))^2 \right) \\
    &\leq& w(p)\left(  \frac{3}{2} \cdot \dist(p,u_p)^2 + 4(\lambda^{20}+1)^2\ \dist(u_p,v_p)^2\right).
\end{eqnarray*}
The second inequality holds because of the choice of $u'_p = \pi_{S}(v_p')$ and the last inequality holds since $\dist(v_p,v_p') \leq \lambda^{20} \cdot \dist(u_p,v_p)$.
\end{itemize}

\medskip
\noindent
\textbf{Bounding the Cost.}
Assume $u_p^* = \pi_S(p)$ for each $p \in X$.
We have
$$ \dist(u_p,v_p) \leq \dist(u_p^*,v_p) \leq \dist(u_p^*,p) + \dist(p,v_p). $$
The first inequality is by the choice of $u_p = \pi_S(v_p)$.
As a result,
\begin{equation}\label{eq:bound-on-d-up-vp}
    \sum_{p \in X} w(p)\ \dist(u_p,v_p)^2
    \leq
    \sum_{p \in X} w(p)\left( 2 \cdot \dist(u^*_p,p)^2 + 2 \cdot \dist(p,v_p)^2 \right)
    =
    2 \cdot \cost(X, S) + 2 \cdot \cost(X, V).
\end{equation}
In each of the cases of fractional assignments of points to centers, the cost of assigning a point $p \in C_v(V, X)$ is at most
$ w(p) \left(\frac{3}{2} \cdot \dist(p, u_p)^2 + 4(\lambda^{20}+1)^2 \cdot \dist(u_p, v_p)^2 \right)$.
As a result, the total cost of this assignment is upper bounded by
\begin{eqnarray*}
    & & \sum_{p\in P} w(p) \left(\frac{3}{2}\dist(p,u_p)^2 + 4(\lambda^{20}+1)^2 \cdot \dist(u_p, v_p)^2\right) 
    \\ &\leq& 
     \sum_{p\in X} w(p) \left(\frac{3}{2} \cdot (\dist(p,v_p) +\dist(u_p,v_p))^2  + 4(\lambda^{20}+1)^2 \ \dist(u_p, v_p)^2\right) \\
     &\leq&
     \sum_{p\in X} w(p) \left(3 \cdot \dist(p,v_p)^2  + 4(\lambda^{20}+2)^2 \ \dist(u_p, v_p)^2\right) \\
    &\leq& 
    3 \cdot \cost(X,V) +8(\lambda^{20}+2)^2 \left( \cost(X,S) + \cost(X,V) \right) \\
    &\leq& 
    9\lambda^{40} \left( \cost(X,S) + \cost(X,V) \right).
\end{eqnarray*}
The third inequality follows by \Cref{eq:bound-on-d-up-vp}.

Finally, since the integrality gap of the LP relaxation is known to be at most $9$ \cite{integrality-gap-k-means}.\footnote{Note that although it is not explicitly mentioned, the result of \cite{integrality-gap-k-means} works for the weighted case.
One way to show this, which only uses the unweighted case as a black box is as follows.
First, we can round all the real weights to some rational weights with an additive error of at most $\epsilon$.
This will not change the feasible solutions and change the optimal objective function only by a factor of $1\pm \epsilon$. 
Then, we scale up all the weights to have integer weights. Now, we replace every point $p$ of weight $w(p)$ with $w(p)$ many unweighted points $p_1,p_2,\cdots,p_{w(p)}$.
It is easy to find a correspondence between any optimal solution (either fractional or integral) of the original weighted LP and any optimal solution of the new unweighted LP.
Hence, we conclude that the integrality gap of the weighted LP is also at most $9 + O(\epsilon)$. Since this is correct for every  $\epsilon > 0$, we conclude that the integrality gap is at most $9$.}
As a result, there exist an integral solution whose cost is at most $81\lambda^{40}\left( \cost(X,S) + \cost(X,V) \right)$, and this solution opens at most 
$ k - \lfloor \frac{m-r}{4} \rfloor $ centers which completes the proof.

\subsection{\texorpdfstring{Proof of \Cref{lem:projection-lemma}}{}}

Assume $S^\star$ is the optimal $k$-means solution on $X$.
Define $S' := \{\pi_{C}(s) \mid s \in S^\star \}$ be the projection of $S^\star$ on $C$.
Let $x \in X$ be arbitrary.
Let $s^\star = \pi_{S^\star}(x)$, $s' = \pi_{C}(s^\star) \in S'$, and $s = \pi_C(x)$.
We have
\begin{align*}
\dist(x,S')^2 &\leq \dist(x,s')^2 \leq (\dist(x,s^\star) + \dist(s',s^\star))^2 \leq (\dist(x,s^\star) + \dist(s,s^\star))^2 \\
&\leq (2 \cdot \dist(x,s^\star) + \dist(s,x))^2 \leq 
8\cdot \dist(x,s^\star)^2 + 2 \cdot \dist(s,x)^2.
\end{align*}
As a result, $\OPT_{k}^C(X) \leq \cost(X, S') \leq 2 \cdot \cost(X, C) + 8 \cdot \OPT_k(X)$.

\subsection{\texorpdfstring{Proof of \Cref{lem:lazy-updates}}{}}

Assume $S \subseteq \Deld$ of size at most $k$ is such that $\cost(X, S) = \OPT_{k}(X)$.
Define
$S' = S + (X \oplus X')$.
Obviously, $S'$ is a feasible solution for the $(k+s)$-means problem on $X'$ and since $S'$ contains $X \oplus X'$, we conclude
$$ \OPT_{k+s}(X') \leq  \cost(X',S') = \cost(X', S + (X \oplus X')) \leq \cost(X,S) = \OPT_{k}(X). $$

\subsection{\texorpdfstring{Proof of \Cref{lem:double-sided-stability}}{}}

Consider the LP relaxation for the $k$-means problem on $X$ for each $k$ as follows.
\begin{align*}
    \min & \sum_{p \in X} \sum_{c \in \Deld}  w(p) \cdot \dist(c,p)^2 \cdot x_{cp} & \\
    \text{s.t.}  & \quad x_{cp} \leq y_{c}  & \forall c \in \Deld, p \in X 
    \\
    &\sum_{c \in \Deld} x_{cp} \geq 1  & \forall p \in X \label{eq:all-open0}\\
    &\sum_{c \in \Deld} y_c \leq k & \\
    &x_{cp}, y_c \geq 0 & \forall c \in \Deld, p \in X
\end{align*}

\newcommand{\FOPT}[0]{\text{FOPT}}

Denote the cost of the optimal fractional solution for this LP by $\FOPT_k$. 
Since the space $X$ is fixed here, we denote $\OPT_k(X)$ by $\OPT_k$.
It is known that the integrality gap of this relaxation is at most $9$ \cite{integrality-gap-k-means}.
So, for every $k$ we have
\begin{equation}\label{eq:int-gap}
    \FOPT_k \leq \OPT_k \leq 9\ \FOPT_k.   
\end{equation}

\begin{claim}
    For every $k_1$, $k_2$ and $0 \leq \alpha, \beta \leq 1$ such that $\alpha + \beta = 1$, we have
    $$\FOPT_{\alpha k_1 + \beta k_2} \leq \alpha \ \FOPT_{k_1} + \beta \ \FOPT_{k_2}. $$
\end{claim}

\begin{proof}
    Assume optimal fractional solutions $(x^*,y^*)$ and $(z^*,t^*)$ for above LP relaxation of $k_1$ and $k_2$-means problems respectively. It is easy to verify that $(\alpha x^* + \beta z^*, \alpha y^* + \beta t^*)$ is a feasible solution for fractional $(\alpha k_1 + \beta k_2)$-means problem, whose cost is $\alpha \ \FOPT_{k_1} + \beta \ \FOPT_{k_2}$, which concludes the claim.
\end{proof}

Now, plug $k_1 = k-r$, $k_2 = k + r/(36\eta)$, $\alpha = 1/(36\eta)$ and $\beta = 1 - \alpha$ in the claim. We have
$$\alpha k_1 + \beta k_2 = \frac{1}{36\eta}(k-r) + \left(1-\frac{1}{36\eta}\right)\left(k+\frac{r}{36\eta}\right) = k - \frac{r}{(36\eta)^2} \leq k. $$
As a result,
$$ \FOPT_k \leq \FOPT_{\alpha k_1 + \beta k_2} \leq \alpha \ \FOPT_{k_1} + \beta \ \FOPT_{k_2}. $$
Together with \Cref{eq:int-gap}, we have
$$ \OPT_k \leq 9\alpha\ \OPT_{k_1} + 9\beta \ \OPT_{k_2}. $$
We also have the assumption that $ \OPT_{k_1} = \OPT_{k-r} \leq \eta \ \OPT_k$, which implies
$$ \OPT_k \leq 9\alpha\eta\ \OPT_{k} + 9\beta\ \OPT_{k_2}. $$
Finally
$$ \OPT_k \leq \frac{9\beta}{1-9\alpha\eta} \cdot \OPT_{k_2} \leq 12 \ \OPT_{k_2} \leq 12\  \OPT_{k +\lfloor r/(36\eta) \rfloor}. $$
The second inequality follows from 
$$ \frac{9\beta}{1 - 9\alpha\eta} \leq \frac{9}{1 - 9\alpha\eta} = \frac{9}{1 - 1/4} = 12. $$

\subsection{\texorpdfstring{Proof of \Cref{lem:robustify-calls-once}}{}}
\label{sec:proof-of-lem-robustify-calls-once}

Assume a $\MakeRbst$ call to a center $w \in W'$ is made and it is replaced by $w_0$.
Let $W^\old$ be the set of centers just before this call to $\MakeRbst$ is made on $w'$.
At this point of time, $t[w_0]$ is the smallest integer satisfying $\lambda^{3t[w_0]} \geq \hat{\dist}(w, W^\old-w) / \lambda^7$.
For the sake of contradiction, assume that another call to $\MakeRbst$  is made on $w_0$.
We assume that this pair $w$, $w_0$ is the first pair for which this occurs.
Let $w' \in W^\old - w$ be such that $\dist(w,\calW^\old - w)=\dist(w,w')$.
Then,
\begin{equation}\label{eq:lambda-to-w0-lower}
    \lambda^{3t[w_0]} \geq \hat{\dist}(w, W^\old-w) / \lambda^{7} \geq \dist(w,w') / \lambda^{7},
\end{equation}
and
\begin{equation}\label{eq:lambda-to-w0-upper}
    \lambda^{3t[w_0]} < \hat{\dist}(w, W^\old-w) / \lambda^{4} \leq \dist(w,w') \cdot \polyup/\lambda^{4},
\end{equation}
by definition of $t[w_0]$ and the guarantee about $\hat{\dist}(w, W^\old-w)$ in \Cref{lem:nearest-neighbor-distance}.
So, by \Cref{lem:robust-property-1}, we have
\begin{equation}\label{eq:dw0w-leq-dwwprime-polyup-lambda5}
   \dist(w_0,w) \leq 4 \cdot \lambda^{3t[w_0]-1} < \dist(w,w')\cdot 4\polyup/\lambda^5. 
\end{equation}
Assume $W^\new$ is the set of centers just after the call to \MakeRbst \ is made on $w_0$.
Since we assumed $w,w_0$ is the first pair that another call to \MakeRbst \ on $w_0$ is made, we have one of the two following cases.
\paragraph{Case 1.} $w' \in W^\new$.
In this case, we have
\begin{align}\label{eq:w0-w-wprime}
    \dist(w_0,W^\new - w_0) &\leq \dist(w_0,w') \leq \dist(w_0,w) + \dist(w,w') \leq (4\polyup/\lambda^5 + 1)\cdot \dist(w,w') \nonumber \\
    &< 2 \cdot \dist(w,w')    
\end{align}
The first inequality holds since $w'$ is present in $W^\new$, the second inequality is triangle inequality, the third inequality holds by \Cref{eq:dw0w-leq-dwwprime-polyup-lambda5} ,and the last inequality holds since $\lambda=(\polyup)^2$ is a large number.
According to \Cref{if:condition-W3} in \Cref{alg:modify:robustify}, we must have $t[w_0] < t$, where $t$ is the smallest integer satisfying $\lambda^{3t} \geq \hat{\dist}(w_0, W^\new-w_0)/\lambda^{10} $, as otherwise we should not made a call to $\MakeRbst$ on $w_0$.
Hence,
\begin{align*}
  \lambda^{3t[w_0]} &< \hat{\dist}(w_0, W^\new-w_0)/\lambda^{10} \leq \dist(w_0, W^\new - w_0) \cdot \polyup/\lambda^{10} \\ 
  &\leq \dist(w,w') \cdot 2 \cdot \polyup/\lambda^{10} < \dist(w,w') /\lambda^7,  
\end{align*}
where the second inequality holds by the guarantee of $\hat{\dist}(w_0,W^\new-w_0)$ in \Cref{lem:nearest-neighbor-distance}, the third inequality holds by \Cref{eq:w0-w-wprime}, and the last inequality holds since $\lambda=(\polyup)^2$ is a large number.
This is in contradiction with \Cref{eq:lambda-to-w0-lower}.

\paragraph{Case 2.} There is $ w'_0 \in W^\new$ that replaced $w'$ by a single call to $\MakeRbst$.
When $\MakeRbst$ is called on $w'$, the algorithm selects the smallest integer $t$ satisfying
$ \lambda^{3t} \geq \hat{\dist}(w', W^{\text{mid}}-w') / \lambda^7 $
and sets $t[w_0'] \gets t$,
where $W^\text{mid}$ is the status of the center set when a call to $\MakeRbst(w')$ is made.
At that point in time, $w_0$ was present in $W^\text{mid}$, which concludes,
\begin{align}\label{eq:dis-wprime-w0prime}
    \dist(w',w_0') &\leq
    4 \cdot \lambda^{3t[w_0']-1}< \hat{\dist}(w',W^\text{mid}-w') \cdot 4/\lambda^5 \leq 
    \dist(w',W^\text{mid} - w') \cdot 4\polyup/\lambda^5  \nonumber \\
    &\leq \dist(w',w_0) \cdot 4\polyup/\lambda^5 \leq (\dist(w',w) + \dist(w,w_0)) \cdot 4\polyup/\lambda^5 \nonumber \\
    &\leq 
    \dist(w,w') \cdot (1 + 4\polyup/\lambda^5) \cdot 4\polyup/\lambda^5 
\end{align}
The first inequality holds by \Cref{lem:robust-property-1}, the second inequality holds by the definition of $t[w_0']$ at the time of call to $\MakeRbst(w')$, the third inequality holds by the guarantee of $\hat{\dist}(w',W^\text{mid}-w')$ in \Cref{lem:nearest-neighbor-distance}, the fourth inequality holds since $w_0$ was present in $W^\text{mid}$ in the time of call to $\MakeRbst(w')$, the fifth inequality is triangle inequality and the last inequality holds by \Cref{eq:dw0w-leq-dwwprime-polyup-lambda5}.
As a result,
\begin{align*}
    \dist(w_0,W^\new - w_0) &\leq \dist(w_0,w_0') \leq \dist(w_0,w) + \dist(w,w') + \dist(w',w_0') \\
    &\leq \left( (4\polyup/\lambda^5) + (1) +
    (1+ 4\polyup/\lambda^5) \cdot 4\polyup/\lambda^5 \right) \cdot \dist(w,w') < 2 \cdot \dist(w,w').
\end{align*}
The first inequality holds since $w'_0$ is present in $W^\new$, the second inequality is triangle inequality, the third inequality follows by \Cref{eq:lambda-to-w0-upper} and \Cref{eq:dis-wprime-w0prime}, and the last inequality holds since $\lambda=(\polyup)^2$ is a large number.
So, we derived the same inequality as \Cref{eq:w0-w-wprime}, and the contradiction follows exactly like the previous case.

\subsection{\texorpdfstring{ Proof of \Cref{claim:remains-robust}}{}}
\label{sec:proof-of-claim-remains-robust}

Since $u$ is $t[u]$-robust at the beginning of the epoch (w.r.t.~$X^{(0)}$), according to \Cref{def:robust}, we conclude that there exists a sequence of points $(x_0,x_1,\ldots,x_{t[u]})$ together with a sequence of sets $(B_0,B_1, \cdots, B_{t[u]})$
such that $x_0 = u$ and the guarantees in \Cref{def:robust} holds for these sequences w.r.t.~$X^{(0)}$.
Since $u \in W_3$ is not contaminated, we infer that the ball around $u=x_0$ of radius $\lambda^{3t[u]+1.5}$ is not changes.
This shows that non of the sets $B_i$ has changes.
This is because if there exists an $0 \leq i \leq t[u]$ such that $B_i \cap (X^{(0)} \oplus X^{(\ell+1)}) \neq \emptyset$, by taking an arbitrary 
$y \in B_i \cap (X^{(0)} \oplus X^{(\ell+1)} )$, we conclude that
$$ \dist(y,x_0) \leq \dist(y,x_i) + \dist(x_i,x_0) \leq \lambda^{3i+1} + 4 \cdot \lambda^{3i-1} \leq \lambda^{3t[u] + 1.5}, $$
where the second inequality follows by $ y \in B_i \subseteq \ball(x_i, \lambda^{3i+1})$ and \Cref{lem:robust-property-1}, and the third inequality follows since $i \leq t[u]$ and $\lambda$ is a large number.
As a result, 
$y \in \ball(x_0, \lambda^{3t[u] + 1.5}) \cap (X^{(0)} \oplus X^{(\ell+1)})$, which is in contradiction with $u = x_0$ not being contaminated.

Hence, all of the $B_i$s remain unchanged during the epoch and it is easy to observe that the same sequences $(x_0,x_1,\ldots,x_{t[u]})$ and $(B_0,B_1, \cdots, B_{t[u]})$ also have the guarantees in \Cref{def:robust} w.r.t.~the new dataset $X^{(\ell+1)}$, which means $u=x_0$ is $t[u]$-robust w.r.t.~$X^{(\ell+1)}$.

\subsection{\texorpdfstring{Proof of \Cref{claim:num-of-contamination}}{}}
\label{sec:proof-of-claim-num-of-contamination}

For the sake of contradiction assume $u,v \in S_\init$ such that $t[u] = t[v] = i$ and $\dist(x,u) \leq \lambda^{3i+2} $ and $\dist(x,v) \leq \lambda^{3i+2}$.
We conclude
$ \dist(u,v) \leq \dist(x,u) + \dist(x,v) \leq 2 \cdot \lambda^{3i+2} $.
By symmetry between $u$ and $v$ assume that $v$ is added to the center set after $u$ through a call $\MakeRbst(v^\old)$ for some $v^\old \in S^\old$ that is replaced by $v$ in the solution and set $t[v] = t$, where $t$ is the smallest integer satisfying $\lambda^{3t} \geq \hat{\dist}(v^\old, S^\old-v^\old)/\lambda^7 $ at that time, where $\hat{\dist}(v^\old, S^\old-v^\old)$ is a $\polyup$ approximation of the distance of $v^\old$ to the center set.
Hence, $\lambda^{3(t-1)} < \hat{\dist}(v^\old, S^\old-v^\old)/\lambda^7$ by the definition of $t$ and $\hat{\dist}(v^\old, S^\old-v^\old) \leq \polyup \cdot \dist(v^\old,u)$ since $u$ was there in the center set at the time of the call $\MakeRbst(v^\old)$.
As a result,
\begin{align*}
   \lambda^{3(t-1)} &< \hat{\dist}(v^\old, S^\old-v^\old)/\lambda^7 \\
   &\leq (\polyup/ \lambda^7) \cdot \dist(v^\old, u) \\
   &\leq (\polyup/ \lambda^7) \cdot (\dist(v^\old, v) + \dist(v, u)) \\
   &\leq (\polyup/ \lambda^7) \cdot (4 \cdot \lambda^{3t-1} + 2 \cdot \lambda^{3i+2}) \\
   & \leq 4 \cdot \lambda^{3t-7} + 2 \cdot \lambda^{3i-4}.
\end{align*}
The fourth inequality follows by \Cref{lem:robust-property-1} and the fifth inequality follows since $\lambda = (\polyup)^2$ (see \Cref{eq:value-of-polyup-lambda}).
As a result,
$$  \lambda^{3t-3} \leq 2 \cdot (\lambda^{3t-3} - 4 \cdot \lambda^{3t-7}) \leq 4 \cdot \lambda^{3i-4} < \lambda^{3i-3}. $$
Here we used the fact that $\lambda$ is a large number.
Thus, $t < i = t[v]$.
This is in contradiction with the definition of $t[v]$ for $v$ since when we made the call to $\MakeRbst(v^\old)$, we set $t[v] = t$ which indicated that $v$ is $t[v]$-robust.
To explain more, the only time that the value of $t[v]$ can be increased is when we have another call on $\MakeRbst(v)$ that replaces $v$ with $v^\new$ which is $t[v^\new] = t'$-robust for a larger integer $t' > t$.
But, we assumed the last call to $\MakeRbst$ and the value of $t[v] = t$ is set in this call.

\subsection{\texorpdfstring{Proof of \Cref{claim:contaminate-only-one}}{}}
\label{sec:proof-of-claim-contaminte-only-one}

Note that $t[u] = i$ since $ u \in S[i] $.
We consider two cases.
\paragraph{Case 1.} $\dist(x,u) \leq \lambda^{3i + 2}$.
According to \Cref{claim:num-of-contamination}, there does not exists $v \in S_\init - u$ satisfying $t[v] = i$ and $\dist(x,v) \leq \lambda^{3i + 2}$. 
Since $\ball(v, \lambda^{3i+1.5}) \subseteq \ball(v, \lambda^{3i+2})$, we conclude that the distance between $x$ and $v$ is strictly greater than $\lambda^{3i + 1.5}$.
So, $x$ does not contaminate $v$.
Hence, there does not exist any $v \in S[i] - u$ contaminated by $x$ and the only center in $S[i]$ that might be contaminated by $x$ is $u$. 
\paragraph{Case 2.} $\dist(x,u) > \lambda^{3i + 2}$.
In this case, we conclude that
$ \lambda^{3i + 2} < \dist(x,u) \leq \polyup \cdot \dist(x, S[i]) $.
Hence, the distance of any center $v$ in $S[i]$ to $x$ is strictly greater than $ \lambda^{3i + 2}/\polyup = \lambda^{3i + 1.5} $ (since $\lambda = (\polyup)^2 $ by \Cref{eq:value-of-polyup-lambda}), which concludes there is no center $v \in S[i]$ contaminated by $x$.

\bibliographystyle{alpha}
\bibliography{references}

\appendix

\section{\texorpdfstring{Barriers Towards Adapting the Algorithm of~\cite{BhattacharyaCGLP24}}{}}
\label{barriers}

In  \cite{BhattacharyaCGLP24}, the authors give a dynamic algorithm for $k$-means on general metric spaces, with $O(1/\epsilon)$-approximation ratio and $\tilde O(k^{\epsilon})$ update time. With the current state-of-the-art tools for designing geometric algorithms in Euclidean spaces, we cannot avoid a tradeoff between a $\poly(1/\epsilon)$ term in the approximation and a $k^\epsilon$ term in the update time for this problem.\footnote{Such tradeoffs are standard in the geometric setting (e.g.~from using LSH). Bypassing them, even in the static setting, is a major open problem.} Thus, it might appear that building on top of the work of \cite{BhattacharyaCGLP24} might be a simpler and more natural approach than building on \cite{BCF24}, since we will eventually incur such a tradeoff in our final algorithm.
Unfortunately, the algorithm of \cite{BhattacharyaCGLP24} is not as well suited to implementation in high-dimensional Euclidean spaces, and it is not clear how to use this algorithm without avoiding an exponential blowup in the approximation ratio, leading to a $\exp(1/\epsilon)$-approximation.

\paragraph{The Algorithm of \cite{BhattacharyaCGLP24}.} The high level approach taken by \cite{BhattacharyaCGLP24} is to maintain a hierarchy of nested solutions $X \supseteq S_0 \supseteq \dots \supseteq S_{\ell}$, where $S_{\ell}$ has size $k$ and is the output of the algorithm. Using a new randomized variant of local search, they can maintain this hierarchy by periodically reconstructing these solutions and maintaining them lazily between reconstructions, leading to an update time of $\tilde O(k^{1 + 1/\ell})$ and a recourse of $\tilde O(k^{1/\ell})$.
To bound their approximation ratio, they show that their randomized local search almost matches the guarantees provided by the \emph{projection lemma}, which says that there must exist some $S_{i + 1} \subseteq S_i$ such that $\cost(X, S_{i + 1}) \leq \cost(X, S_i) + O(\OPT_k(X))$. Applying this argument $O(\ell)$ times, they show that $\cost(X, S_\ell) \leq O(\ell) \cdot \OPT_k(X)$.

\paragraph{The Barrier to Using \cite{BhattacharyaCGLP24}.} If we want to implement \cite{BhattacharyaCGLP24} efficiently in high-dimensional Euclidean spaces, the most natural approach is to modify some components within the algorithm to use approximate nearest neighbor data structures, which can be implemented efficiently in this setting. However, the hierarchical nature of this algorithm requires us to have very strong guarantees on the quality of the solutions in order for the analysis to go through. For example, a guarantee that $\cost(X, S_{i + 1}) \leq O(\cost(X, S_i) + \OPT_k(X))$ for the solutions in the hierarchy is \emph{not} enough to obtain an approximation ratio of $O(\ell)$, and instead would lead to an approximation ratio of $2^{O(\ell)}$ if analyzed in the naive manner. Unfortunately, if we use approximate nearest neighbor data structures to speed up the randomized local search, it is not clear how we can get the sufficiently strong guarantees required to obtain an approximation ratio of $\poly(\ell)$ with this algorithm.
To summarize, the lack of robustness to the constants in the guarantees makes the algorithm less well-suited to implementation in high-dimensional Euclidean spaces.

\end{document}